\documentclass[11pt,a4paper,reqno]{amsbook}%
\usepackage{amsbsy,amsopn,amscd,amstext,amsgen,amsthm,amsmath,amsfonts,amssymb,amsxtra,appendix,bookmark,dsfont,latexsym,color,epsfig,xspace,mathrsfs,dsfont,srcltx}

\makeatletter
\usepackage{hyperref}
\usepackage{delarray,a4,color}
\usepackage{euscript}
\usepackage[latin1]{inputenc}

\numberwithin{equation}{section}
\numberwithin{section}{chapter}

\newtheorem{theorem}{Theorem}[section]
\newtheorem{mtheorem}[theorem]{Statement}
\newtheorem{lemma}[theorem]{Lemma}
\newtheorem{corollary}[theorem]{Corollary}

\newtheorem{assumption}[theorem]{Assumption}
\theoremstyle{definition}
\newtheorem{definition}[theorem]{Definition}

\theoremstyle{remark}
\newtheorem{remark}[theorem]{Remark}

\numberwithin{equation}{section}

\usepackage{color}

\numberwithin{equation}{section}

\newcommand{\norm}[1]{\left\lVert #1 \right\rVert}

\newcommand{\beq}{\begin{equation}}
\newcommand{\eeq}{\end{equation}}

\newcommand{\lf}{\left}
\newcommand{\ri}{\right}
\newcommand{\half}{\frac{1}{2}}

\newcommand{\1}{{\ensuremath {\mathds 1} }}
\newcommand{\ii}{\infty}
\newcommand{\Z}{\mathbb{Z}}
\newcommand{\R}{\mathbb{R}}
\newcommand{\N}{\mathbb{N}}

\newcommand{\C}{\mathbb{C}}

\newcommand{\F}{\mathcal{F}}
\newcommand{\E}{\mathcal{E}}
\newcommand{\cE}{\mathcal{E}}

\newcommand{\cF}{\mathcal{F}}
\newcommand{\cL}{\mathscr{L}}

\newcommand{\B}{\mathcal{B}}

\newcommand{\LL}{\mathcal{L}}
\newcommand{\Q}{\mathcal{Q}}

\newcommand{\cP}{\mathcal{P}}
\newcommand{\cB}{\mathcal{B}}
\newcommand{\cM}{\mathcal{M}}
\newcommand{\OO}{O}

\newcommand{\gS}{\mathfrak{S}}
\newcommand{\gH}{\mathfrak{H}}

\newcommand{\ep}{\varepsilon}
\newcommand{\g}{\gamma}
\newcommand{\Om}{\Omega}
\newcommand{\om}{\omega}

\newcommand{\dd}{\partial}

\newcommand{\supp}{\mathrm{supp}}
\newcommand\pscal[1]{{\ensuremath{\left\langle #1 \right\rangle}}}
\newcommand{\wto}{\rightharpoonup}
\renewcommand{\epsilon}{\varepsilon}
\def\Xint#1{\mathchoice
   {\XXint\displaystyle\textstyle{#1}}%
   {\XXint\textstyle\scriptstyle{#1}}%
   {\XXint\scriptstyle\scriptscriptstyle{#1}}%
   {\XXint\scriptscriptstyle\scriptscriptstyle{#1}}%
   \!\int}
\def\XXint#1#2#3{{\setbox0=\hbox{$#1{#2#3}{\int}$}
     \vcenter{\hbox{$#2#3$}}\kern-.5\wd0}}

\def\dashint{\Xint-}

\DeclareMathOperator{\Tr}{{\rm Tr}}
\DeclareMathOperator{\Tro}{{\rm Tr}_0}
\newcommand{\tr}{\Tr}
\newcommand{\bral}{\left<}

\newcommand{\ketr}{\right>}

\newcommand{\dem}{\varepsilon_{\rm M}}

\newcommand{\ptgf}{\mathcal{E}^{\rm P}_{\dem}}
\newcommand{\ptge}{E^{\rm P}_{\dem}}

\newcommand{\ptgint}{F^{\rm P}_{\dem}}

\newcommand{\FSm}{\gamma_{\rm per} ^0}

\newcommand{\FSh}{H_{\rm per} ^0}
\newcommand{\FSl}{\epsilon_{\rm F}}
\newcommand{\FSd}{\rho_{\rm per} ^0}

\newcommand{\FSp}{V_{\rm per} ^0}

\newcommand{\Wrho}{W_{\rho}}

\newcommand{\rhoQ}{\rho_Q}

\newcommand{\rhoP}{\rho_{\Psi}}

\newcommand{\tote}{E_m}

\newcommand{\crysf}{\mathcal{F}_{\rm crys} }
\newcommand{\cryse}{F_{\rm crys} }

\newcommand{\Eperf}{\mathcal{E}^{\rm per}_m}
\newcommand{\Epere}{E ^{\rm per}_m}
\newcommand{\Eperm}{u_m^{\rm per}}

\newcommand{\Eperelim}{E ^{\rm per}}

\newcommand{\Qpp}{Q ^{++}}
\newcommand{\Qpm}{Q ^{+-}}
\newcommand{\Qmp}{Q^{-+}}
\newcommand{\Qmm}{Q^{--}}

\newcommand{\one}{{\ensuremath {\mathds {1}} }}
\newcommand{\oneep}{\one_{\left(-\infty, \FSl \right)}}

\newcommand{\Ker}{\mathcal L_\ell^N}

\newcommand{\coulker}{|\: . \:| ^{-1}}
\newcommand{\rhogam}{\rho_{\gamma}}

\newcommand{\munucper}{\mu ^{0} _{\rm per}}

\newcommand{\rhff}{\E ^{\rm rHF}}

\newcommand{\Sch}{\mathfrak{S}}

\newcommand{\dist}{\mathrm{dist}}

\newcommand{\rv}{\mathbf{r}}
\newcommand{\eps}{\varepsilon}
\newcommand{\nablap}{\nabla^{\perp}}

\newcommand{\curl}{\mbox{curl}}

\newcommand{\Ofirst}{\Omega_{\mathrm{c_1}}}

\newcommand{\gpf}{\mathcal{E}^{\mathrm{GP}}}
\newcommand{\gpe}{E^{\mathrm{GP}}}
\newcommand{\gpm}{\Psi^{\mathrm{GP}}}

\newcommand{\tff}{\mathcal{E}^{\mathrm{TF}}}
\newcommand{\tfd}{\mathcal{D}^{\mathrm{TF}}}
\newcommand{\tfe}{E^{\mathrm{TF}}}
\newcommand{\tfm}{{\rho^{\mathrm{TF}}}}
\newcommand{\rtf}{{R^{\mathrm{TF}}}}

\newcommand{\rbulk}{R_{\rm bulk}}

\newcommand{\tfchem}{\lambda^{\mathrm{TF}}}

\newcommand{\tfr}{R ^{\rm TF}}

\newcommand{\tfI}{I ^{\rm TF}}
\newcommand{\tfIc}{\mathcal{I} ^{\rm TF}}
\newcommand{\tfpot}{F ^{\rm TF}}
\newcommand{\tfH}{H ^{\rm TF}}
\newcommand{\musta}{\mu_{\ast}}

\newcommand{\hgpf}{\hat{\mathcal{E}}^{\mathrm{GP}}}
\newcommand{\hgpe}{\hat{E}^{\mathrm{GP}}}

\newcommand{\game}{\gamma_{\eps}}

\newcommand{\Hex}{h_{\mathrm{ex}}}
\newcommand{\Hc}{H_{\mathrm{c}1}}
\newcommand{\Hcc}{H_{\mathrm{c}2}}
\newcommand{\Hccc}{H_{\mathrm{c}3}}
\newcommand{\Hcccc}{H_{\mathrm{c}4}}

\newcommand{\glm}{\Psi^{\mathrm{GL}}}

\newcommand{\aav}{\mathbf{A}}
\newcommand{\aavm}{\mathbf{A}^{\mathrm{GL}}}
\newcommand{\hex}{b}
\newcommand{\theo}{\Theta_0}

\newcommand{\glfe}{\mathcal{G}_{\eps}^{\mathrm{GL}}}
\newcommand{\glee}{E_{\eps}^{\mathrm{GL}}}

\newcommand{\annd}{\mathcal{A}_{\rm bl}}

\newcommand{\Ehp}{\mathcal{E}_{\rm hp}}

\newcommand{\es}{\mathbf{e}_s}

\newcommand{\fk}{f_{k}}

\newcommand{\fO}{f_0}

\newcommand{\fone}{\E^{\mathrm{1D}}}

\newcommand{\eone}{E^{\mathrm{1D}}}

\newcommand{\eoneo}{E ^{\rm 1D}_{0}}

\newcommand{\alk}{\alpha(k)}

\newcommand{\alO}{\alpha_0}
\newcommand{\fc}{\E ^{\rm{corr}}}

\newcommand{\pot}{V_{\om,k}}

\newcommand{\LLLN}{\mathfrak{H} ^{N}}
\newcommand{\LLL}{\mathfrak{H}}
\newcommand{\LLLh}{H ^{\rm L}}
\newcommand{\LLLf}{\E ^{\rm L}}
\newcommand{\LLLe}{E ^{\rm L}}
\newcommand{\LLLm}{\Psi ^{\rm L}}

\newcommand{\Lgv}{L_{\rm qh}}

\newcommand{\Barg}{\B}
\newcommand{\BargN}{\B ^{N}}
\newcommand{\Bargh}{H ^{\B}}
\newcommand{\Bargm}{F^ {\B}}
\newcommand{\Barge}{E ^{\B}}

\newcommand{\cLau}{c _{\rm Lau}}

\newcommand{\PKerp}{P_{\mathrm{Ker} (\mathcal{L}^N_2) ^{\perp}}}

\newcommand{\gap}{\mathrm{gap}}
\newcommand{\PsiGV}{\Psi ^{\rm qh}}

\newcommand{\mopt}{m_{\rm opt}}

\newcommand{\muN}{\mu_{N}}
\newcommand{\muNone}{\mu_{N} ^{(1)}}
\newcommand{\ZN}{\mathcal{Z}_{N}}

\newcommand{\Vm}{V_m}

\newcommand{\PsiLau}{\Psi_{\rm Lau}}

\newcommand{\MFf}{\E ^{\rm MF}}

\newcommand{\rhoMF}{\varrho ^{\rm MF}}

\newcommand{\mb}{\underline{m}}

\newcommand{\VD}{\mathcal{V}_2 ^D}
\newcommand{\ED}{E_2 ^D (N)}

\newcommand{\xbf}{\mathbf{x}}

\def\curl{{\rm curl\,}}

\def\dist{\text{dist}\ }
\def\div{\mathrm{div} \ }

\def\({\left(}
\def\){\right)}

\def\ep{\varepsilon}

\def\hal{\frac{1}{2}}

\def\indic{\mathds{1}}

\def\Lp{{L^p_{loc}(\mr^d,\mr^d)}}

\def\mn{\mathbb{N}}
\def\mr{\mathbb{R}}

\def\nab{\nabla}

\def\ro{\rho}

\def\supp{\text{Supp}}

\def\w{{H_n}}
\def\W{\mathcal{W}}

\def\Zbeta{Z_n^\beta}

\def\xbf{{\mathbf x}}

\def\j{\mathbf{E}}

\def\g{w}

\def\Gibbs{\mathbb{P}_{n,\beta}}
\def\Qk{\mathbb{P}_{n,\beta} ^{(k)}}
\def\Qone{\mathbb{P}_{n,\beta} ^{(1)}}

\def\bam{\overline{\mathcal{A}}_m}
\def\bai{\overline{\mathcal{A}}_1}
\def\P{\mathcal{P}}
\def\F{\mathcal{F}}
\def\En{\mathcal{E}}
\def\In{\mathcal{I}_n}

\newcommand{\Fnbeta}{\F_{n,\beta}}
\newcommand{\Fnbetae}{F_{n,\beta}}
\newcommand{\mubet}{\mu_{\beta}}

\def\dist{\mathrm{dist}}

\newcommand{\mubf}{\boldsymbol{\mu}}

\newcommand{\black}[1]{\textcolor{black}{#1}}

\newcommand{\weff}{w_{\rm eff}}
\newcommand{\Enls}{\mathcal{E}^{\rm nls}}
\newcommand{\enls}{E^{\rm nls}}
\newcommand{\unls}{u_{\rm nls}}
\newcommand{\Mnls}{\mathcal{M}_{\rm nls}}

\newcommand{\Gammat}{\tilde{\Gamma}}
\newcommand{\gammat}{\tilde{\gamma}}

\newcommand{\EH}{\mathcal{E} _{\mathrm{H}}}
\newcommand{\eH}{e _{\mathrm{H}}}
\newcommand{\uH}{u _{\mathrm{H}}}

\newcommand{\FNL}{F_{\mathrm{NL}}}
\newcommand{\Aext}{\mathbf{A}_{\rm ext}}
\newcommand{\Aan}{\mathbf{A}_{\rm any}}
\newcommand{\PsiQH}{\Psi ^{\rm qh}}
\newcommand{\cQH}{c _{\rm qh}}
\newcommand{\Heff}{H^{\mathrm{eff}}}

\newcommand{\bA}{\mathbf{A}}

\newcommand{\EAF}{E ^{\mathrm{af}}}
\newcommand{\cEAF}{\cE ^{\mathrm{af}}}

\newcommand{\sym}{\mathrm{sym}}

\newcommand{\cMAF}{\mathcal{M}^{\mathrm{af}}}

\numberwithin{equation}{section}

\title[]{\Huge{Some contributions to many-body quantum mathematics}
\bigskip 
\bigskip 
\bigskip
\bigskip 
\bigskip 
\bigskip
\begin{center}
\small Th\`ese d'habilitation \`a diriger des recherches de l'auteur, r\'esumant des r\'esultats obtenus entre 2011 et 2016 en collaboration avec: Michele Correggi, Mathieu Lewin, Douglas Lundholm, Phan Th\`anh Nam, Robert Seiringer, Sylvia Serfaty et Jakob~Yngvason.
\end{center}
\bigskip
\begin{center}
\small La soutenance a eu lieu le 8 Novembre 2016 au LPMMC (Universit\'e Grenoble-Alpes \& CNRS). 
\end{center}
\bigskip 
\bigskip 
\small \underline{Composition du jury}:
\begin{itemize}
 \item \small Patrick G\'erard (Pr\'esident)
 \item \small Alain Joye (Coordinateur)
 \item \small Volker Bach (Rapporteur)
 \item \small Cl\'ement Mouhot (Rapporteur)
 \item \small Maxim Olshanii (Rapporteur)
 \item \small \'Eric Canc\`es (Membre)
 \item \small Djalil Chafa\"i (Membre)
 \item \small St\'ephane Ouvry (Membre)
\end{itemize}
\large
}

\author[Nicolas Rougerie]{Nicolas Rougerie\\ 
L\lowercase{aboratoire de }P\lowercase{hysique et }M\lowercase{od\'elisation des }M\lowercase{ilieux} C\lowercase{ondens\'es}\\
U\lowercase{niversit\'e} G\lowercase{renoble}-A\lowercase{lpes} \& CNRS.}
\address{Universit\'e Grenoble-Alpes \& CNRS, LPMMC, UMR 5493, BP 166, 38042 Grenoble, France.}
\email{nicolas.rougerie@lpmmc.cnrs.fr}

\date{March 2016.}

\begin{document}

%%%%%%%%%%%%%%%%%%%%%%%%%%%%%% Texte en francais %%%%%%%%%%%%%%%%%%%%%%%%%%%%%%%%%
% 
% \renewcommand{\contentsname}{Sommaire}
% \renewcommand{\refname}{R\'ef\'erences}
% \renewcommand{\abstractname}{R\'esum\'e}
% \renewcommand{\appendixname}{Appendice}

%%%%%%%%%%%%%%%%%%%%%%%%%%%%%%%%%%%%%%%%%%%%%%%%%%%%%%%%%%%%%%%%%%%%%%%%%%%%%%%%%%%%

\maketitle

\bigskip

\bigskip

\setcounter{tocdepth}{1}
\tableofcontents 

\chapter*{Introduction}

\section*{Foreword}

The results summarized here are intended as rigorous mathematical statements on various physical models coming from condensed matter physics, statistical mechanics (classical and quantum), quantum field theory and cold atoms physics. The main tools are mostly those of the mathematical analyst: partial differential equations, functional analysis, spectral theory, calculus of variations, with some incursions into probability theory. A running thread is the construction, by rigorous asymptotic analysis most of the time, of bridges between different levels of mathematical modeling of physical phenomena.

\medskip

In choosing a title for this memoir I have followed an interesting piece of advice put forward in~\cite{Klainerman-00}. Quoting Sergiu Klainerman: ``I think we need to reevaluate our current preconception about what subjects we consider as belonging properly within Mathematics.'' Personally, I would also agree to the same statement had ``Physics'' been substituted for ``Mathematics''. 

The pun with ``Many-body quantum mathematics'', as compared with ``Many-body quantum physics'' or ``Many-body quantum mechanics'', stresses that the mathematical analysis of models and problems from many-body physics has long since grown into a field of its own, within the more general mathematical physics literature. 

There is an ever-going debate trying to decide whether the latter field consists of the intersection of mathematics with physics, or of their union, or something in between. I could personally be happy with the conservative alternative, namely the intersection. It is my impression however that the overlap is considerably larger than usually agreed on. In particular, all the material discussed in the sequel fits therein.

``Many-body quantum mathematics'' provides a unifying theme every chapter can be loosely related to, although not all deal with quantum mechanics and not all with many-body theory. It is my hope that the contributions summarized in this memoir will be of interest not only to pure mathematical physicists\footnote{A ``purely interdisciplinary'' field seems to contain somewhat of a contradiction anyway.} but also to mathematicians and physicists with an interest in rigorous mathematical foundations of physical theories.  

\medskip

Let me emphasize that, this being a habilitation thesis, everything contained herein is completely and (almost) unashamedly biased towards my own work. Only papers to which I contributed (all in collaboration with subsets of the co-authors list on the front page) are discussed in some detail, along with the necessary elements of context and perspective. 

This memoir is thus not a review paper: contributions by other authors are of course mentioned when appropriate, but I have neither tried to be exhaustive nor to give the amount of details that would be required. I still hope that this text can be useful to an extended readership, in particular because some of the chapters contain the first streamlined discussion (as opposed to the original research papers) of projects I have been involved in. Any comment is welcome and can be sent to the author by email. 

\section*{Publications summarized in this memoir.} 

The results I will discuss were obtained in the original papers listed below (in chronological order):

\begin{enumerate}
\item  \textsc{M. Lewin, N. Rougerie}, Derivation of Pekar's Polarons from a Microscopic Model of Quantum Crystals, \textit{SIAM Journal on Mathematical Analysis} \textbf{45}, 1267--1301 (2013) 
\item \textsc{M. Lewin, N. Rougerie}, On the binding of polarons in a mean-field quantum crystal, \textit{ESAIM: Control, Optimization, Calculus of Variations} \textbf{19}, 629--656  (2012) 
\item \textsc{M. Correggi, N. Rougerie}, Inhomogeneous Vortex Patterns in Rotating Bose-Einstein Condensates, \textit{Communications in Mathematical Physics} \textbf{321}, 817--860 (2013)
\item \textsc{N. Rougerie, S. Serfaty, J. Yngvason}, Quantum Hall states of bosons in rotating anharmonic traps, \textit{Physical Review A}, \textbf{87}, 023618 (2013) 
\item \textsc{N. Rougerie, S. Serfaty, J. Yngvason}, Quantum Hall phases and plasma analogy in rotating trapped Bose gases, \textit{Journal of Statistical Physics} \textbf{154}, 2-50 (2014)  
\item \textsc{M. Lewin, P.T. Nam, N. Rougerie}, Derivation of Hartree's theory for generic mean-field Bose systems, \textit{Advances in Mathematics} \textbf{254}, 570-621 (2014)
\item  \textsc{N. Rougerie, S. Serfaty}, Higher Dimensional Coulomb Gases and Renormalized Energy Functionals, \textit{Communications on Pure and Applied Mathematics} \textbf{69} (3), 519 (2016) 
\item \textsc{M. Correggi, N. Rougerie}, On the Ginzburg-Landau functional in the surface superconductivity regime,  \textit{Communications in Mathematical Physics} \textbf{332}  (3), 1297-1343 (2014) 
\item  \textsc{M. Lewin, P.T. Nam, N. Rougerie}, Remarks on the quantum de Finetti theorem for bosonic systems,  \textit{Applied Mathematics Research Express} \textbf{2015}, 48-63 (2015) 
\item \textsc{N. Rougerie, J. Yngvason}, Incompressibility estimates for the Laughlin phase, \textit{Communications in Mathematical Physics}  \textbf{336} (3), 1109-1140 (2015) 
\item \textsc{M. Lewin, P.T. Nam, N. Rougerie}, The mean-field approximation and the non-linear Schr\"odinger functional for trapped Bose gases,  \textit{Transactions of the American Mathematical Society} \textbf{368}, 6131-6157 (2016)
\item\textsc{M. Correggi, N. Rougerie}, Boundary behavior of the Ginzburg-Landau order parameter in the surface superconductivity regime, \textit{Archive for Rational Mechanics and Analysis} \textbf{219} (1),  553-606 (2016)
\item \textsc{M. Lewin, P.T. Nam, N. Rougerie}, Derivation of nonlinear Gibbs measures from many-body quantum mechanics,  \textit{Journal de l'Ecole Polytechnique} \textbf{2}, 65-115 (2015)
\item \textsc{N. Rougerie, J. Yngvason}, Incompressibility estimates for the Laughlin phase, part II,  \textit{Communications in Mathematical Physics} \textbf{339}, 263-277 (2015)
\item \textsc{P.T. Nam, N. Rougerie, R. Seiringer}, Ground states of large bosonic systems: The Gross-Pitaevskii limit revisited, \textit{Analysis and PDEs} \textbf{9} (2), 459--485 (2016)
\item \textsc{D. Lundholm, N. Rougerie}, The average field approximation for almost bosonic extended anyons, \textit{Journal of Statistical Physics} \textbf{161} (5), 1236--1267 (2015)
\item  \textsc{M. Lewin, P.T. Nam, N. Rougerie}, A note on 2D focusing many-bosons systems, arXiv:1509.09045 (2015) 
\item \textsc{M. Correggi, N. Rougerie}, Effects of  boundary curvature  on surface superconductivity, \textit{Letters in Mathematical Physics} \textbf{106} (4), 445-467 (2016)
\item \textsc{D. Lundholm, N. Rougerie}, Emergence of fractional statistics for tracer particles in a Laughlin liquid, \textit{Physical Review Letters} \textbf{116}, 170401 (2016)
\end{enumerate}

\bigskip

Some of the results pertaining to Chapters~\ref{sec:pol}, \ref{sec:FQHE states}, \ref{sec:bosons GS} and~\ref{sec:Gibbs} have already been reviewed~\cite{LewNamRou-ICMP,Rougerie-LMU,Rougerie-cdf,Rougerie-INSMI,Rougerie-xedp13,Rougerie-xedp15} in lecture notes, conference proceedings or other texts. 

The papers~\cite{BlaRou-08,CorRouYng-11,Rougerie-11,Rougerie-11b,Rougerie-12}, included in my PHD thesis~\cite{Rougerie-these}, are not discussed here. Neither are the papers~\cite{CorPinRouYng-11a,CorPinRouYng-11b,CorPinRouYng-12}, whose subject is very much related to that of my PHD. They are not included in the memoir~\cite{Rougerie-these} but have already been reviewed in conference proceedings~\cite{CorPinRouYng-12b,CorPinRouYng-13}.

\section*{Content of Chapter~\ref{sec:pol}}

\noindent\textbf{Papers 1 and 2, joint work with Mathieu Lewin.}

\medskip

A polaron is the quasi-particle formed by the interaction between a charged particle and a deformation of an underlying polarizable medium. We shall discuss a new model for coupling several charged particles to a mean-field theory of an infinite quantum crystal. The model for the latter is microscopic. We may thus describe the small polaron regime where the immersed particles' extension is comparable to the crystal's lattice spacing. 

We first study the well-posedness of the new model, at the level of the ground state. We prove binding (existence of a ground state) for a single polaron and reduce the question of binding for several polarons to the study of HVZ-type inequalities. Physically, this corresponds to the fact that each polaron is trapped in a polarization field it creates in the medium. 

Next, we consider the macroscopic limit of the new model. When the immersed polarons live on a scale much larger than the underlying crystal's lattice spacing, we recover the Landau-Pekar-Tomasevich model. In this description, the polarons interact with a continuous polarizable medium, characterized by the crystal's dielectric matrix. 

\section*{Content of Chapter~\ref{sec:superfluids}}

\noindent\textbf{Papers 3, 8, 12 and 18, joint work with Michele Correggi.}

\medskip

We study two related mean-field models for the response of superfluids/superconductors to external gauge fields. 

In Section~\ref{sec:GP vortex} we address the nucleation of vortex lattices in rotating Bose-Einstein condensates, within the framework of 2D Gross-Pitaevskii theory. We consider an asymptotic limit corresponding to strong interparticle interactions, with the rotation frequency tuned so as to capture the critical value where vortices first appear. In this regime, vortices are tightly packed in the condensate and their density  is inhomogeneous. We investigate their distribution's subtle dependence on the underlying matter density profile via the derivation of an effective electrostatic model. As expected, for large rotation speeds (compared to the first critical value), the vortex density fills the whole sample and becomes homogeneous.

In Section~\ref{sec:GL} we focus on the surface superconductivity regime of type II superconductors. In this phase, the sample is normal in the bulk but superconductivity persists close to the boundary. We use 2D Ginzburg-Landau theory to obtain a detailed description of the phenomenon, in the asymptotic limit of ``extreme-type II superconductors''. We derive a limit model for the distribution of superconductivity along the boundary, deducing a few important facts: absence of normal inclusions in the surface superconductivity layer, estimate of the phase circulation along the boundary, description of the order parameter's dependence on the curvature of the sample's boundary.

\section*{Content of Chapter~\ref{sec:FQHE states}}

\noindent\textbf{Papers 4, 5, 10 and 14, joint work with Sylvia Serfaty and Jakob Yngvason.}

\medskip

The Laughlin wave-function and its descendants form the basis of our current understanding of the fractional quantum Hall effect. They describe strongly correlated quantum fluids, where the interplay between a strong external magnetic field and repulsive interactions forces a very special and rigid form on the ground state of a 2D many-body system. Most of the chapter is concerned with the response of the Laughlin state to external fields modeling trapping and/or disorder. 

We first discuss results bearing on a model for trapped rotating bosons, where we investigate the effect of varying the shape of the trapping potential. This is 
motivated by experimental issues, in particular the instability of the system due to the centrifugal force. 

Then we turn to a more general study of incompressibility/rigidity properties. We formulate and prove some asymptotic bounds, coined ``incompressibility estimates''. Those give some rigorous basis to the conjectured stiffness of the Laughlin state in its reaction to external fields. The proofs proceed via an analogy with the 2D Coulomb gas (one-component plasma), and generalizations thereof.

\section*{Content of Chapter~\ref{sec:Coulomb gas}}

\noindent\textbf{Paper 7, joint work with Sylvia Serfaty.}

\medskip

The classical Coulomb gas is a much studied model of statistical mechanics, both for its own sake (that is, as a basic model for bulk matter and trapped plasmas) and for its relation to other fields (random matrix theory, classical and quantum vortices, the fractional quantum Hall effect ...). In this chapter we investigate the ground and thermal states of the trapped Coulomb gas, in any spatial dimension larger than two, and in the mean-field regime. We characterize fluctuations around mean-field theory in terms of a ``renormalized jellium energy'' giving the energy per unit volume of an infinite system of point charges in a neutralizing background. We obtain an asymptotic expansion of the ground state configurations, showing that the renormalized jellium energy governs the distribution of charges at the microscopic scale. We next prove large deviations estimates for the Gibbs states at small enough temperature. We also use the renormalized energy to derive new estimates on the local charge deviations and 
correlation functions.

\section*{Content of Chapter~\ref{sec:bosons GS}}

\noindent\textbf{Papers 6, 9, 11, 15 and 17, joint work with Mathieu Lewin, Phan Th\`anh Nam and Robert Seiringer.}

\medskip

This long chapter is devoted to a summary of several related works on the interacting Bose gas, in connection with the Bose-Einstein condensation phenomenon. We introduced a general method, based on (new versions of) the quantum de Finetti theorem, in order to derive mean-field descriptions of the ground state of a generic Bose system. We avoid as much as possible to rely on specific properties of physical models, using instead structural results on large bosonic states. This allows to treat a large variety of models and scaling limits in a unified way. For so-called mean-field limits, essentially no other ingredient is needed to derive non-linear Schr\"odinger models from many-body quantum mechanics. For the more delicate dilute limits, we have to rely in addition on some new many-body estimates for the ground states of specific models, but the general strategy still constitutes the backbone of the proofs.

\section*{Content of Chapter~\ref{sec:Gibbs}}

\noindent\textbf{Paper 13, joint work with Mathieu Lewin and Phan Th\`anh Nam.}

\medskip

Here we move on to studying thermal states of the interacting Bose gas, with the long-term goal of discussing the Bose-Einstein phase transition. For an interacting gas, this is a very delicate problem, and we mostly obtain preliminary results. The general strategy is to relate, in a mean-field limit with an appropriate scaling of the temperature, many-body Gibbs states to non-linear Gibbs measures (classical field theory). This is a semi-classical limit, which requires specific tools (inter alias, a new Berezin-Lieb inequality), in addition to the quantum de Finetti theorem previously alluded to. For the 1D Bose gas we obtain a full derivation of the non-linear Gibbs measure based on the mean-field energy functional. For larger spatial dimensions, well-known obstructions to the measure's construction lead us to consider simplified models, with a smooth interaction. The analysis is nevertheless highly non trivial because the limit objects are not trace-class operators.

\section*{Content of Chapter~\ref{sec:anyons}}

\noindent\textbf{Papers 16 and 19, joint work with Douglas Lundholm.}

\medskip

Anyons are (hitherto hypothetical) quantum particles with statistical properties not falling in the ubiquitous bosons/fermions dichotomy. They have been conjectured to emerge as quasi-particles of low-dimensional systems. One of the most promising candidate physical systems comes from fractional quantum Hall physics: it is widely believed that the elementary excitations of the Laughlin (and related) states should behave as anyons. In this chapter we do two things:

First, we give a new argument for the emergence of anyon statistics, in the context of the Laughlin phase. While not mathematically rigorous, we believe that the derivation settles some issues with previous arguments. In brief, we show that tracer particles in a Laughlin liquid, by coupling to quasi-holes of the Laughlin wave-function, acquire an emergent anyonic behavior. 

We next consider the rigorous derivation of simplified models for 2D anyon systems. Because of the intricacies of the many-anyons problem, there is a crucial need for those. We make a step in this direction by rigorously deriving an effective one-body energy functional for almost bosonic anyons. Its main unusual feature is that it contains a self-consistent magnetic field. The general proof strategy is inspired by that of Chapter~\ref{sec:bosons GS}, but the peculiarities of the anyon Hamiltonian require some specific tools and techniques.
 
\chapter{Polarons in quantum crystals}\label{sec:pol}

In condensed matter physics (see e.g.~\cite{AleDev-09}), one calls a \emph{polaron} the system formed by the interaction between a crystal and a charged particle (an electron, most of the time). Similarly, one calls \emph{multi-polaron} the system formed by several charged particles interacting among themselves and with a crystal. A remarkable fact is the possibility to obtain bound states in such a situation. Indeed, in vacuum, the dispersive nature of the quantum kinetic energy and the Coulomb repulsion between particles of the same charge prevent bound states of electrons from being energetically favorable. The interaction with the crystal results in an attractive force that allows electrons to form a stable bound state. Roughly, the mechanism is the following: each electron creates a local polarization of the crystal's atoms by attracting their nuclei and repelling their electrons. This local polarization around each electron is in turn felt by the other electrons forming the multi-polaron. The effective 
interaction thus mediated by the crystal can eventually overcome the Coulomb repulsion between the charged particles one adds to the crystal. 

The polaron is the effective quasi-particle formed by the electron added to the crystal, together with the polarization field it creates by locally moving the crystal's charges. The dressing of the electron by the interaction with the Fermi sea of the crystal is somewhat similar to the dressing all actual electrons experience through their interaction with the ``fictitious'' particles of the Dirac sea in quantum electrodynamics. In the following presentation we shall refer to the charged quantum particles added to an otherwise stable quantum crystal as ``polarons'', in order to distinguish them from the electrons of the crystal, which will not be modeled on the same footing.

In this chapter, which summarizes the findings of~\cite{LewRou-13,LewRou-13b}, we discuss a new polaron model obtained by coupling the $N$-body Schr\"odinger equation with a mean-field model of a quantum crystal with a charge defect. The main results are:
\begin{itemize}
 \item existence of bound states, i.e. minimizers of the energy functional of the model.
 \item derivation of the Pekar functional from the new model, in a macroscopic limit.
\end{itemize}
The present exposition is based on~\cite{Rougerie-xedp13} where the french reader can find a more detailed review of the context and results. 

\section{Models for quantum crystals}\label{sec:pol crys} 

\subsection{The perfect crystal}

Our model for the host quantum crystal is Born-Oppenheimer-like, in that the nuclei are considered as fixed classical particles, but the electrons as mobile quantum particles that optimize their configuration to take the nuclei's distribution into account. The electron's configuration is determined within a mean-field approximation, through the minimization of a reduced Hartree-Fock (a.k.a. Hartree) functional.

Let thus $\munucper$ be a positive measure on $\R^3$, periodic with respect to a given lattice $\cL$ corresponding to the nuclei's periodic arrangement.  We should first give a meaning to the reduced Hartree-Fock functional
\begin{equation}\label{eq:cri per formel}
\rhff_{\munucper} [\gamma] = \Tr \left[- \frac{1}{2}\Delta \gamma\right] - \iint_{\R^3 \times \R^3} 
\frac{\rhogam (x)\munucper (y)}{|x-y|}dxdy + \iint_{\R^3 \times \R^3} 
\frac{\rhogam (x)\rhogam (y)}{|x-y|}dxdy,
\end{equation}
giving the energy of a given infinite electron configuration. Here, $\gamma$ is a reduced one-particle density matrix, a self-adjoint, positive, trace-class operator on $L ^2 (\R^3)$. Identifying $\gamma$ with its kernel,  we denote $\rhogam (x) = \gamma (x,x)$ the corresponding one-body density. The above is a mean-field model where correlations beyond those imposed by the Pauli principle are neglected. The first term is the quantum kinetic energy, the second one the Coulomb attraction to the nuclei, and the third one the Coulomb repulsion between electrons. Note the absence of an exchange term, which would be of the form 
\[
- \iint_{\R^3 \times \R^3} \frac{\left|\gamma(x,y)\right| ^2}{|x-y|}dxdy, 
\]
and would render the functional non-convex in $\gamma$. We neglect such a term for several reasons, both physical (it usually plays a second order role) and mathematical (due to its non-convexity, the behavior of the functional including exchange is much more involved). The one-body density matrix $\gamma$ should satisfy 
\begin{equation}\label{eq:pol Pauli}
0 \leq \gamma \leq \one 
\end{equation}
as appropriate for fermions (this inequality embodies the Pauli exclusion principle).

The first challenge here is that we want to see the crystal as an ambient medium. To avoid boundary effects we have to give a meaning to the above energy functional for an \emph{infinite} locally neutral system of nuclei and electrons. This is done by considering a thermodynamic limit (infinite volume with fixed charge density) as in~\cite{LieSim-77b,CatBriLio-98,CatBriLio-01,CatBriLio-02}. The variational problem obtained by this procedure is posed on the set of one-body density matrices that commute with the translations of the lattice $\cL$   
\begin{equation}\label{eq:etats parfaits}
\cP_{\rm per} := \left\{ \gamma \in \cB \left( L ^2 (\R ^3) \right)| \; \gamma = \gamma ^*,\; 0\leq \gamma \leq 1,\; \gamma \tau_{k} =  \tau_{k} \gamma  \: \forall k\in \cL \right\} 
\end{equation}
where $\cB (\gH)$ is the set of bounded operators on some Hilbert space $\gH$. The energy (per unit volume) can then be given a meaning through Bloch-Floquet theory~\cite{ReeSim4}. We will not recall all the details. Most relevant for the forthcoming discussion of the polaron model is the Euler-Lagrange equation of the infinite-volume energy functional. It is established in~\cite{CatBriLio-02,CanDelLew-08a} that the unique minimizing density matrix $\FSm$ satisfies
\begin{equation}\label{eq:EEL crist parf}
\begin{cases}
\displaystyle \FSm = \oneep (\FSh)\\
\displaystyle \FSh = -\half \Delta + \FSp \\
\displaystyle \FSp = \left(\munucper - \FSd \right) \ast \coulker 
\end{cases}
\end{equation}
with $\FSd (x) = \FSm (x,x)$ the associated density and $\FSl$ the Lagrange multiplier (or chemical potential, or Fermi level). Here $\oneep (\FSh)$ denotes the spectral projector up to the Fermi level $\FSl$. The interpretation of~\eqref{eq:EEL crist parf} is as follows: because of the Pauli principle, non-interacting electrons fill the successive energy levels of their Hamiltonian. In the presence of mean-field interactions, the electrons of the crystal feel an effective self-consistent Hamiltonian $\FSh$ and fill its energy levels up to the Fermi level. The latter is set so as to ensure local neutrality of the crystal.

In all the following we assume that the host crystal is an insulator, i.e. that the spectrum of $\FSh$ (made of Bloch band) has gaps, and that the Fermi level $\FSl$ lies in one of those. 

\subsection{Crystal with local charge defects}

The natural idea to insert polarons in the previous model is to see them as a local charge defect in an otherwise perfect crystal. To this end, we make use of the model introduced in~\cite{CanDelLew-08a,CanDelLew-08b} where the energy of the crystal in presence of a defect is defined by reference to that of the perfect crystal discussed previously. The following description is inspired by the works~\cite{ChaIra-89,ChaIraLio-89,HaiLewSer-05a,HaiLewSer-05b,HaiLewSer-09} where the energy of relativistic particles in the presence of a Dirac sea of virtual electrons is rigorously defined (see~\cite{HaiLewSerSol-07,Lewin_proc-12} for reviews). 

\medskip

In the presence of a charge defect $\nu$ perturbing the periodic nuclei density $\munucper$, we write the density matrix $\gamma$ of the crystal's electrons as a perturbation of the density matrix of the perfect crystal:
\[
\gamma = \FSm + Q. 
\]
The energy as a function of $Q$ is given by the formal computation
\begin{equation}\label{eq:cri def formel}
 \rhff_{\munucper+\nu}[\gamma] - \rhff_{\munucper+\nu}[\FSm] = \Tr \left[  \FSh   Q \right] + D(\nu,\rhoQ) +\frac{1}{2} D (\rhoQ,\rhoQ)
\end{equation}
with
\begin{equation}\label{eq:Coul}
D(\nu,\mu) = \int_{\R ^3} \nu \left(\mu \ast \coulker\right) 
\end{equation}
the Coulomb interaction energy and $\rhoQ$ the charge density associated with $Q$. Of course, in infinite volume, both terms of the left-hand side of~\eqref{eq:cri def formel} are infinite. One can however give a rigorous meaning to their difference in the sense that, when first defined in finite volume, it has a well-defined thermodynamic limit which coincides with the right-hand side.  

Actually (we again refer to~\cite{CanDelLew-08a,CanLew-10} for more details) the appropriate functional setting for the minimization of~\eqref{eq:cri def formel} is highly non-trivial. This is because the response of the crystal to a local charge defect is highly non-local. Due to the long range nature of the Coulomb interaction, the charge density of $Q$ has long range oscillations and fails to be in $L^1 (\R^3)$ as one might have naively hoped. Then, defining an appropriate notion for the trace of $Q$ is highly not immediate. We follow~\cite{BacBarHelSie-99,HaiLewSer-05a,HaiLewSer-05b} and split $Q$ as 
\begin{align}\label{eq:decomp trace}
 Q &= \FSm Q \FSm + (1-\FSm) Q (1-\FSm) + \FSm Q (1-\FSm) + (1-\FSm) Q \FSm \nonumber \\
 &= \Qmm + \Qpp + \Qmp + \Qpm.
\end{align}
One can then define a generalized trace in the manner
\begin{equation}\label{eq:gene trace}
 \Tro [Q] := \Tr [\Qmm] + \Tr[\Qpp]. 
\end{equation}
Of course, for a trace-class operator $Q$, the generalized trace coincides with the usual notion, but this need not be the case in general. We shall also make use of the fact that, because of the Pauli principle~\eqref{eq:pol Pauli}, $Q$ should satisfy the operator inequality
\begin{equation}\label{eq:Pauli def}
-\FSm \leq Q \leq 1-\FSm 
\end{equation}
which turns out to be equivalent to
\begin{equation}\label{eq:Pauli def 2}
 \Qpp - \Qmm \geq Q ^2.
\end{equation}
The natural interpretation of the kinetic energy term in~\eqref{eq:cri def formel} is then 
\[
\Tro \left[  \FSh   Q \right] := \Tr \left[ \FSh \Qpp \right] + \Tr \left[ \FSh \Qmm \right]= \Tr \left[\left| \FSh-\FSl \right| \left( \Qpp - \Qmm \right) \right] + \FSl \Tro Q.  
\]
The coercivity property
\begin{equation}\label{eq:coer}
\Tro \left[  \left( \FSh -\FSl \right)  Q \right] \geq  \Tr \left[\left| \FSh-\FSl \right| Q ^2 \right]
\end{equation}
allows to define a natural variational space\footnote{The reader should bear in mind the rule of thumb $|\FSh - \FSl| \simeq -\Delta = |\nabla| ^2$.} 
\begin{equation}\label{eq:vari space crys}
 \Q = \left\lbrace |\nabla | Q \in \Sch ^2,\:Q = Q^* ,\: |\nabla| \Qpp |\nabla| \in \Sch ^1, \: |\nabla| \Qmm |\nabla| \in \Sch ^1\right\rbrace
\end{equation}
where the $p$-th Schatten class is
\[
\Sch ^p (L^2 (\R^3)):= \left\lbrace Q \in \cB\left( L^2 (\R^3)\right),\: \Tr \left[ |Q| ^p \right] ^{1/p} < +\infty \right\rbrace. 
\]
For all $Q\in \Q$, there exists a charge density $\rhoQ \in L^2 (\R ^3)$) such that
\[
 \Tro \left[ Q V\right] = \int_{\R ^3} \rhoQ V
\]
for all sufficiently regular $V$, and one has $D(\rhoQ,\rhoQ) < \infty$. In this setting, one can give a meaning to the expression
\begin{equation}\label{eq:crys def func}
\crysf [\nu,Q] = \Tro \left[ \left( \FSh - \FSl \right) Q \right] + D(\nu,\rhoQ) +\frac{1}{2} D (\rhoQ,\rhoQ)
\end{equation}
for the energy change of the crystal in the presence of the charge defect $\nu$. The minimal energy is
\begin{equation}\label{eq:crys def ener}
F_{\rm crys}\big[\nu\big] := \inf_{- \FSm \leq Q \leq 1 -\FSm} \left\{ \Tro \left[ \left( \FSh - \FSl \right) Q \right] + D(\nu,\rhoQ) +\frac{1}{2} D (\rhoQ,\rhoQ) \right\}.
\end{equation}
and it is shown in~\cite{CanDelLew-08a} that a minimizer exists. Some of its properties are worked out in~\cite{CanLew-10}.

\section{A new polaron model, existence of bound states}\label{sec:pol model} 

Now we can define the polaron model introduced in~\cite{LewRou-13,LewRou-13b}. The idea is simply to set the charge defect $\nu$ in~\eqref{eq:crys def func} equal to the charge density of the (multi)-polaron, and add the energy of the (multi)-polaron in vacuum. One should not forget to include the energy due to the polaron feeling the electrostatic potential of the perfect crystal.

Thus, for one polaron with wave-function $\psi \in L ^2 (\R ^2)$, we have
\begin{equation}
\cE[\psi]=\frac1{2}\int_{\R^3}|\nabla\psi(x)|^2\,dx+\int_{\R^3}V^0_{\rm per}(x)|\psi(x)|^2\,dx+F_{\rm crys}\big[|\psi|^2\big]
\label{eq:model-intro1}
\end{equation}
and for $N$ interacting polarons,
\begin{multline}\label{eq:model-introN}
\cE[\Psi] = \int_{\R^{3N}}\left(\frac12\sum_{j=1}^N \left|\nabla_{x_j}\Psi(x_1,...,x_N)\right|^2+\sum_{1\leq k<\ell\leq N}\frac{|\Psi(x_1,...,x_N)|^2}{|x_k-x_\ell|}\right)dx_1\cdots dx_N  \\
+\int_{\R ^3} \FSp \rhoP + \cryse [\rhoP]
\end{multline}
where $\Psi \in L ^2 (\R ^{3N})$ is the multi-polaron's many-body wave-function, with associated charge density 
\begin{equation}\label{eq:rhoP}
 \rhoP (x):= \int_{\R ^{3(N-1)}} |\Psi (x,x_2\ldots,x_N)| ^2 dx_2 \ldots dx_N.  
\end{equation}
Since $\Psi$ describes electrons, it should satisfy the Pauli exclusion principle, i.e. be antisymmetric:
\begin{equation}\label{eq:antisym}
\Psi (x_1,\ldots,x_i,\ldots,x_j,\ldots,x_N)= - \Psi (x_1,\ldots,x_j,\ldots,x_i,\ldots,x_N)
\end{equation}
for all $i\neq j$ and all $(x_1,\ldots,x_N) \in \R ^{3N}$. Note that this assumption is not crucial for the results of~\cite{LewRou-13,LewRou-13b}. We could deal with the case where the polarons are actually bosons, so that the wave-function is symmetric (no minus sign in the right-hand side of~\eqref{eq:antisym}).

In the above, the potential $\FSp$ is that generated by the perfect crystal in~\eqref{eq:EEL crist parf} and the energy $\cryse$ is defined via~\eqref{eq:crys def ener}.
% \begin{equation}\label{eq:crys def ener}
% \cryse[\nu] = \inf_{- \FSm \leq Q \leq 1 -\FSm} \left\{ \Tro \left( \left( \FSh - \FSl \right) Q \right) + D(\nu,\rhoQ) +\frac{1}{2} D (\rhoQ,\rhoQ) \right\}.
% \end{equation}
A few comments are in order:
\begin{itemize}
\item The main novelty of this model is that it makes no assumption on the ratio of the size of the (multi)-polaron to the crystal's lattice spacing. It is thus appropriate for describing a ``small polaron'' as opposed to the more famous and widely used Fr\"ohlich and Pekar models, that can only handle the case of a ``large polaron''. 
\item The Fr\"ohlich and Pekar models are based on an effective description of the crystal's excitations in terms of phonons. Our model can take into account the detailed structure of the host quantum crystal.
\item The price to pay for these advantages is the highly non-linear nature of the above functional. It is the source of very serious mathematical difficulties in the proofs. 
\item We have neglected possible correlations between the polarons and the crystal. In that sense it does not cover some of the physics included in the Fr\"ohlich model. However, the Pekar model can be derived from ours, as we shall discuss in Section~\ref{sec:pol Pekar}.
\end{itemize}

We now state the main results of~\cite{LewRou-13b}, referring to the original paper for their proofs and to~\cite{Rougerie-xedp13} for streamlined comments. We denote 
\begin{align}
E(1): =& \inf \left\lbrace \E [\psi],\: \psi \in H ^1 (\R ^{3}),\ \int_{\R ^{3} } |\psi| ^2 = 1 \right\rbrace 
\label{eq:energyN}\\
E(N): =& \inf \left\lbrace \E [\Psi],\: \Psi \in H ^1 (\R ^{3N}), \: \Psi \mbox{ satisfies \eqref{eq:antisym}, } \int_{\R ^{3N} } |\Psi| ^2 = 1 \right\rbrace 
\end{align}
the minimal energies for a single polaron and a multi-polaron, respectively. We have

\begin{theorem}[\textbf{Existence of small polarons} \cite{LewRou-13b}]\label{theo:E_1}\ \\
% For $N=1$, we have  
% \begin{equation}
% E(1)<E_{\rm per}:=\inf\sigma\left(-\frac1{2m}\Delta+V^0_{\rm per}\right).
% \label{eq:E_1} 
% \end{equation}
All minimizing sequences for $E(1)$ converge to a minimizer, strongly in $H^1(\R^3)$, up to extraction and translation.
\end{theorem}

In other words, the model~\eqref{eq:model-intro1} does account for the binding of a single polaron in an insulating quantum crystal: the attraction due to the crystal's deformation is sufficient to overcome the dispersive nature of the quantum kinetic energy.  

The case of a multi-polaron is more subtle since the attraction mediated by the host crystal also has to overcome the usual Coulomb repulsion between the additional electrons. The existence of a bound state depends on the balance between the two forces, and we can give necessary and sufficient conditions for binding to occur:

\begin{theorem}[\textbf{HVZ theorem for small multi-polarons} \cite{LewRou-13b}]\label{theo:HVZ}\ \\
Let $N\geq2$. The following assertions are equivalent:
\begin{enumerate}
\item One has the binding inequality
\begin{equation}\label{eq:binding petit}
 E(N) < E(N-k) + E(k) \mbox{ for all } k=1,\ldots, N-1.
\end{equation}
\item All minimizing sequences for $E(N)$ converge to a minimizer, strongly in $H^1(\R^{3N})$, up to extraction and translation.
\end{enumerate}
\end{theorem}

This is the natural generalization of the celebrated HVZ theorem~\cite{Hun-66,VanWinter-64,Zhislin-71} to the non-linear setting of the multi-polaron. The binding inequalities~\eqref{eq:binding petit} express the fact that it is not favorable to split the $N$-polaron into two smaller objects, a $k$-polaron and a $(N-k)$-polaron. It is easy to see that 
 $$E(N) \leq E(N-k) + E(k) \mbox{ for all } k=1,\ldots, N-1$$
always holds, and the case of equality corresponds to dichotomy in the concentration-compactness principle~\cite{Lions-82a,Lions-84,Lions-84b}. One can still imagine that there are minimizers in this case, but then not all minimizing sequences will converge, see~\cite{FraLieSei-12} where such a case is discussed. Note finally that checking the binding inequalities in practice would be very difficult for our model. One can however use that its macroscopic limit is given by the simplest Pekar functional, for which some results on binding are known~\cite{GriMol-10,FraLieSeiTho-10,FraLieSeiTho-11,Lewin-11}.  

As regards the proofs of the above results, they are based on a standard concentration-compactness approach in the case of one polaron, and on the natural generalization to non-linear quantum many-body systems thereof~\cite{Lewin-11} in the case of the multi-polaron. The main difficulty is to rule out the dichotomy case in the concentration-compactness principle. This is made highly non-trivial by the long range oscillations in the response of the Fermi sea to a local charge defect. It is not at all obvious that if the defect splits into two independent, well-separated pieces, so does the perturbation of the Fermi sea. Proving that this nevertheless happens, in a sufficiently strong sense, is the key lemma of~\cite{LewRou-13b}.

\section{Derivation of the Pekar functional}\label{sec:pol Pekar} 

\subsection{Macroscopic limit}

The Pekar model is appropriate for a ``large polaron'', where large means ``compared to the host crystal's lattice spacing''. To derive it from the model we previously defined, we thus introduce a small parameter $m \ll 1$ to quantify the ratio between macroscopic and microscopic length scales:
\begin{equation}
\cE_m[\psi]:=\frac1{2}\int_{\R^3}|\nabla\psi(x)|^2\,dx+m^{-1}\int_{\R^3}V^0_{\rm per}(x/m)|\psi(x)|^2\,dx+m^{-1}F_{\rm crys}\big[m^{3}|\psi(m\cdot)|^2\big],
\label{totf1}
\end{equation}
with $V^0_{\rm per}$ et $F_{\rm crys}$ defined as previously. The corresponding ground state energy is
\begin{equation}
E_m(1)=\inf \left\lbrace \cE_m [\psi],\ \int_{\R ^3} |\psi| ^2 = 1  \right\rbrace.
\label{tote1}
\end{equation}
Similarly, for $N\geq2$ the energy functional of the $N$-polaron is defined as 
\begin{multline}
\cE_{m}[\Psi]=\int_{\R^{3N}}\left(\frac12\sum_{j=1}^N|\nabla_{x_j}\Psi(x_1,...,x_N)|^2+\sum_{1\leq k<\ell\leq N}\frac{|\Psi(x_1,...,x_N)|^2}{|x_k-x_\ell|}\right)dx_1\ldots dx_N\\
+m^{-1}\int_{\R^3}V^0_{\rm per}(x/m)\,\rho_\Psi(x)\,dx+m^{-1}F_{\rm crys}\big[m^{3}|\rhoP(m\cdot)|^2\big],
\label{totfN}
\end{multline}
with the one-body density $\rhoP$ given by~\eqref{eq:rhoP}. The associated ground state energy is
\begin{equation}\label{toteN}
\tote(N) := \inf \left\lbrace \cE_m [\Psi], \: \int_{\R ^{3N}} |\Psi| ^2 = 1,\: \Psi \mbox{ satisfying \eqref{eq:antisym} }  \right\rbrace.
\end{equation}
Here we impose the scale separation by choosing a lattice whose unit cell has a characteristic length of order $m$. The weights of the different terms of the functional are chosen in view of the scaling properties of electrostatic energies, so that all terms contribute in the limit $m\to 0$.

A change of length and density units $\tilde\psi=m^{3/2}\psi(m\cdot)$ shows that, from the (microscopic) point of view of the crystal (the macroscopic point of view of the polaron being expressed in~\eqref{totf1}), the energy looks like
\begin{equation}
\cE_{m}[\tilde\psi]=m^{-1}\left(\frac1{2m}\int_{\R^3}|\nabla\tilde\psi(x)|^2\,dx+\int_{\R^3}V^0_{\rm per}(x)|\tilde\psi(x)|^2\,dx+F_{\rm crys}\big[|\tilde\psi|^2\big]\right).
\label{eq:Pekar-abstract-intro-micro}
\end{equation}
Our choice of scaling is thus equivalent to having a fixed crystal and a very light polaron, with mass $m\ll 1$ (which justifies the choice of notation). 

In the rest of this section we explain how to recover the Pekar functional in the limit $m\to0$. Note that this is very different from the usual derivation starting from the Fr\"ohlich model. In the latter derivation~\cite{DonVar-83,LieTho-97,MiySpo-07}, the point is to understand why, in a strong coupling limit, the polaron and phonons degrees of freedom decorrelate. Here we have already neglected possible polaron/crystal correlations in writing our small polaron model, and our task is to see how, in a large polaron limit, one can reduce the description of the crystal to its macroscopic dielectric properties.

\subsection{Anisotropic Pekar functional}

The model we will obtain in the limit takes the dielectric matrix $\dem$ of the crystal as a parameter. This is a scalar only in the case of a cubic crystalline lattice, $\cL = \Z ^3$. Here is the Pekar functional defining the limit problem:
\begin{equation}\label{ptgf1}
\ptgf [\psi] :=  \frac{1}{2}\int_{\R ^{3}} |\nabla \psi| ^2 dx  +  \ptgint \big[|\psi|^2\big]
\end{equation}
where $\ptgint$ is defined using Fourier variables
\begin{equation}\label{ptgint}
 \ptgint [\rho]: = 2\pi \int_{\R^3} \left| \hat{\rho} (k)\right| ^2 \left(\frac{1}{k^T \dem k} - \frac{1}{|k| ^2} \right)  \,dk.
\end{equation}
Equivalently, one can define $\ptgf$ and $\ptgint$ by considering the solution $W_\rho$ to the Poisson equation
\begin{equation}\label{poteff}
-{\rm div}\left( \dem \nabla \Wrho \right) = 4 \pi \rho.
\end{equation}
Then
\begin{equation}\label{ptgf1bis}
\ptgf [\psi] :=  \frac{1}{2}\int_{\R ^{3}} |\nabla \psi| ^2 dx  + \frac{1}{2} \int_{\R^3} |\psi|^2 \left( W_{|\psi|^2} - |\psi|^2 \ast |\cdot| ^{-1}\right)
\end{equation}
and the corresponding ground state energy is
\begin{equation}\label{ptge}
\ptge(1) = \inf \left\lbrace \ptgf [\psi],\ \int_{\R ^3} |\psi| ^2 = 1  \right\rbrace.
\end{equation}
Since a dielectric matrix always satisfies $\dem > 1$ as an operator, it is clear that the interaction $\ptgint$ is attractive.

For a $N$-polaron we define similarly
\begin{multline}\label{ptgfN}
\cE^{\rm P}_{\varepsilon_{\rm M}}[\Psi]=\int_{\R^{3N}}\left(\frac12\sum_{j=1}^N|\nabla_{x_j}\Psi(x_1,...,x_N)|^2+\sum_{1\leq k<\ell\leq N}\frac{|\Psi(x_1,...,x_N)|^2}{|x_k-x_\ell|}\right)dx_1\cdots dx_N\\  + \ptgint [\rhoP]
\end{multline}
and
\begin{equation}\label{ptgeN}
\ptge (N)= \inf \left\lbrace \ptgf[\Psi],\: \int_{\R ^{3N}} |\Psi| ^2 = 1, \: \Psi \mbox{ satisfying \eqref{eq:antisym} }  \right\rbrace.
\end{equation}
The case mostly considered in the literature~\cite{Lieb-77,Lewin-11,DonVar-83,FraLieSeiTho-10,FraLieSeiTho-11,GriMol-10} has $\dem$ proportional to the identity matrix (isotropic Pekar model), but see~\cite{Ricaud-14} for results on the anisotropic functional.

\subsection{The dielectric matrix of the crystal}

Theres is an explicit formula for the dielectric matrix of a given crystal~\cite{BarRes-86,CanLew-10}, but the most convenient definition in the present context is through the approach of~\cite{CanLew-10}. The matrix $\dem$ is obtained by considering a macroscopic excitation of the Fermi sea: take a charge defect of the form 
\[
\nu_m(x)=m^3\,\nu(mx)                                                                                                                                                                                                                                                                                                                                                                 \]
with $\nu$ fixed, and insert this in the model discussed in Section~\ref{sec:pol crys}. We call $Q_m$ an associated solution to~\eqref{eq:crys def ener} and  $W_m$ the corresponding electrostatic potential, rescaled back: 
\begin{equation}\label{intro:recaled potential}
W_m (x):=m^{-1}\big(\nu-\rho_{Q_m}\big)\ast|\cdot|^{-1}(x/m).
\end{equation}
Canc\`es and Lewin proved in~\cite[Theorem 3]{CanLew-10} that there exists a matrix $\dem$ (that one can calculate given $\munucper$ and $\cL$) such that 
\begin{equation}\label{eq:CanLew}
W_m \wto W_{\nu}, 
\end{equation}
with $W_\nu$ the unique solution to Equation~\eqref{poteff} where $\rho$ is set equal to $\nu$. This will serve as our definition of the dielectric matrix: take a charge defect, smear it over a large length scale while fixing the charge, compute the electrostatic response of the Fermi, properly rescaled. 

We first state a theorem giving the link between~\eqref{eq:crys def ener} and~\eqref{ptgint}, where the matrix $\dem$ entering the above quantity is defined as above: 

\begin{theorem}[\textbf{Macroscopic limit of the crystal's energy} \cite{LewRou-13}]\label{theo:main crys}\mbox{}\\
Let $(\psi_m)_m$ be a bounded sequence in $H^s(\R^3)$, for some $s>1/4$. Then
\begin{equation}\label{eq:limit crys}
\lim_{m\to0}\Big(m^{-1}F_{\rm crys}\big[m^3|\psi_m(m\cdot)|^2\big]-F^{\rm P}_{\dem}\big[|\psi_m|^2\big]\Big)=0.
\end{equation}
\end{theorem}

This is the key ingredient in our approach. The assumption on the sequence $(\psi_m)_m$ ensures local compactness to be able to pass to the limit. Without such an assumption, the sequence might concentrate on too small a scale for the energy to be given only in terms of macroscopic dielectric properties (see~\cite[Section 3.3]{LewRou-13} for an example).  

\subsection{The anisotropic Pekar model as the macroscopic limit of the small polaron}

In order to discuss the main results in~\cite{LewRou-13}, we note that the potential $\FSp(./m)$ appearing in~\eqref{totf1} lives on the scale of the crystal, much smaller than that of the polaron. One should thus expect that the wave-function of the polaron will have fast oscillations to accommodate this potential. Those are neglected in the Pekar functional, but can be taken into account via a very simple eigenvalue problem: we denote $u^{\rm per}_m$ the unique positive solution to 
\begin{equation*}
\Epere = \inf \left\lbrace\Eperf [v] \: :\: v\in H^1_{\rm per}(\Gamma), \int_{\Gamma} |v| ^2 = |\Gamma| \right\rbrace = \Eperf [\Eperm] 
\end{equation*}
where $\Gamma$ is the unit cell of the lattice $\cL$ and 
\begin{equation}\label{eq:Eperf defi}
\Eperf [v] = \int_{\Gamma} \frac{1}{2m} |\nabla v| ^2 + \FSp |v| ^2,
\end{equation}
gives the energy of particle of mass $m$ in the periodic potential~$\FSp$.

By periodicity one can extend $\Eperm$ to the whole of $\R^3$. A simple perturbative argument shows that $u^{\rm per}_m\to1$ in $L^\ii(\R^3)$ when $m\to0$,  and that
\begin{equation}
\Eperelim:=\lim_{m\to0} m ^{-1} \Epere = \lim_{m\to0}m^{-1}\Eperf[\Eperm] = \int_{\Gamma} V^0_{\rm per} f^{\rm per}
\label{eq:def_E_per}
\end{equation}
where $f^{\rm per}\in L^\ii(\R^3)$ is the unique $\cL$-periodic solution to 
\[
\begin{cases}
\Delta f^{\rm per}=2V^0_{\rm per} \\
\int_{\Gamma}f^{\rm per}=0.
\end{cases}
\]
This function shows up in a perturbative expansion of $u^{\rm per}_m$, see \cite[Section 2]{LewRou-13}:
\begin{equation}\label{eq:uper expansion}
\left\|u^{\rm per}_m-1-m f^{\rm per}\right\|_{L^{\infty} (\R ^3)}\leq Cm^2.
\end{equation}
One can in fact rather easily get rid of the microscopic oscillations of the polaron's wave function by an energy decoupling. This leads to the main result of this section:

\begin{theorem}[\textbf{Pekar's functional as a macroscopic limit \cite{LewRou-13}}]\label{theo:pol main}\mbox{}\\
We denote $\dem>1$ the dielectric matrix defined by~\eqref{poteff} and \eqref{eq:CanLew}. Let $N\geq1$ be an integer.   
\begin{itemize}
\item\textbf{(Convergence of the ground state energy)}. We have
\begin{equation}\label{resulte1}
\boxed{\lim_{m\to0}E_m(N)=N\,\Eperelim + \ptge(N)}
\end{equation}
with $\Eperelim$ given in~\eqref{eq:def_E_per}, and $\ptge(N)$ the Pekar energy defined by~\eqref{ptge} for $N=1$ and~\eqref{ptgeN} for $N\geq2$.

\bigskip

\item \textbf{(Convergence of states)}. Let $(\Psi_m)_m$ be a sequence of approximate minimizers for $E_m(N)$, in the sense that
$$\lim_{m\to0}\left(\cE_m[\Psi_m]-E_m(N)\right)=0.$$
Define $\Psi^{\rm pol}_m$ by setting
\begin{equation}\label{minimiseur pol}
\boxed{\Psi_m (x_1,...,x_N)= \prod_{j=1}^N\Eperm (x_j/m)\;\, \Psi^{\rm pol}_m (x_1,...,x_N).}
\end{equation}
Then $(\Psi^{\rm pol}_m)_m$ is a minimizing sequence for the Pekar energy $\ptge(N)$.  

If either $N=1$ or $N\geq2$ and binding inequalities hold for the Pekar model, there exists a sequence of translations $(\tau_m)_m\subset\R^3$ and a minimizer $\Psi^{\rm P}_{\dem}$ for $\ptge(N)$ such that
\begin{equation}\label{state converge}
\Psi^{\rm pol}_m (x_1-\tau_m,...,x_N-\tau_m) \to \Psi^{\rm P}_{\dem}(x_1,...,x_N) \mbox{ strongly in } H^1 (\R ^{3N})  
\end{equation}
along a subsequence when $m\to 0$.
\end{itemize}
\end{theorem}

Two concluding comments:
\begin{itemize}
\item The result~\eqref{minimiseur pol} shows the micro/macro scale separation. The polaron's wave-function follows at the macroscopic scale a profile given by minimizing the Pekar functional, but also incorporates microscopic oscillations on the scale of the crystal, described by $\Eperm$. We could just as well use the expansion~\eqref{eq:uper expansion} to replace $\Eperm$. The fast oscillations would not appear in a result phrased in terms of a norm not involving derivatives. In view of~\eqref{resulte1} they however contribute to the energy at the same level as the Pekar energy and it is crucial to extract them to obtain convergence in $H^1$.
\item Combining with a result of Lewin~\cite[Theorem~28]{Lewin-11}, we see that for certain values of $N$ and $\dem$, and in the macroscopic limit, minimizing sequences of our small polaron model do converge strongly, thus correctly accounting for the binging phenomenon, at least for large polarons.
\end{itemize}

\chapter{Mean-field models of superfluids and superconductors}\label{sec:superfluids}

This chapter describes results on two related models: Ginzburg-Landau theory for type II superconductors and Gross-Pitaevskii theory for rotating superfluids. In both cases, the main question is the response of a superfluid (or superconductor, which is often thought of as a charged superfluid) to external gauge fields. This response is indeed subtle and beautiful, and constitutes a hallmark of superfluidity/superconductivity. In this memoir we describe works in two directions:
\begin{itemize}
 \item on the appearance and theoretical description of vortex lattices in rotating trapped Bose gases, Section~\ref{sec:GP vortex}.
 \item on the surface superconductivity state of type II superconductors, Section~\ref{sec:GL}.
\end{itemize}
Previous contributions to the field by the author (in collaborations with Xavier Blanc, Michele Correggi, Florian Pinsker and Jakob Yngvason) can be found in~\cite{BlaRou-08,CorRouYng-11,Rougerie-11,Rougerie-11b,Rougerie-12,CorPinRouYng-11a,CorPinRouYng-11b,CorPinRouYng-12}. They shall not be discussed here, see rather~\cite{CorPinRouYng-12b,CorPinRouYng-13,Rougerie-these} for overviews. 

\section{Inhomogeneous vortex lattices in rotating Bose-Einstein condensates}\label{sec:GP vortex}

One of the remarkable properties of Bose-Einstein condensates is their superfluidity. It leads in particular to the nucleation of quantized vortices, i.e. isolated zeros of the matter density surrounded by a quantized phase circulation, in response to rotation. Among other spectacular observations, that of triangular lattices\footnote{`Abrikosov lattices' of vortices analogous to those occurring in type-II superconductors.} containing up to hundreds of vortices (see, e.g., \cite{MadChevWohDal-00,BreStoSeuDal-04,AboRamVogKet-01,CodETALCor-04,RamETALKet-01}) gave a strong motivation for theoretical studies. 

The rigorous mathematical understanding of vortex patterns in superfluids and superconductors has grown into a field of its own in the past 20 years~\cite{Aftalion-07,BetBreHel-94,SanSer-07}. In particular, for a rotating trapped Bose-Einstein condensate, several transitions characterized by the number and arrangement of vortices have been rigorously investigated~\cite{Aftalion-07,AftAlaBro-05,AftJerRoy-11,IgnMil-06,IgnMil-06b,CorRinYng-07,CorRinYng-07b,CorPinRouYng-11a,CorPinRouYng-12,Rougerie-11b,Rougerie-12}. In this section we discuss the results of~\cite{CorRou-13}, where the onset of vortex lattices is studied in a two-dimensional rotating condensate, starting from 2D Gross-Pitaevskii theory. 

The main finding is that, if the rotation frequency is larger than, but of the same order of magnitude as, a certain critical value, a dense pattern of vortices fills part of the condensate. One can derive an explicit formula for the mean distribution of vortices, showing that in this regime, the distribution is \emph{not} uniform but responds to the underlying inhomogeneous matter density. This contrasts with the situation for larger angular velocities, where the vortex density is uniform, although the matter density might vary in space. 

\subsection{2D Gross-Pitaevskii theory of the rotating Bose gas}

The mathematical framework is as follows: we consider a two-dimensional rotating BEC confined by a trapping potential. After a suitable scaling of length units (see \cite[Section 1.1]{CorPinRouYng-12}), the GP energy functional can be written
\beq
	\label{gpf}
	\gpf[\Psi] : = \int_{\R^2}  \frac{1}{2} \lf| \nabla \Psi \ri|^2 - \Omega \Psi^* L \Psi + \frac{V(r)}{\eps^2} |\Psi| ^2+ \frac{\lf| \Psi \ri|^4}{\eps^2} ,
\eeq
where $ \Omega $ is the angular velocity, the axis of rotation being perpendicular to the plane of the condensate. For simplicity, we consider a homogeneous trapping potential of the form 
\beq
	\label{ext pot}
	V(r) : = r^s,	\hspace{1cm}	s \geq 2.
\eeq 
In~\eqref{gpf}, $ L $ stands for the vertical component of the angular momentum. In polar coordinates $ \rv = (r ,\vartheta ) $, 
$$ L = - i \partial_{\vartheta} $$ 
or equivalently 
$$ L = \rv \cdot \nablap \mbox{ with } \nablap : = (-\partial_y, \partial_x).$$ 
The coupling parameter $ \eps >0 $ is going to be assumed small ($ \eps \ll 1 $), i.e., we study the so called Thomas-Fermi (TF) limit of strong interactions. The ground state energy of the system is obtained by the minimization of $ \gpf$:
\beq
	\label{gpe}
	\gpe : = \inf_{\int_{\R ^2}|\Psi| ^2 = 1 } \gpf[\Psi] = \gpf [\gpm].
\eeq
We denote by $ \gpm $ any associated minimizer (there is no uniqueness in the presence of vortices~\cite{Seiringer-02,CorRinYng-07,CorPinRouYng-12}).

\subsection{Thomas-Fermi approximation for the density profile} In the regime we shall study, the vortex distribution in the condensate will depend on the variations of the underlying matter density profile. The latter, given by $|\gpm| ^2$, can be approximated in the regime $\eps\to 0$, by minimizing the simplified Thomas-Fermi (TF) functional
\beq
	\label{tff}
	\tff[\rho] : = \eps^{-2} \int_{\R^2} \lf[ r^s + \rho \ri] \rho,
\eeq
obtained by dropping the kinetic terms in \eqref{gpf}. Here $\rho\geq 0$ plays the role of the matter density, normalized in $L ^1 (\R ^d)$. The minimizer of \eqref{tff} is the explicit radial function 
\beq
	\label{tfm}
	\tfm(r) = \half \lf[\tfchem - r^s\ri]_+,
\eeq
where $ [ \: \cdot \: ]_+ $ stands for the positive part and $ \tfchem $ is a normalization parameter that ensures 
$$ \int_{\R ^2} \tfm  = 1 .$$
Note that $ \tfm $ has compact support in a ball of radius $ \tfr $ and that a rather simple computation yields
\beq
	\label{tfchem}
	\tfe = \frac{\pi s}{4(s+1) \eps^2} \lf(\tfchem\ri)^{2(s+1)/s},	\quad \rtf = \lf(\tfchem\ri)^{1/s},	\quad \tfchem = \lf( \frac{2(s+2)}{\pi s} \ri)^{s/(s+2)},
\eeq
with $ \tfe $ standing for the TF ground state energy. Essentially, $B(0,\tfr)$ is the region occupied by the condensate: we are able to prove that that $\gpm$ decays exponentially (both as a function of $\rv$ and $\ep$) in $\R ^2 \setminus B(0,\tfr)$.

\subsection{The cost function and the critical velocity} The mechanism for vortex nucleation in rotating superfluids is now rather well understood~\cite{AftAlaBro-05,IgnMil-06,IgnMil-06b,CorRouYng-11,Rougerie-12}. One can encode in a so-called cost function the energetic contribution of a vortex located at a given point. We construct this function from two ingredients:
\begin{itemize}
 \item The energetic cost for vortex nucleation, given by
\begin{equation}\label{eq:intro vcost}
\pi |\log \ep| |d_j|\tfm (a_j)
\end{equation}
with $d_j$ the degree of the vortex and $a_j$ its location. This is the kinetic energy that is needed to generate the phase circulation around the vortex. The dependence on $\eps$ is due to the fact that the healing length of the condensate, which sets the vortex core radius, is equal to $\eps$.
\item The energetic gain brought by the vortex locally compensating the rotation field. This, it turns out, we can express using a potential function
\begin{equation}\label{eq:intro pot}
\tfpot (r) = - \frac{\Omega}{|\log \eps|} \int_{r} ^{\tfr}  t \: \tfm(t).
\end{equation}
The energetic gain associated with a single vortex of degree $d_j$ located at $a_j$ is then 
\begin{equation}\label{eq:intro vgain}
2 \pi |\log \ep| d_j F(a_j). 
\end{equation}
\end{itemize}
The cost function we shall use to determine whether or not a vortex can become energetically favorable at a given location is then
\begin{equation}\label{eq:into costf}
\tfH (r)= \half \tfm(r) + \tfpot (r). 
\end{equation}
Here we take into account the facts that vortices should have positive degrees $d_j$, because $\tfpot$ is clearly negative, and that several vortices of degree $1$ are known to be more favorable than a single vortex of larger degree.

Looking for the minimum of $\tfH$, one finds that it lies at $r=0$, indicating that this is where a vortex is most favorable. Equating gain and cost of a vortex at the origin, one can see that the critical speed for vortex nucleation is given by 
\beq
	\label{eq:first critical speed}
	\Ofirst = \Omega_1 |\log\eps|,	\hspace{1cm}	\Omega_1 : =  \frac{\pi}{2} \lf( \frac{2(s+2)}{\pi s} \ri)^{s/(s+2)},
\eeq
namely $\tfH(0) > 0$ for $\Om < \Ofirst $  and $\tfH(0) < 0$ for $\Om > \Ofirst$. The above value is that already found in~\cite{IgnMil-06,AftJerRoy-11}.

\subsection{The limit problem for the vorticity} 
The question we address here is what happens when $\Omega$ is chosen of the form
\begin{equation}\label{eq:intro Omega}
\Om = \Om_0 |\log \ep|, \qquad \Om_0 > \Om_1,\qquad \Om_0 = \OO(1), 
\end{equation}
i.e., when $\Omega$ is strictly larger than the first critical speed but of the same order of magnitude when $\ep \to 0$. In this regime many vortices are nucleated in the condensate, and their mean distribution is found by minimizing a certain energy functional that we now describe. 

The energy of a given vorticity measure $\nu$, i.e. a given mean distribution of vortices is given in units of $ |\log\eps|^2 $ by the expression
\beq
	\label{ren energy}
	\tfIc  [\nu] = \int_{\tfd} \left\{ \frac{1}{2\tfm} |\nabla h_{\nu}|^2 + \frac{1}{2} \tfm |\nu| + \tfpot \nu \right\} ,
\eeq
where
\beq
	\label{domain D}
	\tfd : = \supp (\tfm) = B(0,\tfr),
\eeq
and the potential $h_\nu$ is determined through the elliptic PDE
\begin{equation}
	\label{eq ren energy}
	 \begin{cases}
		-\nabla \left( \frac{1}{\tfm} \nabla h_{\nu}\right) = \nu	 \mbox{ in } \tfd, \\
		h_{\nu} = 0 								 \mbox{ on } \partial \tfd.
	\end{cases} 
\end{equation}
The last two terms of the functional give, as previously mentioned, the balance of local energy cost and gain of individual vortices. The first term accounts for the long-range, Coulomb-like, interaction between vortices. Note indeed the analogy between this expression and classical electrostatics: $h_{\nu}$ corresponds to a potential generated by $\nu$ in a medium with inhomogeneous conductivity (encoded in the variations of $\tfm$). 

The minimization of $\tfIc$ in its natural energy space
\beq
	\label{eq: intro measure class}
	\cM_{\mathrm{TF}} (\tfd) = \left\lbrace \nu \in \left(C^0_c (\tfd)\right) ^*,\: \int_{\tfd} \lf\{ \frac{1}{\tfm} |\nabla h_{\nu}| ^2 +  \tfm |\nu| \ri\} < +\infty \right\rbrace                                                                                                                                        \eeq
turns out to be explicit. We proved that $ \tfIc $ has a unique minimizer $ \musta $ among the measures in $\cM_{\tfm} (\tfd)$. It is explicitly given by 
\beq
	\label{musta}
	\musta =  \lf[ \nabla \lf( \frac{1}{\tfm} \nabla \tfH \ri) \ri]_+ \one_{\left \lbrace\tfH \leq 0\right\rbrace},
\eeq
and we also get 
\begin{equation}\label{eq:tfI}
\tfI: = \tfIc [\musta] = \frac{1}{2} \int_{\supp(\musta)} \tfH \musta.
\end{equation}
Note that~\eqref{musta} can be rewritten as 
\begin{equation}
\label{eq:musta 2}
\musta = \left[  \frac{1}{2}   \partial_r ^2 \log \lf( \tfm \ri) + 2 \Omega_0 \right] _{+} \one_{\left \lbrace\tfH \leq 0\right\rbrace},
\end{equation}
a formula that had previously been obtained by formal means in~\cite{SheRad-04,SheRad-04b}. Our goal will be to derive it rigorously from the full GP theory.

\subsection{Definition of the vorticity measure}

It is convenient to be a little bit more precise in our way of describing the variations of matter density. Corrections to the Thomas-Fermi profile previously introduced can be taken into account through the energy functional
\beq
	\label{hgpf}
	\hgpf[f] : = \int_{\R^2}  \half |\nabla f |^2 + \eps^{-2} \lf[ r^s + f^2 \ri] f^2, 
\eeq
i.e. the original GP functional~\eqref{gpf} restricted to \emph{real} functions. It is a standard fact that the ground state energy 
\beq
	\label{hgpdom}
	\hgpe : = \min_{\int_{\R ^2} |f| ^2 = 1 }  \hgpf[f] = \hgpf [g]
\eeq
is achieved by a function $g$ which is unique up to sign. We take the convention that $g\geq 0$ and one can show that $g$ does not vanish. This function corresponds to a vortex-free profile, taking into account radial kinetic energy corrections to the TF minimizer due to the bending of the matter density.

It is natural to write $\gpm$ in the form
\begin{equation}\label{eq:decouple formal}
\gpm = g u,  
\end{equation}
where $u$ is essentially a phase factor accounting for the vortices of $\gpm$. A convenient way of spotting those is to use the so-called \emph{vorticity measure}
\begin{equation}\label{eq:intro mu}
\mu:= |\log \eps| ^{-1} \curl \left[ \frac{i}{2}\left(u\nabla u ^* - u ^* \nabla u \right) \right]
\end{equation}
which is (up to the $|\log \ep| ^{-1}$ factor) nothing but the curl of the superfluid current 
\begin{equation}\label{eq:intro j}
\mathbf{j} :=  \frac{i}{2}\left(u\nabla u ^* - u ^* \nabla u \right)
\end{equation}
and thus (in analogy with fluid mechanics) a good candidate to count the vortices of $u$. On can indeed show that, in the regime of our interest, this ``intrinsic'' measure is very-well approximated by an ``explicit'' measure counting the vortices of $\gpm$, weighted by their degrees.

\subsection{Main theorem: the onset of vortex lattices} 

We may now state the main result of~\cite{CorRou-13}. As usual~\cite{AftAlaBro-05,IgnMil-06b,Rougerie-12}, one cannot spot vortices lying too close to the boundary of the domain. Indeed, $\tfm$ vanishes on $\dd \tfd$ and thus, according to \eqref{eq:intro vcost}, a vortex lying close to the boundary carries very little energy. We will thus limit ourselves to analyzing the behavior of $\mu$ in the smaller ball $B(0,\rbulk)$ with $\rbulk$ satisfying 
\begin{equation}\label{eq:intro Rc}
\rbulk < \tfr,\qquad \left| \rbulk - \tfr\right| = \OO(\Om ^{-1}).
\end{equation}
The convergence of the vorticity measure will be quantified in the norm
\begin{equation}\label{eq:norm}
\norm{\nu}_{\rm TF}:= \sup_{\phi \in C^1_c (B(0,\rbulk))} \frac{\bigg| \displaystyle\int_{B(0,\rbulk)} \nu \phi \bigg|}{\displaystyle \left(\int_{B(0,\rbulk)}  \frac{1}{\tfm} |\nabla \phi| ^2\right)^{1/2} + \norm{ \nabla \phi} _{L ^{\infty} (B(0,\rbulk)) }}
\end{equation}
which turns our to be well-adapted to both the study of the limit problem~\eqref{ren energy} and to the methods of proofs we use. The statement is as follows:

\begin{theorem}[\textbf{Asymptotics for the vorticity measure}]\label{thm:vorticity}\mbox{}\\
Let $\mu$ and $\musta$ be defined respectively in \eqref{eq:intro mu} and \eqref{musta}. Then we have
\begin{equation}\label{eq:vortic asympt}
\left\| \mu - \musta \right\|_{\rm TF} \leq \OO \left( \frac{\log |\log \ep| ^{1/2} }{|\log \ep| ^{1/4}}\right)
\end{equation}
in the limit $\ep \to 0$.
\end{theorem}

In the course of the proof we establish the following energy estimate, showing that the vortex energy gives the first correction to the density profile energy in the limit $\eps \to 0$:

\begin{theorem}[\textbf{Ground state energy asymptotics}]
	\label{thm:gse asympt}
	\mbox{}	\\
	If $ \Omega = \Omega_0 |\log\eps| $, with $ \Omega_0 > \Omega_1 $, then
	\beq
		\label{gse asympt}
		\gpe = \hgpe +\tfI  |\log\eps|^2 \left( 1 + \OO\left(\frac{\log |\log \ep|}{|\log\eps| ^{1/2}}\right)\right)
	\eeq
	in the limit $\eps \to 0$.
\end{theorem}

To conclude, we discuss several properties of the (average) vortex distribution that can be read off from~\eqref{eq:musta 2}. 
\begin{itemize}
\item When $\Omega_0 < \Omega_1$ in \eqref{eq:first critical speed}, a straightforward computation reveals that the cost function $\tfH$ is positive everywhere. Vortices are thus not favorable and the vortex distribution vanishes identically. We thus recover the expression of~\cite{IgnMil-06,AftJerRoy-11} for the first critical speed.
\item In the regime \eqref{eq:intro Omega}, one can compute that there is a non-empty region where $\tfH <0$ and $\nabla \frac{1}{\tfm} \nabla \tfH > 0$, whose size increases with increasing $\Omega_0$ until it finally fills the whole sample in the limit $\Omega_0 \to \infty$. This region is, according to \eqref{eq:musta 2}, filled with vortices, a behavior that is reminiscent of the `obstacle problem regime' in Ginzburg-Landau theory \cite{SanSer-07}. 
\item The first term in the right-hand side of \eqref{eq:musta 2} cannot be constant, except when $\tfm$ is constant itself, which can only happen in the flat trap, where one formally sets $s=\infty$. This term is thus responsible for an inhomogeneity of the vortex distribution, whereas the second term yields a constant contribution of $2 \Omega_0$ ($2 \Omega$ in the physical variables) units of vorticity per unit area.
\item The first term in \eqref{eq:musta 2}, responsible for the inhomogeneity, is independent of $\Omega_0$. Its importance relative to the second one thus diminishes with increasing $\Omega_0$. The inhomogeneity then becomes a second order correction in the limit $\Omega_0 \to \infty$, which corresponds to $\Omega \gg |\log \ep|$. This is observed in experiments~\cite{CodETALCor-04} and numerical simulations~\cite{Danaila-05} and bridges with the situation considered in \cite{CorPinRouYng-12}, where we proved that, if $\Omega \gg |\log \ep|$, the vortex density is to leading order constant and proportional to $2\Omega$.
\end{itemize}

\section{The surface superconductivity regime in Ginzburg-Landau theory}\label{sec:GL}

The response of a superconducting material to an applied magnetic field displays very rich physics. It is commonly~\cite{SanSer-07,FouHel-10} studied within the effective mean-field/semi-classical Ginzburg-Landau theory. The latter has been introduced on phenomenological grounds in~\cite{GinLan-50} and later derived~\cite{Gorkov-59,deGennes-66} from the microscopic Bardeen-Cooper-Schrieffer~\cite{BarCooSch-57} theory. The derivation has been made mathematically rigorous only very recently~\cite{FraHaiSeiSol-12}. 

Here we describe the results of~\cite{CorRou-14,CorRou-16,CorRou-16b} on the surface superconductivity state of type II superconductors. This is one of the two main types of mixed phases occurring for large Ginzburg-Landau parameter when a magnetic field is applied. Indeed, the phenomenology is as follows:
\begin{itemize}
 \item For low applied magnetic field $\Hex$, the full sample is superconducting and the applied magnetic field is expelled (zero magnetic field in the sample). This is the so-called Meissner effect.
 \item When the applied field reaches a first critical value $\Hc$, vortices start to appear in the sample. These are small regions where the material is in its normal, non-superconducting, phase, surrounded by a sea of the superconducting phase. The vortices carry a phase circulation and organize in triangular lattices, see~\cite{Hess-89} for experiments. The situation in this regime is reminiscent of what has been discussed in Section~\ref{sec:GP vortex} for trapped superfluids (see~\cite{SanSer-07} for review) but shall not concern us here.
 \item At a second critical value $\Hcc$ of the applied field, superconductivity is lost uniformly in the bulk of the material, but survives close to the boundary. This is the \emph{surface superconductivity} phenomenon we shall focus on, first predicted in~\cite{JamGen-63}. Although clear experimental evidence has been reported pretty soon after the prediction~\cite{Strongin-64}, it seems that direct experimental pictures of this phase became available only very recently~\cite{Ning-09}.
 \item Finally, for applied fields larger than a third critical value $\Hccc$, even surface superconductivity is destroyed by the magnetic field, and the full sample is in the normal phase.
\end{itemize}
Note that the above phenomenology stricto sensu applies only to samples with smooth boundaries. If the boundary has corners, superconductivity can survive in those for fields $\Hccc \leq \Hex \leq \Hcccc$ where $\Hcccc$ is yet another, distinct, critical value. We shall not discuss this case here: we restrict to smooth boundaries, and thus there are only three critical fields. We shall work in the regime $\Hcc \leq \Hex \leq \Hccc$. A review of the state of the art on this regime, prior to our papers~\cite{CorRou-14,CorRou-16,CorRou-16b}, can be found in the monograph~\cite{FouHel-10}.  

\subsection{2D Ginzburg-Landau theory} We consider an infinitely long superconducting cylinder of cross-section $\Omega \subset \R ^2$. Throughout, we assume that $\Omega$ is compact and simply connected, with smooth boundary. The state of a superconductor is described by an order parameter $\Psi:~\Omega\to\C$ and an induced magnetic vector potential $ \aav:\Omega \to \R ^2 $ generating an induced magnetic field 
$$h= \frac{1}{\eps^2} \: \curl \, \aav$$
which is \emph{different} from the applied field. 

We write the Ginzburg-Landau energy functional in units convenient for the discussion of surface superconductivity (other conventions are commonly used~\cite{SanSer-07,FouHel-10}):
\beq\label{eq:GL func eps}
	\glfe[\Psi,\aav] = \int_{\Omega} \lf\{ \lf| \nabla_{\aav}\Psi \ri|^2  \frac{1}{2 \hex \eps^2} \lf( |\Psi|^4 - 2|\Psi|^2 \ri) + \frac{1}{\eps^4} \lf| \curl \aav - 1 \ri|^2 \ri\}.
\eeq
Here we use the standard covariant derivative
$$ \nabla_{\aav} : =  \nabla + \frac{i}{\eps^2}  \aav, $$
$ b $ and $ \eps  $ are positive parameters depending on the material and the applied field, that we assume to be constant. Units have been chosen in such a way that $ \eps^{-2} $ measures the intensity of the external magnetic field. We shall study the asymptotic regime $\eps \to 0$, for values of $b$ satisfying
\begin{equation}\label{eq:surf regime}
 1 < b < \theo ^{-1} \simeq 1.7
\end{equation}
where $\theo$ is a spectral parameter (minimum ground state energy of the shifted harmonic oscillator on the half-line):
$$ \theo := \min_{\alpha \in \R } \: \min \left\{  \int_{0} ^{+ \infty}  |\dd_t u | ^2 + (t+\alpha) ^2 |u| ^2, \quad \int_{0} ^{+\infty} |u| ^2 = 1 \right\}.$$
The case $b=1$ (respectively $b=\theo ^{-1}$) corresponds to the second critical field (respectively the third critical field).

The above functional is invariant under the change of gauge
$$ \Psi \to \Psi e ^{i\varphi}, \quad \aav \to \aav - \eps ^2 \nabla \varphi.$$
Thus, the only physically relevant quantities are the gauge invariant ones such as the density $|\Psi| ^2$, which provides the local relative density of superconducting electrons (bound in Cooper pairs). Any minimizing $\Psi$ must satisfy $|\Psi|  \leq 1$. A~value $|\Psi| = 1$ (respectively, $|\Psi| = 0$) corresponds to the superconducting (respectively, normal) phase where all (respectively, none) of the electrons form Cooper pairs. We are interested in the ground state problem
\begin{equation}\label{eq:GL GSE}
 \glee  = \min_{(\Psi, \aav)} \glfe[\Psi,\aav]  
\end{equation}
and denote $ (\glm,\aavm) $ any minimizing pair.  

\subsection{Effective functionals} In the surface superconductivity regime~\eqref{eq:surf regime}, the order parameter $\glm$ is concentrated on a length scale of order $\eps$ close to the boundary of the sample. In addition, the applied magnetic field essentially penetrates the sample, so that 
$$\curl\aavm \simeq 1.$$
Justifying these expectations mathematically is a hard task~\cite{FouHel-10}. Accepting them, one is lead to natural reduced models describing the physics in the boundary layer. We use scaled boundary coordinates $s$ and $t$ (tangential and normal coordinates, in units of $\eps ^{-1}$). Performing a suitable choice of gauge, one may obtain:
\begin{itemize}
\item A functional corresponding to the case of a half-plane sample 
\begin{equation}\label{eq:GL hp func}
\Ehp [\psi] = \int_{s=0} ^{|\partial \Omega|\eps ^{-1}} \int_{t=0} ^{+\infty} \lf\{ \left|\left( \nabla - i t \es \right) \psi\right| ^2  + \frac{1}{2b} |\psi| ^4 - \frac{1}{b} |\psi| ^2 \ri\}.  
\end{equation}
Here one approximates the physical boundary by a straight line using the fact that all the physics happens at distance $O(\eps)$ of the boundary and neglecting lower order corrections in $\eps$.
\item A 1D functional obtained by inserting an ansatz $\psi(s,t) = f (t) e^{-i\alpha s}$ in the above:  
\begin{equation}\label{eq:GL 1D func 0}
 \fone_{0,\alpha}[f] : = \int_0^{+\infty} \lf| \partial_t f \ri|^2 + (t + \alpha )^2 f^2 + \frac{1}{2b} \lf( f^4 -2 f^2 \ri) . 
\end{equation}
It is certainly not obvious, but nevertheless true and part of our results, that such an ansatz is optimal for computing the energy in the regime~\eqref{eq:surf regime}.
\item The above functionals neglect the boundary's curvature altogether. One can, and must, go beyond this approximation for the next results. This is done by introducing a refined 1D functional 
\begin{equation}\label{eq:GL 1D func k}
\fone_{k,\alpha}[f] : = \int_0^{c_0|\log\eps|} (1-\eps k t )\lf\{ \lf| \partial_t f \ri|^2 + \frac{(t + \alpha  - \frac12 \eps k t ^2 )^2}{(1-\eps k t ) ^2} f^2 + \frac{1}{2b} \lf(f^4 - 2 f^2 \ri) \ri\} 
\end{equation}
where $k$ plays the role of a local curvature and $c_0$ is an (essentially) arbitrary constant. This is the natural refinement of~\eqref{eq:GL 1D func 0} one should introduce in the case of a disk sample of radius $R = k ^{-1}.$ One can then use it to locally approximate the physical sample's boundary by the auscultating circle rather than by the tangent.
\end{itemize}

\subsection{Energy and density asymptotics} The link between the full Ginzburg-Landau energy and the above reduced functionals is given by the next theorem. Minimizing~\eqref{eq:GL 1D func k} with respect to both the function $f$  and the real number $\alpha$ we obtain an optimal energy $\eone_\star (k)$, an optimal profile $\fk$ and an optimal phase $\alk$. The rationale is that essentially
$$
\glm(\rv) = \glm (s,t) \approx f_{k(s)} \left( t \right)  \exp \left( - i \alpha_{k(s)} s \right)  
$$
up to a suitable gauge choice, i.e. that for a fixed tangential coordinate $s$, $\glm$ is well-approximated by minimizing the reduced energy~\eqref{eq:GL 1D func k}.

\begin{theorem}[\textbf{Surface superconductivity: energy and density asymptotics}]\label{thm:GL energy}\mbox{}\\
		Denote $s\mapsto k(s)$ the curvature of the boundary $\dd \Omega$. For any fixed $1<b<\theo ^{-1}$, in the limit $ \eps \to 0$, we have
		\begin{equation}\label{eq:energy GL}
			\glee = \frac{1}{\eps} \int_0^{|\partial \Omega|} \eone_\star \left(k(s)\right) ds  + \OO (\eps |\log \eps| ^a). 
		\end{equation}
		and
		\begin{equation}\label{eq:main density}
			\left\Vert |\glm| ^2 -  \left|f_{k(s)}  (t)\right| ^2  \right\Vert_{L ^2 (\Om)} = \OO(\eps ^{3/2} |\log \eps| ^a)
		\end{equation}
		for some power $a>0$.
	\end{theorem}
	
Note that the error terms in the estimates are much smaller than the subleading order of the main terms, as one can easily realize. The above generalizes and complements previous results of~\cite{Pan-02,AlmHel-06,FouHelPer-11} in two directions:
\begin{itemize}
 \item The theorem is valid for the full surface superconductivity regime~\eqref{eq:surf regime}. The only previous result holding in this generality is that of~\cite{Pan-02} where an energy estimate in terms of~\eqref{eq:GL hp func} is obtained. The reduction to a 1D functional was performed in~\cite{AlmHel-06,FouHelPer-11} but only for $1.25 < b$. This was accomplished by perturbative methods: the limiting regime $b\to \theo ^{-1}$ being essentially linear, one can use several spectral theory tools there and try to push the regime of applicability as far as one can. It turns out that this approach does not cover the full physical regime~\eqref{eq:surf regime} and indeed, we use a completely different method for the 1D reduction.
 \item The estimates above contain information at subleading order: the minimization of~\eqref{eq:GL 1D func 0} gives the leading order of Problem~\eqref{eq:GL 1D func k} and previous results were all limited to that level of precision. In this theorem, one actually computes the energetic contribution of the domain's curvature, something that had so far remained elusive in the regime~\eqref{eq:surf regime}. This bridges with previous results~\cite{FouHel-10} in the regime $b\to \theo^{-1}$ where the role of curvature is much better understood. 
\end{itemize}

\subsection{Absence of vortices and estimate of the phase circulation} The above theorem only gives averaged (in the $L^2$ sense) information about the Ginzburg-Landau minimizer. In particular, in the absence of uniform estimates, it does not rule out the possibility that the surface superconductivity layer might contain normal inclusions, e.g. vortices. This seems rather unnatural from a physical point of view, but had never been ruled out in a mathematically rigorous manner. We can deduce from the refined energy estimates of Theorem~\ref{thm:GL energy} a $L^\infty$ estimate in the boundary layer
\beq
\label{eq:annd}
\annd : = 
%\lf\{ \rv \in \Om \: : \: \fO \lf(\eps ^{-1} \tau \ri) \geq \game \ri\} \subset 
\lf\{ \dist(\rv,\dd \Om) \leq \half \eps \sqrt{|\log\game|} \ri\},
\eeq
where $\rm{bl}$ stands for ``boundary layer'' and $ 0 < \game \ll 1 $ is any quantity such that\footnote{$ a > 0 $ is a suitably large constant.
}
\beq
\label{eq:game}
\game \gg \eps^{1/6}|\log \eps| ^{a}.
\eeq

\begin{theorem}[\textbf{Uniform density estimates and Pan's conjecture}]\label{thm:Pan}\mbox{}\\
		Under the assumptions of the previous theorem, it holds
		\beq
			\label{eq:Pan plus}
			\lf\| \lf|\glm(\rv)\ri| - \fO \lf(t \ri) \ri\|_{L^{\infty}(\annd)} \leq C \game^{-3/2} \eps^{1/4} |\log \eps| ^a  \ll 1
		\eeq
		for some constant $a>0$. In particular for any $ \rv  \in \partial \Omega $ we have
		\begin{equation}\label{eq:Pan}
			\lf| \lf| \glm(\rv) \ri| - \fO (0) \ri| \leq  C \eps^{1/4} |\log \eps| ^a \ll 1,
		\end{equation}
		where $C$ does not depend on $\rv$.
	\end{theorem}
	
Estimate~\eqref{eq:Pan} solves the original form of a conjecture by Xing-Bin Pan~\cite[Conjecture 1]{Pan-02}. In addition, since $\fO$ is strictly positive, the stronger estimate~\eqref{eq:Pan plus} ensures that $\glm$ does not vanish in the boundary layer~\eqref{eq:annd}. A physical consequence of the theorem is thus that normal inclusions such as vortices may not occur in the surface superconductivity phase. 

Since $\glm$ cannot vanish on $\dd \Omega$, its phase is well-defined there, and one could try to estimate it. Of course, only gauge invariant quantities make sense physically. One such quantity is the winding number (a.k.a. phase circulation or topological degree) of $\glm$ around the boundary $\dd \Om$:
\beq\label{eq:GL degree}
		2 \pi \: \mathrm{deg} \lf(\Psi, \partial \Omega \ri) : = - i \int_{\partial \Omega} \frac{|\Psi|}{\Psi} \partial_{s} \lf( \frac{\Psi}{|\Psi|} \ri) ds,
	\eeq
$ \partial_{s} $ standing for the tangential derivative. Theorem~\ref{thm:Pan} ensures that $\mathrm{deg}\lf(\Psi, \partial \Omega \ri)\in \Z$ is well-defined. Roughly, this quantity measures the number of quantized phase singularities (vortices) that $\glm$ has inside $\Om$. Our estimate is as follows:

	\begin{theorem}[{\bf Winding number of $ \glm $ on the boundary}]
		\label{thm:GL circulation}
		\mbox{}	\\
		Under the previous assumptions, any GL minimizer $ \glm $ satisfies
		\beq\label{eq:GL degree result}
			\mathrm{deg} \lf(\glm, \partial \Omega\ri) = \frac{|\Omega|}{\eps^2} + \frac{|\alO|}{\eps} + \OO(\eps^{-3/4}|\log\eps|^{\infty})
		\eeq
		in the limit $\eps \to 0$.
	\end{theorem}

Note that the last two theorems only capture the leading 1D contribution of~\eqref{eq:GL 1D func 0}, and do not show curvature corrections. Their proofs however require the full precision of the energy estimate~\eqref{eq:energy GL}. Indeed, one can easily see that the energetic cost of a normal inclusion would scale as $O(\eps)$. It is thus necessary to expand the energy at least to that order if one is to rule this possibility out.

\subsection{Influence of boundary curvature} As already mentioned, one of the main novelties of Theorem~\ref{thm:GL energy} is that it exhibits energetic contributions due to the sample's curvature. It is desirable to obtain more precise results, expliciting the distribution of superconducting electrons along the boundary, as a function of curvature. For obvious physical reasons, it would be desirable to have estimates of the density $|\glm| ^2$. This we cannot do with the desired precision, but we can obtain a rather explicit estimate of $|\glm| ^4$: 

\begin{theorem}[\textbf{Curvature dependence of the order parameter}]\label{thm:Gl curv}\mbox{}\\
Let $\glm$ be a Ginzburg-Landau minimizer and $D\subset \Omega$ a measurable set whose boundary $\dd D$ intersects $\dd \Omega$ with $\pi/2$ angles. Denote $s\mapsto k(s)$ the curvature of $\dd \Om$ as a (smooth) function of the tangential coordinate $s$.  

For any $1<b<\theo ^{-1}$, in the limit $\eps \to 0$,
\begin{equation}\label{eq:GL curv main estimate}
\int_{D}  |\glm| ^4 = \eps \, C_1(b) |\dd \Om \cap \dd D| + \eps ^2 C_2 (b) \int_{\dd D\cap \dd \Om}  k(s) ds + o(\eps ^2),
\end{equation}
where $d s$ stands for the 1D Lebesgue measure along $\dd \Om$ and
\begin{align}\label{eq:GL curv lead ord}
C_1(b) &= - 2 b \eoneo = \int_{0} ^{+\infty} \fO ^4 > 0 
\\
\label{eq:curv corr} C_2 (b) &=  2 b \, \fc_{\alO} [f_0] = \frac{2}{3}b \fO ^2 (0) - 2b \alO \eoneo.
\end{align}
Moreover, for $|b-\theo ^{-1}|$ small enough (independently of $\eps$), $C_2 (b) > 0.$
\end{theorem}

The leading order term in~\eqref{eq:GL curv main estimate} had previously been computed~\cite{FouKac-11,Kachmar-14,Pan-02}, with a less explicit expression of the constant $C_1 (b)$ however. The above theorem shows that the concentration of superconducting electrons along the boundary depends linearly on the curvature. Close to the third critical field, we are able to show that curvature, counted inwards, favors surface superconductivity. This is in accord with previous results (see again~\cite{FouHel-10} for review) obtained in the perturbative regime $b\to \theo ^{-1}$, in particular with known corrections to the expression of the third critical field. Note that the assumption that the set $D$ intersects $\dd \Omega$ at right angles is probably necessary to state the result in the form above: integration along lines of constant tangential coordinate and delicate estimates of tangential variations of the order parameter are involved in the proof.

Let us explain briefly how one obtains the above expressions. The variations due to the boundary's geometry are encoded in the behavior of the 1D functional~\eqref{eq:GL 1D func k}. Since we work in the regime $\eps \to 0$, it is natural to consider a perturbative expansion of $\eone_\star \left(k(s)\right)$. We then get that the leading order is given by the $k=0$ functional~\eqref{eq:GL 1D func 0} and that the first correction is $-\eps k$ times 
\begin{equation}\label{eq:corr func}
\fc _\alpha [f] := \int_{0} ^{c_0 |\log \eps|}  t\lf\{ \lf| \partial_t f \ri|^2 + f ^2 \left( -\alpha (t+\alpha) -\frac{1}{b} + \frac{1}{2b} f ^2\right)\ri\}
\end{equation}
which is obtained by retaining only linear terms in $\eps k$ when expanding~\eqref{eq:GL 1D func k}. It turns out that this correction, evaluated at $f= f_0$, is equal to $C_2 (b)$ defined above. Then, if one can localize in the set $D$ the energy estimate~\eqref{eq:energy GL}, a version of~\eqref{eq:GL curv main estimate} with $|\glm|^4$ replaced by the Ginzburg-Landau energy density results. Integrating the Ginzburg-Landau variational equation over $D$ and estimating a boundary term (a delicate step) leads to the result for the distribution of $|\glm| ^4$.

\chapter{Rigorous studies of fractional quantum Hall states}\label{sec:FQHE states}

The Hall effect has been discovered in the 19th century: if an electric current is imposed on a 2D sample in the presence of a perpendicular magnetic field, a voltage drop perpendicular to both the current and the magnetic field develops, due to the Lorentz force acting on charge carriers. This effect is in fact a popular way to measure the carriers' density and charge in a sample. The more recent discovery of the quantum Hall effects (integer, and then fractional) has lead to two Nobel prizes in physics (von Klitzing in 1985, Laughlin-St\"ormer-Tsui in 1998). The conductance associated to the perpendicular voltage drop is in fact quantized in certain fractions\footnote{With $e$ and $h$ the charge of the electron and Planck's constant, respectively.} of $e ^2 /h$. The precision of the quantization is astonishing, independently of the sample's details. So much so that it is nowadays used to set the resistivity standard.

The integer quantum Hall effect (IQHE) can be understood within non-interacting electron models. While a full theoretical understanding is lengthy and difficult (see e.g.~\cite{BelSchEls-94,Goerbig-09} for reviews), the basic ingredients are the well-known Landau quantization of kinetic energy levels in presence of a magnetic field, and the Pauli principle. The fractional quantum Hall effect (FQHE) on the other hand has been recognized very early as being due to the occurrence of new, strongly correlated, electronic ground states (see~\cite{Girvin-04,Goerbig-09,StoTsuGos-99,Laughlin-99,Jain-07} for reviews). The theoretical understanding of these new states of matter is still one of the most challenging problems of condensed matter physics. Very little is known at a full mathematical level of rigor. In this chapter we describe some results aiming at filling part of this gap. We focus on studies of the most basic fractional quantum Hall states, introduced by Laughlin in 1983.

As happened with many condensed matter phenomena, it has been suggested~\cite{BloDalZwe-08,Viefers-08,Cooper-08} that insight on the FQHE could be gained by emulating (part of) the physics in a cold atoms system. This has so far remained elusive (see however~\cite{Spielman-09,Spielman-12} for recent experimental progress), but this promising possibility opens new theoretical questions. These regard both the proposed specific cold atoms realizations of FQHE physics, and basic properties of FQHE states that could be probed more easily in those than in actual 2D electron gases. 

Section~\ref{sec:FQHE frame} below is a short general introduction to the topics addressed in this chapter. In Section~\ref{sec:FQHE Bose gas} we discuss FQHE states in rotating trapped Bose gases, with a special emphasis on how the trap should be taken into account. This issue, and related ones in the 2D electron gas setting, leads to the introduction of a new variational problem and related asymptotic estimates that we study in Section~\ref{sec:FQHE incomp}. We refer to these inequalities as ``incompressibility estimates'' because they shed light on the rigidity that the Laughlin state displays in its response to external potentials, a crucial property in FQHE physics. A (very) streamlined review of this material has already appeared in~\cite{Rougerie-INSMI}.

\section{General framework}\label{sec:FQHE frame}

The situation of interest for this chapter is that of interacting particles confined in two space dimensions, submitted to a perpendicular magnetic field of constant strength. The Hamiltonian of the full system is given by
\begin{equation}\label{eq:first princ hamil}
H_N = \sum_{j=1} ^N \left(  \left( -i\nabla_j +  \mathrm i\mathbf A(\mathbf x_j)
%B \xbf_j ^{\perp}
\right) ^2 + V (\xbf_j) \right) + \lambda \sum_{1\leq i<j \leq N} w (\xbf_i-\xbf_j) 
\end{equation}
where $\mathbf A$ is the vector potential of the applied magnetic field of strength $B$, given by
$$\mathbf{A}(\mathbf {x})=  \frac{B}{2} (-y,x)$$
in the symmetric gauge. The pair interaction potential is denoted $w$, the coupling constant~$\lambda$. Here we think of repulsive interactions, 
$$w\geq 0, \quad \lambda \geq 0.$$
Units are chosen so that  Planck's constant $\hbar$, the velocity of light and the charge are equal to 1, and the mass is equal to 1/2. The potential $V$ can model both trapping of the particles and disorder in the sample. We will study systems made of fermions or bosons, so $H_N$ will act either on $\bigotimes_{\rm asym}  ^N L ^2 (\R ^2)$ or $\bigotimes_{\rm sym}  ^N L ^2 (\R ^2)$, the anti-symmetric or symmetric $N$-fold tensor products of~$L ^2 (\R ^2)$.  

The quantum Hall effect (integer or fractional) occurs in the regime of very large applied magnetic fields, $B\gg 1$. We shall thus here restrict attention to single particle states corresponding to the lowest eigenvalue of $  \left( -i\nabla + \mathbf A(\mathbf x)\right) ^2$, i.e., the lowest Landau level~(LLL)
\begin{equation}\label{eq:intro LLL}
\LLL := \left\{ \psi\in L ^2 (\R ^2):\ \psi (\xbf) = f(z) e ^{ - B |z| ^2 / 4},\: f \mbox{ holomorphic } \right\} 
\end{equation}
as the single-particle Hilbert space. Here and in the sequel we identify vectors $\xbf =(x,y) \in \R ^2$ and complex numbers $z =x+iy \in \C$. We choose  units so that the magnetic field is~$2$, in order to comply with the conventions for rotating Bose gases as in~\cite{RouSerYng-13a,RouSerYng-13b} where $B$ corresponds to~$2$ times the angular velocity.

The best understood situation leading to a FQHE is that where the interactions are repulsive enough to force the many-body  wave-function to vanish whenever particles come close together. Laughlin~\cite{Laughlin-83} suggested the very fruitful idea that the ground state of a 2D electron gas could be described, in such a parameter regime, by many-body wave-functions of the form
\begin{equation}\label{eq:intro Laughlin}
\PsiLau (z_1,\ldots,z_N) = \cLau \prod_{1\leq i<j \leq N} (z_i-z_j) ^{\ell}  e ^{- \sum_{j=1} ^N |z_j| ^2 / 2}. 
\end{equation}
Here $\ell \geq 3$ is an odd integer and $z_1,\ldots,z_N$ the coordinates of the 2D electrons, identified with complex numbers, and measured in units of the magnetic length. The constant $\cLau$ normalizes the state in $L^2(\R ^{2N})$. One can trivially adapt the wave-function to the bosonic case by taking $\ell$ even, $\ell = 2$ being the most relevant case. Much of FQHE physics is based on the strong inter-particle correlations included in Laughlin's wave-function. 

The proposed trial state~\eqref{eq:intro Laughlin} satisfies three requirements:
\begin{itemize}
 \item it is built solely out of one-particle orbitals of the lowest Landau level, and thus minimizes the magnetic kinetic energy;
 \item it vanishes when two electrons collide, at a faster rate than that imposed by the Pauli principle (which corresponds to $\ell = 1$);
 \item one can argue that it corresponds to a filling factor $\nu$ (number of electrons per unit magnetic area) of approximately $1/\ell$. 
\end{itemize}
It historically served as the basis of a theory explaining the FQHE at $\nu = 1/3$, the first case where it was observed. It also allowed to predict the effect at $\nu = 1/5$, observed later. Noteworthily, this wave-function is proposed as a ground state ansatz for the Hamiltonian~\eqref{eq:first princ hamil} with $V\equiv 0$. The rationale is that the FQHE happens when the energy scales are, by order of importance\footnote{Actually, one assumes first that the electrons' spins are all aligned with the magnetic field, but this is already implicit in~\eqref{eq:first princ hamil}.}, first the magnetic kinetic energy, then the repulsive interactions and last the potential energy due to trapping and disorder.

More generally, one may restrict the wave function to live in the space 
\begin{equation}\label{eq:Ker}
\Ker = \left\{  \Psi\in L^2(\R ^{2N}):\  \Psi (z_1,\ldots,z_N) = \PsiLau (z_1,\ldots,z_N) F (z_1,\ldots,z_N), \: F \in \BargN \right\} 
\end{equation}
where 
\begin{equation}\label{eq:BargN}
\BargN := \bigotimes_{\mathrm{sym}}  ^N \Barg =  \left\{ F \mbox{ holomorphic and symmetric } | \:  F(z_1,\ldots,z_N) e^{- \sum_{j=1} ^N |z_j| ^2 /2  } \in L ^2 (\R ^{2N})\right\} 
\end{equation}
is the $N$-body bosonic Bargmann space (symmetric here means invariant under exchange of two particles $z_i$, $z_j$), with scalar product 
\begin{equation}\label{eq:scalar BargN}
\left\langle F, G \right\rangle_{\BargN} : = \left\langle F e^{-\sum_{j=1} ^N |z_j| ^2 /2} , G e^{-\sum_{j=1} ^N |z_j| ^2 /2}  \right\rangle_{L^2 (\R ^{2N})}.
\end{equation}
In allowing this more general class of states we want to have some flexibility to accommodate the external potential $V$ in~\eqref{eq:first princ hamil}. We still assume that the main energy scales are set by the magnetic kinetic energy and the repulsive interactions, but this leaves some freedom, which is important in applications:
\begin{itemize}
 \item In the 2D electron gas where the FQHE is observed, the disorder potential due to impurities is absolutely crucial to explain the effect. Indeed, a large part of Laughlin's theory concerns the response of his proposed ground state to disorder.
 \item In experiments with cold Bose gases, a trapping potential is ubiquitous, and it is of importance to understand how the Laughlin state responds to it. Actually, some recent proposals~\cite{MorFed-07,RonRizDal-11,ParFedCirZol-01} to emulate FQHE physics with cold atoms require some non-trivial engineering of the trap.
\end{itemize}
It is of some importance in Laughlin's theory that the system actually chooses simple functions such as
\begin{equation}\label{eq:Laughlin phase 2}
\Psi (z_1,\ldots,z_N) = c_f \PsiLau (z_1,\ldots,z_N) \prod_{j=1} ^N f (z_j)
\end{equation}
for some one-particle analytic function $f$, rather than having a genuine general $F$ in~\eqref{eq:Ker}. Comparisons with experiments are pretty convincing, but the profound reason for the emergence of such forms is not quite obvious. A large part of the works presented here is motivated by this issue.

\section{The fractional quantum Hall regime in rotating trapped Bose gases}\label{sec:FQHE Bose gas}

The formal similarity between the Hamiltonian of a rotating Bose gas and that of a 2D electron gas in a magnetic field suggests the possibility of studying {\em bosonic} analogues of the FQHE in cold quantum gases set in rapid rotation. In particular, a phase transition from a Bose-Einstein condensate (BEC) with a vortex lattice (studied e.g. in~\cite{AftBlaDal-05,AftBla-06,AftBlaNie-06a,AftBlaNie-06b}) to a strongly correlated Laughlin state should happen~\cite{BerPap-99,CooWil-99,CooWilGun-01} at total angular momentum $\propto N ^2$ (filling factor of order $1$), where $N\gg 1$ denotes the particle number.  This phase transition has not been observed yet because it is extremely difficult to reach such high angular momenta in the laboratory.

The reasons for this difficulty can be understood as follows. The usual experimental setup consists of bosonic atoms in a rotating trap modeled by a harmonic potential. The usual approximations for large rotation frequency are that the gas is essentially 2D, that single-particle orbitals are confined to the lowest Landau level and that the inter-particles interactions are described by a Dirac delta potential. These can to some extent be backed with rigorous mathematics \cite{LewSei-09}. Accepting them, the relevant Hamiltonian boils down to 
\begin{equation}\label{eq:LLLh harm'}
H_N=  \sum_{j=1} ^N  \om |x_j| ^2+ g \sum_{i<j} \delta (x_i-x_j)
\end{equation}
operating on wave functions $\Psi(x_1,\dots, x_N)$, $x_i\in\mathbb R^2$, in the  lowest Landau level (LLL) of a two-dimensional magnetic Hamiltonian. Here $\om>0$ is half the difference between the squares of the frequency of the harmonic trap and of the rotation velocity, measured in units of the latter, and $g\geq 0$ is the coupling constant of the interaction. When $\om/g$ is large, one can prove~\cite{LieSeiYng-09} that the ground state of $H_N$ shows Bose-Einstein condensation, i.e. correlations between particles almost vanish. To obtain correlated states, one wants to favor the interaction, i.e., consider small values of $\om/g$. Clearly this corresponds to a singular limit where the gas is no longer confined against centrifugal forces. This is the main source of  difficulty that has to be overcome to create the Bose-Laughlin state (see~\cite{RonRizDal-11} for a more quantitative discussion). 

A natural proposal \cite{BreStoSeuDal-04,Viefers-08,MorFed-07} is to add a confining term that should be stronger than quadratic (i.e., stronger than the centrifugal force). Here we investigate asymptotic properties of ground states after such an addition to the Hamiltonian.  We focus in a regime where the parameters are tuned to favor strong correlations. A simple model for the additional term is $\sum_j k|x_j|^4$ with $k>0$, that has already been considered in the context of rotating BECs~\cite{AftAlaBro-05,BreStoSeuDal-04,BlaRou-08,CorPinRouYng-11b,CorPinRouYng-12,Rougerie-11}. A new problem that has to be tackled is the lack of commutativity of the anharmonic potential term, projected onto the lowest Landau level,  and the interaction. This implies that the Laughlin state is not an exact eigenstate of the modified Hamiltonian as it is for the Hamiltonian in a purely harmonic trap. In this section we summarize the papers~\cite{RouSerYng-13a,RouSerYng-13b} where we studied this situation from a 
mathematical point of view.

\subsection{Model for the quantum Hall regime} Taking for granted the approximations mentioned previously, we shall consider a model with one-particle states reduced to the lowest Landau level~\eqref{eq:intro LLL}. For a rotating Bose gas, the role of the magnetic field $B$ is played by twice the rotation frequency. Setting the latter equal to $1$ for simplicity, the $N$-particle Hilbert space is thus  
\begin{equation}\label{eq:LLL N particles}
\LLLN = \left\{ \Psi (z_1,\ldots,z_N) = F(z_1,\ldots,z_N) e^{-\sum_{j=1} ^N |z_j| ^2 / 2 } \in L ^2 (\R ^{2N}),\: F \in \BargN \right\}.  
\end{equation}
where the Bargmann space for $N$ bosons is defined as 
\begin{equation}\label{eq:Barg N particles}
\BargN =  \left\{ F \mbox{ holomorphic and symmetric such that } F(z_1,\ldots,z_N) e^{-\sum_{j=1} ^N |z_j| ^2 / 2 } \in L ^2 (\R ^{2N})\right\}. 
\end{equation}
This is a Hilbert space with scalar product inherited from the usual $L^2$ product:
\begin{equation}\label{eq:BargN scalar}
\bral F,\:G \ketr_{\BargN} = \bral F e^{-\sum_{j=1} ^N |z_j| ^2 / 2 },\:G e^{-\sum_{j=1} ^N |z_j| ^2 / 2 } \ketr_{L ^2 (\R ^{2N})}.
\end{equation}
We are interested in an energy functional corresponding to~\eqref{eq:LLLh harm'}, with an additional quartic term in the one-body potential:
\begin{equation}\label{eq:ener LLL}
\LLLf [\Psi] = N \int_{\R ^2} \pot (r) \rhoP (z) dz + 4 N(N-1) g  \int_{\R^{2(N-1)}} \left| \Psi (z_2,z_2,z_3\ldots,z_N)\right| ^2 dz_2\ldots dz_N
\end{equation}
where 
\begin{equation}\label{eq:QHE density}
\rhoP (z) = \int_{\R ^{2(N-1)}} \left| \Psi (z,z_2,\ldots,z_N)\right| ^2 dz_2\ldots dz_N
\end{equation}
is the one-body density of $\Psi$, normalized to $1$, and
\begin{equation}\label{eq:pot}
\pot (r) = \om r ^2 + k r ^4, \quad r = |z| ^2. 
\end{equation}
The second term in~\eqref{eq:ener LLL} corresponds to contact interactions, where only the value of the wave-function along the diagonals of configuration space is penalized. For a general wave-function $\Psi$ this is of course ill-defined, but we are interested in minimizing amongst the extremely regular functions of the $N$-particles LLL~\eqref{eq:LLL N particles}:
\begin{equation}\label{eq:intro LLLe}
\LLLe := \inf\left\{\LLLf[\Psi] = \bral \Psi, \LLLh \: \Psi \ketr,\: \Psi \in \LLLN, \Vert \Psi \Vert_{L^2 (\R ^{2N})} = 1  \right\}.
\end{equation}
Then, the energy functional~\eqref{eq:ener LLL} is perfectly well-defined, and one can show that it is  in fact the quadratic form associated with a certain bounded operator on the Bargmann space~\eqref{eq:Barg N particles}:
\begin{equation}\label{eq:Bargh}
\Bargh := N\left( \om + 2 k \right) + \sum_{j=1} ^N \left(\left( \om + 3 k \right) L_j + k L_j ^2 \right)+ g \sum_{1\leq i<j\leq N} \delta_{ij} 
\end{equation}
with 
$$L _j = z_j \dd _{z_j}$$
the angular momentum operator in the $j$-th variable and
\begin{equation}\label{eq:inter Barg} 
\delta _{ij} F (\dots,z_i,\dots,z_j\dots)=\frac 1{2\pi} F \big(\dots,\half(z_i+z_j), \dots,\half(z_i+z_j),\dots\big).
\end{equation} 
We state the correspondence between~\eqref{eq:ener LLL} and~\eqref{eq:Bargh} as a lemma, whose proof is a combination of well-known arguments (see~\cite{AftBla-08,PapBer-01,RouSerYng-13b}): 

\begin{lemma}[\textbf{Hamiltonian in the Bargmann space}]\label{lem:ham barg}\mbox{}\\
We have, for any $\Psi\in \LLLN$,
\begin{equation}\label{eq:barg LLL plus}
\LLLf[\Psi] = \bral F,\Bargh F \ketr_{\BargN} \mbox{ with } \Psi = F \exp\left(-\sum_{j=1} ^N \frac{|z_j| ^2}{2} \right),\: F \in \BargN.   
\end{equation}
In particular
\begin{equation}\label{eq:barg LLL}
\LLLe = \Barge := \inf \sigma_{\BargN} \Bargh. 
\end{equation}
Moreover, $\Barge$ is an eigenvalue of $\Bargh$ with (possibly non unique) eigenfunction $\Bargm$. One may choose $\Bargm$ to have a definite total angular momentum, say 
$$ \LL_N \Bargm = L_0 \Bargm$$
where 
\begin{equation}\label{eq:QHE tot ang mom}
\LL_N :=\sum_{j=1} ^N z_j \dd_{z_j}. 
\end{equation}
Equivalently, the infimum in~\eqref{eq:intro LLLe} is attained. One may choose a minimizer $\LLLm$ with definite angular momentum $L_0$.
\end{lemma}

In the sequel we are interested in properties of a ground state with fixed angular momentum, in the regime where correlations between particles are so strong that the wave-function essentially belongs to a space of the form~\eqref{eq:Ker}, and thus contains (at least) the correlations built-in in the Laughlin state.

\subsection{Criteria for strong correlations in the ground state}  It is easy to see that the space~\eqref{eq:Ker}, for $\ell = 2$, corresponds to the null space of the total interaction operator
\begin{equation}\label{eq:inter tot}
\mathcal{I}_N = \sum_{1\leq i<j \leq N} \delta_{ij}, 
\end{equation}
i.e. for $F\in \BargN$
$$ \mathcal{I}_N F = 0 \Longleftrightarrow \Psi := F e ^{-\sum_{j=1} ^N |z_j| ^2 / 2} \in \mathcal{L}^N_2.$$
We first state criteria ensuring that the ground state of our problem almost fully lives in~$\mathcal{L}^N_2$. These are stated in terms of the gaps of the so-called yrast curve\footnote{Note that the total angular momentum commutes with the total interaction.}
\begin{equation}\label{eq:yrast gap}
\gap \left(L \right):= \min \left(\sigma \left( \mathcal{I}_N | _{ \{ \LL_N = L \}} \right) \setminus \left\{ 0 \right\}\right)
\end{equation}
whose dependence on the total particle number $N$ and angular momentum $L$ is unfortunately unknown at present. 

\begin{theorem}[\textbf{Criteria for strong correlations in the ground state}]\label{thm:correl}\mbox{}\\
Let $\LLLm$ be a minimizer of \eqref{eq:intro LLLe} and $\PKerp$ the projector on the orthogonal complement of $\mathcal{L}^N_2$. Define the spectral gaps
\begin{equation}\label{eq:defi gaps}
\Delta_{1} =  \gap(2 N^2), \quad 
\Delta_{3} = \gap(\Lgv + \sqrt 3 N^{2}),\quad
\Delta_4  = \gap(\Lgv+ \sqrt 3 \Lgv ^{1/2} N),
\end{equation}
where
\begin{equation}\label{eq:mom GV}
\Lgv =  - \frac{\om N}{2k} + O(1).
\end{equation}
We have
\begin{equation}\label{eq:correl}
\left\Vert \PKerp \LLLm \right\Vert  \to 0  
\end{equation}
in the limit $N\to \infty$, $\om,k\to 0$ if one of the following conditions holds~:\smallskip

\noindent {\bf Case 1.} $\om \geq 0$ and 
\[
\left( g \,\Delta_1 \right)^{-1}(\omega N^2+kN^3)\rightarrow 0. 
\]
%In this case, for $N$ large enough, 
%\begin{equation}\label{eq:estim correl 1}
%\Vert \PKerp \LLLm\Vert^2\leq\frac 13 \left(a \,\Delta_{1,2} (C) \right)^{-1} kN^3 (1+o(1)).
%\end{equation}

\noindent {\bf Case 2.} $- 2 k N\leq \om \leq 0$ and
\[
\left(g \,\Delta_1  \right)^{-1}({N \om ^2}/{k} + \om N ^2 + k N ^3) \rightarrow 0. 
\]
%Again \eqref{eq:estim correl 1} holds under this condition for $N$ large.
\smallskip

\noindent{\bf Case 3.} $\om \leq - 2 k N$ and
\[
\left(g \: \Delta_3  \right)^{-1} k N ^{3}\rightarrow 0.
\]
%Here 
%\begin{equation}\label{eq:estim correl 2}
%\Vert \PKerp \LLLm \Vert^2\leq \frac 13 \left(a \: \Delta_3 (C) \right)^{-1} k\,N^3(1+o(1)).
%\end{equation}
\noindent{\bf Case 4.}
$\om \leq - 2 k N$ and
$$
\left(g \: \Delta_4 
\right)^{-1} |\omega|N\rightarrow 0.
$$
%Here 
%\beq\label{eq:estim correl 3}
%\Vert P^\perp \Psi_0\Vert^2\leq \frac 32 \left(a \: \Delta_4 (C) \right)^{-1} |\omega| N (1+o(1)).
%\eeq 
\end{theorem}

We believe that these criteria for entering the strongly correlated Laughlin phase $\mathcal{L}^N_2$ are essentially optimal. They set bounds on parameters of a putative experiment needed to obtain the new physics (beyond that of Bose-Einstein condensates) of the FQHE regime.

Note that obtaining such a theorem in the case $k=0$ is essentially trivial since in that case the Hamiltonian is made of two parts, one proportional to the total angular momentum and the other to the total interaction. These commute with one another (because of rotational invariance of the full interaction) and may thus be jointly diagonalized, so that the ground state must lie on the yrast line, i.e. achieve the infimum in~\eqref{eq:yrast gap} for some~$L$. The difficulties in the proof of Theorem~\ref{thm:correl} thus reside in the quartic addition to the trapping potential, which destroys the nice structure of the $k=0$ case. 

Of course, one would be interested in gaining more knowledge on the $N$-dependence of the gaps~\eqref{eq:defi gaps} setting the conditions for strong correlations. In view of existing numerical results~\cite{RegChaJolJai-06,RegJol-04,RegJol-07,VieHanRei-00}, a natural conjecture seems to be that 
\begin{align}\label{eq:conjecture gap}
\gap \left( L \right) &= \gap(N(N-1) -N)\mbox{ for any } L\geq N(N-1)-N,\\
\gap(L) &\geq \gap (N(N-1)-N) \geq C \mbox{ for any } L\in \N, 
\end{align}
so that one could replace $\Delta_1, \Delta_2, \Delta_3$ by universal constants in the statement, but this remains an open problem. The different cases in the theorem correspond to changes in the nature of the ground state, as we discuss next. 

\subsection{Angular momentum estimates} We have just introduced criteria for entering the strongly correlated regime. In the purely harmonic case $k=0$, one can easily show that as soon as the ground state essentially lives on $\mathcal{L}_2 ^N$, then it is essentially proportional to the Laughlin state~\eqref{eq:intro Laughlin}. This is because the Laughlin state clearly has the smallest angular momentum amongst functions canceling the interaction operator. When $k>0$ and $\omega< 0$, the one-body potential develops a local maximum at the origin, which results in new physics: a state of $\mathcal{L}^N_2$ with larger angular momentum can be preferred over the Laughlin state.

\begin{theorem}[\textbf{Angular momentum estimates}]\label{thm:momentum}\mbox{}\\
In the limit $N\to \infty$, $\om,k \to 0$ the angular momentum $L_0$ of a ground state of $\LLLh$ satisfies
\begin{enumerate}
\item If $\om \geq - 2 k N$,
\begin{equation}\label{eq:result mom 1}
L_0 \leq 2 N ^2 
\end{equation}
\item If $\om \leq -2 kN$ and $|\om| /k \ll N ^2$,
\begin{equation}\label{eq:result mom 2}
\left| L_0 - \Lgv \right| \leq \sqrt{3} N ^2,
\end{equation}
where
\begin{equation}\label{eq:mom GV bis}
\Lgv =  - \frac{\om N}{2k} + O(1).
\end{equation}
In particular  $L_0/\Lgv \to 1$ if $ N \ll |\omega|/k\ll N 2$.
\item If $\om \leq -2 kN$ and $|\om| /k \gg N ^2$
\beq\label{eq:result mom 3}
\left| L_0 - \Lgv \right| \leq \sqrt{3} \Lgv ^{1/2} N. 
\eeq
In particular  $L_0/\Lgv \to 1$ if $|\omega|/k \gg N^{2}$.
\end{enumerate}
\end{theorem}

The first estimate is compatible with the Laughlin state staying an (approximate) ground state, and we indeed believe this is the case. On the contrary, in the other two cases, the theorem shows that the ground state has a strictly larger angular momentum (for $\ell = 2$, the Laughlin state has angular momentum $N(N-1)$). There should  be a phase transition associated with the passage from a strictly increasing trap when $\omega >0$ to a mexican-hat potential when $\omega < 0 $. We are not able at present to characterize more precisely the ground states. A possible track to do so is discussed in Section~\ref{sec:FQHE incomp}. Before proceeding to this, we illustrate the expected behavior by considering trial states. 

\subsection{Behavior of trial states} The proof of the above results relies heavily on the construction and analysis of suitable trial states. We believe , but cannot prove at present, that those are essentially optimal in the regime of parameters discussed in Theorem~\ref{thm:correl}. The structural changes in these trial functions provides the rationale for the occurrence of different regimes, and in particular for the changes in the total angular momentum of the ground state obtained in Theorem~\ref{thm:momentum}.

As noticed above, the one-body potential of our problem changes from having a local minimum at the origin to having a local maximum when $\om$ is decreased. We  need to have some flexibility in the matter density of our trial states to adapt to this behavior, which leads us to the form 
\begin{equation}\label{eq:trial states}
\PsiGV _m = c_m \prod_{j=1} ^N z_j ^m \prod_{1 \leq i<j\leq N} \left( z_i - z_j\right) ^2 e^{-\sum_{j=1} ^N |z_j| ^2 / 2}
\end{equation}
were $c_m$ is a normalization constant. One recovers the pure Laughlin state for $m=0$ and the states with $m>0$ are usually referred to as Laughlin quasi-holes  (whence the label $\rm qh$). The factor $\prod_{j=1} ^N z_j ^m$ is interpreted as a multiply quantized vortex located at the origin. Its role is to deplete the density of the state when $\omega < 0$ to reduce potential energy. We shall be lead to taking $m$ rather large, and will thus also refer to the above functions as ``Laughlin + giant vortex'' states.

Our task is to calculate the potential energy 
\begin{equation}\label{eq:intro pot ener}
N \int_{\R ^2} \pot (r) \rhoP (z) dz 
\end{equation}
of our trial states, with $\rhoP$ as in~\eqref{eq:QHE density}. Given a candidate trial state we thus want to evaluate precisely what the corresponding matter density is. It is convenient to work with scaled variables, defining (we do not emphasize the dependence on $m$)
\begin{equation}\label{eq:intro scaling}
\muN (Z) := N ^N \left| \PsiGV_m (\sqrt{N} Z )\right| ^2.
\end{equation}
We compare the one-body density corresponding to \eqref{eq:intro scaling}, i.e.
\[
\muNone (z) = \int_{\R ^{2(N-1)}} \muN (z,z_2,\ldots,z_N) dz_2\ldots dz_N, 
\]
with the minimizer $\rhoMF$ of the mean-field free energy functional
\begin{equation}\label{eq:intro MFf}
\MFf [\rho] = \int_{\R ^2}  \Vm \rho + 2 D(\rho,\rho) + N ^{-1} \int_{\R ^2} \rho \log \rho 
\end{equation}
amongst probability measures on $\R ^2$. Here 
\[
\Vm (r) = r ^2 - 2 \frac{m}{N} \log r  
\]
and the notation 
\begin{equation}\label{eq:2D Coulomb}
D(\rho,\rho) = -\iint_{\R^2\times \R^2} \rho(x) \log |x-y|  \rho(y) dxdy
\end{equation}
stands for the 2D Coulomb energy. Our results on the trial states' density profiles may be summarized as follows:

\begin{theorem}[\textbf{Density estimates for Laughlin + giant vortex trial states}]\label{theo:dens QH trial}\mbox{}\\
There exists a constant $C>0$ such that for large enough $N$ and any smooth function $V$ on $\mathbb R^2$
\begin{equation}\label{densitydiff}
\left\vert \int_{\R^2} V \left(\muNone - \rhoMF\right)  \right\vert \leq  CN^{-1/2} \log N \Vert \nabla V \Vert_{L ^2 (\R ^2)} + C N ^{-1/2}\Vert \nabla V \Vert_{L ^{\infty} (\R ^2)}
\end{equation}
if $m\lesssim N ^2$, and  
\begin{equation}\label{densitydiff2}
\left\vert \int_{\R ^2} V \left(\muNone - \rhoMF \right)  \right\vert \leq  CN^{1/2} m ^{-1/4} \Vert V \Vert_{L ^{\infty} (\R ^2)}
\end{equation}
if $m\gg N ^2$.
\end{theorem}

Our main tool is the well-known plasma analogy, originating in \cite{Laughlin-83,Laughlin-87}, wherein the density of the Laughlin state is interpreted as the Gibbs measure of a classical 2D Coulomb gas (one component plasma). More precisely, after the scaling of space variables, one can identify the $N$-particle density of the Laughlin state with the Gibbs measure of a 2D jellium with mean-field scaling, that is a system of $N$ particles in the plane interacting via weak (with a prefactor $N^{-1}$) logarithmic pair-potentials and with a constant neutralizing background. Within this analogy the  vortex of degree $m$ in \eqref{eq:trial states} is interpreted as an additional point charge pinned at the origin.

Existing knowledge  about the mean-field limit for classical particles then suggests the approximation made rigorous in the above theorem. For our purpose, precise estimates that were not available in the literature are required, so we developed a new strategy for the study of the mean field limit that gives explicit and quantitative estimates on the fluctuations around the mean field density that are of independent interest.

This result allows to evaluate the one-particle density of our trial states in a simple and explicit way and deduce estimates of the potential energy \eqref{eq:intro pot ener}. Optimizing these over $m$, we find 
\begin{equation}\label{eq:intro m opt}
\mopt = \begin{cases}
               0 \mbox{ if } \om \geq - 2 k N \\
               - \frac{\om}{2 k} - N \mbox{ if } \om < - 2 k N.
              \end{cases} 
\end{equation}
We interpret this as a strong indication that, within the fully correlated regime, a transition occurs for $\om < 0$ and $|\om| \propto k N$ between a pure Laughlin state and a correlated state with a density depletion at the origin. Interestingly we also find that the character of the mean-field density $\rhoMF$ and thus that of the one-particle density of the state \eqref{eq:trial states} strongly depends on $m$: For $m\ll N ^2$ it is correctly approximated by a flat density profile located in a disc or an annulus (depending on the value of $m$), whereas for $m\gg N ^2$ the density profile is approximately a radial Gaussian centered on some circle. This is due to a transition from a dominantly electrostatic to a dominantly thermal behavior of the 2D Coulomb gas to which we compare our trial states. Using the expression of the optimal value of $m$ given in \eqref{eq:intro m opt} this suggests a further transition in the ground state of \eqref{eq:intro LLLe} in the regime $|\om|\propto k N ^2$ . Establishing 
these phenomena rigorously remains a challenging open problem, but they provide a rationale for the estimates on the true ground states that we stated previously.

\section{Incompressibility estimates for the Laughlin phase}\label{sec:FQHE incomp} 

In the previous section we were interested in a cold-atoms physics situation where ground states should be of the form~\eqref{eq:Ker}, and we were looking for a factor $F$ (or more simply, a $f$ in~\eqref{eq:Laughlin phase 2}) that would optimize the potential energy in a trapping potential. As discussed earlier this kind of situation also occurs in FQHE physics proper (i.e. in 2D electron gases) where the one-body potential accounts for trapping, but also, more importantly, for disorder. The recurrence of this kind of problem has lead us to consider in the papers~\cite{RouYng-14,RouYng-15}  a general framework for studying such questions. We summarize this work in this section. In short, the programme is:
\begin{enumerate}
 \item Assume, as is often done in FQHE physics, that the interaction is strong enough to force the ground state to live in~\eqref{eq:Ker} for some $\ell$.
 \item Assume further, in first approximation, that choosing such a form renders the interaction energy negligible.
 \item Consider then minimizing the remaining part of the energy, namely that due to the one-body potential, within the class~\eqref{eq:Ker}. 
 \item What can be said in general of the minimization problem, of the ground state energy ? 
 \item What general features do the one-body densities of states of the form~\eqref{eq:Ker} have in common ?
 \item In particular, can one argue rigorously for some form of incompressibility, as expected from physical arguments ?
 \item Can one show that, in this simplified context, minimizers always (approximately) take the form~\eqref{eq:Laughlin phase 2} ?
\end{enumerate}
The physical motivation comes from Items 6 and 7. Incompressibility is expected to be an essential feature of FQHE states, and trial states of the form~\eqref{eq:Laughlin phase 2} play a very prominent role in the theory. Note that the trial states~\eqref{eq:trial states} we have discussed previously are of this form. The above programme, if completed, would confirm that they are energetically optimal, as we discuss further below.

\subsection{A new variational problem} The programme sketched above leads us to the study of a very simple energy functional. Given a one-body potential $V:\R ^2 \to \R $ we consider the energy of a state $\Psi \in \Ker$
\begin{equation}\label{eq:start energy}
\E [\Psi] = N \int_{\R ^2} V (z) \rhoP (z) dz  
\end{equation}
depending only on the 1-particle probability density $\rhoP$ 
\begin{equation}\label{eq:intro density}
\rhoP (z):= \int_{\R ^{2(N-1)}} |\Psi (z,z_2,\ldots,z_N)| ^2 dz_2 \ldots dz_N. 
\end{equation}
We will be interested in studying the ground state problem of minimizing~\eqref{eq:start energy}, \emph{within the Laughlin phase $\Ker$.} We take for granted the reduction to the LLL and cancellation of the interactions by the vanishing of the wave function along the diagonals of configuration space. We wish to see whether the Laughlin state, or a close variant, emerges as the natural ground state in a given potential. In particular we investigate the robustness of the correlations of Laughlin's wave function when the trapping potential is varied. 

In view of~\eqref{eq:Ker} we are looking for a wave-function of the form
$$
\Psi_F (z_1,\ldots,z_N)= c_F \PsiLau (z_1,\ldots,z_N) F (z_1,\ldots,z_N)
$$
where $F\in\BargN$ and $c_F$ is a normalization factor. A natural {\it guess} is that, whatever the one-body potential, the correlations stay in the same form and the ground state is well-approximated by a wave function 
\begin{equation}\label{eq:intro guess}
\Psi (z_1,\ldots,z_N) = c_{f} \PsiLau (z_1,\ldots,z_N) \prod_{j=1} ^N f (z_j)  
\end{equation}
where the additional holomorphic factor $F$ that characterizes functions of $\Ker$ is uncorrelated. Indeed, the energy functional dictating the form of $\Psi$ contains only one-body terms once the form~\eqref{eq:Ker} is assumed. It thus does not seem favorable to correlate the state more than necessary. Although intuitively appealing, this reasoning is too simplistic, and the emergence of the ansatz~\eqref{eq:intro guess} is far from trivial. This is because~\eqref{eq:start energy} is a one-body functional in terms of $\Psi$, but the correlation factor $F$ itself  really sees an effective, complicated, many-body Hamiltonian. This is due to the factor $\PsiLau$ it is combined with to form the state $\Psi_F$.

\subsection{Incompressibility of the Laughlin phase} 
The energy functional~\eqref{eq:start energy} is of course very simple and all the difficulty of the problem lies in the intricate structure of the variational set~\eqref{eq:Ker}. The expected rigidity of the strongly correlated states of~\eqref{eq:Ker} 
should manifest itself through the property that their densities are essentially bounded above by a universal constant
\begin{equation}\label{eq:intro incomp bound}
\rhoP \lessapprox \frac{1}{\pi\ell N} \mbox{ for any } \Psi \in \Ker. 
\end{equation}
This is the incompressibility notion we will investigate, in the limit $N\to \infty$. In view of existing numerical computations of the Laughlin state~\cite{Ciftja-06}, Estimate~\eqref{eq:intro incomp bound} can hold only in some appropriate weak sense. We now discuss the approach we follow, formulated in terms of the previous variational problem. 

It is a well-known fact~\cite{Laughlin-83} (a rigorous justification follows from Theorem~\ref{theo:dens QH trial} above) that the one-particle density of the Laughlin state is approximately constant over a disc of radius $\sim \sqrt{N}$ and then quickly drops to $0$. It is thus natural to also consider external potentials that live on this scale, which amounts to scale space variables\footnote{That we scale lengths by a factor $\sqrt{N-1}$ instead of $\sqrt{N}$ is of course irrelevant for large~$N$. It only serves to simplify some expressions. 
} and consider rather the energy functional
\begin{equation}\label{eq:scale ener}
\E_N [\Psi] = (N-1) \int_{\R^2} V\left(\xbf\right) \rhoP \left(\sqrt{N-1} \: \xbf \right) 
\end{equation}
where $\Psi$ is of the form~\eqref{eq:Ker} and $\rhoP$ is the corresponding matter density. We shall be concerned with the large $N$ behavior of the ground state energy
\begin{equation}\label{eq:energy Nn}
E(N) := \inf \left\{ \E_N [\Psi_F], \; \Psi_F \in \mathcal L^N_\ell,\; \norm{\Psi}_{L^2} = 1 \right\}.
\end{equation} 
By scaling, the bound
$$ (N-1) \rhoP \left(\sqrt{N-1} \: \xbf \right) \lesssim \frac{1}{\pi \ell}$$
follows from the conjectured estimate~\eqref{eq:intro incomp bound}. We thus expect the bath-tub energy~\cite[Theorem~1.14]{LieLos-01}
\begin{equation}\label{eq:bath tub}
E_V (\ell):= \inf\left\{ \int_{\R^2} V \rho \: | \: \rho \in L^1 (\R ^2), \: 0 \leq \rho \leq \frac{1}{\pi \ell},\ \int_{\R^2}\rho=1 \right\}.
\end{equation}
to play a role in this context. We recast~\eqref{eq:intro incomp bound} in a ``weak dual formulation'' by conjecturing that 
\begin{equation}\label{eq:incomp conj}
\liminf_{N\to \infty} E(N) \geq E_V (\ell) 
\end{equation}
for any (reasonable) one-body potential $V$. More generally, introducing
\begin{equation}\label{eq:bath tub C}
E_V \left(\ell / C \right):= \inf\left\{ \int_{\R^2} V \rho \: \Big| \: \rho \in L^1 (\R ^2), \: 0 \leq \rho \leq \frac{C}{\pi \ell},\ \int_{\R^2}\rho=1 \right\}
\end{equation}
we shall refer to a bound of the form 
\begin{equation}\label{eq:incomp conj C}
\liminf_{N\to \infty} E(N) \geq E_V (\ell/C) 
\end{equation}
as an incompressibility estimate corresponding to a density bound
\begin{equation}\label{eq:intro incomp bound C}
\rhoP \lessapprox \frac{C}{\pi\ell N} \mbox{ for any } \Psi \in \Ker. 
\end{equation}
These bounds state that, whatever the external potential $V$ imposed on the system, it is never possible to create a density bump above a certain maximum value. Doing so might be energetically favorable to accommodate the variations of $V$, but is prohibited by the imposed form $\Psi\in\Ker$.

\subsection{Unconditional incompressibility estimates} We can now state the main result of~\cite{RouYng-15}:

\begin{theorem}[\textbf{Unconditional incompressibility for $\Ker$}]\label{thm:incomp 1}\mbox{}\\
Let $V\in C ^{2} (\R^2)$ be increasing at infinity in the sense that
\begin{equation}\label{eq:increase V}
\min_{|x|\geq R} V(\xbf)\to \infty\quad\text{for}\quad R\to\infty.
\end{equation}
Then 
\begin{equation}\label{eq:main incomp}
\liminf_{N\to \infty} E(N) \geq E_V\left( \ell / 4 \right)
\end{equation}
where the bath-tub energy $E_V\left( \ell / 4 \right)$ is defined as in~\eqref{eq:bath tub C}.
\end{theorem}

As explained above, since this result holds for a very large class of potentials $V$, it can be interpreted as a weak formulation of the density bound 
$$
\rhoP \lessapprox \frac{4}{\pi\ell N} \mbox{ for any } \Psi \in \Ker. 
$$
It thus misses the conjectured optimal bound~\eqref{eq:intro incomp bound} by a factor $4$, but has the correct order of magnitude (mind the scaling in the definition of the energy). To appreciate this we make the

\begin{remark}[Illustrative comparisons]\label{rem:comparisons}\mbox{}\\
How would the energy~\eqref{eq:scale ener} behave in less constrained variational sets ? Three interesting cases are worth considering, keeping in mind that $\min_{\R ^2} V$ sets the energy scale, and can with no loss of generality be chosen to be $0$:
\begin{itemize}
\item \textit{Particles outside the} LLL. Suppose the single-particle Hilbert space was the full space $L^2 (\R ^2)$ instead of the constrained space $\LLL$. The minimization is then of course very simple and we would obtain $E(N) = \min V$ by taking a minimizing sequence concentrating around a minimum point of $V$. 
\item \textit{Uncorrelated bosons in the} LLL. For non interacting bosons in a strong magnetic field, one should consider the space $\bigotimes_\sym ^{N} \LLL$, the symmetric tensor product of $N$ copies of the LLL. The infimum in~\eqref{eq:scale ener} can then be computed considering uncorrelated trial states of the form $f ^{\otimes N}$, $f\in \LLL$. LLL functions do satisfy a kind of incompressibility property 
because they are of the form holomorphic $\times$ gaussian. This can be made precise by the inequality~\cite{AftBlaNie-06b,Carlen-91,LewSei-09}
\begin{equation}\label{eq:LLL incomp}
\sup_{z\in \C} \left|f (z) e^{-|z| ^2 / 2 }\right| ^2 \leq  \left\Vert f(.) e ^{-|\:.\:|/2} \right\Vert_{L ^2 (\R ^2)} ^2. 
\end{equation}
This is a much weaker notion however: it only leads to $\rhoP \leq 1$ for $\Psi\in\Ker$. One may still construct a sequence of the form $f ^{\otimes N}$ concentrating around a minimum point of $V$ without violating~\eqref{eq:LLL incomp}. The liminf in~\eqref{eq:main incomp} is thus also equal to $\min V$  in this case.  
\item \textit{Minimally correlated fermions in the} LLL. Due to the Pauli exclusion principle, fermions can never be uncorrelated: the corresponding wave functions have to be antisymmetric w.r.t. exchange of particles: 
\begin{equation}\label{eq:Pauli}
\Psi(\xbf_1,\ldots,\xbf_i,\ldots,\xbf_j,\ldots,\xbf_N) = -\Psi(\xbf_1,\ldots,\xbf_j,\ldots,\xbf_i,\ldots,\xbf_N). 
\end{equation}
For  LLL wave functions (which are continuous) this implies 
$$\Psi(\xbf_i=\xbf_j) = 0 \mbox{ for any } i\neq j,$$
i.e. the wave function vanishes on the diagonals of the configuration space. Due to the holomorphy constraint, any $N$-body LLL fermionic wave-function is then of the form~\eqref{eq:Ker}, for $\ell = 1$. If $F = f ^{\otimes N}$, one could then talk of ``minimally correlated'' fermions. This case is covered by our theorem, and one obtains $E_V(1/4)$ as a lower bound to the energy. 
\end{itemize}
\hfill\qed
\end{remark}

We call the above an ``unconditional'' result, because it is obtained under the sole assumption that the variational set for the energy $E(N)$ is given by~\eqref{eq:Ker}. This is in contrast with the result we discuss next, where the optimal bound~\eqref{eq:incomp conj} is obtained under an additional assumption.

\subsection{Conditional optimal estimates} We can obtain a sharper result than Theorem~\ref{thm:incomp 1} if we follow Laughlin's original intuition that particles are correlated only pairwise. This means that $F$ contains only two-body correlation factors, i.e. that it can be written in the form
\begin{equation}\label{eq:choice correlations}
F (z_1,\ldots,z_N) = \prod_{j = 1} ^N f_1 (z_j)\prod_{( i,j ) \in \{1,\ldots,N \} } f_2 (z_i,z_j)
\end{equation}
with $f_1$ and $f_2$ two polynomials ($f_2$ being in addition a symmetric function of $z,z'$). We also assume
\beq\label{eq:degree}
\deg (f_1) \leq D N, \quad \deg(f_2) \leq D
\eeq
for some constant $D$ independent of $N$. The degree of $f_2(z,z')$ is here understood as the degree of the polynomial in $z$ with $z'$ fixed (and vice versa).  Assumptions~\eqref{eq:choice correlations} and~\eqref{eq:degree} are, of course, restrictive, but still cover a huge class of functions with possibly very intricate correlations. 

Then, defining 
\begin{multline}\label{eq:var set 2 D}
\VD = \Big\{ F \in \Barg ^N \, : \, \mbox{ there exist } (f_1,f_2) \in \Barg \times \Barg ^2, \deg (f_1) \leq D N, \deg(f_2) \leq D, 
\\  F (z_1,\ldots,z_N) = \prod_{j = 1} ^N f_1 (z_j)\prod_{1\leq i < j \leq N } f_2 (z_i,z_j) \Big\}
\end{multline}
and the corresponding ground state energy
\begin{equation}\label{eq:energy Nn bis}
\ED:= \inf \left\{ \E_N [\Psi_F], \; \Psi_F \mbox{ of the form  (\ref{eq:Ker}) } \mbox{ with } F \in \VD , \norm{\Psi_F}_{L^2} = 1 \right\},
\end{equation}
we have

\begin{theorem}[\textbf{Conditional optimal incompressibility estimates}]\label{thm:incomp 2}\mbox{}\\
Let $V\in C ^{2} (\R^2)$ be increasing at infinity as in Theorem~\ref{thm:incomp 1}. Then 
\begin{equation}\label{eq:main incomp bis}
\liminf_{N\to \infty} \ED \geq E_V(\ell).  
\end{equation}
\end{theorem}

This is a weak formulation of the optimal density bound~\eqref{eq:intro incomp bound}, under a reasonable but unproven assumption on the states we minimize over. To see that the lower bound~\eqref{eq:main incomp bis} is indeed optimal and can, at least in some cases, be achieved with a function of the form~\eqref{eq:Laughlin phase 2}, we state a result that follows rather straightforwardly from the analysis of ``quasi-holes trial states'' in~\cite{RouSerYng-13b} (see Theorem~\ref{theo:dens QH trial}):

\begin{corollary}[\textbf{Optimization of the energy in radial potentials}]\label{cor:radial}\mbox{}\\
Let $V:\R ^2 \mapsto \R$ be as in Theorems~\ref{thm:incomp 1} and~\ref{thm:incomp 2}. Assume further that $V$ is radial, has at most polynomial growth at infinity,  and satisfies one of the two following assumptions 
\begin{enumerate}
\item $V$ is radially increasing,
\item $V$ has a ``mexican-hat'' shape: $V(|\xbf|)$ has a single local maximum at the origin and a single global minimum at some radius $R$.
\end{enumerate}
Then for $D$ large enough we have
\begin{equation}\label{eq:incomp opt}
\limsup_{N\to \infty} E(N) \leq \lim_{N\to \infty} \ED = E_V(\ell).  
\end{equation}
More precisely, in case $(1)$ 
\begin{equation}\label{eq:trial increas}
\E_N [\PsiLau] \to  E_V(\ell) \mbox{ when } N\to \infty
\end{equation}
and in case $(2)$ one can find a fixed number $\mb \in \R$ and a sequence $m(N) \in \N$ with $m \sim \mb N $ as $N\to \infty$ such that, defining 
\begin{equation}\label{eq:trial mexican 1}
\Psi_m (z_1,\ldots,z_N) := c_m \PsiLau (z_1,\ldots,z_N) \prod_{j=1} ^N z_j ^m 
\end{equation}
we have 
\begin{equation}\label{eq:trial mexican 2}
\E_N [\Psi_m] \to E_V(\ell) \mbox{ when } N\to \infty.
\end{equation}
\end{corollary}

Note that the potentials considered in Section~\ref{sec:FQHE Bose gas} are of the form above. That the Laughlin state itself is always (asymptotically) optimal in case $1$, whatever the rate of growth of $V$, is a strong manifestation of its rigidity. The only way to accommodate a steeper potential barrier by creating a density maximum around the origin is to leave the Laughlin phase~\eqref{eq:Ker}, at the expense of increasing either the kinetic energy or the interaction energy. 

\subsection{Hints of proofs} 
As for the proof of Theorem~\ref{theo:dens QH trial}, the starting point is Laughlin's plasma analogy, which is reminiscent of the log-gas analogy in random matrix theory~\cite{AndGuiZei-10,Dyson-62a,Forrester-10,Ginibre-65,Mehta-04}. The crucial observation~\cite{Laughlin-83,Laughlin-87} is that the absolute square of the wave-function~\eqref{eq:intro Laughlin} can be regarded as the Gibbs state of a 2D Coulomb gas (one-component plasma). To prove Theorems~\ref{thm:incomp 1} and~\ref{thm:incomp 2}, we generalize this idea and map  any state of $\mathcal{L}_\ell ^N$ to a Gibbs state of a classical Hamiltonian.

Let us consider a function of the form
\begin{equation}\label{eq:start state}
\Psi_F = c_F \PsiLau (z_1,\ldots,z_N) F(z_1,\ldots,z_N) 
\end{equation}
where $\PsiLau$ is the Laughlin state \eqref{eq:intro Laughlin}, $F$ is holomorphic and symmetric and $c_F$ is a normalization factor. The idea is simply to write 
$$ |\Psi_F| ^2 = c_F ^2 \exp\left( - 2\log |\PsiLau| -2 \log |F| \right)$$
and interpret $2\log |\PsiLau| + 2\log |F|$ as a classical Hamilton function. One of the reasons why this is an effective procedure is that we can take advantage of the good scaling properties of $|\PsiLau| ^2$ to first change variables and obtain a classical Gibbs states \emph{with mean-field two-body interactions} and \emph{small effective temperature}. Specifically, we define
\begin{equation}\label{eq:Gibbs state}
\mu_N (Z) := (N-1)^{N} \left| \Psi_F \left(\sqrt{N-1} \: Z\right) \right| ^2 = \frac{1}{\ZN} \exp \left(-\frac{1}{T} H_N (Z) \right)
\end{equation}
where $\ZN$ ensures normalization of $\mu_N$ in $L^1(\R ^{2N})$,
$$T=\frac{1}{N},$$
and the classical Hamiltonian is of the form
\begin{equation}\label{eq:class Hamil}
H_N (Z) = \sum_{j=1} ^N |z_j| ^2 + \frac{2\ell}{N-1} \sum_{1 \leq i<j \leq N } w(z_i-z_j) + \frac{1}{N-1} W(Z).
\end{equation}
We have here written
\begin{equation}\label{eq:coul pot}
w(z) := - \log |z| 
\end{equation}
for the 2D Coulomb kernel and defined
\begin{equation}\label{eq:weird pot}
W(Z) := - 2 \log \left| F\left( \sqrt{N-1} \: Z  \right) \right|. 
\end{equation}
We are ultimately interested in an incompressibility bound of the form 
\begin{equation}\label{eq:formal incomp}
\mu_N ^{(1)} \lesssim \frac{{\rm const.}}{\pi \ell} 
\end{equation}
in an appropriate weak sense, independently of the details of $W$. Here, the $n$-th marginal of $\mu_N$ is defined as
$$ \mu_N ^{(n)}(z_1,\ldots,z_n):= \int_{\R ^{2{N-n}}} \mu_N (z_1,\ldots,z_N)dz_{n+1}\ldots dz_N.$$

The function \eqref{eq:weird pot} can be rather intricate and represents in general a genuine $N$-body interaction term of the Hamiltonian \eqref{eq:class Hamil}. The only thing we know a priori about $W$ is that it is \emph{superharmonic in each of its variables}:
\begin{equation}\label{eq:subharm}
-\Delta_{z_j} W \geq 0 \quad \forall j=1\ldots N 
\end{equation}
which follows from the fact that $F$ is holomorphic. Under the additional assumption that $F$ belongs to $\VD$, it is possible to regard $W$ as a few-body interaction in a mean-field like scaling. Then, a mean-field approximation for $\mu_N$ can be rigorously proven to be accurate: We first justify that an ansatz of the form
\begin{equation}\label{eq:MF ansatz}
\mu_N \approx \rho ^{\otimes N}, \quad \rho \in L ^1 (\R ^2) 
\end{equation}
is effective for deriving  density bounds on $\mu_N ^{(1)}$. We then show that the appropriate $\rho$ should minimize an effective mean-field energy functional, and thus satisfy a variational equation, depending on $W$. Using this equation jointly with~\eqref{eq:subharm} we obtain that the optimal $\rho$ always satisfies
$$\rho \leq \frac{1}{\pi \ell}$$  
and Theorem~\ref{thm:incomp 2} follows. The difficult part in this approach is the justification of the ansatz~\eqref{eq:MF ansatz}. In the general case there does not seem to be any reason why it should be a good approximation, and another strategy is called for. 

There is another way to bound the one-particle density in ground states of Coulomb systems without using any mean-field approximation. The argument is due to Lieb (unpublished), and variants thereof have been used recently in~\cite{PetSer-14,RotSer-14,RouSer-14}. It is based solely on some properties of the Coulomb kernel and general superharmonicity arguments, so that it can be adapted to Hamiltonians of the form~\eqref{eq:class Hamil}. Our strategy to prove Theorem~\ref{thm:incomp 1} is then the following:
\begin{itemize}
\item We adapt Lieb's argument to obtain a bound on the {\it minimal separation of points} in the ground state configurations
%{\red for a modification of the classical Hamiltonian~\eqref{eq:class Hamil}, including a small perturbation}
. This implies a local density upper bound at the level of the ground state,  but not yet the Gibbs state \eqref{eq:Gibbs state}.
\item We exploit the fact that~\eqref{eq:Gibbs state} is a Gibbs state for $H_N$ with small temperature $T \to 0$ in the limit $N\to \infty$. It  is thus reasonable to expect the density bound on the ground state to also apply to the Gibbs state.
\item More precisely, appropriate upper and lower bounds to the free-energy $- T \log \ZN$ confirm that it is close to the ground state energy of $H_N$ in the limit $T\to 0$. Applying these bounds to a suitably perturbed Hamiltonian gives the desired estimates on the density by a Feynman-Hellmann type argument, i.e. by evaluating the derivative of the free-energy with respect to the perturbation.
\end{itemize}

\chapter{On the classical Coulomb gas}\label{sec:Coulomb gas}

Classical Coulomb systems are most central objects (see~\cite{Serfaty-14,Serfaty-15} for reviews) in statistical mechanics, random matrix theory and fractional quantum Hall physics (see Chapter~\ref{sec:FQHE states} for this latter aspect). They can be seen as a toy model for matter, containing the truly long-range nature of electrostatic interactions. As was pointed out by Wigner \cite{Wigner-55,Wigner-67} and Dyson \cite{Dyson-62a,Dyson-62b,Dyson-62c}, two-dimensional Coulomb systems are directly related to Gaussian random matrices, more precisely the Ginibre ensemble, and such random matrix models have also received much attention for their own sake. A similar connection exists between ``log-gases" in dimension 1 and the GUE\footnote{Gaussian Unitary Ensemble.} and GOE\footnote{Gaussian Orthogonal Ensemble.} ensembles of random matrices, as well as more indirectly to orthogonal polynomial ensembles. For more details on these aspects we refer to \cite{Forrester-10}, and for an introduction to the random 
matrix aspect to the texts \cite{AndGuiZei-10,Mehta-04}. 

\medskip 
 
In this chapter we discuss the main results of~\cite{RouSer-14}. We shall be concerned with the mean-field regime for many-particles trapped Coulomb gases in any dimension $d\geq 2$, thus in particular the most physically relevant case $d=3$, and the case of interest for random matrix theory, vortices in superfluids and quantum Hall physics $d=2$. Some of the methods can be applied to the one-dimensional log-gas (where the interaction kernel is $- \log |x|$) see~\cite{SanSer-14a}. Our goal is to characterize, in terms of a new object, the renormalized jellium energy, the fluctuations around mean-field theory in low temperature equilibrium states. Results include 
\begin{itemize}
 \item a subleading order expansion of the ground state energy and asymptotic results on corresponding equilibrium configurations that reveal a two scale structure, 
 \item a characterization of the microscopic distribution of charges in terms of the jellium problem,
 \item similar results for the Gibbs states at sufficiently low temperature, together with free-energy estimates leading to a natural conjecture on the liquid to crystal phase transition,
 \item precise estimates on the accuracy of mean-field theory for different observables of the systems.
\end{itemize}
We refer to~\cite{PetSer-14,RotSer-14,Leble-15,Leble-15b,Leble-16,LebSer-15} for more recent developments.

\medskip

We start by a short introduction to the problem in Section~\ref{sec:Coul MF}, briefly reviewing what is known about the mean-field limit on a macroscopic scale. We then heuristically discuss the physics that should be expected to govern the microscopic scale in Section~\ref{sec:Coul Jellium}. This leads us to defining the  jellium renormalized energy, in a new, more flexible way (as opposed to the previous works~\cite{SanSer-12,SanSer-14}). We then state our main results in two separate subsections, one concerning the ground state problem, Section~\ref{sec:Coul GS}, one concerning the equilibrium states at positive temperature~\ref{sec:Coul Gibbs}. Finally, we state results quantifying the precision of mean-field theory, obtained as corollaries of the main line of attack, in Section~\ref{sec:Coul fluctu}.

\section{The Coulomb gas in the mean-field limit}\label{sec:Coul MF} 

Denoting $x_1,\ldots,x_n$ the positions of the particles, the total energy at rest of our system is given by the Hamiltonian 
\begin{equation}\label{wn}
\w(x_1, \dots, x_n)=   \sum_{i\neq j} \g (x_i-x_j)  +n \sum_{i=1}^n V(x_i)
\end{equation}
where 
\begin{equation}\label{eq:coul ker}
\begin{cases}
\displaystyle \g(x)=\frac{1}{|x|^{d-2}}& \text{if}\  d \ge 3\\
\displaystyle \g(x)= -   \log |x|& \text{if } \ d=2
\end{cases}
\end{equation} 
is  a multiple of the Coulomb potential in dimensions $d\ge 2$, i.e. we have
\begin{equation}
-\Delta \g= c_d \delta_0  \end{equation}
with
\begin{equation}\label{defc}
c_2 = 2\pi, \qquad c_d = (d-2)|\mathbb{S}^{d-1}| \ \text{when} \ d\ge 3
\end{equation}
and $\delta_0$ is the Dirac mass at the origin. The one-body potential $V:\mr^d \to \mr$ is a continuous function, growing at infinity (\emph{confining} potential). More precisely, we assume 
\begin{equation}\label{eq:trap pot}
\begin{cases}
\displaystyle \lim_{|x|\to \infty} V(x) = +\infty  & \text{ if }  d \ge 3\\
\displaystyle \lim_{|x|\to \infty} \left(\frac{V(x)}{2} - \log |x|\right) = +\infty  & \text{ if }  d=2.
\end{cases}
\end{equation} 
For simplicity, we assume in this chapter that $V$ is smooth, see~\cite{RouSer-14} for the optimal assumptions. Note the factor $n$ in front of the one-body term (second term) in \eqref{wn} that puts us in a mean-field scaling where the one-body energy and the two-body energy (first term) are of the same order of magnitude. This choice is equivalent to demanding that the pair-interaction strength be of order $n^{-1}$. By scaling, one can always reduce to this situation in the particular case where the trapping potential $V$ has some homogeneity, which is particularly important in applications. 

We are interested in equilibrium properties of the system in the regime $n\to \infty$, that is on the large particle number asymptotics of the ground state and the Gibbs state at given temperature. In the former case we consider configurations $(x_1,\ldots,x_n)$ that minimize the total energy \eqref{wn}. We will denote 
\begin{equation}\label{eq:gse}
E_n := \min_{\R ^{dn}} \w 
\end{equation}
the ground state energy. To characterize ground state configurations we shall consider the asymptotics of the empirical measure
\begin{equation}\label{eq:coul emp meas}
\mu_n := \frac{1}{n} \sum_{j=1} ^n \delta_{x_j}. 
\end{equation}
At positive temperature we study the Gibbs state at inverse temperature $\beta$, i.e. the probability law\footnote{Remark the conventions in the units, chosen to comply with standard choices in random matrix theory.}
\begin{equation}\label{eq:defi Gibbs state}
\Gibbs (x_1,\ldots,x_n) =\frac{1}{\Zbeta} e^{-\frac{\beta}{2} \w (x_1,\ldots,x_n)} \, dx_1 \ldots dx_n
\end{equation}
where $\Zbeta$ is a normalization constant (partition function). By definition, the Gibbs measure minimizes the $n$-body free energy functional
\begin{equation}\label{eq:free ener N}
\Fnbeta [\mubf] := \int_{\R ^{dn}} \mubf(\xbf) H_n (\xbf) d\xbf + \frac{2}{\beta} \int_{\R ^{dn}} \mubf(\xbf) \log (\mubf (\xbf)) d\xbf
\end{equation}
over probability measures $\mubf \in \P (\R ^{dn})$. We also have the standard relation 
\begin{equation}\label{eq:free ener N min}
\Fnbetae := \inf _{\mubf \in \P (\R ^{dn})} \Fnbeta [\mubf]  =  \Fnbeta [\Gibbs] = - \frac{2}{\beta} \log \Zbeta.
\end{equation}
between the minimal free-energy and the partition function $\Zbeta$. 

In the mean-field regime of our interest, it is fairly standard to derive the leading-order behavior of these objects when $n\to \infty$, as we now recall.

\subsection{The mean-field regime: continuum limit} It is well-known that to leading order 
\begin{equation}\label{eq:gse first order}
E_n = n ^2 \En[\mu_0] (1+o(1)) 
\end{equation}
in the limit $n\to \infty$ where 
\begin{equation}\label{eq:def MF ener}
\En[\mu] = \iint_{\mr^d\times \mr^d} \g(x-y) \, d\mu(x)\, d\mu(y)+ \int_{\mr^d}V(x)\, d\mu(x)
\end{equation}
is the mean-field energy functional defined for Radon measures $\mu$, and $\mu_0$ (the equilibrium measure) is the minimizer of $\En$ amongst probability measures on $\R ^d$. The functional \eqref{eq:def MF ener} is nothing but the continuum energy corresponding to \eqref{wn}: the first term is the classical Coulomb interaction energy of the charge distribution $\mu$ and the second the potential energy in the potential $V$. If the interaction potential $\g$ was regular at the origin we could use the empirical measure~\eqref{eq:coul emp meas} to write 
\begin{align*}
\w (x_1,\ldots,x_n) &= n ^2 \left( \int_{\mr^d}V(x)\, d\mu_n(x) + \iint_{\mr^d\times \mr^d} \g(x-y) \, d\mu_n(x)\, d\mu_n(y)\right) - n \g(0) \\
&= n ^2 \En[\mu_n] \left( 1+ O(n ^{-1})\right)
\end{align*}
and \eqref{eq:gse first order} would  easily follow from a simple compactness argument. In the case where $\g$ has a singularity at the origin, a regularization procedure is needed but this mean-field limit result still holds true, meaning that  for minimizers of $\w$, the empirical measure $\mu_n$ converges to $\mu_0$, the minimizer of \eqref{eq:def MF ener}:
\begin{equation}\label{eq:coul gs first}
\mu_n \wto \mu_0 \mbox{ when } n\to \infty
\end{equation}
as measures. This is standard and can be found in  a variety of sources. The physics here is that (at least if $V$ is reasonably regular), the many charges get confined in a bounded region of space and their distribution thus converges to a continuum limit, given by minimizing the continuum energy functional~\eqref{eq:def MF ener}.

At small enough temperature (large enough $\beta$) one should essentially expect that the above results hold $\Gibbs$-almost surely. One can in fact derive a large deviation principle vindicating this point, see~\cite{BenGui-97,BenZei-98,ChaGozZit-13,Serfaty-15}. For larger temperatures, the typical behavior of the system is more easily guessed through an alternate point of view on the mean-field regime, that we discuss next.

\subsection{The mean-field regime: vanishing of correlations} Another way to think of the mean-field limit, less immediate in the present context but more suited for generalizations (see e.g. Chapter~\ref{sec:bosons GS}), is as a regime of low correlations. In reality, particles are indistinguishable, and the configuration of the system should thus be described by a probability measure $\mubf(\xbf) = \mubf (x_1,\ldots,x_n)$, which is  symmetric under particle exchange:
\begin{equation}\label{eq:symmetry}
\mubf(x_1,\ldots,x_n) = \mubf (x_{\sigma(1)},\ldots,x_{\sigma(n)}) \mbox{ for any permutation } \sigma.  
\end{equation}
A ground state configuration $\mubf_n$ is found by minimizing the $n$-body energy functional
\begin{equation}\label{eq:ener N}
\In [\mubf] := \int_{\R ^{dn}}  H_n (\xbf) \mubf(d\xbf)  
\end{equation}
amongst symmetric probability measures $\mubf \in \P_\sym (\R ^{dn})$. It is immediate to see that $\mubf_n$ must be a convex superposition of measures of the form $\delta_{(x_1,\ldots,x_n)}$ with $(x_1,\ldots,x_n)$ minimizing~$\w$ (in other words it has to be a symmetrization of some $\delta_{\xbf}$ for a minimizing configuration $\xbf$). The infimum of the functional \eqref{eq:ener N} of course coincides with 
\begin{equation*}
\inf_{\mubf \in \P_\sym (\R ^{dn})} \int_{\R ^{dn}} H_n (\xbf) \mubf(d\xbf)  = E_n 
\end{equation*}
and a way to understand the asymptotic formula \eqref{eq:gse first order} is to think of the minimizing $\mubf_n$ as being almost factorized 
\begin{equation}\label{eq:intro factor}
\mubf_n (x_1,\ldots,x_n) \approx \rho ^{\otimes n} (x_1,\ldots,x_n) = \prod_{j=1} ^n \rho (x_i)
\end{equation}
with a regular probability measure $\rho \in \P (\R ^d)$. Plugging this ansatz into \eqref{eq:ener N} we indeed obtain 
\[
\In [\rho ^{\otimes n}] =  n ^2 \En[\rho] \left( 1+ O(n ^{-1})\right)
\]
and the optimal choice is $\rho=\mu_0$. The mean-field limit can thus  also be understood as one where correlations amongst the particles of the system vanish in the limit $n\to \infty$, which is the meaning of the factorized ansatz. 

This point of view readily generalizes to the Gibbs state. Imagine that $\Gibbs$ approximately factorizes: 
\begin{equation}\label{eq:intro factor bis}
\Gibbs (x_1,\ldots,x_n) \approx \rho ^{\otimes n} (x_1,\ldots,x_n) = \prod_{j=1} ^n \rho (x_i)
\end{equation}
and insert this ansatz in the free-energy functional~\eqref{eq:free ener N}. One obtains the mean-field free-energy functional
\begin{equation}\label{eq:MF free ener func}
\F [\rho] =  \En [\rho] + \frac{2}{n\beta} \int \rho \log \rho
\end{equation}
and the optimal choice of $\rho$ in~\eqref{eq:intro factor bis} should minimize this object. It is then clear that for $n\beta \gg 1$, the entropy term becomes negligible and we should expect~\eqref{eq:intro factor bis} to hold with $\rho = \mu_0$. At low temperature we thus see no difference between Gibbs states and ground states. For $n\beta \lesssim 1$ however the entropy is \emph{not} negligible. Then~\eqref{eq:intro factor bis} should hold with $\rho = \mubet,$ the minimizer of~\eqref{eq:MF free ener func}. One way to give a rigorous meaning to this is to consider the marginals   
\begin{equation}\label{eq:marginal Gibbs}
\Qk (x_1,\ldots, x_k) = \int_{ \xbf' \in \R ^{d(n-k)}}  \Gibbs (x_1,\ldots, x_k,\xbf') \, d\xbf',
\end{equation}
interpreted as the probability density for having one particle at $x_1$, one particle at $x_2$, $\ldots$, and one particle at $x_k$. Then one can show~\cite{CagLioMarPul-92,MesSpo-82,Kiessling-89,Kiessling-93,KieSpo-99,Rougerie-LMU} that for $n\beta$ and $k$ fixed 
\begin{equation}\label{eq:coul MF beta}
\Qk \wto \mubet ^{\otimes k} \mbox{ as measures when } n\to \infty 
\end{equation}
and that for $n\beta \to \infty$ and $k$ fixed 
\begin{equation}\label{eq:coul MF 0}
\Qk \wto \mu_0 ^{\otimes k} \mbox{ as measures when } n\to \infty. 
\end{equation}
This essentially implies (see~\cite{HauMis-14,Mischler-11,MisMou-13} and references therein for discussions of this point) that for $\Gibbs$-almost every $(x_1,\ldots,x_n)$ the empirical measure converges to $\mubet$ in case~\eqref{eq:coul MF beta} and to $\mu_0$ in case~\eqref{eq:coul MF 0}.

\section{Beyond mean-field: the renormalized jellium energy}\label{sec:Coul Jellium}

The main results of~\cite{RouSer-14} concern the behavior of the Coulomb gas beyond the mean-field asymptotics we just recalled. In this section we motivate, and then define rigorously, an energy functional that is the cornerstone of our approach. Previous versions of this object had been defined, in 2D only, by Sandier and Serfaty~\cite{SanSer-12,SanSer-14}. The main novelty here is to propose an alternate definition which is more wieldy than the original one, in that it easily generalizes to higher dimensions and significantly simplifies some proofs.

\subsection{Heuristics}

As already mentioned, points minimizing $\w$ tend to be densely packed in a bounded region of space (the support of $\mu_0$, that we shall denote $\Sigma$) in the limit $n\to \infty$. Their distribution (i.e. the empirical measure) has to follow $\mu_0$ on the macroscopic scale but this requirement still leaves a lot of freedom on the configuration at the {\it microscopic scale}, that is on length scales  of order $n ^{-1/d}$ (the mean inter-particle distance). A natural  idea is thus to blow-up at scale $n ^{-1/d}$ in order to consider configurations where points are typically separated by distances of order unity, and investigate which microscopic configuration is favored. On such length scales, the equilibrium measure $\mu_0$ varies slowly so the 
points will want to follow a constant density given locally by the value of $\mu_0$. Since the problem is electrostatic in nature it is intuitive that the correct way to measure the distance between the configuration of points and the local value of the equilibrium measure should use the Coulomb energy. This leads to the idea that the local energy around a blow-up origin should be the electrostatic energy of what is often called a {\it  jellium} in physics: an infinite collection of interacting particles in a constant neutralizing background of opposite charge, a model originally introduced in \cite{Wigner-34}. At this microscopic  scale the pair-interactions will no longer be of mean-field type, their strength will be of order $1$.  A splitting formula will allow us to separate exactly the Coulomb energy of this jellium as the next to leading order term, except that what will come out is more precisely some average of all the  energies of the jellium configurations obtained after blow-up around all possible 
origins.

Of course, it is a  delicate matter to define the energy of the infinite jellium in a mathematically rigorous way: one has to take into account the pair-interaction energy of infinitely many charges, without assuming any local charge neutrality, and the overall energy may be finite only via screening effects between the charges and the neutralizing background that are difficult to quantify. This has been done  for the first time in 2D in \cite{SanSer-12,SanSer-14}, the energy functional for the jellium being  the renormalized energy $W$ alluded to above.  As already mentioned, one of the main contributions of the work reviewed here is to present an alternate definition $\W$ that generalizes better to higher dimensions. 

\subsection{The renormalized jellium energy}

The renormalized energy is defined via the electric field $\j$ generated by the full charged system: a  (typically infinite) distribution of point charges in a constant neutralizing background.
Note first that the classical Coulomb interaction of two charge distributions (bounded Radon measures) $f$ and $g$,
\begin{equation}\label{eq:def Coul ener}
D(f,g):=\iint_{\mr^d \times \mr^d} \g(x-y)\, df(x) \, dg(y)
\end{equation}
is linked to the (electrostatic) potentials $h_f =  \g * f$, $h_g =  \g * g $ that they generate via the formula
\begin{equation}\label{eq:ener field}
D(f,g) = \int_{\R ^d} f h_g  = \int_{\R ^d} g h_f =\frac{1}{c_d} \int_{\R ^d} \nabla h_f \cdot \nabla h_g
\end{equation}
where we used the fact that by definition of $\g$, 
$$ -\Delta h_f = c_d f,\quad -\Delta h_g = c_d g.$$
The electric field generated by the distribution $f$ is given by $\nab h_f$, and its square norm thus  gives a constant times  the electrostatic energy density of the charge distribution $f$: 
$$D(f,f) = \frac{1}{c_d}\int_{\R ^d} |\nabla h_f| ^2.$$

The electric field generated by a jellium  is of the form described in the following definition.

\begin{definition}[\textbf{Admissible electric fields}]\label{def:adm field}\mbox{}\\ 
Let $m >0$. Let $\j$ be a vector field in $\mr^d$. We say that $\j$ belongs to the class $\bam$ if $\j = \nab h$ with
\begin{equation}\label{curlj}
-\Delta h = c_d \Big(\sum_{p\in \Lambda}N_p \delta_p - m \Big) \quad \text{in} \ \mr^d
\end{equation}
for some discrete set $\Lambda \subset \mr^d$, and $N_p$ integers in $\mn^*$.
%We say that $\j$ belongs to the class $\mathcal{A}_m$ if the same holds with $N_p$ all equal to $1$.
\end{definition}

In the present definition $h$ corresponds to the electrostatic potential  generated by the jellium and $\j$ to its electric field, while the constant $m$ represents the mean number of particles per unit volume, or the density of the neutralizing background. An important difficulty is that the electrostatic energy $D(\delta_p,\delta_p)$ of a point charge, where $\delta_p$ denotes the Dirac mass at $p$, is infinite, or in other words, that  the electric field generated by point charges fails to be in $L^2_{loc}$. This is where the need for a ``renormalization" of this infinite contribution comes from.

To remedy this, we replace point charges by smeared-out charges, as in Onsager's lemma~\cite{Onsager-39}: We pick some arbitrary fixed {\it radial} nonnegative function $\ro$, supported in $B(0,1)$ and with integral $1$. For any point $p$ and $\eta>0$ we introduce   the smeared charge 
\begin{equation}\label{eq:smeared charge}
\delta_p^{(\eta)}= \frac{1}{\eta^d}\ro \left(\frac{x}{\eta}\right) *  \delta_p.
\end{equation}
A simple example is to take 
$ \ro = \frac{1}{|B(0,1)|} \indic_{B(0,1)},$
in which case 
$$\delta_p^{(\eta)} = \frac{1}{|B(0,\eta)|} \indic_{B(0,\eta)}.$$
We also define
\begin{equation}\label{kapd}
\kappa_d:= c_d D(\delta_0^{(1)}, \delta_0^{(1)}) \quad \text{for} \  d\ge 3,  \quad \kappa_2=c_2, \qquad \gamma_2=  c_2 D(\delta_0^{(1)}, \delta_0^{(1)})  \ \text{for} \ d=2.
\end{equation} 
The numbers $\kappa_d$, $\gamma_2$,  depend only on the choice of the function $\ro$ and on the dimension. This non symmetric  definition is due to the fact that the  logarithm behaves differently from power functions under rescaling, and is  made to simplify formulas below.

Newton's theorem  \cite[Theorem 9.7]{LieLos-01} asserts that the Coulomb potentials generated by the smeared charge $\delta_p^{(\eta)}$ and the point charge $\delta_p$ coincide outside of $B(p,\eta)$. A consequence of this is that there exists a radial function $f_\eta$ solution to
\begin{equation}\label{eqf0}
\begin{cases}
 -  \Delta f_\eta= c_d\left( \delta_0^{(\eta)} - \delta_0 \right) \quad \text{in} \ \mr^d
\\   
 f_\eta\equiv 0 \quad  \text{in} \ \mr^d \backslash B(0,\eta).
\end{cases}
\end{equation}
and it is easy to define the field $\j_\eta$ generated by a jellium with smeared charges starting from the field of the jellium with (singular) point charges, using $f_\eta$:

\begin{definition}[\textbf{Smeared electric fields}]\label{def:smear field}\mbox{}\\
For any vector field $\j=\nab h$ satisfying
\begin{equation}\label{dij}
-\div \j= c_d\Big(\sum_{p \in \Lambda}N_p \delta_p -m\Big)
\end{equation}
in some subset $U$ of $\mr^d$, with $\Lambda \subset U$ a discrete set of points, we let
$$\j_\eta := \nab h + \sum_{p\in \Lambda} N_p \nab f_\eta(x-p) \qquad h_\eta= h + \sum_{p \in \Lambda} N_p f_\eta(x-p).$$
We have
\begin{equation}
\label{delp}
-\div  \j_\eta  = - \Delta h_\eta =  c_d\Big(\sum_{p \in \Lambda}N_p \delta_p^{(\eta)} -m\Big)\end{equation}
and denoting by $\Phi_\eta$ the map $\j \mapsto \j_\eta$,  we note that $\Phi_\eta$ realizes a bijection from the set of vector fields satisfying \eqref{dij} and which are gradients,  to those satisfying \eqref{delp} and which are also gradients.
\end{definition}

% Note that the above definition  in principle  depends implicitly on the set $U$, whose choice will be clear from the context in the sequel (most of the time we will take $U= \R ^d$).

For any fixed $\eta>0$ one may then define the electrostatic energy per unit volume of the infinite jellium with smeared charges as
\begin{equation}\label{eq:Weta pre}
\limsup_{R\to \infty} \dashint_{K_R}  |\j_\eta|^2 := \limsup_{R\to \infty} |K_R| ^{-1} \int_{K_R}  |\j_\eta|^2 
\end{equation}
where $\j_\eta$ is as in the above definition and $K_R$ denotes the cube $[-R,R]^d$. Note that the quantity $\dashint_{K_R}  |\j_\eta|^2$ may  not have a limit (this does occur for somewhat pathological configurations). This motivates the use of the $\limsup$ instead.

This energy is now well-defined for $\eta>0$ and blows up as $\eta \to 0$, since it includes the self-energy of each smeared charge in the collection, absent in the original energy  (i.e. in the Hamiltonian \eqref{wn}). One should thus \emph{renormalize} \eqref{eq:Weta pre} by removing the self-energy of each smeared charge before taking the limit $\eta \to 0$. We will see that  the leading order  energy of a smeared charge is $\kappa_d \g(\eta)$, and this is the quantity that should be removed for each point. But  in order for the charges to efficiently screen the neutralizing background,  configurations will need to have  the same charge density as the neutralizing background (i.e. $m$ points per unit volume). We are then led to the definition  

\begin{definition}[\textbf{The renormalized energy}]\label{def:renorm ener}\mbox{}\\
For any $\j \in  \bam$, we define
\begin{equation}\label{We}
\W_\eta(\j) = \limsup_{R\to \infty} \dashint_{K_R}  |\j_\eta|^2 - m(\kappa_d  \g(\eta)+\gamma_2 \indic_{d=2})
\end{equation}
and the renormalized energy is given by 
\begin{align}\label{weta}
\W(\j) = \liminf_{\eta\to 0}\W_\eta(\j) \nonumber =\liminf_{\eta\to 0}\left(\limsup_{R\to \infty} \dashint_{K_R}   |\j_\eta|^2 - m(\kappa_d \g(\eta)+\gamma_2 \indic_{d=2})\right).
\end{align} 
\end{definition}

It is easy to see that if $\j \in \bam$, then $\j' := m ^{1/d - 1} \j (m ^{-1/d} .)$ belongs to $\bai$ and 
\begin{equation}\label{eq:scale renorm}
\begin{cases}
\W_\eta (\j) = m ^{2-2/d} \mathcal{W}_{\eta m^{1/d}}  (\j ') & \  \W (\j) = m ^{2-2/d} \W (\j ')  \: \mbox{ if } d\geq 3 \\
\W_\eta (\j) = m \left( \mathcal{W}_{\eta m^{1/d}}  (\j ')-\frac{\kappa_2}{2} \log m \right)  & \ \W  (\j) = m \left(\W (\j ')-\frac{\kappa_2 }{2}\log m \right)  \: \mbox{ if } d=2,
\end{cases}
\end{equation}thus the same scaling formulae hold for $\inf_{\bam} \W$.
 One may thus reduce to the study of $\W(\j)$ on $\bai$, for which we have the following result:

\begin{theorem}[\textbf{Minimization of the renormalized energy}]\label{thm:renorm ener}\mbox{}\\
The infimum 
\begin{equation}\label{eq:min ren ener}
\alpha_d:= \inf_{\j \in \bai} \W (\j) 
\end{equation}
is achieved and is finite. Moreover, there exists a sequence $(\j_n)_{n\in \mathbb{N}}$ of periodic vector fields (with diverging period in the limit $n\to \infty$) in $\bai$ such that 
\begin{equation}\label{eq:min period}
\W (\j_n) \to \alpha_d \mbox{ as } n \to \infty. 
\end{equation}
\end{theorem}

We should stress at this point that the definition of the jellium (renormalized) energy we use is essential for our approach to the study of equilibrium states of \eqref{wn}. In particular it is crucial that we are allowed to define the energy of an \emph{infinite} system via a local density (square norm of the electric field). For a different definition of the jellium energy, the existence of the thermodynamic limit was previously proved in \cite{LieNar-76}, using ideas from \cite{LebLie-69,LieLeb-72}. This approach does not transpose easily in our context however, and our proof that $\alpha_d$ is finite follows a different route.

The natural, but so far unproven, conjecture is that $\alpha_d$ can be attained by a periodic (or crystalline) distribution of charge. In 2D, it is known that the the optimal lattice is the triangular one~\cite{SanSer-12}, but in higher dimensions, even the optimal arrangement amongst regular lattices is unknown.

\section{Asymptotics for the ground state}\label{sec:Coul GS}

We are now equipped to investigate the corrections to~\eqref{eq:gse first order} and~\eqref{eq:coul gs first}. The subleading behavior of the ground state energy is most easily stated:

\begin{theorem}[\textbf{Expansion of the ground state energy}]\label{thm:coul gse}\mbox{}\\
Let $E_n$ denote the ground state energy~\eqref{eq:gse} and $\alpha_d$ denote the minimum of $\W$ for a jellium of density $1$ in dimension $d$, defined in Theorem~\ref{thm:renorm ener}. In the limit $n\to \infty$, we have
\begin{equation}\label{eq:gse second order}
E_n = 
\begin{cases} 
 \displaystyle      n ^2 \En [\mu_0] + \frac{n ^{2-2/d}}{c_d}\alpha_d \int \mu_0 ^{2-2/d} (x)  dx + o \left(n ^{2-2/d}\right) \mbox{ if } d\geq 3\\ 
 \displaystyle      n ^2 \En [\mu_0] - \frac{n}{2}\log n + n\left(\frac{\alpha_2}{2\pi} -\hal \int \mu_0(x) \log \mu_0(x) \,dx \right) + o\left( n\right)
       \mbox{ if } d= 2.
      \end{cases}
\end{equation} 
\end{theorem}

This result encodes the double scale nature of the charge distribution: the first term is due to the points following the macroscopic distribution  $\mu_0$ -- we assume that the probability $\mu_0$ has a density that we also denote $\mu_0$ by abuse of notation. The next order correction, which happens to lie at the order $n^{2-2/d}$, governs the configurations at the microscopic scale {and the crystallization, by selecting configurations which }  minimize $\W$ (on average, with respect to the blow-up centers). 

The factors involving $\mu_0$ in the correction come from scaling, and  from the fact that the points locally see a neutralizing background whose charge density is given by the value of $\mu_0$. The scaling properties of the renormalized energy imply that the minimal energy of a jellium  with neutralizing density $\mu_0 (x)$  is $\alpha_d \mu_0 ^{2-2/d} (x)$ (respectively $\mu_0(x)\alpha_2 - \mu_0 (x) \log \mu_0(x)$ when $d=2$). Integrating this energy density on the support of $\mu_0$ leads to the formula.
 
The interpretation of the correction is thus that around (almost) any point $x$ in the support of $\mu_0$ there are approximately $\mu_0(x)$ points per unit volume, distributed so as to minimize a jellium energy with background density $\mu_0(x)$. Due to the properties of the jellium, this implies that, up to a $\mu_0(x)$-dependent rescaling, the local distribution of particles are the same around any blow-up origin $x$ in the support of $\mu_0$. This can be interpreted as a result of universality with respect to the potential $V$ in \eqref{wn}, in connection with recent works on the 1D log gas \cite{BouErdYau-12,BouErdYau-14,Leble-15,SanSer-14a}.  

\medskip

We now proceed to give a rigorous meaning to the above, somewhat vague, assertions about the ground state configurations. This requires a bit of formalism. To describe the behavior of minimizers at the microscopic scale we follow the same approach as in~\cite{SanSer-14} and perform a blow-up: For a given $(x_1, \dots , x_n)$, we let $x_i'= n^{1/d} x_i$ and
\begin{equation}
\label{hpn}
h_n'(x')=  w \ast \left(\sum_{i=1}^n \delta_{x_i'} - \mu_0(n^{-1/d}x')\right),
\end{equation}
Note that the associated electric field $\nab h_n'$ is in $\Lp$ if and only if $p<\frac{d}{d-1}$. We next embed $(\mr^d)^n$ into the set of probabilities on $X=\Sigma\times\Lp$, for some $1<p<\frac{d}{d-1}$.  For any $n$ and $\xbf =(x_1,\dots, x_n)\in(\mr^d)^n$ we let $i_n(\xbf) = P_{\nu_n},$
where $\nu_n=\sum_{i=1}^n \delta_{x_i}$ and $P_{\nu_n}$ is the push-forward of the normalized Lebesgue measure on $\Sigma$ by
$$x\mapsto \left(x, \nab h_n'(n^{1/d}x + \cdot)\right).$$ 
Explicitly:
\begin{equation}\label{eq:Pnun}
i_n (\xbf) = P_{\nu_n} = \dashint_{\Sigma} \delta_{(x,\nabla h'_n (n ^{1/d} x + \cdot) )} dx.
\end{equation}
This way  $i_n(\xbf)$ is an element of $\P(X)$, the set of probability measures on $X = \Sigma\times\Lp$ (couples of blown-up centers, blown-up electric fields) that measures  the probability of having a given blown-up electric field around a given blow-up point in $\Sigma$. As suggested by the above discussion, the natural object we should look at is really $i_n (\xbf)$ { and its limits up to extraction,~$P$}. 

Due to the fact that the renormalized jellium functional describing the small-scale physics is invariant under translations of the electric field, and by definition of $i_n$, we should of course expect the objects we have just introduced to have a certain translation invariance, formalized as follows:

\begin{definition}[\textbf{$T_{\lambda(x)}$-invariance}]\label{invariance}\mbox{}\\
We say a probability measure $P$ on $X$ is $T_{\lambda(x)}$-invariant if $P$ is invariant  by $(x,\j)\mapsto\left(x,\j(\lambda(x)+\cdot)\right)$, for any $\lambda(x)$ of class $C^1$  from $\Sigma$ to $\mr^d$.
\end{definition}
Note that $T_{\lambda(x)}$-invariant implies translation-invariant (simply take $\lambda \equiv 1$). 

\begin{definition}[\textbf{Admissible configurations}]
We say that $P \in \mathcal{P}(X)$ is admissible if its first marginal is the normalized Lebesgue 
measure on $\Sigma$, if it holds for $P$-a.e. $(x,\j)$ that $\j \in \overline{\mathcal{A}}_{\mu_0(x)}$, and if $P $ is $T_{\lambda(x)}$-invariant.
\end{definition}

We next define an ``average with respect to blow-up centers" of $\W$:
\begin{equation}\label{wtilde}
\widetilde{\W}(P):= \frac{|\Sigma|}{c_d} \int_{X} \W(\j) \, dP(x,\j),
\end{equation} 
if $P$ is admissible, and $+\infty$ otherwise. In view of the scaling relation \eqref{eq:scale renorm}, we may guess that 
\begin{equation}\label{eq:gamma d}
\min_{P \ \text{admissible}} \widetilde{\W}=
\xi_d:=  \begin{cases}\displaystyle
             \frac{1}{c_d} \min_{\bai}  \W  \int_{\R   ^d} \mu_0 ^{2-2/d} \quad \mbox{if } d\geq 3\\
             \displaystyle \frac{1}{2\pi} \min_{\bai}  \W - \hal \int_{\R   ^2} \mu_0 \log \mu_0  \quad \mbox{if } d=2.
            \end{cases}
\end{equation} 
Note that the minima on the right-hand side exist thanks to Theorem~\ref{thm:renorm ener}. Also, the left-hand side is clearly larger than the right-hand side in view of the definition of $\widetilde{\W}$, that  of admissible, and \eqref{eq:scale renorm}. That there is actually equality is a consequence of the following theorem
on the behavior of ground state configurations (i.e. minimizers of \eqref{wn}) :

\begin{theorem}[\textbf{Microscopic behavior of ground state configurations}]\label{thm:coul GS}\mbox{}\\
Let $(x_1, \dots, x_n)\in (\mr^d)^n$ minimize $\w$. Let $h_n' $ be associated via \eqref{hpn} and  
$$P_{\nu_n}= i_n(x_1, \dots, x_n) $$
be defined as in \eqref{eq:Pnun}.  

Up to extraction of a subsequence we have that 
$$P_{\nu_n} \wto P$$
in the sense of probability measures on $X$, where $P$ is admissible. Moreover, $P$ is a minimizer of $\widetilde{\W}$ and $\j$ minimizes $\W$ over $\overline{\mathcal{A}}_{\mu_0(x)}$  for $P$-a.e. $(x,\j)$.
%\item We have $\sum_{i=1}^n \zeta(x_i) = o(n^{1-2/d}),$ as $n \to \infty$.
\end{theorem}

 If the crystallization conjecture is correct, and if the only minimizers of $\W$ are periodic configurations, this means that after blow-up around (almost) any point in $\Sigma$, one should see a crystalline configuration of points, the one minimizing $\W$, packed at the scale corresponding to $\mu_0(x)$.   

\section{Asymptotics for the Gibbs states}\label{sec:Coul Gibbs}

We now turn to the case of positive temperature and study the Gibbs measure \eqref{eq:defi Gibbs state}. We shall need a slight strengthening of the assumption \eqref{eq:trap pot}: we assume that there exists $\beta_1>0$ such that
\begin{equation}\label{integrabilite}
\begin{cases}
\int e^{-\beta_1 V(x)/2} \, dx <\infty  & \mbox{ when } d\geq 3\\
\int e^{-\beta_1( \frac{V(x)}{2}- \log |x|)} \, dx <\infty & \mbox{ when } d=2. 
\end{cases}
\end{equation}
We first state a generalization of Theorem~\ref{thm:coul gse} to positive temperature:

\begin{theorem}[\textbf{Free energy/partition function estimates}]\label{thm:partition}\mbox{}\\
The following estimates hold in the limit $n\to \infty$.
\begin{enumerate}
\item (Low temperature regime).  Let $\xi_d$ be as in~\eqref{eq:gamma d}. For $\beta \gg n ^{2/d-1}$ we have 
\begin{equation}\label{eq:partition low T}
\Fnbetae = n ^2 \En [\mu_0]  + n ^{2-2/d} \xi_d (1+o(1))
\end{equation} 
if $d\geq 3$ and 
\begin{equation}\label{eq:partition low T 2d}
\Fnbetae = n ^2 \En [\mu_0] - \frac{n}{2}\log n + n\xi_2 (1+o(1))
\end{equation} 
if $d=2$.

\item (High temperature regime). If $d\geq 3$ and $ n ^{-1} \lesssim \beta \lesssim n ^{2/d - 1}$ 
\begin{equation}\label{eq:partition high T}
\Fnbetae = n ^2 \En [\mubet]  + \frac{2 n}{\beta} \int_{\R ^d} \mubet \log \mubet (1+o(1))
\end{equation} 
respectively if $d=2$ and $n ^{-1} \lesssim \beta \lesssim  (\log n) ^{-1}$ 
\begin{equation}\label{eq:partition high T 2d}
 \Fnbetae = n ^2 \En [\mubet] + \frac{2 n}{\beta} \int_{\R ^d} \mubet \log \mubet (1+o(1)).
\end{equation} 
\end{enumerate}
\end{theorem}

The dichotomy in the above suggests strongly that the liquid to crystal transition must happen in the temperature regime $\beta \sim n ^{2/d-1}$. Indeed, for lower temperatures, the free-energy can basically not be distinguished from the ground state energy, and the occurrence of the jellium ground state energy as the subleading order suggests crystallization of typical thermal configurations. On the contrary, for larger temperature\footnote{Remark the small gap between the two regimes in 2D, due to scaling properties of the $\log$ kernel.}, the subleading term is of entropic nature, which suggests a melting of the crystal. This picture has been recently strengthened in~\cite{LebSer-15,Leble-15b} by a more detail study of the critical regime $\beta \sim n ^{2/d-1}$. 

\medskip

Our next result exposes the consequences of the low temperature estimates \eqref{eq:partition low T}-\eqref{eq:partition low T 2d} for the Gibbs measure itself. Roughly speaking, we prove that it charges only configurations whose renormalized energy  {$\widetilde{\W}$} is below a certain threshold. This is a large deviations type upper bound, at speed $n^{2-2/d}$. The threshold of renormalized energy vanishes in the limit $\beta \gg n ^{2/d-1}$, showing that the Gibbs measure charges only configurations that minimize the renormalized energy at the microscopic scale. This is a step towards a proof of crystallization, in that the crystallization for the thermal states in the low temperature regime is reduced to a question (still open) about ground states of the jellium problem.

The statement below uses the same framework as Theorem \ref{thm:coul GS}, in terms of the limiting probability measures $P$, which should now be seen as random, due to the temperature. We  consider the limit of the probability that the system is in a state $(x_1,\ldots,x_n)\in A_n$ for a given sequence of sets $A_n \subset (\R ^{d})^n$ in configuration space. Associated to this sequence we introduce
\begin{equation}\label{ainf} 
A_\infty = \left\{ P \in \P(X) : \exists \: \xbf_n \in A_n, \, P_{\nu_n} \wto P \  \text{up to a subsequence}\right\}
\end{equation}
where $P_{\nu_n}$ is as in \eqref{eq:Pnun} with $\xbf=\xbf_n$ and the convergence is weakly as measures.  The set $A_\infty$ can be thought of as the limsup of the sequence $i_n(A_n)$.

\begin{theorem}[\textbf{Behavior of thermal states in the low temperature regime}]\label{thm:coul Gibbs}\mbox{}\\
For any $n>0$ let $A_n\subset (\mr^d)^n$ and $A_\infty$ be as above.
{Let 
$$\bar{\beta}:=\limsup_{n\to+\infty} \beta n^{1-2/d}$$
and $\xi_d$ be as in~\eqref{eq:gamma d}. Assume $0<\bar{\beta}\le + \infty$. There exists $C_{\bar{\beta}}$ such that $C_{\bar{\beta}}=0$ for $\bar{\beta} = \infty$ and 
\begin{equation}
\label{ldr}
\limsup_{n\to \infty} \frac{ \log \Gibbs(A_n)}{n ^{2-2/d}} \le -\frac{\beta}{2} \left(\inf_{P\in A_\infty}\widetilde{\W}(P)  - \xi_d  - C_{\bar{\beta}}  \right).
\end{equation}}
% Moreover, for fixed $\beta>0$,  letting $\widetilde{ \Q}$ denote the push-forward of $\Q$ by $i_n$ (defined in \eqref{eq:Pnun}), $\{\widetilde{\Q}\}_{n}$ is tight and converges as $n\to \infty$, up to a subsequence, to a probability measure on $\mathcal P(X)$ which is concentrated on admissible probabilities satisfying  $\widetilde{\W}(P) \le \xi_d + C_{\bar{\beta}}$. 
\end{theorem}

Note that the error term {$C_{\bar{ \beta}}$} becomes negligible when $\beta \gg n ^{2/d-1}$. Thus the Gibbs measure concentrates on minimizers of  {$\widetilde{\W}$} in that regime. When $\beta =\bar{\beta} n^{2/d-1}$, we have instead a threshold phenomenon: the Gibbs measure concentrates on configurations whose {$\widetilde{ \W} $} is below the minimum plus {$C_{\bar{\beta}}$}.

\section{Fluctuations around mean-field theory}\label{sec:Coul fluctu}

Our methods also yield some quantitative estimates on the fluctuations around the equilibrium measure. These are obtained through nice coercivity properties of (precursors of) the renormalized energy functional. We shall focus on the low temperature regime for concreteness and stress that the following estimates are insufficient for a strong crystallization result (estimates on the location of charges down to the optimal scale $n^{-1/d}$). They however set bounds on the possible deviations from mean-field theory in a stronger sense than what has been considered so far in this memoir. In this sense they can come in handy if such estimates are needed in applications (compare with Theorem~\ref{theo:dens QH trial} in the previous chapter).

We shall estimate the following quantity, defined for any configuration $(x_1,\ldots,x_n)\in\R^{dn}$ 
\begin{equation}\label{eq:intro charge fluctu}
D(x,R):=  \nu_n (B(x,R)) - n \int_{B(x,R)} \mu_0, 
\end{equation}
i.e. the deviation of the charge contained in a ball $B(x,R)$ with respect to the prediction drawn from the equilibrium measure. Here 
\begin{equation}\label{eq:empirical measure}
\nu_n  = \sum_{i=1} ^n  \delta_{x_i}.
\end{equation}
Note that since $\mu_0$ is a fixed bounded function, the second term in \eqref{eq:intro charge fluctu} is of order $n R^d$, and typically the difference of the two terms is much smaller than this, since the distribution of points at least approximately follows the equilibrium measure.

We state three typical results:
\begin{enumerate}
\item the probability of a large deviation of the number of points in a ball $B(x,R)$ (i.e. a deviation of order $n R ^d$) is exponentially small as soon as $R\gtrsim n ^{-1/(d+2)}$.
\item the probability of a charge deviation of order $n ^{1-1/d}$ in a macroscopic ball $B(x,R)$ with $R = O(1)$ is exponentially small.
\item an estimate on the discrepancy between $\nu_n$ and the equilibrium measure $n\mu_0$ in weak Sobolev norms.
\end{enumerate}
We denote $W^{-1, q}(\Omega)$ the dual of the Sobolev space $W^{1,q'}_0(\Omega)$ with $1/q+1/q'=1$, in particular $W^{-1,1}$ is the dual of Lipschitz functions.

\begin{theorem}[\textbf{Charge fluctuations}]\label{thm:charge fluctu}\mbox{}\\
Assume there is a constant $C>0$ such that $\beta \geq C n ^{2/d - 1}$. Then the following holds for any $x\in \R ^d$.
\begin{enumerate}
\item (Large fluctuations on microscopic scales). Let $R_n$ be a sequence of radii satisfying $R_n \geq C_R \: n ^{-\frac{1}{d+2}}$ for some  constant $C_R$. Then, for any $\lambda>0$ we have, for $n$ large enough,
\begin{equation}\label{eq:fluctu charge micro}
\Gibbs \left( |D(x,R_n)| \geq \lambda n R_n ^{d}\right) \leq C e ^{-C \beta n ^{2-2/d}(C_R\lambda^2 - C)},
\end{equation}
for some $C$ depending only on dimension.
\item (Small fluctuations at the macroscopic scale). Let $R>0$ be a fixed radius. There is a constant  $C$ depending only on dimension such that for any $\lambda>0$, for $n$ large enough, 
\begin{equation}\label{eq:fluctu charge macro}
\Gibbs \left( |D(x,R)| \geq  \lambda n^{1-1/d} \right) \leq C e ^{-C \beta n ^{2-2/d}( \min(\lambda^2  R^{2-d}, \lambda^4 R^{2-2d})   -C) }. 
\end{equation}
\item (Control in weak Sobolev norms). Let $R>0$ be some fixed radius. Let $1 \leq q < \frac{d}{d-1}$ and define
\begin{align*}
t_{q,d} &= 2 - \frac{1}{d} - \frac{1}{q}<1\\
\tilde{t}_{q,d} &= 3   -\frac{1}{d} - \frac{2}{dq}  >0.
\end{align*}
There is a constant $C_R>0$ such that the following holds for $n$ large enough, and any $\lambda$ large enough, 
\begin{equation}\label{eq:fluctu field}
\Gibbs \left( \left\Vert \nu_n - n\mu_0 \right\Vert_{W^{-1,q} (B_R)} \geq \lambda n ^{t_{q,d}}\right) \leq C e ^{- \beta C_R \lambda^2 n ^{ \tilde{t}_{q,d}}}.
\end{equation}
\end{enumerate}
\end{theorem}

As regards \eqref{eq:fluctu field}, note that only estimates in $W^{-1,q}$ norm with $q<\frac{d}{d-1}$ make sense, since a Dirac mass is in $W^{-1,q}$  if and only if $q<\frac{d}{d-1}$  (hence the same for $\nu_n$). In view of the values of the parameters $t_{q,d}$ and $\tilde{t}_{q,d}$ defined above, our results give meaningful estimates for any such norm: a large deviation in a ball of fixed radius would be a deviation of order~$n$. Since $t_{q,d}<1$, equation \eqref{eq:fluctu field} above implies that such deviations are exponentially unlikely in $W^{-1,q}$ for any $q<\frac{d}{d-1}$, in particular in $W^{-1,1}$ for any space dimension.

\medskip

Finally, we return to the ``uncorrelated point of view'' on the mean field limit, stating some consequences for the reduced densities~\eqref{eq:marginal Gibbs} of the Gibbs state:

\begin{corollary}[\textbf{Reduced densities in the low temperature regime}]\label{thm:marginals}\mbox{}\\
Let $R>0$ be some fixed radius, $1\leq q<\frac{d}{d-1}$. Under the same assumptions as Theorem \ref{thm:charge fluctu}, there exists a constant $C>0$ depending only on the  dimension such that the following holds:  
\begin{itemize}
\item (Estimate on the one-particle reduced density).
\begin{equation}\label{eq:result marginal 1}
\left\Vert \Qone - \mu_0 \right\Vert_{W^{-1,q} (B_R)} \leq C n^{1-1/d-1/q}=o_n(1). 
\end{equation}
\item (Estimate on $k$-particle reduced densities). Let $k\geq 2$ and $\varphi : \R ^{dk} \mapsto \R$ be a smooth function with compact support, symmetric w.r.t. particle exchange. Then we have
\begin{equation}\label{eq:result marginal k}
\left| \int_{\R ^{dk}} \left( \Qk - \mu_0 ^{\otimes k}\right) \varphi \right| \leq C k n ^{1-1/d-1/q}\sup_{x_1 \in \R ^d} \ldots \sup_{x_{k-1} \in \R ^d }   \left\Vert \nabla \varphi (x_1,\ldots,x_{k-1}, \:. \: ) \right\Vert_{L ^p (\R ^d)},
\end{equation}
where $1/p = 1 - 1/q$.
\end{itemize}
\end{corollary}

As our other results, Corollary \ref{thm:marginals} concerns the low temperature regime. One may use the same technique to estimate the discrepancy between the $n$-body problem and mean-field theory in other regimes. When $\beta n $ becomes small however (large temperature), in particular when $\beta \sim n ^{-1}$ so that entropy and energy terms in \eqref{eq:MF free ener func} are of the same order of magnitude, a different method can give slightly better estimates, see the original paper for details.

\chapter{Ground states of large bosonic systems}\label{sec:bosons GS}

The quest for Bose-Einstein condensation (BEC) is a long story that spanned most of the 20th century, from the original theoretical prediction~\cite{Bose-24,Einstein-24} to the first experimental realizations~\cite{CorWie-nobel,Ketterle-nobel}. In short, one of the main questions is to understand why, how, and when, a system of bosons\footnote{i.e. quantum particles that do not satisfy the Pauli exclusion principle.} shows macroscopic occupancy, with most particles residing in a single quantum state. This phenomenon paves the way for the direct observation of macroscopic quantum effects, such as superfluidity, alluded to in Section~\ref{sec:GP vortex} of this memoir. 

Ultimately, one would like to derive the existence of a critical temperature under which BEC occurs in thermal equilibrium states of an interacting bosonic system. This question is still mostly open from a mathematical point of view, and shall be touched upon only in Chapter~\ref{sec:Gibbs} below. In fact, for interacting Bose gases, even the existence of BEC in the zero-temperature ground states is a very delicate question. In this chapter we shall summarize the papers~\cite{LewNamRou-14,LewNamRou-14b,LewNamRou-14c,LewNamRou-15,NamRouSei-15} that deal with this issue. The general goal is to show that, in appropriate scaling limits, the many-body bosonic problem reduces, at the level of the ground state, to the study of effective one-body nonlinear Schr\"odinger (NLS) theories. 

This problem corresponds to the existence of ``molecular chaos'' in related fields such as kinetic theory. It logically precedes the ``propagation of chaos'' problem. The latter deals with the same question at the level of the appropriate evolution equation. Indeed, in the dynamical problem, one assumes BEC and emergence of NLS theory (or molecular chaos) in the initial datum and shows that it is propagated in time by the many-body Schr\"odinger flow. The rationale is that, in most experiments, the initial datum is a ground state for some many-body Hamiltonian, and the evolution is generated by a sudden perturbation thereof.

\medskip

The chapter is organized as follows:
\begin{itemize}
 \item Section~\ref{sec:bosons intro} sets the stage and discusses the various mathematical problems addressed here. There is actually a large variety of models and scalings for which one would like to study basically the same kind of questions. Several of these are described here, emphasizing various key difficulties one may encounter in the analysis.
 \item Section~\ref{sec:bosons deF} addresses the quantum de Finetti theorems, focusing on contributions from the author, together with Mathieu Lewin and Phan Th\`anh Nam~\cite{LewNamRou-14,LewNamRou-14b}. These tools offer a lot of flexibility in dealing with the problems discussed in Section~\ref{sec:bosons intro}. Indeed, they describe remarkable structural properties of general admissible states for large bosonic systems. They can thus be applied (almost) independently of the particular physics at hand, provided one is dealing with bosons. This observation is the main key in the research program summarized in this chapter.
 \item Section~\ref{sec:bosons MF} contains the main results of~\cite{LewNamRou-14,LewNamRou-14c}. They concern so-called mean-field limits  where the occurrence of BEC is a consequence of inter-particles interactions being weak but frequent. In this case, very general proof strategies can be implemented, resting solely on the quantum de Finetti theorems and general tools of many-body quantum mechanics. Very little of the properties of the Hamiltonians at hand enter the proof strategy.
 \item Section~\ref{sec:bosons dilute} is concerned with the more challenging, but also more experimentally relevant, dilute regimes wherein BEC occurs because interactions are strong but rare. Here one has to complement the very general tools discussed so far with methods and estimates more directly linked to specific models. We shall discuss a new treatment of the Gross-Pitaevskii limit for the 3D Bose gas~\cite{NamRouSei-15} and a derivation of 2D NLS theory in the presence of possibly attractive interactions~\cite{LewNamRou-15}.
\end{itemize}

More general discussions and extensive lists of references on various aspects of the topics discussed here may be found in~\cite{BloDalZwe-08,DalGioPitStr-99,DalGerJuzOhb-11,PetSmi-01,PitStr-03,Fetter-09,Cooper-08} (physics literature) and ~\cite{Aftalion-06,BenPorSch-15,Golse-13,Lewin-ICMP,LieSeiSolYng-05,Rougerie-LMU,Rougerie-cdf,Schlein-08,Seiringer-ICMP10,Verbeure-11} (mathematics literature). 

\section{Models and scaling regimes}\label{sec:bosons intro}

Let us first introduce the general type of problems we will be interested in. All these being formulated in roughly the same way, they should ideally also be treated in the same way. The results of this chapter can be seen as means to this end. 

\subsection{A general Hamiltonian}

The Hamilton operator we will consider in this chapter can be written in the general form 
\begin{equation}\label{eq:bosons hamil gen}
H_N := \sum_{j=1} ^N \left\{ \left( -i\nabla_j + A (x_j) \right) ^2 +  V (x_j) \right\} + \sum_{1\leq i < j \leq N} w_{\beta,N} ( x_i - x_j)
\end{equation}
acting on 
\begin{equation}\label{eq:bosons L2 sym}
\gH^N := L ^2 _{\rm sym} (\R ^{d N}) \simeq \bigotimes_{\rm sym} ^N L ^2 (\R ^d) 
\end{equation}
as appropriate for a system made of bosons, whose many-body wave function must satisfy the symmetry condition 
\begin{equation}\label{eq:bosons L2 sym bis}
\Psi_N (x_1,\ldots,x_N) =  \Psi_N (x_{\sigma(1)},\ldots,x_{\sigma(N)})
\end{equation}
for every permutation $\sigma$ of $N$ indices.

This model is not the most general one that can be dealt with using the tools discussed here. However, it can already include most of the technical difficulties one might encounter, as discussed below. As far as applications are concerned, this is, to the best of the author's knowledge, the most general form relevant for the cold atoms experiments where BEC is observed. Indeed
\begin{itemize}
 \item The space dimension $d$ can be $1,2,3$, all three cases being relevant to current experiments. One can even allow for larger $d$, with less practical relevance of course.
 \item We allow for the presence of a gauge vector potential $ A : \R ^d \mapsto \R^d$, which can reflect the fact that the problem is formulated in a rotating frame, model a (possibly artificial) magnetic field ... 
 \item The one-body potential $V:\R ^d \mapsto \R$ can be either trapping (increasing)  or decaying at infinity.
 \item We make no structural assumption on the interaction potential $w_{\beta,N}$. We do not assume a priori that it has a sign, nor that its Fourier transform has. Provided stability holds, we can thus deal with either attractive or repulsive interactions.
 \item We consider a general family of scaling limits by choosing $w_{\beta,N}$ of the form
 \begin{equation}\label{eq:bosons gen int pot}
 w_{\beta,N} (x) = \frac{1}{N-1} N ^{d\beta} w \left( N ^\beta x \right).
 \end{equation}
Here $\beta$ shall be a fixed parameter in the limit $N\to \infty$. The larger it is, the more singular the interaction potential. The prefactor $(N-1) ^{-1}$ makes the one-body and two-body terms of~\eqref{eq:bosons hamil gen} of the same order of magnitude, so that one may hope for a well-defined limit problem. The  potential $N ^{d\beta} w \left( N ^\beta x \right)$ has a fixed integral, so that one should think of it as being of order $1$ when $N\to \infty$.
\end{itemize}

In this chapter we are interested in the ground state problem of minimizing the energy corresponding to~\eqref{eq:bosons hamil gen}:
\begin{equation}\label{eq:bosons GSE gen}
E (N) := \inf \left\{ \left\langle \Psi_N, H_N \Psi_N \right\rangle, \: \Psi_N \in \gH ^N , \: \int_{\R ^{dN}} |\Psi_N| ^2 = 1\right\} = \inf \sigma (H_N)
\end{equation}
and in the associated minimizers (if they exist) or minimizing sequences. Equivalently, we consider the lowest eigenvalue (or bottom of the essential spectrum), and associated (generalized) eigenfunctions, of $H_N$.

\subsection{Limit problems} Roughly speaking, one is interested in showing that a minimizer of~\eqref{eq:bosons GSE gen} has, in a certain sense to be discussed below, all its particles in the same quantum state. If one assumes this to be true and picks as ansatz a pure BEC of the form 
$$ \Psi_N (x_1,\ldots,x_N) = u ^{\otimes N} (x_1,\ldots,x_N) = \prod_{j=1} ^N u(x_j), \quad u\in L ^2 (\R ^d),$$
evaluating the energy leads to a NLS-like functional
\begin{equation}\label{eq:bosons NLS gen}
N ^{-1} \left\langle u ^{\otimes N}, H_N u ^{\otimes N} \right\rangle \simeq \Enls [u] = \int_{\R ^d} \left| (-i\nabla + A ) u \right| ^2 + V |u| ^2 + \frac{1}{2} \left( \weff * |u| ^2 \right) |u| ^2.
\end{equation}
Here, $A$ and $V$ are inherited from~\eqref{eq:bosons hamil gen} but $\weff$ can depend on the context. For $\beta = 0$ in~\eqref{eq:bosons gen int pot} one makes the obvious choice $\weff = w$. For $\beta >0$ the natural guess should be $\weff = \left( \int_{\R ^d} w \right) \delta_0$ but this can be misleading in some cases (an important example shall be discussed in the sequel). 

If we expect the guess $\Psi_N \simeq u ^{\otimes N}$ to be correct, then the linear many-body minimization problem~\eqref{eq:bosons GSE gen} should be connected for $N$ large to the nonlinear one-body minimization problem
\begin{equation}\label{eq:bos NLS gse}
\enls := \inf\left\{ \Enls [u],\: u \in L ^2 (\R ^d), \: \int_{\R ^d} |u| ^2 = 1\right\}. 
\end{equation}
We shall denote a minimizer by $\unls$ if it exists. It might not be unique, in which case we denote by $\Mnls$ the set of all minimizers.

\subsection{A meta-theorem} We now state a vague general theorem relating the minimization problems~\eqref{eq:bosons GSE gen} and~\eqref{eq:bos NLS gse}. This statement will serve as a model and will be declined in several rigorous avatars in the sequel, depending on the various concrete choices of settings we will make. 

The main thing we still have to specify is in which sense the approximation
$$ \Psi_N \underset{N \to \infty}{\simeq} \unls ^{\otimes N}$$
for a ground state (or approximate ground state) $\Psi_N$ of~\eqref{eq:bosons hamil gen} is correct, where $\unls$ is a ground state for~\eqref{eq:bosons NLS gen}.
Clearly, this cannot be a convergence since both objects depend on~$N$. There are known obstructions to obtaining an estimate in norm~\cite{DerNap-13,GreSei-13,LewNamSerSol-13,NamSei-14,Seiringer-11}, connected to Bogoliubov's theory. Moreover, there might not be a unique relevant ~$\unls$. 

For all these reasons, the convergence of states will be expressed in terms of the reduced density matrices of $\Psi_N$ rather than on the wave-function itself. Let us recall the definition: the reduced $k$-particles density matrix of $\Psi_N$ is the positive trace-class operator on $\gH ^k$ defined by
\begin{equation}\label{eq:bos reduced DM}
\gamma_N ^{(k)} := \tr_{k\to N} \left[ | \Psi_N \rangle \langle \Psi_N | \right] 
\end{equation}
where $| \Psi_N \rangle \langle \Psi_N |$ is the orthogonal projector onto $\Psi_N$, and the partial trace can be taken, because of the symmetry~\eqref{eq:bosons L2 sym bis}, on any choice of $N-k$ variables amongst the $N$ original ones. Note our choice of normalization: $\gamma_N ^{(k)}$ has trace $1$ for all $k$ if $\Psi_N$ is $L ^2$-normalized.

We can now state the informal statement that will serve as a model for more rigorous versions in the rest of the chapter:

\begin{mtheorem}[\textbf{Derivation of NLS theory, informal version}]\label{thm:bosons meta}\mbox{}\\
Suppose that the data $A,V,w$ and $0 \leq \beta < \beta_0 \in \R ^+$ of the many-body Hamiltonian~\eqref{eq:bosons hamil gen} are such that $H_N$ makes sense as a self-adjoint operator, and is bounded from below by $-C N$ for $C$ fixed. Then the energy functional~\eqref{eq:bos NLS gse} is bounded from below and we have 
\begin{enumerate}
 \item \underline{Convergence of the ground-state energy}:
\begin{equation}\label{eq:bos meta ener}
\frac{E(N)}{N} \underset{N\to \infty}{\longrightarrow} \enls. 
\end{equation}
\item \underline{Convergence of states}: let $(\Psi_N)_{N\in \N}$ be a sequence of approximate minimizers for $E(N)$, i.e. 
$$ \left\langle \Psi_N, H_N \Psi_N \right\rangle = E(N) + o (N)$$
for large $N$. Let $\gamma_N ^{(k)}$ be the corresponding reduced density matrix, defined as in~\eqref{eq:bos reduced DM}. Then 
\begin{equation}\label{eq:bos meta state}
\gamma_N ^{(k)} \wto_* \int |u ^{\otimes k} \rangle \langle u ^{\otimes k} | d\mu (u)  
\end{equation}
weakly-$*$ in the trace-class when $N\to \infty$, where $\mu$ is a Borel probability measure concentrated on the weak limits of minimizing sequences for~\eqref{eq:bos NLS gse}.
\item \underline{Bose-Einstein condensation}: if the infimum in~\eqref{eq:bos NLS gse} is reached at a unique (modulo a constant phase) $\unls$, then 
\begin{equation}\label{eq:bos meta BEC}
\gamma_N ^{(k)} \to |\unls ^{\otimes k} \rangle \langle \unls ^{\otimes k} |  
\end{equation}
strongly in the trace-class when $N \to \infty.$   
\end{enumerate}
\end{mtheorem}

A few comments:
\begin{itemize}
 \item Since this is not supposed to be a rigorous statement yet, we are a bit sloppy on the assumptions. Ideally they reduce to those ensuring that both the $N$-body theory and the one-body NLS theory are reasonably well-defined. In most cases, we manage to get pretty close to that minimal requirement. 
 \item The choice of the effective interaction potential $\weff$ in~\eqref{eq:bosons NLS gen} depends on the value of $\beta$. The larger it is, the harder the analysis, and one ideally wants to be able to push it as far as one can. 
 \item The fact that $H_N \geq -CN$ as an operator will in practice be part of the statement, under specified, more explicit, assumptions on the data of the problem.
 \item Equation~\eqref{eq:bos meta state}, involving a statistical mixtures of NLS ``minimizers'' described by the measure $\mu$, is due to the fact that the latter might not be unique. This case does occur in several physically relevant situations and then one cannot hope for a stronger result. 
 \item The statement~\eqref{eq:bos meta state} is complicated by the fact that $\enls$ might not have minimizers. In this case, minimizing sequences for the NLS problem will only converge weakly in $L^2$. This case does occur if the one-body potential $V$ is not strong enough to bind all particles (loss of compactness at infinity). If all minimizing sequences for $\enls$ converge strongly, the measure $\mu$ charges only minimizers, and since they all have mass $1$, the convergence in~\eqref{eq:bos meta state} is strong in the trace-class.
 \item The result~\eqref{eq:bos meta BEC} implies BEC in the usual sense: the eigenvalues of the one-body density matrix are interpreted as the (normalized) occupation numbers of single-particle orbitals present in a many-body state. The convergence in trace-class to a projector implies that one and only one of these occupation numbers is macroscopic. It corresponds to the single quantum state occupied by most of the particles.
\end{itemize}

Before we proceed to state several rigorous versions of Statement~\ref{thm:bosons meta} in Sections~\ref{sec:bosons MF} and~\ref{sec:bosons dilute}, we explain the general character of the result by listing several key variations on a same theme that are included in the above. We hope this will make it clear that general tools, independent of the details of the Hamiltonian~\eqref{eq:bosons hamil gen}, are needed. Section~\ref{sec:bosons deF} is devoted to the discussion of the most useful of such tools, that is used repeatedly to prove the rigorous versions of Statement~\ref{thm:bosons meta}.

\subsection{Several sources of mathematical difficulties} The ideal case for the above problem is that where $A\equiv 0$, $V$ is trapping, $\beta = 0$ and $w$ is of positive type (positive Fourier transform $\hat{w}\geq 0$). When the Hamiltonian has such a nice structure, one may use it to give a relatively simple proof of Theorem~\ref{thm:bosons meta}. This approach has been pioneered in~\cite{BenLie-83} and is used repeatedly in the literature~\cite{Seiringer-11,GreSei-13,SeiYngZag-12}. If instead $\hat{w}\leq 0$, one can still proceed~\cite{LieYau-87} by using a trick due to L\'evy-Leblond~\cite{LevyLeblond-69,LieThi-84}. One can actually merge the two methods~\cite{Lewin-ICMP} to deal with a general $w$ but, at this stage, things are already considerably trickier. In fact, violating any of the above conditions leads to considerable difficulties, as we now explain.

\begin{enumerate}
\item If $A\neq 0$, the Hoffmann-Ostenhof inequality~\cite{Hof-77} is not available, and one has thus no simple way of connecting the many-body kinetic energy to a one-body object. Moreover, for $A\neq 0$ the bosonic ground state  in general does \emph{not} coincide with the absolute ground state (with no symmetry restriction). Deriving the NLS functional then requires to take bosonic symmetry explicitly into account. On top of that, the minimizer for the limit problem is then generically not unique, which causes difficulties in the proof.
\item If $V$ is not trapping, it is far from obvious that one can pass to the limit in the energy, because a lack of compactness occurs at infinity. Mass may be lost if the potential is not strong enough to bind all particles, leading to a dichotomy situation in the concentration-compactness principle~\cite{Lions-82a,Lions-84,Lions-84b}. This is reflected by the convergence~\eqref{eq:bos meta state} being a priori only weak-$*$. 
\item If $\beta >0$ there are in fact two limits to deal with at the same time: large particle number and short-range interaction. Compactness arguments are then either totally useless, or at least very difficult to employ. Moreover, the singularity of the interaction for large $\beta$ could quite simply render the result false. 
\item Note also that there is a clear physical dichotomy between the cases $\beta \leq 1/d$ where the range of the interaction is much larger than the typical interparticle distance, and $\beta \geq 1/d$ where it is the opposite. In the former case, interactions are frequent but weak, in the latter they are rare but strong. Nevertheless, the general result looks pretty much the same, although the choice of the effective potential can become subtle for large $\beta$.
\item One way to understand the emergence of~\eqref{eq:bosons NLS gen} is that the particles become almost independent and identically distributed. That the interactions take a mean-field form is then a consequence of the law of large numbers. If $w$ has no particular structure, it is not so easy to understand why particles should become approximately independent in the limit. 
\end{enumerate}

Strategies based on structural properties of $H_N$ are particularly sensitive to these difficulties, especially if at least two of them occur at the same time. The general idea in the sequel is to avoid as much as possible using the structure of the Hamiltonian. One can instead rely on a remarkable property of bosonic states themselves, that is of wave-functions satisfying~\eqref{eq:bosons L2 sym bis}. For ease of notation we focus on the model~\eqref{eq:bosons hamil gen} which includes the above sources of difficulties, but let us quickly mention other possible generalizations not explicitly considered here:
\begin{itemize}
 \item The dispersion relation could be more complicated, that is we could replace the (magnetic) Laplacian by a fractional (magnetic) Laplacian or a variant thereof. This is of relevance for relativistic bosonic systems.
 \item The interaction term could be more general. Three-body or $k$-body interactions for arbitrary $k$ could be included in the model. The interaction operator could also be a more general operator, e.g. mixing position and momentum variables. We refer to Section~\ref{sec:MF anyons} where we shall be lead to consider such cases.
 \item The Hilbert space could be restricted, and the Hamiltonian projected accordingly, to model particles constrained by very strong external forces. The particles could live on a lattice instead of in the continuum for example (tight-binding approximation) or in the lowest Landau level (large magnetic field limit).
 \item More generally, at least for a trapped mean-field-like situation corresponding to $\beta = 0$ and $V$ trapping, one could consider a fully abstract model, with a general separable Hilbert space $\gH$ replacing $L^2 (\R^d)$, and an abstract operator acting on~$\gH ^N$. 
 \end{itemize}

\section{Structure of bosonic states: quantum de Finetti theorems}\label{sec:bosons deF}

Roughly speaking, the quantum de Finetti theorem asserts that \emph{any} $N$-body bosonic state looks, for large $N$, as a convex combination (statistical mixture) of pure BECs. Let $\gH$ be any separable complex Hilbert space and  $\Gamma_N$ a (sequence of) bosonic $N$-body states. Thus $\Gamma_N$ is a self-adjoint, positive, trace-class operator on $\gH ^N = \bigotimes_{\rm sym} \gH$, e.g. $\Gamma_N = |\Psi_N \rangle \langle \Psi_N |$ for $\Psi_N \in \gH ^N$. Then, modulo extraction of a subsequence, one should have in mind an approximation of the form
\begin{equation}\label{eq:deF inform}
\Gamma_N \underset{N\to \infty}{\simeq} \int_{\gH} |u ^{\otimes N} \rangle \langle u ^{\otimes N} | d\mu (u)
\end{equation}
with $\mu$ a Borel probability measure on the one-body Hilbert space $\gH$. Thus, for large $N$, there are so to say \emph{only} pure BECs in the $N$-body state space. Of course the latter is convex, so it must also obviously contain convex combinations of BECs, as in the right-hand of~\eqref{eq:deF inform}, but these are, for many practical purposes, the only relevant possibilities. The energy functional appearing in~\eqref{eq:bosons GSE gen} being a linear function of $|\Psi_N \rangle \langle \Psi_N |$, it becomes intuitively clear that it is energetically favorable to have a single BEC as an approximate ground state, provided one is able to pass to the large $N$ limit in the energy.

There is actually a condition that $H_N$ should satisfy, namely it should be a few-body Hamiltonian including only $k$-body terms for $k$ fixed, or at least $k$ small compared to $N$. This is because the correct meaning one should give to~\eqref{eq:deF inform} involves reduced density matrices: let 
\begin{equation}\label{eq:bos red DM bis}
\gamma_N ^{(k)} := \tr_{k\to N} \left[ \Gamma_N \right]. 
\end{equation}
The true statement is that, modulo extraction of a subsequence independent of $k$, 
\begin{equation}\label{eq:deF inform bis}
\gamma_N ^{(k)}  \underset{N\to \infty}{\simeq} \int_{\gH} |u ^{\otimes k} \rangle \langle u ^{\otimes k} | d\mu (u)
\end{equation}
where the measure $\mu$ does not depend on $k$. This kind of result is a quantum analogue of the famous classical de Finetti/Hewitt-Savage theorem~\cite{DeFinetti-31,DeFinetti-37,HewSav-55,DiaFre-80,HauMis-14,MisMou-13,Lions-CdF,Aldous-85,Kallenberg-05}, used extensively to deal with classical mean-field limits. The quantum version originates in~\cite{Stormer-69,HudMoo-75}. Recent variants and improvements may be found in~\cite{AmmNie-08,AmmNie-09,AmmNie-11,AmmNie-15,ChrKonMitRen-07,FanVan-06,KonRen-05,Chiribella-11,Harrow-13}. In this section we describe two results obtained jointly with Mathieu Lewin and Phan Th\`anh Nam, respectively in~\cite{LewNamRou-14b} and~\cite{LewNamRou-14}. Put together, these lead to a fully constructive proof of (the rigorous version of)~\eqref{eq:deF inform bis}. I refer to my lecture notes~\cite{Rougerie-LMU,Rougerie-cdf} for a more detailed exposition of this proof, which is reminiscent of a semi-classical approach due to Ammari and Nier~\cite{Ammari-hdr,AmmNie-08,AmmNie-09}. 

\subsection{A finite dimensional quantitative version}

The main content of~\cite{LewNamRou-14b} is a new proof of results originating in~\cite{ChrKonMitRen-07,Chiribella-11}. It has also been found independently by Lieb and Solovej~\cite{LieSol-15} (for other purposes than the quantum de Finetti theorem), and yet another proof appeared in~\cite{Harrow-13}. Variants of the result have also been considered~\cite{CavFucSch-02,ChrKonMitRen-07,ChrTon-09,RenCir-09,Renner-07,BraHar-12}.

A question one might want to ask is whether~\eqref{eq:deF inform bis} can be made quantitative. The answer is known in the classical setting~\cite{DiaFre-80,HauMis-14}: it is possible to construct an explicit measure for which explicit error estimates can be proven. Unfortunately, this construction cannot be used in a quantum setting (it is based on empirical measures, a purely classical concept). A quantum analogue of the Diaconis-Freedman theorem is available only in the case where the one-body Hilbert space is finite dimensional:

\begin{theorem}[\textbf{Quantitative quantum de Finetti theorem}]\label{thm:deF finite dim}\mbox{}\\
Let $\gH$ be a finite dimensional complex Hilbert space 
$$ \dim (\gH) = d.$$
Let $\Gamma_N \in \gS ^1 (\gH_\sym ^N)$ be a $N$-body bosonic state and $\gamma_N ^{(k)}$ its reduced density matrices. Let $\mu_N \in \cP (S\gH)$ be the  probability measure 
\begin{align}\label{eq:def CKMR}
d\mu_N (u) &:=  \dim \gH ^N \tr_{\gH ^N} \left[\Gamma_N | u^{\otimes N} \rangle  \langle u^{\otimes N} |\right] du \nonumber\\
&= \dim \gH_\sym^N \tr_{\gH ^N}  \left\langle u^{\otimes N}, \Gamma_N u^{\otimes N} \right\rangle  du, 
\end{align}
with $du$ the Lebesgue measure on the unit sphere $S\gH$ of $\gH$.

Denoting 
\begin{equation}\label{eq:finite deF etat}
\Gammat _N := \int_{u\in S \gH} |u ^{\otimes N}\rangle \langle u ^{\otimes N}| d\mu_N(u) 
\end{equation}
the associated state and $\gammat_N ^{(k)}$ its reduced density matrices, we have
\begin{equation}\label{eq:CKMR exact}
\gammat _N^{(k)} = {{N+k+d-1}\choose k}^{-1}\sum_{\ell=0}^{k} {N \choose \ell}  \gamma_N^{(\ell)} \otimes _s \1_{\gH^{n-\ell}}.
\end{equation}
This implies in particular the bound
\begin{equation} \label{eq:error finite dim deF}
\Tr_{\gH^k} \Big| \gamma_N ^{(k)} - \widetilde \gamma _N^{(k)}  \Big| \leq \frac{2 k(d+2k)}{N}
\end{equation}
for all $k=1 \ldots N$.
\end{theorem}

A few comments:

\begin{itemize}
 \item As mentioned previously, the result is not new. The construction was first introduced in~\cite{ChrKonMitRen-07}, where a bound similar to~\eqref{eq:error finite dim deF} is derived. The explicit formula~\eqref{eq:CKMR exact} originates in~\cite{Chiribella-11}. In~\cite{LewNamRou-14b} we re-derived the result using an analogy with well-known semi-classical ideas (coherent states, upper and lower symbols, cf~\cite{Lieb-73b,Simon-80}) and the second-quantized formalism (the CCR algebra).
 \item The error bound~\eqref{eq:error finite dim deF} easily follows from the explicit formula~\eqref{eq:CKMR exact}. The latter does come in handy in some cases, e.g. in the proof of some results of Chapter~\ref{sec:Gibbs} below.  
  \item Note that the estimate~\eqref{eq:error finite dim deF} is interesting mostly when $d\ll N$ (think of $k$ as being fixed), i.e. when there are many more particles than available one-body states. It is rather intuitive that something special should happen in this case.
  \item It is very satisfying that one has a simple explicit construction such as~\eqref{eq:def CKMR}. The error bound we obtain can be shown to be (essentially) optimal with this construction, and it remains an open problem to see whether another one would improve the error. 
\item However, the particular form of the approximate de Finetti measure used above is crucial in the proof of the results of Chapter~\ref{sec:Gibbs}. In short, this has the form of a lower symbol in a coherent state basis, which allows to apply (variants of) the celebrated Berezin-Lieb inequalities~\cite{Berezin-72,Lieb-73b,Simon-80}.
\end{itemize}

\subsection{The weak quantum de Finetti theorem} 

For infinite dimensional one-body Hilbert spaces such as $L ^2 (\R ^d)$,  no explicit error bound quantifying the precision of~\eqref{eq:deF inform bis} is known at present. One has to be content with formulating~\eqref{eq:deF inform bis} as a limit statement, or as an equality in the case of states having infinitely many particles. This formulation is actually the original one~\cite{Stormer-69,HudMoo-75}. It provides useful information on the limit of $N$-body states \emph{only if} one is able to pass to the strong $\gS ^1 (\gH ^k)$ - limit in $\gamma_N  ^{(k)}$, where $\gS ^1 (\gH ^k)$ denotes the trace class. This requirement can turn into a severe restriction in applications: if the system under consideration is not confined, mass may be lost in the limit, and one cannot hope for a strong limit.

To remedy this, we introduced and proved in~\cite{LewNamRou-14} a so-called weak quantum de Finetti theorem, which generalizes the original statement to the case  where one is able to pass to the weak-$*$ limit\footnote{Whence the name we coined. Beware that the weak theorem is actually stronger than the usual version, refered to as ``strong quantum de Finetti'' in~\cite{LewNamRou-14} and subsequent papers.}. Our proof is constructive, reducing the infinite dimensional statement to the finite dimensional case via ``geometric localization'' arguments as presented in~\cite{Lewin-11}. 

\begin{theorem}[\bf Weak quantum de Finetti]\label{thm:deF weak}\mbox{}\\
Let $\gH$ be a complex separable Hilbert space and $(\Gamma_N)_{N\in \N}$ a sequence of bosonic states with $\Gamma_N \in \gS ^1 (\gH ^N)$. We assume that for all $k\in\N$ 
\begin{equation}\label{eq:def faible convergence}
\gamma_N ^{(k)} \wto_\ast \gamma ^{(k)}
\end{equation}
in $\gS ^1 (\gH ^k)$. Then there exists a unique probability measure $\mu \in \cP (B\gH)$ on the unit ball $B\gH = \left\{ u \in\gH, \norm{u} \leq 1\right\}$ of $\gH$, invariant under the action of $S^1$, such that
\begin{equation}
\gamma^{(k)}=\int_{B\gH}|u^{\otimes k}\rangle\langle u^{\otimes k}| \, d\mu(u)
\label{eq:melange faible}
\end{equation}
for all $k\geq0$.
\end{theorem}

A few comments:
\begin{itemize}
\item The above essentially describes the generic case : given a sequence of $N$-body states, one can always extract diagonal subsequences to ensure that~\eqref{eq:def faible convergence} holds. 
\item When the convergence in~\eqref{eq:def faible convergence} is strong, the mass is conserved in the limit, and thus $\mu$ has to be concentrated on the unit sphere of $\gH$. The converse is also true. Also, since $\mu$ does not depend on $k$, we see that strong convergence for one particular $k\geq 1$ implies convergence for all $k$.
\item The constructive proof via geometric localization yields corollaries that turn useful when dealing with a mean-field problem with loss of mass at infinity.
\item Ammari and Nier have more general results~\cite{Ammari-hdr,AmmNie-08,AmmNie-09,AmmNie-11}. In particular, it is not necessary to start from a state with a fixed particle number. One can consider a state on Fock space provided suitable bounds on its particle number (seen as a random variable in this framework) are available. See Chapter~\ref{sec:Gibbs} below for a discussion of this situation.
\end{itemize}

\section{Mean-field limits}\label{sec:bosons MF}

We now turn to stating rigorous versions of the meta-Theorem~\ref{thm:bosons meta} in two cases where inter-particles interactions are frequent but weak. One usually refers to this situation as the \emph{mean-field regime}. Since every particle interacts weakly with essentially all the others, one can hope that correlations effects will average out. Then particles should behave independently, and their interactions converge to a mean self-consistent effect, because of the law of large numbers. The quantum de Finetti theorem discussed above is the main tool we use to substantiate this picture. 

\subsection{Mean-field regime} The main focus in the paper~\cite{LewNamRou-14} is the case where the one-body potential $V$ is not trapping. Then, some particles can fly apart and the issue is to prove that those which stay bound do form a Bose-Einstein condensate. Note that there might be several clusters of bound particles: those staying in the well of the one-body potential, and those who are bound by the attractive part of the interaction potential if it is strong enough. The proof of the following theorem requires to show that all these clusters form BECs.

In this section we have to rely on several (concentration-)compactness arguments. We can thus not deal with two limits at the same time and take $\beta = 0$ in~\eqref{eq:bosons hamil gen}. As for the external and interaction potentials $A,V,w$ we make standard assumptions~\cite{ReeSim2,ReeSim4} ensuring self-adjointness of the many-body Hamiltonian 
\begin{equation}\label{eq:bos decay V}
V = f_1 + f_2 ,\quad
|A| ^2 = f_3 + f_4,\quad
w = f_5 + f_6
\end{equation}
where $f_1,f_3,f_5 \in L ^p (\R ^d)$ for $\max(1,d/2) \leq p <\infty$ and $f_2,f_4,f_6 \in L ^{\infty} (\R ^d)$ and decay at infinity. To state the result, we  denote 
\begin{equation}\label{eq:bos Hartree}
\EH [u] := \int_{\R ^d } \left| \left( -i\nabla + A\right)u \right| ^2 + V |u| ^2 + \frac{1}{2} \iint_{\R ^d \times \R ^d} |u(x)| ^2 w (x-y) |u(y)|^2 dxdy  
\end{equation}
the appropriate NLS functional, usually called Hartree functional in this context where the interaction is non-local. To describe particles escaping the potential wells, we also need the functional without external fields:
\begin{equation}\label{eq:bos Hartree 0}
\EH ^0 [u] := \int_{\R ^d } \left| \nabla u \right| ^2 + \frac{1}{2} \iint_{\R ^d \times \R ^d} |u(x)| ^2 w (x-y) |u(y)|^2 dxdy.  
\end{equation}
The mass of minimizing sequences can be split into several clusters. Accordingly, we denote 
\begin{align}\label{eq:bos Hartree min ener}
\eH (\lambda) &= \inf \left\{ \EH [u], \int_{\R ^d} |u| ^2 = \lambda \right\}\nonumber\\
\eH ^0 (\lambda) &= \inf \left\{ \EH ^0 [u], \int_{\R ^d} |u| ^2 = \lambda \right\}.\nonumber\\
\end{align}
We may now state (a particular case of) the main result of~\cite{LewNamRou-14}, where $H_N$ and the corresponding $E(N)$ are defined as in Section~\ref{sec:bosons intro}:

\begin{theorem}[\textbf{Derivation of Hartree's theory in the mean-field regime}]\label{thm:bos hartree}\mbox{}\\
Under the preceding assumptions, we have:

\medskip

\noindent$(i)$ \underline{Convergence of energy.}
\begin{equation}
\lim_{N\to\ii}\frac{E(N)}{N}=\eH (1).
\label{eq:general ener}
\end{equation}

\medskip

\noindent $(ii)$ \underline{Convergence of states.} Let $\Psi_N$ be a sequence of $L ^2 (\R ^{dN})$-normalized quasi-minimizers for $H_N$:  
\begin{equation}\label{eq:quasi min}
\pscal{\Psi_N,H_N \Psi_N}= E (N)+o(N) \mbox{ when } N \to \infty, 
\end{equation}
and $\gamma^{(k)}_N$ be the corresponding reduced density matrices. There exists $\mu \in \cP (B\gH)$ a probability measure on the unit ball of  $\gH$ with $\mu (\cM ) = 1$, where
\begin{equation}
\cM =\Big\{u\in B\gH\ :\ \EH [u]=\eH(\norm{u}^2)=\eH(1)-\eH^0(1-\norm{u}^2)\Big\},
\label{eq:def_M_V} 
\end{equation}
such that, along a subsequence, 
\begin{equation}
\gamma^{(k)}_{N_j}\wto_*\int_{\cM}|u^{\otimes k}\rangle\langle u^{\otimes k}|\,d\mu(u)
\label{eq:general state}
\end{equation}
weakly-$\ast$ in $\gS^1(\gH^k)$, for all $k\geq1$.

\medskip

\noindent $(iii)$ If in addition the strict binding inequality 
\begin{equation}
\eH(1)<\eH(\lambda)+\eH^0(1-\lambda)
\label{eq:hartree liaison stricte}
\end{equation}
holds for all $0\leq\lambda<1$, the measure $\mu$ has its support in the sphere $S\gH$ and the limit~\eqref{eq:general state} holds in trace-class norm. In particular, if $\eH(1)$ has a unique (modulo a constant phase) minimizer $u_H$, then for the whole sequence
\begin{equation}
\gamma^{(k)}_{N}\to|u_H^{\otimes k}\rangle\langle u_H^{\otimes k}| \mbox{  strongly in } \gS ^1(\gH ^k)
\label{eq:general BEC}
\end{equation} 
for all fixed $k\geq1$.
\end{theorem}

\begin{itemize}
\item The main originality here is to deal with the case where mass is lost at infinity. When $V$ is trapping, similar results had been obtained previously~\cite{FanSpoVer-80,PetRagVer-89,RagWer-89,Werner-92}.
\item The measure appearing in the above is of course the de Finetti measure defined by Theorem~\ref{thm:deF weak}. The latter, together with its constructive proof and corollaries, is basically the only tool needed to prove Theorem~\ref{thm:bos hartree}. 
\item It is a classical fact that $\cM$  coincides with the set of weak limits of minimizing sequences for~\eqref{eq:bos Hartree}.
% 
% $$\cM  = \left\lbrace u\in B\gH \: | \: \exists (u_n) \mbox{ minimizing sequence for } \eH (1) \mbox{ such that } u_n \wto u \mbox{ weakly in } L ^2 (\R ^d) \right\rbrace.$$
This follows from the usual concentration-compactness method for non-linear one-body variational problems. Thus, the weak limits for the $N$-body problem are fully characterized in terms of the weak limits for the one-body problem. Also, for large $N$, the HVZ-like binding inequalities for the many-body problem reduce to binding inequalities for the Hartree functional. 
\item It is certainly possible for the convergence in the above statement to be only weak-$\ast$, which covers a certain physical reality. If the one-body potential is not attractive enough to retain all particles, we will typically have a scenario where
\[
\begin{cases}
\eH(\lambda) < \eH (1) \mbox{ for } 0 \leq \lambda < \lambda_c \\
\eH (\lambda) = \eH (1) \mbox{ for } \lambda_c \leq \lambda \leq 1 
\end{cases}
\]
where $\lambda_c$ is a critical mass that can be bound by the potential $V$. In this case, $\eH(\lambda)$ will not be achieved if $\lambda_c < \lambda \leq 1$ and one will have a minimizer for Hartree's energy only for a mass $0 \leq \lambda \leq \lambda_c$. If for example the minimizer  $\uH$ at mass $\lambda_c$ is unique modulo a constant phase, the above result implies 
\[
\gamma^{(k)}_{N} \wto_\ast  |\uH^{\otimes k}\rangle\langle \uH^{\otimes k}|
\]
and one should note that the limit has a mass $\lambda_c ^k <1$. This scenario happens for example in the case of a ``bosonic atom''~\cite{BenLie-83,Bach-91,BacLewLieSie-93,Solovej-90,Hogreve-11}. 
\item Note that, in actual experiments, one-body potentials that decay at infinity are always used to cool the gas below the critical temperature for BEC, by allowing the hotter atoms to escape the trap (evaporative cooling). Such a situation can be modeled by the above setting.

\end{itemize}

\subsection{Weak NLS regime}\label{sec:weak NLS}

We next consider the case of small but positive $\beta$ in~\eqref{eq:bosons hamil gen}. Then the effective interaction is short-range (a Dirac delta), although the interaction range is much larger than the typical interparticle distance. This is still a mean-field situation, but the fact that two limits need to be dealt with at the same time adds some difficulties. In particular, a quantitative argument is required: one cannot simply use compactness arguments to pass to the limit. The estimates we rely on are those of Theorem~\ref{thm:deF finite dim}. Since they are valid only in finite dimension we need a way to project the many-body Hamiltonian and control the error thus made. In~\cite{LewNamRou-14c} we project on the low-lying eigenstates of the one-body Hamiltonian and control their number (which enters the error in~\eqref{eq:error finite dim deF}) via a semi-classical inequality \`a la Cwikel-Lieb-Rosenblum (CLR)~\cite{Cwikel-77,Lieb-80,Rosenbljum-76}. This method requires $V$ to be trapping, but does 
not use any sophisticated estimate 
beyond Theorem~\ref{thm:deF finite dim} and the CLR inequality. It is thus very general, and pays the price for this by being limited to rather small $\beta$, i.e. not too singular interactions. We shall go beyond this limitation in Section~\ref{sec:bosons dilute} below. 

We here deal with the case of trapped bosons, assuming that for some constants~$c,C~>~0$   
\begin{equation}\label{eq:GP asum V}
V(x) \geq c |x| ^s - C.
\end{equation}
In this case, the one-body Hamiltonian $-\Delta + V$ has a discrete spectrum on which we have a good control thanks to Lieb-Thirring-like inequalities. 
As for the vector potential $A$, that can model a  magnetic field or Coriolis forces acting on rotating particles, it is sufficient to assume 
\begin{equation}\label{eq:asum A}
A \in L ^2_{\rm loc} (\R ^d, \R ^d).
\end{equation}
 %\marginpar{\tt Vector potential modified depending on $d$.}
A particularly relevant example is given by $A(x) = \Omega (-x_2,x_1,0)$ for $d=3$ and $A(x)=\Omega(-x_2,x_1)$ for $d=2$, corresponding to Coriolis forces due to a rotation at speed $\Omega$ around the $x_3$-axis, or a constant magnetic field of strength $\Omega$ pointing in this direction.

Since $w_{\beta,N}$ becomes singular in the limit anyway,  we make comfortable assumptions on the unscaled potential $w$: 
\begin{equation}\label{eq:GP asum w}
w\in L^\infty(\R^d,\R) \mbox{ and } w(x) = w(-x) .
\end{equation}
We also assume that 
\begin{equation}\label{eq:replace delta}
x \mapsto (1+|x|) w (x) \in L ^1 (\R ^d), 
\end{equation}
which simplifies the replacement of $w_{\beta,N}$ by a Dirac mass to obtain the limit functional
\begin{equation}\label{eq:bos nls mf}
\Enls [u] := \int_{\R ^d } \left| \left( -i\nabla + A\right)u \right| ^2 + V |u| ^2 + \frac{a}{2} \int_{\R ^d } |u(x)| ^4   
\end{equation}
and the associated ground state energy
\begin{equation}\label{eq:bos nls mf gse}
\enls := \inf\left\{ \Enls [u],\quad \int_{\R ^d} |u| ^2 = 1  \right\}. 
\end{equation}
The latter being finite depends on the sign of the effective coupling constant, which will in this section be taken as
$$ a = \int_{\R ^d} w. $$
In order for the limit problem to be well-defined, we need some assumptions on $w$. Indeed, if $a\leq 0$, the interaction is attractive and it is not obvious that the quantum kinetic energy is strong enough to stabilize the system. It in fact turns out that the validity of the mean-field approximation demands slightly stronger assumptions than those ensuring that~\eqref{eq:bos nls mf gse} to be finite:

\smallskip

\noindent$\bullet$ When $d=3$, a ground state for~\eqref{eq:bos nls mf gse} exists if and only if $a \geq 0$. This is because the cubic non-linearity is super-critical\footnote{One may for example consult~\cite{KilVis-08,Tao-06} for a classification of non-linearities in the NLS equation.} in this case: the interaction energy dominates the quantum kinetic energy in stability issues. It is moreover easy to see that $N ^{-1} E(N) \to -\infty$ if $w$ is negative at the origin. The optimal assumption happens to be a classical stability notion for the interaction potential:
\begin{equation}\label{eq:GP hyp 3}
\iint_{\R ^d \times \R ^d} \rho (x) w(x-y) \rho(y) dx dy \geq 0, \mbox{ for all } \rho \in L^1 (\R ^d), \: \rho \geq 0. 
\end{equation}
This is satisfied as soon as $w= w_1 + w_2$, $w_1 \geq 0$ and $\hat{w_2}\geq 0$ with $\hat{w_2}$ the Fourier transform of $w_2$. This assumption clearly implies $\int_{\R ^d} w \geq 0$, and one may easily see by changing scales that if it is violated for a certain $\rho \geq 0$, then $E(N)/N \to -\infty$.

\medskip

\noindent$\bullet$ When $d=2$, the cubic non-linearity is critical, i.e. it competes with the quantum kinetic energy at the same order of magnitude. A minimizer for~\eqref{eq:bos nls mf gse} exists if and only if $a > - a ^*$ with 
\begin{equation}\label{eq:GP a star}
a ^* = \norm{Q}_{L ^2} ^2  
\end{equation}
where $Q\in H ^1 (\R^2)$ is the unique~\cite{Kwong-89} (modulo translations) solution to 
\begin{equation}\label{eq:GP defi Q}
-\Delta Q + Q - Q ^3 = 0.  
\end{equation}
The critical interaction parameter $a ^*$ is the best possible constant in the interpolation inequality
\begin{equation}\label{eq:GP interpolation}
\int_{\R ^2} |u| ^4 \leq C \left(\int_{\R^2} |\nabla u | ^2\right) \left(\int_{\R ^2} |u | ^2\right). 
\end{equation}
We refer to~\cite{GuoSei-13,Maeda-10} for the existence of a ground state and to~\cite{Weinstein-83} for the inequality~\eqref{eq:GP interpolation}. A pedagogical discussion may be found in~\cite{Frank-14}.

In view of the above conditions, it is clear that in 2D we have to assume $\int w \geq - a ^*$, but this is in fact not sufficient: as in 3D, if the interaction potential is sufficiently negative at the origin, one may see that $N ^{-1} E(N) \to -\infty$. The appropriate assumption is now
\begin{equation}\label{eq:GP hyp 2}
\norm{u}_{L^2}^2\norm{\nabla u}_{L^2}^2+\frac12\iint_{\R^2\times \R^2}|u(x)|^2|u(y)|^2 w(x-y)\,dx\,dy > 0
\end{equation}
for all $u\in H^1(\R^2)$. Replacing $u$ by $\lambda u(\lambda x)$ and taking the limit $\lambda\to 0$ shows that this implies
$$\int_{\R^2}w(x)\,dx\geq-a^*.$$
If the strict inequality in~\eqref{eq:GP hyp 2} is reversed for a certain $u$, then $E(N)/N \to - \infty$. The case where equality may occur in~\eqref{eq:GP hyp 2} is left aside here. It requires a more thorough analysis, see~\cite{GuoSei-13} where this is provided at the level of the NLS functional.

\noindent$\bullet$ When $d=1$, the cubic non-linearity is sub-critical and there always exists a NLS minimizer. In this case we need no further assumption because the quantum kinetic energy always dominates the interaction energy. 

\medskip

The following theorem is proved in~\cite{LewNamRou-14c}:

\begin{theorem}[\textbf{Weak NLS limit of the many-body problem}]\label{thm:bos nls mf}\mbox{}\\
Assume that either $d=1$, or $d=2$ and~\eqref{eq:GP hyp 2} holds, or $d= 3$ and~\eqref{eq:GP hyp 3} holds. For $d\geq 2,$ further suppose that 
\begin{equation}\label{eq:GP beta 0}
0 < \beta \leq \beta_0 (d,s) :=\frac{s}{2ds + s d^2 + 2d ^2} < 1.
\end{equation}
where $s$ is the exponent appearing in~\eqref{eq:GP asum V}. We then have  
\begin{enumerate}
\item \underline{Convergence of the energy:}
\begin{equation}\label{eq:GP ener convergence}
\frac{E(N)}{N} \to \enls \mbox{ when } N\to \infty.
\end{equation}
\item \underline{Convergence of states:} Let $\Psi_N$ be a ground state of~\eqref{eq:bosons hamil gen} and $ \gamma_N ^{(k)}$
its reduced density matrices. Along a subsequence we have, for all $k\in \N$,  
\begin{equation}\label{eq:GP state convergence}
\lim_{N\to \infty}\gamma_N ^{(k)} = \int_{u\in \Mnls} d\mu (u) |u ^{\otimes n} \rangle \langle u ^{\otimes n}| 
\end{equation}
strongly in  $\gS ^1 (L ^2 (\R ^{dk}))$. Here $\mu$ is a probability measure supported on
\begin{equation}\label{eq:GP nls set}
\Mnls = \left\{ u \in L ^2 (\R ^d), \norm{u}_{L ^2} = 1, \Enls [u] = \enls \right\}. 
\end{equation}
In particular, when~\eqref{eq:bos nls mf} has a unique minimizer $\unls$ (up to a constant phase), we have for the whole sequence
\begin{equation}\label{eq:GP BEC}
\lim_{N\to \infty} \gamma_N ^{(k)} =|\unls ^{\otimes k} \rangle \langle \unls ^{\otimes k}|.
\end{equation}
\end{enumerate}
\end{theorem}

\begin{itemize}
 \item Apart from the $d=1$ case (where there is no restriction), the assumption~\eqref{eq:GP beta 0} we make on $\beta$ is rather strong. It can be improved a little bit, see the original paper for a discussion. In any case, the methods of~\cite{LewNamRou-14c} always require $\beta < 1/d$, i.e. a mean-field-like regime.
 \item As previously mentioned, the assumptions on $w$ are close to optimal, thanks to the method of proof based on the quantum de Finetti theorem. In particular, the above result seems to be the first derivation of ground states with attractive interactions. This is particularly tricky in 2D, since such minimizers exist only for small enough values of the effective coupling constant.
\end{itemize}

\section{Dilute limits}\label{sec:bosons dilute}

The next step of the program discussed in this chapter is to consider dilute limits where interactions are rare but strong. In other words: improve the range of $\beta$ obtained in Theorem~\ref{thm:bos nls mf}. This is very important in practical applications: in the cold atoms context, the emergence of BEC in interacting Bose gases has much more to do with diluteness of the quantum gases under consideration than with the weakness of the interactions. 

The mathematical state of the art on this question, prior to the results discussed below, is contained in a seminal series of papers by Lieb, Seiringer and Yngvason~\cite{LieYng-98,LieYng-01,LieSeiYng-00,LieSeiYng-01,LieSeiYng-02b,LieSeiYng-09,LieSeiYng-05,LieSei-02,LieSei-06,Seiringer-03}, parts of which are reviewed in~\cite{LieSeiSolYng-05,Seiringer-ICMP10,Yngvason-10}. In this approach, structural properties of the Hamiltonian play a key role in the proofs. In particular, the interaction potential is required to be positive. As we mentioned previously, in the dilute regime, delicate estimates are required, and we cannot avoid relying on the structure of the Hamiltonian to prove these. Nevertheless, the use of the quantum de Finetti theorem simplifies a lot the proofs. 

In the following we present two results, proved respectively in~\cite{NamRouSei-15} and~\cite{LewNamRou-15}: 
\begin{itemize}
 \item We first revisit the so-called Gross-Pitaevskii limit  ($\beta =1$) for the ground state of a trapped 3D Bose gas. Relying on the quantum de Finetti theorem and a new two-body estimate we simplify the Lieb-Seiringer-Yngvason approach of this problem. Interestingly, the new approach is rather insensitive to the presence of a gauge field $A$, which is a considerable simplification. Indeed, improving the results of~\cite{LieSeiYng-00} to the case of non zero $A$ was a technical tour de force  in~\cite{LieSei-06}.
 \item Next we discuss the case of the attractive 2D Bose gas. Using specific two-body estimates and a bootstrap argument, we can improve the range of $\beta$ available to $\beta < (s+1)/(s+2)$ with $s$ as in~\eqref{eq:GP asum V}. Note that~\cite{LieYng-01,LieSeiYng-01} can go much further in the case of repulsive interactions. The novelty of~\cite{NamRouSei-15} is thus to reach the dilute limit (it is always possible to take $\beta > 1/d = 1/2$ with this method) in the case of the attractive 2D Bose gas.  
\end{itemize}

\subsection{3D case, the Gross-Pitaevskii regime}
In this section we start again from~\eqref{eq:bosons hamil gen}, with $d=3$. The one-body potential $V$ shall be trapping, but we do not need a specific assumption on its growth at infinity:
\begin{align} \label{eq:assumption-V}
0\le V\in L^1_{\rm loc}(\R^3),  \quad \lim_{|x|\to \infty} V(x)=+\infty. 
\end{align}
As for the gauge vector potential $A$, only mild assumptions are required:
\begin{align} \label{eq:assumption-A}
A\in L^3_{\rm loc}(\R^3,\R^3),\quad \lim_{|x|\to \infty} |A(x)| e^{-b|x|}=0
\end{align}
for some constant $b>0$. For simplicity, we assume that the interaction potential $w$ is a fixed function which is non-negative, radial and of finite range, $$\1(|x|>R_0)w(x)  \equiv 0  \mbox{ for some constant } R_0>0.$$
The main point in this section is that we shall take $\beta = 1$. This is essentially the largest relevant value of $\beta$ in 3D, and the result changes qualitatively in this regime, as we now explain.

For $\beta = 1$, the naive approximation 
\begin{equation}\label{eq:w delta}
 N ^{3\beta} w (N ^{\beta} x) \underset{N\to \infty}{\to} \left(\int_{\R ^3} w\right) \delta_0 
\end{equation}
does {\em not} capture the leading order behavior of the physical system. In fact, strong correlations at short distances $O(N^{-1})$ yield a nonlinear correction, which essentially amounts to replacing the coupling constant $\int w$ by $(8\pi) \times $ the scattering 
length of $w$, whose definition we now recall (see~\cite[Appendix~C]{LieSeiSolYng-05}): Under our assumption on $w$, the zero-energy scattering equation 
$$ (-2\Delta+w(x))f(x) = 0, \quad \lim_{|x|\to \infty}f(x)=1,$$
has a unique solution and it satisfies
$$ f(x)=1-\frac{a}{|x|},~~\forall |x|>R_0$$
for some constant $a \ge 0$ which is called the {\em scattering length} of $w$. In particular, if $w$ is the potential for hard spheres, namely $w(x)\equiv \black{+\infty}$ when $|x|<R_0$ and $w(x)\equiv \black{0}$ when $|x|\ge R_0$, then the scattering length of $w$ is exactly $R_0$. In a dilute gas, the scattering length can be interpreted as an effective range of the interaction: a quantum particle far from the others is felt by them as a hard sphere of radius $a$. A useful variational characterization of $a$ is as follows:
\begin{equation}\label{eq:var scat}
8\pi a = \inf\left\lbrace \int_{\R ^3} 2 |\nabla f| ^2 + w |f| ^2, \quad \lim_{|x|\to \infty}f(x)=1 \right\rbrace. 
\end{equation}
Consequently, $8\pi a$ is smaller than $\int w$. Moreover, a simple scaling shows that the scattering length of $w_N=N^2 w(N.)$ is $a/N$, so that the scenario $\beta =1$ corresponds to fixing the scattering length of the interaction potential, the so-called Gross-Pitaevskii limit. The limit problem then consists in a NLS functional where the strength of the interactions is set by the scattering length:
\begin{equation}\label{eq:GP func}
 \cE_{\rm GP}[u]:= \int_{\R ^3} \left| \left( -i\nabla + A \right) u \right| ^2 + 4\pi a \int_{\R^3} |u(x)|^4 dx 
\end{equation} 
Note that this functional is \emph{not} obtained by testing $H_N$ with factorized states of the form~$u ^{\otimes N}$. We shall denote  
\begin{equation}\label{eq:GP ener}
e_{\rm GP} := \inf \left\{ \cE_{\rm GP}[u] , \quad \int_{\R^2} |u| ^2 =1 \right\} 
\end{equation}
the Gross-Pitaevskii ground state energy.

\begin{theorem}[\textbf{Derivation of the Gross-Pitaevskii functional}]\label{thm:cv-GP}\mbox{}\\
Under the above assumptions on $A,V$ and $w$, we have, for $\beta = 1$,
\begin{align} \label{eq:cv-energy}
\frac{E(N)}{N} \to e_{\rm GP} \mbox{ when } N \to \infty. 
\end{align}
Moreover, if $\Psi_N$ is an approximate ground state for $H_N$, namely
$$
\lim_{N\to \infty}  \frac{\langle \Psi_N, H_N \Psi_N\rangle}{N} = e_{\rm GP},
$$
then there exists a Borel probability measure $\mu$ supported on the set of minimizers of $\cE_{\rm GP}(u)$ such that, along a subsequence, 
\begin{equation}\label{eq:GP state convergence bis}
\lim_{N\to \infty}\gamma_N ^{(k)} = \int d\mu (u) |u ^{\otimes n} \rangle \langle u ^{\otimes n}|,\quad \forall k\in \mathbb{N} 
\end{equation}
strongly in  $\gS ^1 (L ^2 (\R ^{3k}))$ where $\gamma_{N}^{(k)}$ is the $k$-particle  reduced density matrix of $\Psi_{N}$. In particular, if $\cE_{\rm GP}(u)$ subject to $\|u\|_{L^2}= 1$ has a unique minimizer $u_0$ (up to a complex phase), then there is complete Bose-Einstein condensation
\begin{align} \label{eq:BEC}
\lim_{N\to \infty}\Tr \left| \gamma_{N}^{(k)} - |u_0^{\otimes k} \rangle \langle u_0^{\otimes k}| \right| =0,\quad \forall k\in \mathbb{N}.
\end{align}
\end{theorem}

As we mentioned previously, the result in itself is not new (only~\eqref{eq:GP state convergence bis} is for $k\geq 2$): it was first proved in~\cite{LieSeiYng-00} in the case where $A\equiv 0$ and then extended to non-zero $A$ in~\cite{LieSei-06}. Both proofs, in particular that for $A\neq 0$, are pretty difficult, and the main contribution in~\cite{NamRouSei-15} is to provide substantial simplifications. We shall discuss only the energy lower bound and convergence of states proofs, the energy upper bound being taken directly from~\cite{LieSeiYng-00,Seiringer-03}.

The main difficulty in dealing with the GP limit is that an ansatz $u ^{\otimes N}$ does \emph{not} give the correct energy asymptotics. In this regime, correlations between particles \emph{do} matter, and one should rather think of an ansatz of the form 
\begin{equation}\label{eq:GP ansatz}
 \prod_{i=1} ^N u(x_i) \prod_{1\leq i< j \leq N} f (x_i-x_i), 
\end{equation}
or a close variant, where $f$ is linked to the two-body scattering process. In the time-dependent literature, one may find different ways of taking this most crucial fact into account, see~\cite{ErdSchYau-09,ErdSchYau-10,BenOliSch-12,Pickl-15}. For the ground state problem we consider here, there seems to exist only one approach, relying on an idea due to Dyson~\cite{Dyson-57}, generalized several times as need arose~\cite{LieYng-98,LieSeiSol-05,LieSei-06,LieSeiSolYng-05}. The trick is to bound the Hamiltonian from below by an effective one which is much less singular (roughly, has a smaller $\beta$), but still encodes the scattering length of the original interaction potential. The philosophy is that if $\Psi_N$ is the ground state of the original Hamiltonian, then roughly 
$$\Psi_N \approx \tilde{\Psi}_N \prod_{1\leq i< j \leq N} f (x_i-x_i)$$
where $f$ encodes the two-body scattering process and $\tilde{\Psi}_N$ is a ground state for the effective Dyson Hamiltonian. 

The next task is then to justify the mean-field approximation  for the ground state $\tilde{\Psi}_N$ of the Dyson Hamiltonian. There are two key difficulties:
\begin{itemize}
\item  The effective Hamiltonian has a three-body interaction term. Obviously one has to show that it can be neglected, e.g. by showing that the probability of having three particles close by is small. 
\item Although the Dyson Hamiltonian has a smaller $\beta$ than the original one, we still have to deal with the mean-field approximation in the ``rare but strong collisions'' limit if we are to be able to neglect three-body interactions. In other words, even though the effective Hamiltonian is much less singular than the original one, we do not have the freedom to reduce the singularity as much as we would like. 
\end{itemize}

It is in treating these two difficulties that our new method significantly departs from the previous works~\cite{LewNamRou-14c,LieSei-06}. We shall rely on a  strong a priori estimate for ground states of the Dyson Hamiltonian. Roughly speaking, we use the variational equation to prove a bound of the form 
\begin{align} \label{eq:intro-h1h2}
\langle \tilde{\Psi}_N, (-\Delta)_1 (-\Delta)_2 \tilde{\Psi}_N \rangle \le C
\end{align}
independently of $N$. This second moment  estimate is the key point in our analysis. It is reminiscent of similar estimates used in the literature for the time-dependent problem~\cite{ErdSchYau-07,ErdSchYau-09,ErdSchYau-10,ErdYau-01}. The use of Dyson's lemma is crucial here: one should not expect~\eqref{eq:intro-h1h2} to hold for the ground state of the original Hamiltonian.

We use~\eqref{eq:intro-h1h2} for two distinct purposes:
\begin{itemize}
 \item Control three-body contributions in the Dyson Hamiltonian, showing that they can safely be neglected. In~\cite{LieSei-06} this was achieved by a complicated argument based on path integrals.
 \item Improve the strategy of~\cite{LewNamRou-14c} (discussed in Section~\ref{sec:weak NLS}) to extend it to the ``rare but strong collisions'' regime. Indeed, although much less singular than the full model, the Dyson Hamiltonian still is in this regime. The two-body estimate~\eqref{eq:intro-h1h2} gives a very strong control on the errors made by projecting the Dyson Hamiltonian onto low kinetic energy modes.  
\end{itemize}
Finally, the use of the quantum de Finetti theorem as a basic building block of the proof gives a simplified interpretation of the emergence of the effective one-body energy. Moreover, it simplifies drastically the convergence of states proof, which goes via a Feynman-Hellmann-type argument. That is, by applying the whole strategy to an auxiliary Hamiltonian with a smooth, small, $k$-body say, perturbation and ultimately taking the size of the perturbation to $0$.

\subsection{2D case, attractive interactions} 

As a last instance of Statement~\ref{thm:bosons meta}, we consider the dilute limit for 2D particles with attractive interactions: 
$$d=2,\: \beta > 1/2, \: \int_{\R ^2} w \leq 0.$$
As alluded to in Section~\ref{sec:weak NLS}, this is the most delicate case in so far as the interplay between quantum kinetic energy and pair interactions is concerned. Indeed, for $d\geq 3$ and $\beta >0$, attractive interactions generically lead to a collapse of the ground state, while for $d=1$, the system is stable for any interaction strength. In 2D, the NLS functional is stable only for small enough values of the coupling constant, which corresponds (beware the scaling of the coupling constant in~\eqref{eq:bosons hamil gen}) to the existence of a critical mass for stability of the many-body system. This mass-critical scenario at the level of the ground state has a natural extension to the dynamical problem, where global existence versus finite time blow-up are strongly conditioned by the mass of the initial datum (see~\cite{
KilVis-08,Tao-06} and references therein).  

As in Section~\ref{sec:weak NLS}, our results depend on the growth at infinity of the one-body potential $V$. We assume that
\begin{align} \label{eq:assumption-V again}
V \in L^1_{\rm loc}(\R^2,\R), \quad A \in L^2_{\rm loc}(\R^2,\R^2)\quad \text{and} \quad V(x)\ge C^{-1} (|A(x)|^2 +|x|^s) - C.
\end{align}
To have a stable system we make an almost optimal assumption on the interactions potential, similarly as in Section~\ref{sec:weak NLS}: 
\begin{align} \label{eq:assumption-w2}
\inf_{u\in H^1(\R^2)}\left(\frac{\displaystyle\iint_{\R^2\times \R ^2}|u(x)|^2|u(y)|^2w(x-y)\,dx\,dy}{2\norm{u}_{L^2(\R^2)}^2\norm{\nabla u}_{L^2(\R^2)}^2}\right)>-1
\end{align} 
It is easy to see that this implies $\int_{\R ^2} w >-a^*$ with $a^*$ the critical coupling constant~\eqref{eq:GP a star} for existence of a NLS ground state. Thus the ground state energy 
\begin{equation}\label{eq:bos nls mf gse 2D}
\enls := \inf\left\{ \Enls [u],\quad \int_{\R ^2} |u| ^2 = 1  \right\}. 
\end{equation}
is well-defined, where
\begin{equation}\label{eq:bos nls 2D}
\Enls [u] := \int_{\R ^2 } \left| \left( -i\nabla + A\right)u \right| ^2 + V |u| ^2 + \frac{a}{2} \int_{\R ^2 } |u(x)| ^4   
\end{equation}
with
$$ a = \int_{\R ^2} w. $$
A simple condition implying~\eqref{eq:assumption-w2} is $\int_{\R ^2} |w_-| \leq a ^*$. 

The results of~\cite{LewNamRou-14c} allow to derive the above limit problem when $N\to \infty$, provided 
$$0 < \beta < \beta_0 (s) := \frac{s}{4(s+1)}.$$
In~\cite{LewNamRou-15} we extend this range to 
\begin{equation}\label{eq:beta condition}
0 < \beta < \beta_1 (s):= \frac{s+1}{s+2}.
\end{equation}
This is a qualitative improvement: while $\beta_0(s) <1/2$, we have $\beta_1 (s) > 1/2$. We can thus treat a dilute limit where interactions are rare but strong:

\begin{theorem}[\textbf{Attractive dilute 2D Bose gases}]\label{thm:cv-nls}\mbox{}\\
Assume that $V$, $A$, $w$ satisfy the above. Then, for every $0<\beta<(s+1)/(s+2)$, 
\begin{align} \label{eq:thm-cv-GSE}
\lim_{N\to \infty} e_N = e_{\rm nls} >-\infty. 
\end{align}
Moreover, for any ground state $\Psi_N$ of $H_N$, there exists a Borel probability measure $\mu$ supported on the ground states of $\cE_{\rm nls}(u)$ such that, along a subsequence, 
\begin{align} \label{eq:thm-cv-DM}
\lim_{N \to \infty}\Tr \left| \gamma_{\Psi _{N}}^{(k)} - \int |u^{\otimes k} \rangle \langle u^{\otimes k}| d\mu(u) \right| =0,\quad \forall k\in \mathbb{N}.
\end{align}
If $\cE_{\rm nls}(u)$ has a unique minimizer $u_0$ (up to a phase), then for the whole sequence
\begin{align} \label{eq:thm-BEC}
\lim_{N \to \infty}\Tr \left| \gamma_{\Psi_{N} }^{(k)} - |u_0^{\otimes k} \rangle \langle u_0^{\otimes k}|  \right| =0,\quad \forall k\in \mathbb{N}.
\end{align}
\end{theorem}

\begin{itemize}
\item The estimates leading to the above result answer a question raised by Xuwen Chen and Justin Holmer in their derivation of focusing time-dependent NLS equations~\cite{CheHol-13,CheHol-15}. In the case of the harmonic trap $s=2$, one can combine their BBGKY-type\footnote{That is, based on the Bogoliubov-Born-Green-Kirkwood-Yvon hierarchy of evolution equations for reduced density matrices.} approach and the results of~\cite{LewNamRou-15} to derive the time-dependent 2D focusing NLS equation in the regime $\beta < 3/4$. If $\beta<1/2$, the next order correction to the 2D focusing quantum dynamics can be obtained using the Bogoliubov approach~\cite{LewNamSch-14,NamNap-15,NamNap-16}  (see \cite{BocCenSch-15} for the defocusing case). 
\item The basic strategy of proof is the same as in~\cite{LewNamRou-14c}, see Section~\ref{sec:weak NLS}: projection onto low-lying energy eigenstates of the one-body Hamiltonian and use of the finite dimensional quantum de Finetti theorem~\ref{thm:deF finite dim}.
\item The main novelty allowing to reach larger values of $\beta$ is the use of a priori moments estimates having the same flavor as~\eqref{eq:intro-h1h2}. Roughly speaking, the variational equation for ground states of~\eqref{eq:bosons hamil gen} implies 
\begin{align}\label{eq:one-two body bound}
\Tr\left[ -\Delta \gamma_{N}^{(1)}\right] &\leq C \left(1 + N ^{-1}|E(N)|\right) \nonumber\\
\Tr \left[(-\Delta_1)(-\Delta_2) \gamma_{N}^{(2)}\right] &\leq C \left( 1 + N ^{-1}|E(N)| \right) ^2.
\end{align}
\item Note that the above estimates are not truly a priori bounds since they involve the many-body ground-state energy $E(N)$ in the right-hand side. The ideal case would be to know stability of the second kind:  $E(N) \geq - CN$. Of course, this is far from obvious: it should be the conclusion and not the premise. To put~\eqref{eq:one-two body bound} to good use, the starting point is an energy lower bound obtained in~\cite{Lewin-ICMP} by the operator methods alluded to in Section~\ref{sec:bosons intro}. This yields stability of the second kind for the mean-field regime $\beta < 1/2$ and may be combined with~\eqref{eq:one-two body bound} and the quantum de Finetti theorem in an induction procedure leading to stability for $\beta < \frac{s+1}{s+2}$.  

\end{itemize}

\chapter{Bosonic thermal states and non-linear Gibbs measures}\label{sec:Gibbs}

We now move one step further in the description of large bosonic systems by considering equilibrium states at positive temperature, with the upshot of gaining a better understanding of the BEC phase transition. Here we present joint results with M. Lewin and P.T. Nam~\cite{LewNamRou-14d}, following mainly the expository note~\cite{LewNamRou-ICMP}. We refer to~\cite{Rougerie-xedp15} for another look at the main results of~\cite{LewNamRou-14d}, with an emphasis on their relation to constructive quantum field theory~\cite{DerGer-13,GliJaf-87,LorHirBet-11,Simon-74,Summers-12,VelWig-73} and the use of invariant measures in the study of non-linear dispersive partial differential equations~\cite{LebRosSpe-88,Bourgain-94,Bourgain-96,Bourgain-97,Bourgain-00,Tzvetkov-08,BurTzv-08,BurThoTzv-10,ThoTzv-10,CacSuz-14}.

Recall that full Bose-Einstein condensation is the phenomenon that, below a certain critical temperature $T_c$, (almost) all particles of a bosonic system must reside in a single quantum state of low energy. Ideally one would like to prove the existence of such a temperature and provide an estimate thereof. For an interacting gas this has so far remained out of reach. In particular, for the homogeneous Bose gas, there is still no proof of Bose-Einstein condensation in the thermodynamic limit, even in the ground state. However, this might not be the main question of interest for the description of cold atoms experiments. Indeed, those are performed in magneto-optic traps, which set a fixed length scale to the problem. The gases in which BEC is observed are thus \emph{not} homogeneous, and the thermodynamic limit is not the most physically relevant regime in this context. 

A lot of progress has been achieved in recent years by considering different scaling regimes, more adapted to the case of inhomogeneous systems, e.g. the mean-field and the Gross-Pitaevskii limits. For trapped systems, BEC in the ground state and its propagation by the many-body Schr\"odinger dynamics is now fairly well understood (see~\cite{BenPorSch-15,Golse-13,Lewin-ICMP,LieSeiSolYng-05,Rougerie-LMU,Rougerie-cdf,Schlein-08,Seiringer-ICMP10} and Chapter~\ref{sec:bosons GS} for reviews). Even in these somewhat more wieldy regimes, very little is known about positive temperature equilibrium states. In particular, an estimate of the critical temperature (or, more generally, temperature regime) is lacking.

In this chapter, we discuss some of our recent results~\cite{LewNamRou-14d}, in the perspective of the previously mentioned issues. The main idea is to relate, in a certain limit, Gibbs states of large bosonic systems to non-linear Gibbs measures built on the associated mean-field functionals. This is a semi-classical method, where the BEC phenomenon can be recast in the context of a classical field theory. Although, as far as the study of BEC is concerned, our results are rather partial, it is our hope that our methods might in the future help to shed some light on the physics we just discussed.

\section{The grand-canonical ensemble}\label{sec:GC ens}

We consider $N$ bosons living in $\R^d$ and work in the grand-canonical ensemble. Let thus $\gH = L ^2 (\R ^d)$, 
$$\gH ^N = \bigotimes_{\rm sym} ^N \gH \simeq L_{\rm sym} ^2 (\R^{Nd})$$
be the symmetric $N$-fold tensor product appropriate for bosons and
\begin{align*}
\cF &= \C \oplus \gH \oplus \gH ^{2} \oplus \ldots \oplus \gH ^{N} \oplus \ldots\\
\cF &= \C \oplus L ^2 (\R ^d) \oplus L_\sym ^2 (\R ^{2d})  \oplus \ldots \oplus L_\sym ^2 (\R ^{Nd}) \oplus \ldots 
\end{align*}
be the bosonic Fock space. We are interested in the positive temperature equilibrium states of the second-quantized Hamiltonian $\mathbb{H}_\lambda$ defined as 
$$\mathbb{H}_\lambda = \mathbb{H}_0 + \lambda \mathbb{W} = \bigoplus_{n=1} ^\infty H_{n,\lambda},$$
with\footnote{We do not consider a gauge field in this discussion. This is for simplicity only, cf the original paper.}
\begin{align*}
\mathbb{H}_0 &= \bigoplus_{n= 1} ^\infty \left( \sum_{j=1} ^n h_j\right) = \bigoplus_{n= 1} ^\infty \left( \sum_{j=1} ^n -\Delta_j + V (x_j) - \nu \right)\\
\mathbb{W} &= \bigoplus_{n= 2} ^\infty \left( \sum_{1 \leq i<j \leq n}  w_{ij} \right).
%= \bigoplus_{n= 2} ^\infty \left( \sum_{1 \leq i<j \leq n}  w (x_i-x_j) \right)
\end{align*}
Here $\nu\in \R$ is a chemical potential, $V$ is a trapping potential, i.e.
$$ V(x) \to + \infty \mbox{ when } |x| \to \infty$$
and $w$ is a \emph{positive}, symmetric, self-adjoint operator on $\gH ^2$. The methods of~\cite{LewNamRou-14d} are limited to rather smooth repulsive interactions, thus $w$ will in general \emph{not} be a multiplication operator. We shall comment on this issue below but, for the time being, think of a finite-rank $w$ with smooth eigenvectors, corresponding to a regularization of a physical interaction. 

We are interested in the asymptotic behavior of the grand-canonical Gibbs state at temperature $T$
\begin{equation}\label{eq:GC Gibbs}
 \Gamma_{\lambda,T} = \frac{1}{Z_\lambda (T)}\exp\left( - \frac1T \mathbb{H}_\lambda\right). 
\end{equation}
The partition function $Z_\lambda (T)$ fixes the trace equal to $1$ and satisfies 
\begin{equation}\label{eq:GC part}
 - T \log Z_\lambda (T)  = F_\lambda (T) 
\end{equation}
where $F_\lambda (T)$ is the infimum of the free energy functional 
\begin{equation}\label{eq:free ener GC}
\F_{\lambda} [\Gamma] = \tr_{\cF} \left[ \mathbb{H}_\lambda  \Gamma \right] + T \tr_{\cF} \left[ \Gamma \log \Gamma \right] 
\end{equation}
over all grand-canonical states (trace-class self-adjoint operators on $\cF$ with trace~$1$). It turns out that an interesting limiting behavior emerges in the regime 
\begin{equation}\label{eq:regime}
 T \to \infty, \quad \lambda = \frac{1}{T}, \quad \nu \mbox{ fixed,} 
\end{equation}
provided one makes the following assumptions:

\begin{assumption}[\textbf{One-body hamiltonian}]\label{asum h}\mbox{}\\
We assume that, as an operator on $\gH = L^2 (\R ^d)$, 
$$ h:= -\Delta + V - \nu > 0$$ 
and that there exists  $p>0$ such that 
\begin{equation}\label{eq:asum trace}
\tr_{\gH} h ^{-p} < \infty. 
\end{equation}
\end{assumption}

Note that one can always find a $p$ such that~\eqref{eq:asum trace} holds. The easiest case, for which our results are the most satisfying, is that where one can take $p=1$ (refered to as the trace-class case). This happens only in 1D and if the trapping potential grows sufficiently fast at infinity, i.e.
\begin{equation}\label{eq:anharm oscill}
 h = -\frac{d^2}{dx^2} + |x| ^a - \nu
\end{equation}
with $\nu$ small enough and $a>2$.

\begin{assumption}[\textbf{Interaction term}]\label{asum w}\mbox{}\\
We pick $w$ a positive self-adjoint operator on $\gH ^2$ and distinguish two cases: 
\begin{itemize}
 \item either one can take $p=1$ in~\eqref{eq:asum trace} and then we assume 
\begin{equation}\label{eq:asum w 1}
 \tr_{\gH ^2} \left[ w\, h ^{-1} \otimes h ^{-1} \right] < \infty
\end{equation}
\item or $p>1$ in~\eqref{eq:asum trace} and we make the stronger assumption that
\begin{equation}\label{eq:asum w 2}
 0 \leq w \leq h ^{1-p'} \otimes h ^{1-p'}
\end{equation}
for some $p'>p$.
\end{itemize}
\end{assumption}

In essence, these (rather restrictive) assumptions correspond to asking that the non-interacting Gibbs state has a well-controlled interaction energy. Indeed, one can compute that its two-body density matrix behaves as $T^2 h ^{-1} \otimes h ^{-1}$ in the limit $T\to \infty$. In the 1D case~\eqref{eq:anharm oscill} where $p=1$, we can take for example $w$ a multiplication operator, e.g. by a bounded function $w(x-y)$.  The assumption we make when $p>1$ does not cover such operators.

\section{Non-linear Gibbs measures}\label{sec:NL Gibbs}

The natural limiting object in the setting we just described turns out to be the non-linear Gibbs measure on one-body quantum states given formally by 
 \begin{equation}
d\mu(u)=\frac{1}{Z}e^{-\cE[u]}\,du,
\label{eq:mu_intro}
 \end{equation}
where $Z$ is a partition function and $\cE [u]$ is the mean-field energy functional
\begin{equation}\label{eq:MF func}
\cE [u] := \left\langle u | -\Delta + V - \nu | u \right\rangle_{\gH} + \frac{1}{2} \left\langle u \otimes u | w | u \otimes u \right\rangle_{\gH^2}.
\end{equation}
The rigorous meaning of~\eqref{eq:mu_intro} is given by the two following standard lemmas/definitions:

\begin{lemma}[\textbf{Free Gibbs measure}]\label{lem:free}\mbox{}\\
Write
$$ h = -\Delta + V - \nu = \sum_{j = 1} ^{\infty} \lambda_j |u_j\rangle \langle u_j|$$
and define the associated scale of Sobolev-like spaces 
\begin{equation}
\gH^s:=D(h^{s/2})=\bigg\{u=\sum_{j\geq1} \alpha_j\,u_j\ :\ \|u\|_{\gH^{s}}^2:=\sum_{j\geq1}\lambda_j^s|\alpha_j|^2<\ii\bigg\}.
\end{equation}
Define a finite dimensional measure on $\mathrm{span} (u_1,\ldots,u_K)$ by setting
$$d\mu_0 ^K (u) := \bigotimes_{j=1} ^K \frac{\lambda_j}{\pi} \exp\left( -\lambda_j |\langle u, u_j\rangle | ^2\right) d \langle u, u_j\rangle$$
where $d\langle u, u_j\rangle = da_j db_j$ and $a_j,b_j$ are the real and imaginary parts of the scalar product.

Let $p>0$ be such that~\eqref{eq:asum trace} holds. Then there exists a unique measure $\mu_0$ over the space $\gH ^{1-p}$ such that, for all $K >0$, the above finite dimensional measure $\mu_{0,K}$ is the cylindrical projection of $\mu_0$ on $\mathrm{span} (u_1,\ldots,u_K)$. 
Moreover 
\begin{equation}\label{eq:DM free meas}
\gamma_0^{(k)}:=\int_{\gH^{1-p}} |u^{\otimes k}\rangle\langle u^{\otimes k}|\;d\mu_0(u) = k!\,(h^{-1})^{\otimes k} 
\end{equation}
where this is seen as an operator acting on $\bigotimes_{\rm sym} ^k \gH$.
\end{lemma}

Note that the free Gibbs measure \emph{never} lives on the energy space $\gH ^1$. It lives on the original Hilbert space $\gH ^0 = \gH$ if and only if~\eqref{eq:asum trace} holds with $p=1$.  We can now define the interacting Gibbs measure as being absolutely continuous with respect to the free Gibbs measure: 

\begin{lemma}[\textbf{Interacting Gibbs measure}]\label{lem:inter}\mbox{}\\
Let 
$$ \FNL [u] := \frac{1}{2}\left\langle u \otimes u | w | u \otimes u \right\rangle_{\gH^2}.$$
If Assumptions~\ref{asum w} hold we have that $ u \mapsto  \FNL [u] $ is in $L ^1 (\gH ^{1-p},d\mu_0)$. In particular 
\begin{equation}\label{eq:int meas}
\mu (du) = \frac{1}{Z_r} \exp\left( - \FNL [u] \right) \mu_0 (du)  
\end{equation}
makes sense as a probability measure over $\gH ^{1-p}$. That is, the relative partition function satisfies 
$$Z_r = \int \exp\left( - \FNL [u] \right) \mu_0 (du) >0  .$$
\end{lemma}

Equation~\eqref{eq:int meas} is the correct interpretation of the formal definition~\eqref{eq:mu_intro}. It is this object that we derive from the bosonic  grand-canonical Gibbs state.

\section{Mean-field/small temperature regime}\label{sec:Gibbs main}

Our main result in~\cite{LewNamRou-14d} relates, in the limit~\eqref{eq:regime}, the grand-canonical Gibbs state~\eqref{eq:GC Gibbs} to the classical Gibbs measure on one-body states defined in Lemma~\ref{lem:inter}. It is useful to recall the notion of $k$-particle reduced density matrix for states over $\cF$. For a ``diagonal'' state of the form
$$\Gamma = \bigoplus_{n=0} ^\infty G_n, \quad G_n: \gH^n \mapsto \gH ^n,$$
this is the operator 
$\Gamma ^{(k)}$ on the $k$-particles space $\gH ^k$ defined by\footnote{The partial trace is taken over the symmetric space only.}
\begin{equation}\label{eq:GC DM}
 \tr_{\gH ^k} \left[ A_k \Gamma^{(k)}\right]=\sum_{n\geq k} {n\choose k} \tr_{\gH ^n} \left[A_k \otimes_s \1 ^{\otimes (n-k)} G_n\right]
\end{equation}
for every bounded operator $A_k$ on $\gH ^k$. These are helpful for rewriting expressions involving only $k$-particle operators. Remark that the normalization convention taken here differs from that of Chapter~\ref{sec:bosons GS}: if $\Gamma$ is a $N$-body state, $\Gamma^{(k)}$ has trace ${N \choose k}$.

\begin{theorem}[\textbf{Derivation of nonlinear Gibbs measures}]\label{thm:NL Gibbs}\mbox{}\\
Under Assumptions~\ref{asum h} and~\ref{asum w} we have 
$$\frac{F_\lambda (T) - F_0 (T)}{T} \underset{T\to \infty}{\longrightarrow}  -\log Z_r $$
where $F_\lambda (T)$ is the infimum of the free-energy functional~\eqref{eq:free ener GC} and $Z_r$ the relative partition function defined in Lemma~\ref{lem:inter}. 

Let furthermore $\Gamma_{\lambda,T} ^{(k)}$ be the reduced $k$-body density matrix of $\Gamma_{\lambda,T}$. We have, 
\begin{equation}
\frac{\Gamma_{\lambda,T}^{(1)}}{T} \underset{T\to \infty}{\longrightarrow}\int_{\gH^{1-p}} |u\rangle\langle u|\,d\mu(u)
\label{eq:limit_DM_S_p}
\end{equation}
strongly in the Schatten space $\gS^p(\gH)$. In case $p=1$ in~\eqref{eq:asum trace} we also have, for any $k\geq 2$,
\begin{equation}\label{eq:result DM}
\frac{k!}{T^k}\Gamma_{\lambda,T}^{(k)} \underset{T\to \infty}{\longrightarrow} \int_\gH |u^{\otimes k}\rangle\langle u^{\otimes k}|\,d\mu(u)
\end{equation}
strongly in the trace-class.
\end{theorem}

An obvious caveat of our approach is the case $p>1$ where the result is not as strong as one would hope. We conjecture that~\eqref{eq:result DM} continues to hold strongly in the Schatten space $\gS^p(\gH)$ also for $k>1$, but this remains an open problem. More importantly, it would be highly desirable to go beyond the stringent assumption~\eqref{eq:asum w 2}. There are however well-known obstructions to the construction of the Gibbs measure $\mu$ in this case. A minima, a Wick renormalization~\cite{DerGer-13,GliJaf-87,Simon-74} must be performed in order to make sense of the measure when $w= w (x-y)$ is a multiplication operator. The quantum Gibbs state one starts the analysis from should also be modified accordingly. We refer to~\cite{FroKnoSchSoh-16} for recent results in this direction. 

Let us comment on the relevance of this result to the BEC phenomenon:
\begin{itemize}
\item First, this theorem seems to be the first giving detailed information on the limit of the thermal states of a Bose system at relatively large temperature. The result~\eqref{eq:limit_DM_S_p} clearly shows the absence of full BEC in the regime we consider. For finite dimensional bosons, things simplify a lot and versions of Theorem~\ref{thm:NL Gibbs} were known before, see~\cite{Gottlieb-05,GotSch-09,JulGotMarPol-13},~\cite[Chapter 3]{Knowles-thesis}, and~\cite[Appendix~B]{Rougerie-LMU}.
\item The asymptotic regime~\eqref{eq:regime} should be thought of as a mean-field limit. In fact, when $p=1$,~\eqref{eq:limit_DM_S_p} indicates that the expected particle number behaves as $O(T)$, so that taking $\lambda = T ^{-1}$ corresponds to the usual mean-field scaling where the coupling constant scales as the inverse of the particle number. Roughly speaking we are thus dealing with the regime
$$ N \to \infty, \quad \lambda = \frac{1}{N}, \quad T \sim N.$$
\item Still in the case $p=1$, one should expect from a natural extrapolation of this theorem that full BEC does occur in the regime 
$$ N \to \infty, \quad \lambda = \frac{1}{N}, \quad T \ll N$$
which would indicate that the critical temperature scales in this case as the particle number. A natural idea would be to confirm this by studying the concentration of the non-linear Gibbs measure $\mu$ on the mean-field minimizer when the chemical potential is varied.
\item A comparison with the case of particles with no quantum statistics (boltzons) is instructive, see~\cite{LewNamRou-ICMP} for details. First, the emergence of the non-linear Gibbs measure truly requires Bose statistics. Next, for trapped boltzons, the critical temperature for BEC scales as $O(1)$ when $N\to \infty$ in the mean-field regime. Only when $T\to 0$ at the same time as $N\to \infty$ does full BEC occur in the Gibbs state. One should expect that the critical temperature is much higher for bosons than for boltzons. This can be seen from comparing the above theorem and~\cite[Section~3.2]{Suto-03,LewNamRou-14}: $T^{\rm boltzons}_c \sim 1$ in the mean-field limit whereas $T^{\rm bosons}_c \gg 1$. In fact, Theorem~\ref{thm:NL Gibbs} suggests that in some situations the critical temperature for bosons can be as high as $O(N)$. 
\end{itemize}

\section{Berezin-Lieb inequality for relative entropies}\label{sec:NL Gibbs proof}
%%%%%%%%%%%%%%%%%%%%%%%%%%%%%%%%%%%%%%%%%%
%%%%%%%%%%%%%%%%%%%%%%%%%%%%%%%%%%%%%%%%%%

The proof of Theorem~\ref{thm:NL Gibbs} is lengthy and shall not be reviewed here. We however discuss its main conceptual building blocks, restricting to the trace-class case $p=1$. First, as one easily realizes, $\Gamma_{\lambda,T}$ minimizes the relative free-energy functional (the second term is a von Neumann relative entropy)
\begin{align}\label{eq:rel free ener}
\F_{\rm rel} [\Gamma] &= \lambda \tr_{\cF} \left[ \mathbb{W} \: \Gamma \right] + T \tr_{\cF} \left[ \Gamma \left(\log \Gamma - \log \Gamma_{0,T} \right)\right] \nonumber\\
&=\lambda \tr_{\gH ^2} \left[ w \: \Gamma ^{(2)}\right] + T \tr_{\cF} \left[ \Gamma \left(\log \Gamma - \log \Gamma_{0,T} \right)\right],
\end{align}
whose minimum is $F_\lambda (T) - F_0 (T)$.  It is crucial to work directly on that difference, characterizing it as the minimum of some functional. We expect that both terms taken individually diverge very fast.

On the other hand, $-\log Z_r$ is the infimum of some variational problem of which $\mu$ is the solution: $\mu$ minimizes 
\begin{align}\label{eq:class rel ener}
\F_{\rm cl} [\nu] &= \int \FNL [u] d\nu (u) + \int \frac{d\nu}{d\mu_0} \log \left( \frac{d\nu}{d\mu_0}\right) d\mu_0\nonumber\\
&= \int \frac{1}{2} \tr_{\gH ^2} \left[w\, |u ^{\otimes 2} \rangle \langle u ^{\otimes 2}| \right] d\nu (u) + \int \frac{d\nu}{d\mu_0} \log \left( \frac{d\nu}{d\mu_0}\right) d\mu_0
\end{align}
amongst all $\mu_0$-absolutely continuous measures $\nu$ on the one-body Hilbert space $\gH$. The second term is the classical relative entropy of $\nu$ relative to $\mu_0$.

The above observations make the semi-classical nature of Theorem~\ref{thm:NL Gibbs} apparent.  We have to relate the minimization problems for the quantum and classical free-energy functionals~\eqref{eq:rel free ener} and~\eqref{eq:class rel ener}. In order to describe quantum objects, operators on a Hilbert space, in terms of classical objects, probability measures, we rely on quantum de Finetti theorems again, as discussed in Section~\ref{sec:bosons deF}. This time we need  Fock-space versions thereof. In the trace-class case, the following result of~\cite{AmmNie-08} is sufficient:

\begin{theorem}[\textbf{Grand-canonical quantum de Finetti theorem}]\label{thm:deF}\mbox{}\\
Let $0\leq \Gamma_n$ be a sequence of states on the Fock space $\cF(\gH)$. Assume that there exists a sequence $0<\epsilon_n\to0$ such that, for all $k\in \N $
\begin{equation}
\epsilon_n ^k  \tr_{\gH ^{k}} \Big[ \Gamma_n ^{(k)}\Big]\leq C_k \ <\infty.
\label{eq:bound_DM}
\end{equation}
Then, there exists a unique Borel probability measure $\nu$ on $\gH$ such that, along a subsequence,  
\begin{equation}
k!(\epsilon_{n})^k\,\Gamma_{n}^{(k)} \wto \int_{\gH} |u^{\otimes k}\rangle \langle u^{\otimes k}| \,d\nu(u)
\label{eq:weak_limit_DM}
\end{equation}
weakly in the trace-class for every integer $k$.
\end{theorem}

Some comments:
\begin{itemize}
\item To deal with the non-trace-class case $p>1$, some improvements are necessary, provided in~\cite{LewNamRou-14d}, but we shall not discuss this here.
\item The case $p=1$ of Theorem~\ref{thm:NL Gibbs} could be rephrased as saying that the de Finetti measure of the sequence $\Gamma_{\lambda,T}$ is the nonlinear Gibbs measure $\mu$. 
  \item Here $\epsilon_n$ should be thought of as a semiclassical parameter essentially forcing the particle number to behave as $\eps_n ^{-1}$. To apply this result in our context, we need to check that~\eqref{eq:bound_DM} holds, with $\epsilon_n = T ^{-1}$ and $\Gamma = \Gamma_{\lambda,T}$. In the trace-class case, this is not too difficult because the particle number is well under control. 
\item Since we have convergence of density matrices, a lower bound to the interaction energy term in~\eqref{eq:rel free ener} follows from this theorem.
\end{itemize}

In view of the last point, the main remaining task is to estimate the relative entropy term in~\eqref{eq:rel free ener}. One should keep the following in mind:
\begin{itemize}
 \item Contrarily to the energy, it is not expressed in terms of reduced density matrices only. It genuinely depends of the full state.
 \item One should not ``undo'' the relative entropy to consider only the von Neumann entropies of $\Gamma_{0,T}$ and $\Gamma_{\lambda,T}$ directly, since we do not expect a good control on those.
\end{itemize}

Well-known tools to link quantum entropy terms to their semiclassical counterparts are the so-called Berezin-Lieb inequalities~\cite{Berezin-72,Lieb-73b,Simon-80}. Since we work in infinite dimensions and we cannot undo the relative entropy to rely on known inequalities, we need the following new result:

\begin{theorem}[\textbf{Berezin-Lieb inequality for the relative entropy}]\label{thm:BL}\mbox{}\\
Let $\epsilon_n \to 0$ and $\{\Gamma_n\}$, $\{\Gamma_n'\}$ be two sequences of states  satisfying the assumptions of Theorem~\ref{thm:deF}. Let $\nu$ and $\nu'$ be their de Finetti measures. Then
\begin{equation}\label{eq:BL ineq}
\liminf_{n\to \infty} \tr_{\cF} \left[ \Gamma_n \left( \log \Gamma_n - \log \Gamma_n' \right)\right] \ge \int \frac{d\nu}{d\nu'} \log \left(\frac{d\nu}{d\nu'}\right) d\nu'. 
\end{equation}
\end{theorem}

Some comments:

\begin{itemize}
\item Clearly, this will give the desired lower bound to the relative entropy term in~\eqref{eq:rel free ener} if we can prove first that the de Finetti measure of the free Gibbs state $\Gamma_{0,T}$ is given by $\mu_0$. Since all the density matrices of $\Gamma_{0,T}$ can be explicitly computed, this is easily shown by direct inspection.
\item For the proof of this result, we rely heavily on a constructive approach to Theorem~\ref{thm:deF}, reminiscent of the proof of Theorem~\ref{thm:deF weak} we discussed previously. The precise way one can approach the measure by explicit functions of the states, as in~\eqref{eq:def CKMR}, is crucial. 
\item The link with semi-classics proceeds by seeing the de Finetti measure both as a lower symbol/covariant symbol/Husimi function/anti-Wick quantization and as an approximate upper symbol/contravariant symbol/Wigner measure/Wick quantization based on a coherent state decomposition.
\item We also rely on deep properties of the quantum relative entropy: joint convexity and monotonicity under two-positive trace preserving maps~\cite{Wehrl-78,OhyPet-93,Carlen-10}. These are in fact equivalent to the strong subadditivity of the quantum entropy, a crucial property proved by Lieb and Ruskai~\cite{LieRus-73a,LieRus-73b}.
\end{itemize}

\chapter{Emergent fractional statistics}\label{sec:anyons}

Traditionally, for sound experimental and theoretical reasons, quantum particles are divided into two classes: bosons and fermions. The many-body wave-function of a system of bosons is symmetric under particle exchange, while that of a system of fermions is anti-symmetric. This implies fundamental differences in their behavior: fermions satisfy the Pauli exclusion principle (at most one particle per quantum state), bosons do not. Consequently, fermionic bulk matter is stable, while bosonic bulk matter is not~\cite{DysLen-67,DysLen-68,Lieb-76,LieSei-09}. Free bosons in thermal equilibrium distribute in single-particle quantum states following the Bose-Einstein statistics, while free fermions follow the Fermi-Dirac statistics. Bosons condense into a single one-particle state (see Chapters~\ref{sec:bosons GS}-\ref{sec:Gibbs} and references therein) while fermions form a Fermi sea. Some fundamental inequalities such as the Lieb-Thirring kinetic energy lower bound hold 
only for fermions~\cite{LieThi-75,LieThi-76,LieSei-09}. Cherished physical principles, such as the classification of elements in the Mendeleev periodic table, are valid because electrons are fermions. Others, such as the possibility of creating a laser, are valid because photons are bosons.

In this chapter we discuss some contributions to the theory of \emph{anyons}, that is, of particles that do \emph{not} fit in the bosons/fermions dichotomy. The very possibility that such exotic objects could exist should be seen as a daring hypothesis. We shall thus first start by a quick introduction to the anyons concept in Section~\ref{sec:intro anyons}, making no claim of originality (see~\cite{Froehlich-90,Lerda-92,Wilczek-90,IenLec-92,Forte-92,Myrheim-99,Khare-05,Ouvry-07} for general references). In brief, we shall recall that anyons can exist only in low spatial dimensions, 1D or 2D, so that they can only be quasi- or dressed particles of constrained systems. Next, in Section~\ref{sec:QH anyons} we summarize the findings of~\cite{LunRou-16} where we vindicate\footnote{Without a mathematical proof, but what we think is a sounder physical proof than what has appeared in the literature so far.} the general belief~\cite{AroSchWil-84,Laughlin-99,Jain-07} that some quasi-particles arising in fractional 
quantum 
Hall physics indeed 
behave as anyons. Finally, in Section~\ref{sec:MF anyons} we summarize~\cite{LunRou-15}, giving modest elements of answer to the question ``Suppose one manages to create anyons in the lab, how can one probe their physics ?'' As we shall explain, appropriate effective models are necessary to answer that question, and we focus on the derivation of one such model from first principles. 

\section{The concept of 2D anyons}\label{sec:intro anyons}

There does not seem to be a universal way of building a quantum theory, starting from a reasonable classical one. The possibility for exotic particle statistics has staid hidden in implicit quantization choices  for years, until those were recognized as true choices instead of obvious necessities~\cite{Girardeau-65,GolMedSha-81,LeiMyr-77,Souriau-67,Souriau-70,Wilczek-82a,Wilczek-82b}. To appreciate this, we shall first review the usual theoretical argument supporting the bosons/fermions dichotomy. 

\subsection{Bosons and fermions}

Let us consider $N$ quantum particles living in the physical space $\R ^d$. The classical configuration space is $\R ^{dN}$ and one describes the quantum system with a complex-valued wave-function $\Psi_N \in L ^2 (\R ^{dN})$. Since $|\Psi_N| ^2$ is interpreted as the direct-space probability density of the system, and $|\hat{\Psi}_N| ^2$ as the momentum-space probability density\footnote{With the convention that the Fourier transform is a $L^2$ isometry.}, we have to impose 
$$ \int_{\R^{dN}} |\Psi_N| ^2 = 1,$$
but this is not the only physical constraint. Indeed, actual particles are \emph{indistinguishable}, so that the labeling of coordinates $(x_1,\ldots,x_N) \in \R ^{dN}$ is arbitrary. Thus, one should obviously demand that 
\begin{equation}\label{eq:any den sym}
 | U_{\sigma} \Psi_N | ^2 = | \Psi_N | ^2 
\end{equation}
for any permutation $\sigma \in \Sigma_N$ of $N$ indices, where $U_\sigma$ is the unitary operator that permutes particle labels:
$$ U_{\sigma} \Psi_N (x_1,\ldots,x_N) = \Psi_N (x_{\sigma(1)},\ldots,x_{\sigma(N)}).$$
It is sometimes improperly claimed that~\eqref{eq:any den sym} implies that $\Psi_N$ is either symmetric or anti-symmetric. There are counter-examples~\cite{Girardeau-65}, but one should realize that indistinguishability is a stronger concept than~\eqref{eq:any den sym}. In quantum mechanics, every possible measurement on a system is represented by a  self-adjoint operator $A$ on $L^2 (\R ^{dN})$. The expected value of the quantity represented by $A$, when measured in the state $\Psi_N$, is given by 
$$ \langle A \rangle_{\Psi_N}:=\langle \Psi_N | A | \Psi_N \rangle.$$
Thus one should demand that any such value does not depend on the labeling of the particles, and the bosons/fermions dichotomy comes about via the following elementary lemma:

\begin{lemma}[\textbf{Bosons and fermions}]\label{lem:bos ferm}\mbox{}\\
Let $\Psi_N \in L ^2 (\R ^{dN})$ with 
$$ \int_{\R ^{dN}} |\Psi_N| ^2 = 1$$
satisfy, for any bounded self-adjoint operator $A$ on $L^2 (\R ^{dN})$ and any permutation $\sigma \in \Sigma_N$
\begin{equation}\label{eq:indistinguishability}
 \langle \Psi_N | A | \Psi_N \rangle = \langle U_\sigma \Psi_N | A | U_\sigma\Psi_N \rangle. 
\end{equation}
Then either 
\begin{equation}\label{eq:any bosons}
 U_\sigma \Psi_N = \Psi_N \quad \forall \sigma \in \Sigma_N \mbox{ \rm{(Bosons)}} 
\end{equation}
or 
\begin{equation}\label{eq:any fermions}
 U_\sigma \Psi_N = (-1) ^{\rm{sign(\sigma)}} \Psi_N \quad \forall \sigma \in \Sigma_N \mbox{ \rm{(Fermions)}}. 
\end{equation}
\end{lemma}

\begin{proof}
Take $A= |\Psi_N \rangle \langle \Psi_N|$. According to~\eqref{eq:indistinguishability}, we must have 
$$ \left| \langle \Psi_N, U_{\sigma} \Psi_N \rangle \right| ^2 = 1$$
for all $\sigma$ and hence 
$$ U_{\sigma} \Psi_N = c_{\sigma} \Psi_N $$
for some complex number $c_{\sigma}$ with modulus one. For consistency we see that 
$$ c_{\sigma \circ \sigma'} = c_{\sigma} c_{\sigma'}$$
must hold for all $\sigma,\sigma' \in \Sigma_N$. The map $\sigma \mapsto c_{\sigma}$ is a thus one-dimensional representation of the permutation group and it follows that either $c_\sigma \equiv 1$ or $c_{\sigma} = (-1) ^{\rm{sign(\sigma)}}.$
\end{proof}

Of course, this argument is mostly academic: that all known particles behave as either bosons and fermions is quite simply an experimental fact. But in view of the above, there seems to be little to no room, even in theory, for other possibilities. One could then claim that the dichotomy is a fundamental fact of nature. That is, provided one trusts the (elementary) quantization procedure just described, which, in hindsight, is more subtle than it looks.

\subsection{Anyons}

The breakthrough leading to the introduction of anyons occurred in~\cite{LeiMyr-77} where it was realized 
\begin{itemize}
 \item that there is one choice implicitly made in the above quantization procedure,
 \item that a different quantization can lead to different particle statistics,
\item that these can occur only in low dimensions $d\leq 2$.
\end{itemize}
A parallel approach is that of~\cite{Wilczek-82a,Wilczek-82b} where it is argued that in low dimensions the spin needs not be quantized. Hence, via the spin-statistics theorem, one may hope for exotic particle statistics. We shall not describe this approach here\footnote{The spin-statistics theorem is a result of relativistic quantum field theory. One can prefer not having to use it to discuss ordinary non-relativistic matter.}, nor that of~\cite{GolMedSha-81}, and now turn to a (very crude) summary of the findings of~\cite{LeiMyr-77}.

In fact, the only choice made in the above quantization procedure that is not a basic axiom of quantum mechanics, and that is not obviously logically necessary, is that we  have \emph{first quantized, then symmetrized}. Namely, we have first introduced wave-functions, and then restricted possible choices thereof by a physical symmetry postulate. What happens if we instead \emph{first symmetrize, then quantize} ? 

We shall not enter into the mathematical description of this process, but state the essential conclusions, restricting\footnote{There is an ambiguity in 1D, in that one cannot continuously exchange particles on a line without making them collide.} to $d\geq 2$:
\begin{itemize}
 \item In dimension $d\geq 3$, the former conclusion remains valid, only bosonic and fermionic wave-functions are logically consistent possibilities.
 \item For $d=2$, there is a whole continuum of other possible choices, labeled by a real number $\alpha \in[ 0,2[$. Bosons correspond to $\alpha = 0$, fermions to $\alpha = 1.$
\end{itemize}
This difference is of topological origin: one removes particles encounters points from the configuration space, because they are somewhat singular. A given particle then moves in a punctured space. The punctured 3D space is simply connected, while the punctured 2D space is not. In 3D and higher, the first homotopy group of the classical configuration space is the permutation group, in 2D it is the braid group. The latter has a continuum of one-dimensional representations, and this is where the continuum of possible statistics originates from~\cite{Wu-84,Wu-91}. 

In short, in 2D there is no inconsistency in demanding that
\begin{equation}\label{eq:anyon gauge}
 U_\sigma \Psi_N = \left( e^{i\pi \alpha }\right) ^{\rm{sgn} (\sigma)} \Psi_N 
\end{equation}
where $\alpha$ is a real number that can be restricted to $[0,2[$ by periodicity. One refers to $\alpha$ as the statistics parameter, and speaks of fractional statistics\footnote{Although, strictly speaking, irrational $\alpha$ is also allowed.} if $\alpha \neq 0,1$, that is when we deal with particles that behave neither as bosons nor as fermions. The term anyons (``any''-ons) was coined by Wilczek~\cite{Wilczek-82b} because the statistics parameter can a priori be any number.

Of course the above object cannot be a function, for it would be multi-valued\footnote{Think of permuting a pair of indices twice; the function then picks a phase $e^{2i\pi \alpha}$.} unless $\alpha = 0$ (bosons) or $\alpha = 1$ (fermions). It should in fact be seen as a fiber bundle, which would complicate any theoretical description and mathematical analysis. Fortunately, there is a natural, albeit formal, way to fall back on the familiar description in terms of wave-functions. We present this next.

\subsection{Magnetic gauge picture}

In the sequel we shall be interested in the quantum mechanics of non-relativistic 2D particles obeying fractional statistics. This means that we are interested in acting on ``functions'' satisfying~\eqref{eq:anyon gauge} with a many-body Hamiltonian of the form 
\begin{equation}\label{eq:anyons many hamil}
H_N = \sum_{j=1} ^N \left( p_j + \Aext (x_j)\right) ^2 + V(x_j) + \sum_{1\leq i < j \leq N} w (x_i-x_j) 
\end{equation}
where $p_j = -i\nabla_j$ is the usual momentum operator for particle $j$, $\Aext$ is the vector potential of an external gauge field, $V$ a one-body potential and $w$ a pair interaction potential. The usual way to render this problem more tractable is to write 
\begin{equation}\label{eq:any change gauge}
\Psi_N = \prod_{1\leq i < j \leq N} e^{i \alpha \arg (x_i-x_j)} \Phi_N 
\end{equation}
where $\arg (x_i-x_j)$ is the angle of the vector $(x_i-x_j)$ in the plane. That way, $\Phi_N$ has, at least formally, bosonic symmetry (its $\alpha$ is $0$ modulo $2$) and one might prefer to describe the system in terms of it. The important observation is that acting on $\Psi_N$ with $H_N$ is formally equivalent to acting on $\Phi_N$ with a new, effective Hamiltonian:
$$ \langle \Psi_N |H_N| \Psi_N \rangle = \langle \Phi_N |H_N ^{\alpha}| \Phi_N \rangle $$
where 
\begin{equation}\label{eq:anyons mag hamil}
H_N ^{\alpha} = \sum_{j=1} ^N \left( p_j + \Aext (x_j) + \alpha \Aan (x_j) \right) ^2 + V(x_j) + \sum_{1\leq i < j \leq N} w (x_i-x_j). 
\end{equation}
The effective gauge vector potential is given by 
\begin{equation}\label{eq:any gauge vect}
\Aan (x_j) = \sum_{k\neq j} \frac{(x_j-x_j) ^\perp}{|x_j-x_k| ^2} 
\end{equation}
and it depends on the position of the particles. The associated magnetic field is a sum of delta functions
\begin{equation}\label{eq:any gauge field}
\curl \Aan (x_j) = 2\pi \sum_{k\neq j} \delta (x_j-x_k).
\end{equation}
One can thus see non-relativistic anyons as ordinary bosons that carry in addition fictitious Aharonov-Bohm-like flux tubes whose vector potential influences the other particles. Of course one could choose a convention where $\Phi_N$ is fermionic instead, with obvious modifications to the above equations. One could even (although this does not seem a very good idea) describe 2D fermions as bosons coupled to flux tubes, and vice-versa. 

The above formalism based on~\eqref{eq:any change gauge} and~\eqref{eq:anyons mag hamil} is called ``magnetic gauge picture'' by opposition to the ``anyon gauge picture'' based on~\eqref{eq:anyon gauge} and~\eqref{eq:anyons many hamil}. The obvious advantage is that we can now think in terms of bosonic (or fermionic) wave-functions. The obvious drawback is that the anyon Hamiltonian is pretty complicated. Even in the free case $V,\Aext,w\equiv 0$, very few exact results on the spectrum or dynamics of $H_N^\alpha$ are known. Well-educated approximations are thus highly desirable, even more so than for usual interacting bosons and fermions  (the effective magnetic interaction contained in~\eqref{eq:anyons mag hamil} is considerably more intricate than a simple pair interaction). Section~\ref{sec:MF anyons} below is concerned with this issue. Before entering this question, we discuss the emergence of fractional statistics for a certain type of quasi-particles in Section~\ref{sec:QH anyons}. Interestingly, 
we 
shall see that it is the magnetic gauge picture, and not the anyon gauge picture, that naturally emerges. For all practical purposes relevant to this memoir we thus can, and shall, consider only the former description. Anyons are thus identified in the sequel as ordinary bosons or fermions feeling an effective Hamiltonian of the form~\eqref{eq:anyons mag hamil}

\section{Anyons in quantum Hall physics}\label{sec:QH anyons}

That fractional statistics can occur only in low dimensions seems a no-go for their relevance to fundamental particles. All of these (electrons, photons, quarks ...) are 3D objects. Even if restricted to lower dimensions in the lab, they retain their 3D statistics, which can only be bosonic or fermionic. The statistics of composite objects, such as protons, neutrons and atoms, are obtained by adding the statistics of their constituents: this can only produce bosons or fermions.

Anyons can thus only exist as quasi-particles of low dimensional systems. In this section we discuss the emergence of 2D fractional statistics for the most prominent candidates that seem to have been considered in the literature so far: quasi-holes of the Laughlin state. We have already studied the Laughlin wave-function in some detail in Chapter~\ref{sec:FQHE states}. Here we shall be interested in some of its elementary excitations, called quasi-holes. We argue that, when coupled to tracer particles, they effectively change the statistics of the latter. The tracer particles, dressed by the quasi-holes of the Laughlin wave-function, become anyons with statistics parameter 
$$\alpha = \alpha_0-\frac{1}{\ell}$$
where $\ell$ is the exponent in the Laughlin wave-function~\eqref{eq:intro Laughlin} and $\alpha_0$ the original statistics parameter of the tracer particles (thus $\alpha_0 = 0$ or $1$, bosons or fermions). The possibility of creating anyons in the fractional quantum Hall effect context seems to be widely accepted~\cite{Girvin-04,Goerbig-09,StoTsuGos-99,Laughlin-99,Jain-07}. It has recently been proposed~\cite{Cooper-08,BloDalZwe-08,MorFed-07,ParFedCirZol-01,RonRizDal-11,Viefers-08} that this physics could be emulated in ultra-cold atomic gases subjected to artificial magnetic fields. It is such a system we have in mind primarily, although the following might still be of relevance to the usual 2D electrons gases where the FQHE is observed. 

\medskip

\noindent\textbf{Warning.} The reader should be aware that the results of this section are \emph{not} mathematically rigorous, contrarily to everything discussed in this memoir so far. They are included because 
\begin{itemize}
\item their authors (this is joint work with Douglas Lundholm~\cite{LunRou-16}) have good hope to render them rigorous, at least in parts, in the near future;
\item they avoid some issues with the usual arguments in favor of fractional statistics; 
\item they give a physical motivation to the rigorous results of Section~\ref{sec:MF anyons} below, and to future investigations.
\end{itemize}

\subsection{Tracer particles in a Laughlin liquid}

Consider the following thought experiment: two different species of 2D particles interact repulsively among themselves, and couple to a uniform magnetic field. Such a set-up has already been suggested~\cite{CooSim-15,ZhaSreGemJai-14,ZhaSreJai-15} as a way to probe the emergence of fractional statistics, using cold atomic gases. 

The Hilbert space for the joint system is 
\begin{equation}\label{eq:Hilbert}
\gH ^{M+N} = L ^2 _{\rm{sym}} (\R ^{2M}) \otimes L ^2_{\rm{sym}} (\R ^{2N}),
\end{equation}
where $M$ is the number of particles of the first type and $N$ the number of particles of the second type. 
For definiteness we assume that the two types of particles be bosons, whence the imposed symmetry in the above Hilbert spaces. 
The following however applies to any choice of statistics for both types of particles.

We write the Hamiltonian for the full system as 
$$
H_{M+N} = H_M \otimes \one + \one \otimes H_N + \sum_{j=1} ^M \sum_{k=1} ^N W_{12} (x_j-x_k), 
$$
where
\begin{equation}\label{eq:start Hamil M}
H_M= \sum_{j=1} ^M \frac{1}{2m}\left(p_{x_j} + e \Aext (x_j)\right) ^2 + \sum_{1\leq i<j\leq M} W_{11} (x_i-x_j)
\end{equation}
is the Hamiltonian for the first type of particles and 
\begin{equation}\label{eq:start Hamil N}
H_N= \sum_{k=1} ^N  \frac{1}{2}\left(p_{y_k} + \Aext (y_k)\right) ^2   + 
\sum_{1\leq i<j\leq N} W_{22} (y_i-y_j)
\end{equation}
that for the second type of particles. We shall denote $X_M= (x_1,\ldots,x_M)$ and $Y_N=(y_1,\ldots,y_N)$ the coordinates of the two types of particles and choose units so that $\hbar = c = 1$, and the mass and charge of the \emph{second} type of particles are respectively $1$ and $-1$. 
We keep the freedom that the first type of particles might have a different mass $m$ and a different charge 
$-e<0$. The first type of particles should be thought of as tracers immersed in a large sea of the second type. We shall accordingly use the terms ``tracer particles'' and ``bath particles'' in the sequel.

We have also introduced:
\begin{itemize}
 \item the usual momenta $p_{x_j} \!= - i \nabla_{x_j}$ and $p_{y_k} \!= -i \nabla_{y_k}$.
 \item a uniform magnetic field of strength $B >0$.
	Our convention is that it points downwards: 
 $$ 
	\Aext (x) := -\frac{B}{2} x^{\perp}
	= -\frac{B}{2} (-x_2,x_1).
 $$
 \item intra-species interaction potentials, $W_{11}$ and $W_{22}$. We think of short-range potentials for simplicity.
 \item an inter-species interaction potential $W_{12}$. Again, $W_{12}$ should be thought of as being short-range.
\end{itemize}

The next step is to assume a strong hierarchy of energy scales. The requirements are rather stringent, but we think they are ultimately necessary for the emergence of fractional statistics. We assume that the energy scales of the problem are set, in order of decreasing importance/magnitude, by 
\begin{enumerate}
 \item the external magnetic field $B$, which forces all bath particles to occupy lowest Landau levels orbitals. Note that the splitting between Landau levels for the tracer particles is rather proportional to $eB/m$ so if $m>e$ it is reasonable to allow that they occupy several Landau levels.
 \item the interaction potential $W_{22}$, which forces the joint wave-function of the system to vanish when two bath particles encounter.
 \item the interaction potential $W_{12}$, which forces the joint wave-function of the system to vanish when a bath and a tracer particles encounter.
 \item all the rest, which, as we will argue, essentially boils down to an effective energy for the tracer particles. 
\end{enumerate}

This leads us to a reasonable ansatz for the joint wave-function of the system:
\begin{equation}\label{eq:ansatz 3}
 \Psi (X_M,Y_N) = \Phi (X_M) \cQH(X_M) \PsiQH (X_M,Y_N)
\end{equation}
where $\PsiQH$ describes a Laughlin state of the $N$ bath particles, coupled to $M$ quasi-holes at the locations of the tracer particles:
\begin{equation}\label{eq:QH}
\PsiQH (X_M,Y_N) 
 =  \prod_{j=1} ^M \prod_{k=1} ^N (\zeta_j - z_k) ^n   \prod_{1\leq i<j \leq N} (z_i-z_j) ^{ \ell }  e ^{- B \sum_{j=1} ^N |z_j| ^2 / 4}. 
\end{equation}
We have here identified $\R ^2 \ni x_j \leftrightarrow \zeta_j \in \C$ and $\R ^2 \ni y_k \leftrightarrow z_k \in \C$. We choose $\cQH (X_M) >0 $ to enforce 
\begin{equation}\label{eq:norm QH}
\cQH (X_M) ^2 \int_{\R ^{2N}}  |\PsiQH (X_M,Y_N)| ^2 \,dY_N = 1 
\end{equation}
for any $X_M$. We thus ensure normalization of the full wave function $\Psi$ by demanding that 
$$ \int_{\R ^{2M}} |\Phi (X_M)| ^2 \,d X_M = 1.$$
In~\eqref{eq:QH}, $n$ and $\ell$ are two positive integers. The latter should be even for symmetry reasons (it would be odd if the bath particles were fermions).

\subsection{Derivation of the magnetic gauge Hamiltonian}

The wave-function $\Phi$ describing the tracer particles is the only  variable part left. One thus wants to choose it properly in order to approximate the ground state of the full system. Our main claim is that this can be achieved by taking $\Phi$ as the ground state of some effective Hamiltonian:

\begin{mtheorem}[\textbf{Emergence of fractional statistics}]\label{thm:any emer}\mbox{}\\
We make the previous assumptions on the energy scales of the problem. In addition we suppose that $N\gg M$ and that $W_{11}$ contains a hard-core of radius larger than $l_B = \sqrt{2/B}$. Then we claim the following:
\begin{enumerate}
 \item An ansatz of the form~\eqref{eq:ansatz 3}  is a good approximation for the ground state of the full system.
 \item The full ground state energy satisfies 
 \begin{equation}\label{eq:any emer ener}
E(N,M) \simeq \left\langle \Psi, H_{M+N} \Psi \right\rangle \simeq \left\langle \Phi, \Heff_M  \Phi \right\rangle + \frac{B N}{2}
\end{equation}
where 
\begin{equation}\label{eq:any emer hamil}
\boxed{\Heff_M= \sum_{j=1} ^M \frac{1}{2m}  \left( p_{x_j} + \left( e - \frac{1}{\ell} \right) \Aext (x_j) - \frac{1}{\ell} \Aan (x_j) \right)^2 + \sum_{1\leq i<j\leq M} W_{11} (x_i-x_j),}
 \end{equation}
$\Aan (x_j)$ being defined as in~\eqref{eq:any gauge vect}.  
 \item Consequently, the optimal choice for $\Phi$ is a ground state of $\Heff_M$ and we have 
\begin{equation}\label{eq:any emer ener bis}
E(N,M) \simeq \inf\left\{ \left\langle \Phi, \Heff_M  \Phi \right\rangle, \: \Phi \in L^2_\sym (\R ^{2M}), \: \int_{\R ^{2M}} |\Phi| ^2 = 1 \right\}+ \frac{B N}{2}.
\end{equation}
\end{enumerate}
\end{mtheorem}

A few comments are in order, keeping in mind that this statement is \emph{not} a rigorous theorem:
\begin{itemize}
 \item We are lead to conjecture that all the physics boils down to the motion of the tracer particles. These are dressed by holes in the density of bath particles and thus become quasi-particles whose physics is determined by an emergent Hamiltonian.
 \item A first remarkable fact about $\Heff_M$ is that the charge of the tracer particles gets reduced: $e\to e - 1/\ell$ in the coupling to the external magnetic field. This comes about because Laughlin quasi-holes carry a fraction of an electron's charge~\cite{Laughlin-83,Laughlin-87}. This is a well documented fact, including experimentally~\cite{SamGlaJinEti-97,MahaluEtal-97,YacobiEtal-04}.
 \item That the Laughlin quasi-holes have fractional statistics had also been conjectured early on~\cite{AroSchWil-84,Haldane-83,Halperin-84}. This has so far not been confirmed experimentally, despite some promising hints~\cite{CamZhoGol-05} and some interesting proposals~\cite{SafDevMar-01,GroSimSte-06,KimLawVisSmiFra-05,KimLawVisSmiFra-06,CooSim-15,ZhaSreGemJai-14,ZhaSreJai-15}. 
 \item In our approach, the emergence of fractional statistics comes about via the vector potential $- \ell ^{-1} \Aan (x_j)$, as discussed in Section~\ref{sec:intro anyons}.
 \item One could formally gauge $- \ell ^{-1} \Aan (x_j)$ away and consider minimizing a usual magnetic Hamiltonian amongst wave functions satisfying~\eqref{eq:anyon gauge} with $\alpha = -1/\ell$. This is not desirable since the magnetic gauge Hamiltonian acting on bosonic wave-functions we derive directly is more tractable, see Section~\ref{sec:MF anyons} below and references therein.
 \item The original paper~\cite{LunRou-16} also discusses possible experimental signatures of the above statement. In some specific experimental regime, one can argue that a measurement of the density of the tracer particles should reveal a clear influence of fractional statistics.
 \item The usual argument in favor of fractional statistics, originating in~\cite{AroSchWil-84} (see also~\cite{KjoLei-97,KjoLei-99,KjoMyr-99,Myrheim-99}), proceeds by first treating the tracer particles/quasi-holes as classical parameters in the function~\eqref{eq:ansatz 3}. That is, we freeze the $\Phi$ degree of freedom and consider $(\zeta_j)_{j=1 \ldots M}$ as a family of parameters for the quasi-holes wave-function. One can compute that, upon moving a $\zeta_j$ around another adiabatically, the quasi-holes wave-function picks a Berry phase $e^{-2i\pi/\ell}$. This Berry phase is then identified with the square of the statistics exchange phase quasi-holes should have in a quantum description. While certainly compelling, this argument is not devoid of issues (see for example~\cite{Forte-91}).
 \item The main leap of faith in the argument of~\cite{AroSchWil-84} is the identification of the Berry phase picked by the bath particles' wave function and the statistics phase of the tracer particles. This is completely avoided in our approach, since we do not appeal to the Berry phase concept at all.
\end{itemize}

\section{Average-field approximation for almost bosonic anyons}\label{sec:MF anyons}

Now that we have seen that anyons can emerge in some concrete physical situations, we  turn to the next obvious question. Namely: how does a system of anyons behave, what characteristic features of it could be observed ? Since the free anyon Hamiltonian~\eqref{eq:anyons mag hamil} is pretty intricate, and basically no general exact results are known (contrarily to the case of free bosons or fermions), this raises the question of deriving effective, simpler theories. In this section we report on the results of~\cite{LunRou-15}, where such an effective theory is rigorously derived, in a mean-field-like limit. 

\subsection{Model for trapped extended anyons}

Here we consider the ground state problem for free anyons in a trapping potential $V:\R^2 \to \R$, that is we are interested in the bottom of the spectrum of the Hamiltonian (we discard $\alpha$ from the notation for shortness)
\begin{equation}\label{eq:any mag hamil trap}
H_N  = \sum_{j=1} ^N \left( p_j + \alpha \Aan (x_j) \right) ^2 + V(x_j).  
\end{equation}
where the gauge vector potential is again given by~\eqref{eq:any gauge vect}. Compared to~\eqref{eq:anyons mag hamil} and~\eqref{eq:any emer hamil} we have discarded the external gauge vector potential and the pair interaction. This is because they would not complicate the picture significantly, at least with the approach we follow in~\cite{LunRou-15}. The above Hamiltonian can be rigorously defined and shown to enjoy crucial functional inequalities~\cite{LarLun-16,LunSol-13a,LunSol-13b,LunSol-14}. However, its complicated structure makes it very hard to draw precise predictions, even about its ground state. To appreciate this, consider formally expanding the squares in~\eqref{eq:any mag hamil trap}. We obtain
\begin{align}\label{eq:any expand hamil}
H_N  &= \sum_{j=1}^N \left( p_j ^2 + V(x_j) \right) \nonumber \\
&+ \alpha \sum_{j\neq k} \left( p_j \cdot \nablap w_0 (x_j - x_k ) + \nablap w_0 (x_j - x_k ) \cdot p_j \right) \nonumber \\
&+ \alpha ^2 \sum_{j\neq k \neq \ell} \nablap w_0 (x_j-x_k) \cdot \nablap w_0 (x_j-x_\ell) \nonumber \\
&+ \alpha ^2 \sum_{j\neq k} |\nabla w_0 (x_j-x_k)| ^2
\end{align}
where we wrote
$$w_0 := \log |\:.\:|$$
for the 2D Coulomb kernel. We are thus dealing with ordinary bosons (recall that we are interested in the action of $H_N$ on symmetric wave-functions) that experience:
\begin{itemize}
 \item a two-body interaction mixing space and momentum variables (second line);
 \item a three-body interaction (third line);
 \item a singular two-body interaction (the terms on the fourth line are not locally integrable).
\end{itemize}
In this section we want to justify that, in spite of these peculiar features, the ground state of $H_N$ can in some regime be approximated in the form of a pure tensor power $u ^{\otimes N}$, in the spirit of Chapter~\ref{sec:bosons GS}. In the context of fractional statistics, such approximations are usually refered to as ``average-field models''~\cite{FetHanLau-89,ChenWilWitHal-89,ChiSen-92,Trugenberger-92b,Trugenberger-92,IenLec-92,Westerberg-93}.

An immediate difficulty is that pure tensor powers $u ^{\otimes N}$ are \emph{not} in the domain of $H_N$, because of the singular pair interaction. To remedy this, we smear the Aharonov-Bohm flux carried by each particle on a disc of radius $R>0$: let
\begin{equation}\label{eq:smeared Coul}
w_R  := \log |\:.\:| \ast \frac{\1_{B(0,R)}}{\pi R^2} . 
\end{equation}
and consider (observe that $\Aan = \Aan ^0$)
\begin{equation}\label{eq:potential}
	\Aan ^R (x_j) := \sum_{k\neq j} \nablap w_R (x_j-x_k), 
\end{equation}
together with the regularized Hamiltonian
\begin{equation}\label{eq:hamil}
	H_N ^R := \sum_{j=1}^N \left( p_j + \alpha \Aan^R (x_j) \right) ^2 + V(x_j).
\end{equation}
This so-called ``extended anyons'' model is discussed in~\cite{LarLun-16,Mashkevich-96,Trugenberger-92b,ChoLeeLee-92}. Besides the fact that it offers a reasonable simplification of the original Hamiltonian, one can motivate it because anyons should be seen as emergent quasi-particles, and have thus a finite radius. In Statement~\ref{thm:any emer}, this is set by the radius of the hard-core we assumed was contained in~$W_{11}$. 

For fixed $R>0$, this operator is self-adjoint on $L_{\rm sym} ^2 (\R ^{2N})$ and we shall denote 
\begin{equation}\label{eq:gse any}
	E ^R (N) := \inf \sigma (H_N ^R)
\end{equation}
the associated ground state energy (lowest eigenvalue) for $N$ extended anyons. We are interested in the asymptotics of $E^R (N)$ and the associated minimizers in the limit $N\to \infty$ and $R\to 0$.  

\subsection{Justification of the average field approximation}

To obtain a well-defined limit problem we set
\begin{equation}\label{eq:scale alpha}
	\alpha = \frac{\beta}{N-1} \to 0 \ \ \text{when} \ \ N\to \infty,
\end{equation}
where $\beta$ is a given, fixed constant. This choice ensures that all terms in~\eqref{eq:any expand hamil} are formally of order $N$ or smaller. Physically it means we consider ``almost bosonic anyons''. One can also perturb around fermions, but we have not studied this case. Note that if we simply took $\alpha \to 0$ at fixed $N$, we would recover ordinary bosons at leading order. One could then only see the effect of the non-trivial statistics in a 
perturbative expansion, a route followed e.g. in~\cite{Sen-91,ComCabOuv-91,SenChi-92,Ouvry-94,ComMasOuv-95}. 

The effective model we shall derive encodes the effect of anyon statistics in the form of a self-consistent gauge field:
\begin{equation}\label{eq:avg func}
	\cEAF [u] := \int_{\R ^2} \left( \left| \left( \nabla 
		+ i \beta \mathbf{A} [|u|^2] \right) u \right|^2 + V|u|^2 \right)
\end{equation}
where 
\begin{equation}\label{eq:avg field}
	\bA [\rho] := \nablap w_0 \ast \rho.
\end{equation}
One obtains this functional by the same kind of steps as in the case of usual bosons with pair interactions: either replacing the magnetic gauge vector potential by a mean-field, or by taking (essentially) an ansatz of the form $u ^{\otimes N}$. We denote $\EAF$ the ground state energy obtained by minimizing~\eqref{eq:avg func} subject to the unit mass constraint
$$ \int_{\R ^2} |u| ^2 = 1.$$

For technical reasons we assume that the one-body potential is confining 
\begin{equation}\label{eq:trap pot any}
		V(x) \ge c |x|^s - C, \quad s>0,
\end{equation}
and that the size $R$ of the extended anyons does not go to zero too fast in 
the limit $N\to \infty$. The following is the main result of~\cite{LunRou-15}:

\begin{theorem}[\textbf{Validity of the average field approximation}]\label{thm:main ener}\mbox{}\\
We consider $N$ extended anyons of radius
	$R \sim N ^{-\eta}$ in an external potential $V$ satisfying~\eqref{eq:trap pot any}. We assume the relation 
	\begin{equation}\label{eq:eta restriction}
		0 < \eta < \eta_0 (s) := \frac{1}{4}\left( 1+ \frac{1}{s}\right)^{-1},
	\end{equation}
	and that the statistics parameter scales as~\eqref{eq:scale alpha}.  
	
	Then, in the limit $N\to \infty$ we have 
	for the ground-state energy
	\begin{equation}\label{eq:main ener}
		\frac{E^R(N)}{N} \to \EAF.
	\end{equation}
	Moreover, if $\Psi_N$ is a sequence of ground states for $H_N ^R$, with 
	associated reduced density matrices $\gamma_N ^{(k)}$, then 
	we have, along a subsequence,
	\begin{equation}\label{eq:main state}
		\gamma_N ^{(k)} \to \int_{\cMAF} |u ^{\otimes k} \rangle \langle u ^{\otimes k} | \,d\mu(u) 
	\end{equation}
	strongly in the trace-class when $N\to \infty$, where $\mu$ is a Borel 
	probability measure supported on the set of minimizers of $\cEAF$,
	$$
		\cMAF := \lbrace u\in L ^2 (\R ^2) : \norm{u}_{L^2} = 1, \: \cEAF [u] = \EAF \rbrace.
	$$
\end{theorem}

A few comments:

\begin{itemize}
\item The method of proof is inspired from that of~\cite{LewNamRou-14c}, cf Chapter~\ref{sec:bosons GS}. The use of the quantum de Finetti theorem is even more crucial here: there does not seem to be any way of using the structure of the magnetic gauge Hamiltonian to justify the approximation.
\item Actually, its peculiar form adds quite a few twists to the proof, compared to~\cite{LewNamRou-14c}. We have to rely on some specific functional inequalities~\cite{Lundholm-15,HofLapTid-08} to control the error terms.
\item Note that the limit problem~\eqref{eq:avg func} comprises an effective 
	self-consistent \emph{magnetic} field. 
	A term in the form of a self-consistent \emph{electric} field is more usual. Taking the $\curl$ of $\mathbf{A}[|u|^2]$ we see that the effective field is directly proportional to the matter density.
\item We deal with the limit $R\to 0$ at the same time as 
	$N \to \infty$ because typically the anyon radius should be much smaller than the total system's size. This is reminiscent of the NLS and GP limits for trapped 
	Bose gases, see again Chapter~\ref{sec:bosons GS} and references therein. 
\item Taking $R$ not too small in the limit $N\to \infty$ is a  requirement in our method of proof. 
It is in fact not clear whether some lower bound on $R$ is a necessary 
condition for the average field description to be correct. 
For very small or zero $R$, it is still conceivable that a description in 
terms of a functional of the form of~\eqref{eq:avg func} is correct in the limit. A first satisfying step would be to allow $R\ll N ^{-1/2}$, the typical inter-particle distance.
\end{itemize}
% 
% \bibliographystyle{siam}
% \bibliography{/home/rougerie/Travail/Documentation/Bibtex/biblio_NR_Fev16}

\begin{thebibliography}{100}

\bibitem{AboRamVogKet-01}
{\sc J.~R. Abo-Shaeer, C.~Raman, J.~M. Vogels, and W.~Ketterle}, {\em
  {Observation of Vortex Lattices in Bose-Einstein Condensates}}, Science, 292
  (2001), pp.~476--479.

\bibitem{Aftalion-06}
{\sc A.~Aftalion}, {\em Vortices in {B}ose--{E}instein Condensates}, vol.~67 of
  Progress in nonlinear differential equations and their applications,
  Springer, 2006.

\bibitem{Aftalion-07}
\leavevmode\vrule height 2pt depth -1.6pt width 23pt, {\em Vortex patterns in
  {B}ose {E}instein condensates}, in Perspectives in nonlinear partial
  differential equations, vol.~446 of Contemp. Math., Amer. Math. Soc.,
  Providence, RI, 2007, pp.~1--18.

\bibitem{AftAlaBro-05}
{\sc A.~Aftalion, S.~Alama, and L.~Bronsard}, {\em {Giant vortex and breakdown
  of strong pinning in a rotating Bose-Einstein condensate}}, Arch. Rational
  Mech. Anal., 178 (2005), pp.~247--286.

\bibitem{AftBla-06}
{\sc A.~Aftalion and X.~Blanc}, {\em Vortex lattices in rotating
  {B}ose--{E}instein condensates}, SIAM Journal on Mathematical Analysis, 38
  (2006), pp.~874--893.

\bibitem{AftBla-08}
\leavevmode\vrule height 2pt depth -1.6pt width 23pt, {\em {Reduced energy
  functionals for a three-dimensional fast rotating {B}ose-{E}instein
  condensates}}, Annales Henri Poincar{\'e}, 25 (2008), pp.~339--355.

\bibitem{AftBlaDal-05}
{\sc A.~Aftalion, X.~Blanc, and J.~Dalibard}, {\em Vortex patterns in a fast
  rotating {B}ose-{E}instein condensate}, Phys. Rev. A, 71 (2005), p.~023611.

\bibitem{AftBlaNie-06b}
{\sc A.~Aftalion, X.~Blanc, and F.~Nier}, {\em Lowest {L}andau level functional
  and {B}argmann spaces for {B}ose-{E}instein condensates}, J. Funct. Anal.,
  241 (2006), pp.~661--702.

\bibitem{AftBlaNie-06a}
\leavevmode\vrule height 2pt depth -1.6pt width 23pt, {\em {Vortex distribution
  in the lowest Landau level}}, Phys. Rev. A, 73 (2006), p.~011601(R).

\bibitem{AftJerRoy-11}
{\sc A.~Aftalion, R.~Jerrard, and J.~Royo-Letelier}, {\em Non existence of
  vortices in the small density region of a condensate}, J. Funct. Anal., 260
  (2011), pp.~2387--2406.

\bibitem{Aldous-85}
{\sc D.~Aldous}, {\em Exchangeability and Related Topics}, Lecture Notes in
  Math., Springer, 1985.

\bibitem{AleDev-09}
{\sc A.~Alexandrov and J.~Devreese}, {\em Advances in Polaron Physics},
  Springer Series in Solid-State Sciences, Springer, 2009.

\bibitem{AlmHel-06}
{\sc Y.~Almog and B.~Helffer}, {\em The distribution of surface
  superconductivity along the boundary: on a conjecture of {X.B. Pan}}, SIAM
  Journal on Mathematical Analysis, 38 (2006), pp.~1715 -- 1732.

\bibitem{Ammari-hdr}
{\sc Z.~Ammari}, {\em Syst\`emes hamiltoniens en th\'eorie quantique des champs
  : dynamique asymptotique et limite classique.}
\newblock Habilitation {\`a} Diriger des Recherches, University of Rennes I,
  February 2013.

\bibitem{AmmNie-08}
{\sc Z.~Ammari and F.~Nier}, {\em Mean field limit for bosons and infinite
  dimensional phase-space analysis}, Ann. Henri Poincar\'e, 9 (2008),
  pp.~1503--1574.

\bibitem{AmmNie-09}
\leavevmode\vrule height 2pt depth -1.6pt width 23pt, {\em Mean field limit for
  bosons and propagation of {W}igner measures}, J. Math. Phys., 50 (2009),
  p.~042107.

\bibitem{AmmNie-11}
\leavevmode\vrule height 2pt depth -1.6pt width 23pt, {\em {Mean field
  propagation of Wigner measures and BBGKY hierarchies for general bosonic
  states}}, J. Math. Pures Appl., 95 (2011), pp.~585--626.

\bibitem{AmmNie-15}
\leavevmode\vrule height 2pt depth -1.6pt width 23pt, {\em {Mean field
  propagation of infinite dimensional Wigner measures with a singular two-body
  interaction potential}}, Ann. Sc. Norm. Sup. Pisa., 14 (2015), pp.~155--220.

\bibitem{AndGuiZei-10}
{\sc G.~W. Anderson, A.~Guionnet, and O.~Zeitouni}, {\em An introduction to
  random matrices}, vol.~118 of Cambridge Studies in Advanced Mathematics,
  Cambridge University Press, Cambridge, 2010.

\bibitem{AroSchWil-84}
{\sc S.~Arovas, J.~Schrieffer, and F.~Wilczek}, {\em Fractional statistics and
  the quantum {H}all effect}, Phys. Rev. Lett., 53 (1984), pp.~722--723.

\bibitem{Bach-91}
{\sc V.~Bach}, {\em Ionization energies of bosonic {C}oulomb systems}, Lett.
  Math. Phys., 21 (1991), pp.~139--149.

\bibitem{BacBarHelSie-99}
{\sc V.~Bach, J.~M. Barbaroux, B.~Helffer, and H.~Siedentop}, {\em On the
  stability of the relativistic electron-positron field}, Commun. Math. Phys.,
  201 (1999), pp.~445--460.

\bibitem{BacLewLieSie-93}
{\sc V.~Bach, R.~Lewis, E.~H. Lieb, and H.~Siedentop}, {\em On the number of
  bound states of a bosonic {$N$}-particle {C}oulomb system}, Math. Z., 214
  (1993), pp.~441--459.

\bibitem{BarCooSch-57}
{\sc J.~Bardeen, L.~N. Cooper, and J.~R. Schrieffer}, {\em Theory of
  superconductivity}, Phys. Rev., 108 (1957), pp.~1175--1204.

\bibitem{BarRes-86}
{\sc S.~Baroni and R.~Resta}, {\em Ab initio calculation of the macroscopic
  dielectric constant in silicon}, Phys. Rev. B, 33 (1986), pp.~7017--7021.

\bibitem{BelSchEls-94}
{\sc J.~Bellissard, H.~Schulz-Baldes, and A.~van Elst}, {\em The non
  commutative geometry of the quantum {Hall} effect}, J. Math. Phys., 35
  (1994), pp.~5373--5471.

\bibitem{BenGui-97}
{\sc G.~Ben~Arous and A.~Guionnet}, {\em Large deviations for {W}igner's law
  and {V}oiculescu's non-commutative entropy}, Probab. Theory Related Fields,
  108 (1997), pp.~517--542.

\bibitem{BenZei-98}
{\sc G.~Ben~Arous and O.~Zeitouni}, {\em Large deviations from the circular
  law}, ESAIM: Probability and Statistics, 2 (1998), pp.~123--134.

\bibitem{BenOliSch-12}
{\sc N.~{Benedikter}, G.~{de Oliveira}, and B.~{Schlein}}, {\em {Quantitative
  Derivation of the Gross-Pitaevskii Equation}}, Comm. Pure App. Math., 68
  (2015), pp.~1399--1482.

\bibitem{BenPorSch-15}
{\sc N.~{Benedikter}, M.~{Porta}, and B.~{Schlein}}, {\em {Effective Evolution
  Equations from Quantum Dynamics}}, Springer Briefs in Mathematical Physics,
  Springer, 2016.

\bibitem{BenLie-83}
{\sc R.~{Benguria} and E.~H. {Lieb}}, {\em {Proof of the Stability of Highly
  Negative Ions in the Absence of the Pauli Principle}}, Phys. Rev. Lett., 50
  (1983), pp.~1771--1774.

\bibitem{Berezin-72}
{\sc F.~A. Berezin}, {\em Convex functions of operators}, Mat. Sb. (N.S.),
  88(130) (1972), pp.~268--276.

\bibitem{BerPap-99}
{\sc G.~Bertsch and T.~Papenbrock}, {\em Yrast line for weakly interacting
  trapped bosons}, Phys. Rev. Lett., 83 (1999), pp.~5412--5414.

\bibitem{BetBreHel-94}
{\sc F.~B\'ethuel, H.~Br\'ezis, and F.~H{\'e}lein}, {\em Ginzburg-{L}andau
  vortices}, Progress in Nonlinear Differential Equations and their
  Applications, 13, Birkh\"auser Boston, Inc., Boston, MA, 1994.

\bibitem{BlaRou-08}
{\sc X.~Blanc and N.~Rougerie}, {\em {Lowest-Landau-Level vortex structure of a
  Bose-Einstein condensate rotating in a harmonic plus quartic trap}}, Phys.
  Rev. A, 77 (2008), p.~053615.

\bibitem{BloDalZwe-08}
{\sc I.~Bloch, J.~Dalibard, and W.~Zwerger}, {\em Many-body physics with
  ultracold gases}, Rev. Mod. Phys., 80 (2008), pp.~885--964.

\bibitem{BocCenSch-15}
{\sc C.~{Boccato}, S.~{Cenatiempo}, and B.~{Schlein}}, {\em {Quantum many-body
  fluctuations around nonlinear {S}chr\"odinger dynamics}}, Preprint (2015)
  arXiv:1509.03837.

\bibitem{Bose-24}
{\sc S.~Bose}, {\em {P}lancks {G}esetz und {L}ichtquantenhypothese}, Z. Phys.,
  26 (1924), pp.~178--181.

\bibitem{BouErdYau-12}
{\sc P.~Bourgade, L.~Erd\"os, and H.-T. Yau}, {\em Bulk universality of general
  $\beta$-ensembles with non-convex potential}, J. Math. Phys., 53 (2012),
  p.~095221.

\bibitem{BouErdYau-14}
\leavevmode\vrule height 2pt depth -1.6pt width 23pt, {\em Universality of
  general $\beta$-ensembles}, Duke Mathematical Journal, 163 (2014),
  pp.~1127--1190.

\bibitem{Bourgain-94}
{\sc J.~Bourgain}, {\em Periodic nonlinear {S}chr\"odinger equation and
  invariant measures}, Comm. Math. Phys., 166 (1994), pp.~1--26.

\bibitem{Bourgain-96}
\leavevmode\vrule height 2pt depth -1.6pt width 23pt, {\em Invariant measures
  for the {2D}-defocusing nonlinear {S}chr\"odinger equation}, Comm. Math.
  Phys., 176 (1996), pp.~421--445.

\bibitem{Bourgain-97}
\leavevmode\vrule height 2pt depth -1.6pt width 23pt, {\em {Invariant measures
  for the Gross-Pitaevskii equation}}, Journal de Math\'ematiques Pures et
  Appliqu\'ees, 76 (1997), pp.~649--02.

\bibitem{Bourgain-00}
\leavevmode\vrule height 2pt depth -1.6pt width 23pt, {\em Invariant measures
  for {NLS} in infinite volume}, Communications in Mathematical Physics, 210
  (2000), pp.~605--620.

\bibitem{BraHar-12}
{\sc F.~Brand\~{a}o and A.~Harrow}, {\em {Quantum de Finetti Theorems under
  Local Measurements with Applications}}.
\newblock arXiv:1210.6367.

\bibitem{BreStoSeuDal-04}
{\sc V.~Bretin, S.~Stock, Y.~Seurin, and J.~Dalibard}, {\em {Fast rotation of a
  Bose-Einstein condensate}}, Phys. Rev. Lett., 92 (2004), p.~050403.

\bibitem{BurThoTzv-10}
{\sc N.~{Burq}, L.~{Thomann}, and N.~{Tzvetkov}}, {\em {Long time dynamics for
  the one dimensional non linear Schr\"odinger equation}}, Ann. Inst. Fourier.,
  63 (2013), pp.~2137--2198.

\bibitem{BurTzv-08}
{\sc N.~Burq and N.~Tzvetkov}, {\em Random data {C}auchy theory for
  supercritical wave equations. {I}. {L}ocal theory}, Invent. Math., 173
  (2008), pp.~449--475.

\bibitem{CacSuz-14}
{\sc F.~Cacciafesta and A.-S. {de Suzzoni}}, {\em Invariant measure for the
  {S}chr\"odinger equation on the real line}, J. Func Anal, 269 (2015),
  pp.~271--324.

\bibitem{CagLioMarPul-92}
{\sc E.~Caglioti, P.-L. Lions, C.~Marchioro, and M.~Pulvirenti}, {\em A special
  class of stationary flows for two-dimensional {E}uler equations: a
  statistical mechanics description}, Comm. Math. Phys., 143 (1992),
  pp.~501--525.

\bibitem{CamZhoGol-05}
{\sc F.~E. Camino, W.~Zhou, and V.~J. Goldman}, {\em {Realization of a Laughlin
  quasiparticle interferometer: Observation of fractional statistics}}, Phys.
  Rev. B, 72 (2005), p.~075342.

\bibitem{CanDelLew-08a}
{\sc {\'E}.~Canc{\`e}s, A.~Deleurence, and M.~Lewin}, {\em A new approach to
  the modelling of local defects in crystals: the reduced {H}artree-{F}ock
  case}, Commun. Math. Phys., 281 (2008), pp.~129--177.

\bibitem{CanDelLew-08b}
\leavevmode\vrule height 2pt depth -1.6pt width 23pt, {\em Non-perturbative
  embedding of local defects in crystalline materials}, J. Phys.: Condens.
  Matter, 20 (2008), p.~294213.

\bibitem{CanLew-10}
{\sc {\'E}.~Canc{\`e}s and M.~Lewin}, {\em The dielectric permittivity of
  crystals in the reduced {H}artree-{F}ock approximation}, Arch. Ration. Mech.
  Anal., 197 (2010), pp.~139--177.

\bibitem{Carlen-91}
{\sc E.~Carlen}, {\em Some integral identities and inequalities for entire
  functions and their application to the coherent state transform}, J. Funct.
  Anal., 97 (1991), pp.~231--249.

\bibitem{Carlen-10}
\leavevmode\vrule height 2pt depth -1.6pt width 23pt, {\em Trace inequalities
  and quantum entropy: an introductory course}, in Entropy and the Quantum,
  R.~Sims and D.~Ueltschi, eds., vol.~529 of Contemporary Mathematics, American
  Mathematical Society, 2010, pp.~73--140.
\newblock Arizona School of Analysis with Applications, March 16-20, 2009,
  University of Arizona.

\bibitem{CatBriLio-98}
{\sc I.~Catto, C.~{Le Bris}, and P.-L. Lions}, {\em The mathematical theory of
  thermodynamic limits: {T}homas-{F}ermi type models}, Oxford Mathematical
  Monographs, The Clarendon Press Oxford University Press, New York, 1998.

\bibitem{CatBriLio-01}
\leavevmode\vrule height 2pt depth -1.6pt width 23pt, {\em On the thermodynamic
  limit for {H}artree-{F}ock type models}, Ann. Inst. H. Poincar{\'e} Anal. Non
  Lin{\'e}aire, 18 (2001), pp.~687--760.

\bibitem{CatBriLio-02}
\leavevmode\vrule height 2pt depth -1.6pt width 23pt, {\em On some periodic
  {H}artree-type models for crystals}, Ann. Inst. H. Poincar\'e Anal. Non
  Lin\'eaire, 19 (2002), pp.~143--190.

\bibitem{CavFucSch-02}
{\sc C.~M. Caves, C.~A. Fuchs, and R.~Schack}, {\em Unknown quantum states: the
  quantum de {F}inetti representation}, J. Math. Phys., 43 (2002), p.~4535.

\bibitem{ChaGozZit-13}
{\sc D.~Chafa\"i, N.~Gozlan, and P.-A. Zitt}, {\em First order asymptotics for
  confined particles with singular pair repulsions}, {Annals of Applied
  Probability}, 24 (2014), pp.~2371--2413.

\bibitem{ChaIra-89}
{\sc P.~Chaix and D.~Iracane}, {\em From quantum electrodynamics to mean field
  theory: {I}. {T}he {B}ogoliubov-{D}irac-{F}ock formalism}, J. Phys. B, 22
  (1989), pp.~3791--3814.

\bibitem{ChaIraLio-89}
{\sc P.~Chaix, D.~Iracane, and P.-L. Lions}, {\em From quantum electrodynamics
  to mean field theory: {I}{I}. {V}ariational stability of the vacuum of
  quantum electrodynamics in the mean-field approximation}, J. Phys. B, 22
  (1989), pp.~3815--3828.

\bibitem{CheHol-15}
{\sc X.~{Chen} and J.~{Holmer}}, {\em {The rigorous derivation of the {2D}
  cubic focusing {NLS} from quantum many-body evolution}}, Preprint (2015)
  arXiv:1508.07675.

\bibitem{CheHol-13}
\leavevmode\vrule height 2pt depth -1.6pt width 23pt, {\em {Focusing Quantum
  Many-body Dynamics: The Rigorous Derivation of the 1D Focusing Cubic
  Nonlinear Schr\"odinger Equation}}, Arch. Rat. Mech. Anal., 221 (2016),
  pp.~631--676.

\bibitem{ChenWilWitHal-89}
{\sc Y.~H. Chen, F.~Wilczek, E.~Witten, and B.~I. Halperin}, {\em On anyon
  superconductivity}, Int. J. Mod. Phys. B, 3 (1989), pp.~1001--1067.

\bibitem{Chiribella-11}
{\sc G.~Chiribella}, {\em On quantum estimation, quantum cloning and finite
  quantum de {F}inetti theorems}, in Theory of Quantum Computation,
  Communication, and Cryptography, vol.~6519 of Lecture Notes in Computer
  Science, Springer, 2011.

\bibitem{ChiSen-92}
{\sc R.~Chitra and D.~Sen}, {\em {Ground state of many anyons in a harmonic
  potential}}, Phys. Rev. B, 46 (1992), pp.~10923--10930.

\bibitem{ChoLeeLee-92}
{\sc M.~Y. Choi, C.~Lee, and J.~Lee}, {\em Soluble many-body systems with
  flux-tube interactions in an arbitrary external magnetic field}, Phys. Rev.
  B, 46 (1992), pp.~1489--1497.

\bibitem{ChrKonMitRen-07}
{\sc M.~Christandl, R.~K{\"o}nig, G.~Mitchison, and R.~Renner}, {\em
  One-and-a-half quantum de {F}inetti theorems}, Comm. Math. Phys., 273 (2007),
  pp.~473--498.

\bibitem{ChrTon-09}
{\sc M.~Christiandl and B.~Toner}, {\em Finite de {F}inetti theorem for
  conditional probability distributions describing physical theories}, J. Math.
  Phys., 50 (2009), p.~042104.

\bibitem{Ciftja-06}
{\sc O.~Ciftj{\'{a}}}, {\em Monte {C}arlo study of {B}ose {L}aughlin wave
  function for filling factors $1/2$, $1/4$ and $1/6$}, Europhys. Lett., 74
  (2006), pp.~486--492.

\bibitem{RenCir-09}
{\sc J.~Cirac and R.~Renner}, {\em {de Finetti Representation Theorem for
  Infinite-Dimensional Quantum Systems and Applications to Quantum
  Cryptography}}, Phys. Rev. Lett., 102 (2009), p.~110504.

\bibitem{CodETALCor-04}
{\sc I.~Coddington, P.~C. Haljan, P.~Engels, V.~Schweikhard, S.~Tung, and E.~A.
  Cornell}, {\em Experimental studies of equilibrium vortex properties in a
  {B}ose-condensed gas}, Phys. Rev. A, 70 (2004), p.~063607.

\bibitem{ComMasOuv-95}
{\sc A.~Comtet, S.~Mashkevich, and S.~Ouvry}, {\em {Magnetic moment and
  perturbation theory with singular magnetic fields}}, Phys. Rev. D, 52 (1995),
  pp.~2594--2597.

\bibitem{ComCabOuv-91}
{\sc A.~Comtet, J.~Mc~Cabe, and S.~Ouvry}, {\em {Perturbative equation of state
  for a gas of anyons}}, Phys. Lett. B, 260 (1991), pp.~372--376.

\bibitem{Cooper-08}
{\sc N.~R. {Cooper}}, {\em {Rapidly rotating atomic gases}}, Advances in
  Physics, 57 (2008), pp.~539--616.

\bibitem{CooSim-15}
{\sc N.~R. Cooper and S.~H. Simon}, {\em {Signatures of Fractional Exclusion
  Statistics in the Spectroscopy of Quantum Hall Droplets}}, Phys. Rev. Lett.,
  114 (2015), p.~106802.

\bibitem{CooWil-99}
{\sc N.~R. Cooper and N.~K. Wilkin}, {\em Composite fermion description of
  rotating {B}ose-{E}instein condensates}, Phys. Rev. B, 60 (1999),
  pp.~R16279--R16282.

\bibitem{CooWilGun-01}
{\sc N.~R. Cooper, N.~K. Wilkin, and J.~M.~F. Gunn}, {\em Quantum phases of
  vortices in rotating {B}ose-{E}instein condensates}, Phys. Rev. Lett., 87
  (2001), p.~120405.

\bibitem{CorWie-nobel}
{\sc E.~A. Cornell and C.~E. Wieman}, {\em {Bose-Einstein condensation in a
  dilute gas, the first 70 years and some recent experiments}}, Rev. Mod.
  Phys., 74 (2002), pp.~875--893.

\bibitem{CorPinRouYng-11a}
{\sc M.~Correggi, F.~Pinsker, N.~Rougerie, and J.~Yngvason}, {\em Critical
  rotational speeds in the {G}ross-{P}itaevskii theory on a disc with
  {D}irichlet boundary conditions}, J. Stat. Phys., 143 (2011), pp.~261--305.

\bibitem{CorPinRouYng-11b}
\leavevmode\vrule height 2pt depth -1.6pt width 23pt, {\em Rotating superfluids
  in anharmonic traps: From vortex lattices to giant vortices}, Phys. Rev. A,
  84 (2011), p.~053614.

\bibitem{CorPinRouYng-12}
\leavevmode\vrule height 2pt depth -1.6pt width 23pt, {\em Critical rotational
  speeds for superfluids in homogeneous traps}, J. Math. Phys., 53 (2012),
  p.~095203.

\bibitem{CorPinRouYng-12b}
\leavevmode\vrule height 2pt depth -1.6pt width 23pt, {\em Vortex phases of
  rotating superfluids}, in Proceedings of the 21st International Laser Physics
  Workshop, Calgary, 2012.

\bibitem{CorPinRouYng-13}
\leavevmode\vrule height 2pt depth -1.6pt width 23pt, {\em {Giant vortex phase
  transition in rapidly rotating trapped Bose-Einstein condensates}}, in Theory
  of quantum gases and quantum coherence, Lyon, 2012, vol.~217, European
  Journal of Physics, special topics, 2013, pp.~183--188.

\bibitem{CorRinYng-07b}
{\sc M.~Correggi, T.~Rindler-Daller, and J.~Yngvason}, {\em Rapidly rotating
  {B}ose-{E}instein condensates in homogeneous traps}, J. Math. Phys., 48
  (2007), p.~102103.

\bibitem{CorRinYng-07}
\leavevmode\vrule height 2pt depth -1.6pt width 23pt, {\em Rapidly rotating
  {B}ose-{E}instein condensates in strongly anharmonic traps}, J. Math. Phys.,
  48 (2007), p.~042104.

\bibitem{CorRou-13}
{\sc M.~{Correggi} and N.~{Rougerie}}, {\em Inhomogeneous vortex patterns in
  rotating {Bose-Einstein} condensates}, Communications in Mathematical
  Physics, 321 (2013), pp.~817 -- 860.

\bibitem{CorRou-14}
\leavevmode\vrule height 2pt depth -1.6pt width 23pt, {\em On the
  {Ginzburg-Landau} functional in the surface superconductivity regime},
  Communications in Mathematical Physics, 332 (2014), pp.~1297 -- 1343.

\bibitem{CorRou-16}
\leavevmode\vrule height 2pt depth -1.6pt width 23pt, {\em {Boundary behavior
  of the Ginzburg-Landau order parameter in the surface superconductivity
  regime}}, {Arch. Ration. Mech. Anal.}, 219 (2016), pp.~553--606.

\bibitem{CorRou-16b}
\leavevmode\vrule height 2pt depth -1.6pt width 23pt, {\em Effects of boundary
  curvature on surface superconductivity}, Letters in Mathematical Physics, 106
  (2016), p.~445.

\bibitem{CorRouYng-11}
{\sc M.~Correggi, N.~Rougerie, and J.~Yngvason}, {\em The transition to a giant
  vortex phase in a fast rotating {Bose-Einstein} condensate}, Communications
  in Mathematical Physics, 303 (2011), pp.~451 -- 508.

\bibitem{Cwikel-77}
{\sc M.~Cwikel}, {\em Weak type estimates for singular values and the number of
  bound states of {S}chrodinger operators}, Annals of Mathematics, 106 (1977),
  pp.~pp. 93--100.

\bibitem{DalGioPitStr-99}
{\sc F.~Dalfovo, S.~Giorgini, L.~P. Pitaevskii, and S.~Stringari}, {\em Theory
  of {B}ose-{E}instein condensation in trapped gases}, Rev. Mod. Phys., 71
  (1999), pp.~463--512.

\bibitem{DalGerJuzOhb-11}
{\sc J.~Dalibard, F.~Gerbier, G.~Juzeli\={u}nas, and P.~\"{O}hberg}, {\em
  Artificial gauge potentials for neutral atoms}, Rev. Mod. Phys., 83 (2011),
  p.~1523.

\bibitem{Danaila-05}
{\sc I.~Danaila}, {\em Three-dimensional vortex structure of a fast rotating
  {Bose-Einstein} condensate with harmonic-plus-quartic confinement}, Phys.
  Rev. A, 72 (2005), p.~013605.

\bibitem{DeFinetti-31}
{\sc B.~{de Finetti}}, {\em Funzione caratteristica di un fenomeno aleatorio}.
\newblock Atti della R. Accademia Nazionale dei Lincei, 1931.
\newblock Ser. 6, Memorie, Classe di Scienze Fisiche, Matematiche e Naturali.

\bibitem{DeFinetti-37}
\leavevmode\vrule height 2pt depth -1.6pt width 23pt, {\em La pr\'evision : ses
  lois logiques, ses sources subjectives}, Ann. Inst. H. Poincar\'e, 7 (1937),
  pp.~1--68.

\bibitem{deGennes-66}
{\sc P.-G. de~Gennes}, {\em Superconductivity of Metals and Alloys}, Westview
  Press, 1966.

\bibitem{MahaluEtal-97}
{\sc R.~de~Picciotto, M.~Reznikov, M.~Heiblum, V.~Umansky, G.~Bunin, and
  D.~Mahalu1}, {\em Direct observation of a fractional charge}, Nature, 389
  (1997), pp.~162--164.

\bibitem{DerGer-13}
{\sc J.~Derezi{\'n}ski and C.~G{\'e}rard}, {\em {Mathematics of Quantization
  and Quantum Fields}}, Cambridge University Press, Cambridge, 2013.

\bibitem{DerNap-13}
{\sc J.~{Derezi{\'n}ski} and M.~{Napi{\'o}rkowski}}, {\em {Excitation spectrum
  of interacting bosons in the mean-field infinite-volume limit}}, Annales
  Henri Poincar\'e,  (2014), pp.~1--31.

\bibitem{DiaFre-80}
{\sc P.~Diaconis and D.~Freedman}, {\em Finite exchangeable sequences}, Ann.
  Probab., 8 (1980), pp.~745--764.

\bibitem{DonVar-83}
{\sc M.~D. Donsker and S.~R.~S. Varadhan}, {\em Asymptotics for the polaron},
  Comm. Pure Appl. Math., 36 (1983), pp.~505--528.

\bibitem{Dyson-57}
{\sc F.~J. Dyson}, {\em Ground-state energy of a hard-sphere gas}, Phys. Rev.,
  106 (1957), pp.~20--26.

\bibitem{Dyson-62a}
\leavevmode\vrule height 2pt depth -1.6pt width 23pt, {\em Statistical theory
  of the energy levels of complex systems. {I}}, J. Math. Phys., 3 (1962),
  pp.~140--156.

\bibitem{Dyson-62b}
\leavevmode\vrule height 2pt depth -1.6pt width 23pt, {\em Statistical theory
  of the energy levels of complex systems. {II}}, J. Math. Phys., 3 (1962),
  pp.~157--165.

\bibitem{Dyson-62c}
\leavevmode\vrule height 2pt depth -1.6pt width 23pt, {\em Statistical theory
  of the energy levels of complex systems. {III}}, J. Math. Phys., 3 (1962),
  pp.~166--175.

\bibitem{DysLen-67}
{\sc F.~J. Dyson and A.~Lenard}, {\em Stability of matter. {I}}, J. Math.
  Phys., 8 (1967), pp.~423--434.

\bibitem{DysLen-68}
\leavevmode\vrule height 2pt depth -1.6pt width 23pt, {\em Stability of matter.
  {II}}, J. Math. Phys., 9 (1968), pp.~1538--1545.

\bibitem{Einstein-24}
{\sc A.~Einstein}, {\em Quantentheorie des einatomigen idealen {G}ases},
  Sitzber. Kgl. Preuss. Akad. Wiss., 1924, pp.~261--267.

\bibitem{ErdSchYau-07}
{\sc L.~Erd{\"{o}}s, B.~Schlein, and H.-T. Yau}, {\em Derivation of the cubic
  non-linear {S}chr\"odinger equation from quantum dynamics of many-body
  systems}, Invent. Math., 167 (2007), pp.~515--614.

\bibitem{ErdSchYau-09}
{\sc L.~Erd{\H{o}}s, B.~Schlein, and H.-T. Yau}, {\em Rigorous derivation of
  the {G}ross-{P}itaevskii equation with a large interaction potential}, J.
  Amer. Math. Soc., 22 (2009), pp.~1099--1156.

\bibitem{ErdSchYau-10}
{\sc L.~Erd{\"{o}}s, B.~Schlein, and H.-T. Yau}, {\em Derivation of the
  {G}ross-{P}itaevskii equation for the dynamics of {B}ose-{E}instein
  condensate}, Ann. of Math., 172 (2010), pp.~291--370.

\bibitem{ErdYau-01}
{\sc L.~Erd{\"o}s and H.-T. Yau}, {\em Derivation of the nonlinear
  {S}chr\"odinger equation from a many body {C}oulomb system}, Adv. Theor.
  Math. Phys., 5 (2001), pp.~1169--1205.

\bibitem{FanSpoVer-80}
{\sc M.~Fannes, H.~Spohn, and A.~Verbeure}, {\em Equilibrium states for
  mean-field models}, J. Math. Phys., 21 (1980), pp.~355--358.

\bibitem{FanVan-06}
{\sc M.~Fannes and C.~Vandenplas}, {\em Finite size mean-field models}, J.
  Phys. A, 39 (2006), pp.~13843--13860.

\bibitem{Fetter-09}
{\sc A.~Fetter}, {\em Rotating trapped {B}ose-{E}instein condensates}, Rev.
  Mod. Phys., 81 (2009), p.~647.

\bibitem{FetHanLau-89}
{\sc A.~L. Fetter, C.~B. Hanna, and R.~B. Laughlin}, {\em Random-phase
  approximation in the fractional-statistics gas}, Phys. Rev. B, 39 (1989),
  pp.~9679--9681.

\bibitem{Forrester-10}
{\sc P.~J. Forrester}, {\em Log-gases and random matrices}, vol.~34 of London
  Mathematical Society Monographs Series, Princeton University Press,
  Princeton, NJ, 2010.

\bibitem{Forte-91}
{\sc S.~Forte}, {\em {Berrys's phase, fractional statistics and the Laughlin
  wave function}}, Mod. Phys. Lett. A, 6 (1991), pp.~3152--3162.

\bibitem{Forte-92}
\leavevmode\vrule height 2pt depth -1.6pt width 23pt, {\em Quantum mechanics
  and field theory with fractional spin and statistics}, Rev. Mod. Phys., 64
  (1992), pp.~193--236.

\bibitem{FouHel-10}
{\sc S.~Fournais and B.~Helffer}, {\em Spectral Methods in Surface
  Superconductivity}, Progress in Nonlinear Differential Equations and their
  Applications, 77, Birkh\"auser Boston, Inc., Boston, MA, 2010.

\bibitem{FouHelPer-11}
{\sc S.~Fournais, B.~Helffer, and M.~Persson}, {\em {Superconductivity between
  $ H_{c_2} $ and $ H_{c_3} $}}, J. Spectr. Theory, 1 (2011), pp.~273--298.

\bibitem{FouKac-11}
{\sc S.~Fournais and A.~Kachmar}, {\em Nucleation of bulk superconductivity
  close to critical magnetic field}, Advances in Mathematics, 226 (2011),
  pp.~1213 -- 1258.

\bibitem{Frank-14}
{\sc R.~L. Frank}, {\em Ground states of semi-linear {PDE}s}.
\newblock Lecture notes, 2014.

\bibitem{FraHaiSeiSol-12}
{\sc R.~L. {Frank}, C.~{Hainzl}, R.~{Seiringer}, and J.~P. {Solovej}}, {\em
  {Microscopic Derivation of Ginzburg-Landau Theory}}, J. Amer. Math. Soc., 25
  (2012), pp.~667--713.

\bibitem{FraLieSei-12}
{\sc R.~L. {Frank}, E.~H. {Lieb}, and R.~{Seiringer}}, {\em {Binding of
  Polarons and Atoms at Threshold}}, Commun. Math. Phys., 313 (2012),
  pp.~405--424.

\bibitem{FraLieSeiTho-10}
{\sc R.~L. Frank, E.~H. Lieb, R.~Seiringer, and L.~E. Thomas}, {\em Bi-polaron
  and ${N}$-polaron binding energies}, Phys. Rev. Lett., 104 (2010), p.~210402.

\bibitem{FraLieSeiTho-11}
\leavevmode\vrule height 2pt depth -1.6pt width 23pt, {\em Stability and
  absence of binding for multi-polaron systems}, Publ. Math. Inst. Hautes
  \'Etudes Sci., 113 (2011), pp.~39--67.

\bibitem{Froehlich-90}
{\sc J.~Fr\"ohlich}, {\em Quantum statistics and locality}, in Proceedings of
  the Gibbs Symposium, Les Houches - Ecole d'Ete de Physique Theorique, 1990,
  pp.~89--142.

\bibitem{FroKnoSchSoh-16}
{\sc J.~Fr\"ohlich, A.~Knowles, B.~Schlein, and V.~Sohinger}, {\em Gibbs
  measures of nonlinear {S}chr\"odinger equations as limits of quantum
  many-body states in dimensions $d\leq 3$}.
\newblock arXiv:1605.07095, 2016.

\bibitem{Ginibre-65}
{\sc J.~Ginibre}, {\em Statistical ensembles of complex, quaternion, and real
  matrices}, J. Math. Phys., 6 (1965), pp.~440--449.

\bibitem{GinLan-50}
{\sc V.~Ginzburg and L.~Landau}, {\em On the theory of superconductivity}, Zh.
  Eksp. Teor. Fiz., 20 (1950), pp.~1064--82.

\bibitem{Girardeau-65}
{\sc M.~Girardeau}, {\em Permutation symmetry of many-particle wave functions},
  Phys. Rev., 129 (1965), pp.~500--508.

\bibitem{Girvin-04}
{\sc S.~Girvin}, {\em Introduction to the fractional quantum {H}all effect},
  S\'eminaire Poincar\'e, 2 (2004), pp.~54--74.

\bibitem{GliJaf-87}
{\sc J.~Glimm and A.~Jaffe}, {\em Quantum Physics: A Functional Integral Point
  of View}, Springer-Verlag, 1987.

\bibitem{Goerbig-09}
{\sc M.~O. Goerbig}, {\em Quantum {H}all effects}.
\newblock arXiv:0909.1998, 2009.

\bibitem{GolMedSha-81}
{\sc G.~Goldin, R.~Menikoff, and D.~Sharp}, {\em {Representations of a local
  current algebra in nonsimply connected space and the Aharonov-Bohm effect}},
  J. Math. Phys., 22 (1981), p.~1664.

\bibitem{Golse-13}
{\sc F.~{Golse}}, {\em {On the Dynamics of Large Particle Systems in the Mean
  Field Limit}}, ArXiv e-prints 1301.5494,  (2013).
\newblock Lecture notes for a course at the NDNS+ Applied Dynamical Systems
  Summer School "Macroscopic and large scale phenomena", Universiteit Twente,
  Enschede (The Netherlands).

\bibitem{Gorkov-59}
{\sc L.~Gork'ov}, {\em {Microscopic derivation of the Ginzburg-Landau equations
  in the theory of superconductivity}}, Soviet Phys. JETP, 9 (1959),
  pp.~1364--1367.

\bibitem{GotSch-09}
{\sc A.~Gottlieb and T.~Schumm}, {\em Quantum noise thermometry for bosonic
  {J}osephson junctions in the mean-field regime}, Phys. Rev. A, 79 (2009),
  p.~063601.

\bibitem{Gottlieb-05}
{\sc A.~D. Gottlieb}, {\em Examples of bosonic de {F}inetti states over finite
  dimensional {H}ilbert spaces}, J. Stat. Phys., 121 (2005), pp.~497--509.

\bibitem{GreSei-13}
{\sc P.~Grech and R.~Seiringer}, {\em The excitation spectrum for weakly
  interacting bosons in a trap}, Comm. Math. Phys., 322 (2013), pp.~559--591.

\bibitem{GriMol-10}
{\sc M.~Griesemer and J.~S. M{\o}ller}, {\em Bounds on the minimal energy of
  translation invariant $n$-polaron systems}, Commun. Math. Phys., 297 (2010),
  pp.~283--297.

\bibitem{GroSimSte-06}
{\sc E.~Grosfeld, S.~H. Simon, and A.~Stern}, {\em Switching noise as a probe
  of statistics in the fractional quantum hall effect}, Phys. Rev. Lett., 96
  (2006), p.~226803.

\bibitem{GuoSei-13}
{\sc Y.~Guo and R.~Seiringer}, {\em Symmetry breaking and collapse in
  {B}ose-{E}instein condensates with attractive interactions}, Lett. Math.
  Phys., 104 (2014), pp.~141--156.

\bibitem{HaiLewSer-05a}
{\sc C.~Hainzl, M.~Lewin, and {\'E}.~S{\'e}r{\'e}}, {\em Existence of a stable
  polarized vacuum in the {B}ogoliubov-{D}irac-{F}ock approximation}, Commun.
  Math. Phys., 257 (2005), pp.~515--562.

\bibitem{HaiLewSer-05b}
\leavevmode\vrule height 2pt depth -1.6pt width 23pt, {\em Self-consistent
  solution for the polarized vacuum in a no-photon {QED} model}, J. Phys. A, 38
  (2005), pp.~4483--4499.

\bibitem{HaiLewSer-09}
\leavevmode\vrule height 2pt depth -1.6pt width 23pt, {\em Existence of atoms
  and molecules in the mean-field approximation of no-photon quantum
  electrodynamics}, Arch. Ration. Mech. Anal., 192 (2009), pp.~453--499.

\bibitem{HaiLewSerSol-07}
{\sc C.~Hainzl, M.~Lewin, {\'E}.~S{\'e}r{\'e}, and J.~P. Solovej}, {\em A
  minimization method for relativistic electrons in a mean-field approximation
  of quantum electrodynamics}, Phys. Rev. A, 76 (2007), p.~052104.

\bibitem{Haldane-83}
{\sc F.~D.~M. Haldane}, {\em Fractional quantization of the {H}all effect: A
  hierarchy of incompressible quantum fluid states}, Phys. Rev. Lett., 51
  (1983), pp.~605--608.

\bibitem{Halperin-84}
{\sc B.~I. Halperin}, {\em Statistics of quasiparticles and the hierarchy of
  fractional quantized {H}all states}, Phys. Rev. Lett., 52 (1984),
  pp.~1583--1586.

\bibitem{Harrow-13}
{\sc A.~Harrow}, {\em The church of the symmetric subspace}, preprint arXiv,
  (2013).

\bibitem{HauMis-14}
{\sc M.~Hauray and S.~Mischler}, {\em On {K}ac's chaos and related problems},
  J. Func. Anal., 266 (2014), pp.~6055--6157.

\bibitem{Hess-89}
{\sc H.~F. Hess, R.~B. Robinson, R.~C. Dynes, J.~M.~V. Jr., and J.~V.
  Waszczak}, {\em Scanning-tunneling-microscope observation of the {Abrikosov}
  flux lattice and the density of states near and inside a fluxoid}, Phys. Rev.
  Lett., 62 (1989), p.~214.

\bibitem{HewSav-55}
{\sc E.~Hewitt and L.~J. Savage}, {\em Symmetric measures on {C}artesian
  products}, Trans. Amer. Math. Soc., 80 (1955), pp.~470--501.

\bibitem{Hof-77}
{\sc M.~{Hoffmann-Ostenhof} and T.~{Hoffmann-Ostenhof}}, {\em Schr{\"o}dinger
  inequalities and asymptotic behavior of the electron density of atoms and
  molecules}, Phys. Rev. A, 16 (1977), pp.~1782--1785.

\bibitem{HofLapTid-08}
{\sc M.~Hoffmann-Ostenhof, T.~Hoffmann-Ostenhof, A.~Laptev, and J.~Tidblom},
  {\em {Many-particle Hardy Inequalities}}, J. London Math. Soc., 77 (2008),
  pp.~99--114.

\bibitem{Hogreve-11}
{\sc H.~Hogreve}, {\em A remark on the ground state energy of bosonic atoms},
  J. Stat. Phys., 144 (2011), pp.~904--908.

\bibitem{HudMoo-75}
{\sc R.~L. Hudson and G.~R. Moody}, {\em Locally normal symmetric states and an
  analogue of de {F}inetti's theorem}, Z. Wahrscheinlichkeitstheor. und Verw.
  Gebiete, 33 (1975/76), pp.~343--351.

\bibitem{Hun-66}
{\sc W.~Hunziker}, {\em On the spectra of {S}chr{\"o}dinger multiparticle
  {H}amiltonians}, Helv. Phys. Acta, 39 (1966), pp.~451--462.

\bibitem{IenLec-92}
{\sc R.~Iengo and K.~Lechner}, {\em Anyon quantum mechanics and
  {C}hern-{S}imons theory}, Phys. Rep., 213 (1992), pp.~179--269.

\bibitem{IgnMil-06b}
{\sc R.~Ignat and V.~Millot}, {\em {Energy expansion and vortex location for a
  two-dimensional rotating Bose-Einstein condensate}}, Rev. Math. Phys., 18
  (2006), pp.~119--162.

\bibitem{IgnMil-06}
\leavevmode\vrule height 2pt depth -1.6pt width 23pt, {\em {The critical
  velocity for vortex existence in a two-dimensional rotating Bose-Einstein
  condensate}}, J. Func. Anal., 233 (2006), pp.~260--306.

\bibitem{Jain-07}
{\sc J.~K. Jain}, {\em {Composite fermions}}, Cambridge University Press, 2007.

\bibitem{JulGotMarPol-13}
{\sc B.~Juli\'a-D\'iaz, A.~Gottlieb, J.~Martorell, and A.~Polls}, {\em Quantum
  and thermal fluctuations in bosonic {J}osephson junctions}, Phys. Rev. A, 88
  (2013), p.~033601.

\bibitem{Kachmar-14}
{\sc A.~Kachmar}, {\em The {Ginzburg-Landau} order parameter neat the second
  critical field}, SIAM Journal on Mathematical Analysis, 46 (2014), p.~572.

\bibitem{Kallenberg-05}
{\sc O.~Kallenberg}, {\em Probabilistic symmetries and invariance principles},
  Probability and its Applications, Springer, 2005.

\bibitem{Ketterle-nobel}
{\sc W.~Ketterle}, {\em {When atoms behave as waves: Bose-Einstein condensation
  and the atom laser}}, Rev. Mod. Phys., 74 (2002), pp.~1131--1151.

\bibitem{Khare-05}
{\sc A.~Khare}, {\em {Fractional Statistics and Quantum Theory}}, World
  Scientific, Singapore, 2nd~ed., 2005.

\bibitem{Kiessling-89}
{\sc M.~K.-H. Kiessling}, {\em On the equilibrium statistical mechanics of
  isothermal classical self-gravitating matter}, Jour. Stat. Phys., 55 (1989),
  pp.~203--257.

\bibitem{Kiessling-93}
\leavevmode\vrule height 2pt depth -1.6pt width 23pt, {\em Statistical
  mechanics of classical particles with logarithmic interactions}, Comm. Pure.
  Appl. Math., 46 (1993), pp.~27--56.

\bibitem{KieSpo-99}
{\sc M.~K.-H. Kiessling and H.~Spohn}, {\em A note on the eigenvalue density of
  random matrices}, Comm. Math. Phys., 199 (1999), pp.~683--695.

\bibitem{KilVis-08}
{\sc R.~Killip and M.~Visan}, {\em Nonlinear {S}chr{\"o}dinger equations at
  critical regularity}.
\newblock Lecture notes for the summer school of Clay Mathematics Institute,
  2008.

\bibitem{KimLawVisSmiFra-05}
{\sc E.-A. Kim, M.~Lawler, S.~Vishveshwara, and E.~Fradkin}, {\em Signatures of
  fractional statistics in noise experiments in quantum {H}all fluids}, Phys.
  Rev. Lett., 95 (2005), p.~176402.

\bibitem{KimLawVisSmiFra-06}
\leavevmode\vrule height 2pt depth -1.6pt width 23pt, {\em Measuring fractional
  charge and statistics in fractional quantum {H}all fluids through noise
  experiments}, Phys. Rev. B, 74 (2006), p.~155324.

\bibitem{KjoLei-97}
{\sc H.~Kj\o{}nsberg and J.~Leinaas}, {\em On the anyon description of the
  {L}aughlin hole states}, Int. J. Mod. Phys. A, 12 (1997), pp.~1975--2002.

\bibitem{KjoLei-99}
\leavevmode\vrule height 2pt depth -1.6pt width 23pt, {\em Charge and
  statistics of quantum {H}all quasi-particles - a numerical study of mean
  values and fluctuations}, Nuclear Physics B, 559 (1999), p.~705?742.

\bibitem{KjoMyr-99}
{\sc H.~Kj\o{}nsberg and J.~Myrheim}, {\em Numerical study of charge and
  statistics of {L}aughlin quasi-particles}, Int. J. Mod. Phys. A, 14 (1999),
  p.~537.

\bibitem{Klainerman-00}
{\sc S.~Klainerman}, {\em {PDEs as a unified subject}}, {Geometric and
  Functional Analysis, special volume GAFA 2000},  (2000), pp.~1--37.

\bibitem{Knowles-thesis}
{\sc A.~Knowles}, {\em Limiting dynamics in large quantum systems.}
\newblock Doctoral thesis, ETH Z\"urich.

\bibitem{KonRen-05}
{\sc R.~K\"{o}nig and R.~Renner}, {\em A de {F}inetti representation for finite
  symmetric quantum states}, J. Math. Phys., 46 (2005), p.~122108.

\bibitem{Kwong-89}
{\sc M.~K. Kwong}, {\em Uniqueness of positive solutions of {$\Delta
  u-u+u^p=0$} in {${\bf R}^n$}}, Arch. Rational Mech. Anal., 105 (1989),
  pp.~243--266.

\bibitem{LarLun-16}
{\sc S.~Larson and D.~Lundholm}, {\em {Exclusion bounds for extended anyons}}.
\newblock arXiv, 2016.

\bibitem{Laughlin-83}
{\sc R.~B. Laughlin}, {\em Anomalous quantum {H}all effect: An incompressible
  quantum fluid with fractionally charged excitations}, Phys. Rev. Lett., 50
  (1983), pp.~1395--1398.

\bibitem{Laughlin-87}
\leavevmode\vrule height 2pt depth -1.6pt width 23pt, {\em Elementary theory :
  the incompressible quantum fluid}, in The quantum {H}all effect, R.~E. Prange
  and S.~E. Girvin, eds., Springer, Heidelberg, 1987.

\bibitem{Laughlin-99}
\leavevmode\vrule height 2pt depth -1.6pt width 23pt, {\em Nobel lecture:
  Fractional quantization}, Rev. Mod. Phys., 71 (1999), pp.~863--874.

\bibitem{Spielman-12}
{\sc L.~J. LeBlanc, K.~Jim\'{e}nez-Garc\'{\i}a, R.~A. Williams, M.~C. Beeler,
  A.~R. Perry, W.~D. Phillips, and I.~B. Spielman}, {\em {Observation of a
  superfluid Hall effect}}, Proc. Nat. Acad. Sci., 109 (2012),
  pp.~10811--10814.

\bibitem{Leble-15b}
{\sc T.~Lebl\'e}, {\em Local microscopic behavior for {2D C}oulomb gases}.
\newblock arXiv:1510.01506, 2015.

\bibitem{Leble-15}
\leavevmode\vrule height 2pt depth -1.6pt width 23pt, {\em A uniqueness result
  for minimizers of the 1d log-gas renormalized energy}, Journal of Functional
  Analysis, 268 (2015), pp.~1649 -- 1677.

\bibitem{Leble-16}
\leavevmode\vrule height 2pt depth -1.6pt width 23pt, {\em {Logarithmic,
  Coulomb and Riesz Energy of Point Processes}}, Journal of Statistical
  Physics, 162 (2016), pp.~887 -- 923.

\bibitem{LebSer-15}
{\sc T.~Lebl\'e and S.~Serfaty}, {\em Large deviation principle for empirical
  fields of {L}og and {R}iesz gases}.
\newblock arXiv:1502.02970, 2015.

\bibitem{LebLie-69}
{\sc J.~L. Lebowitz and E.~H. Lieb}, {\em Existence of thermodynamics for real
  matter with {C}oulomb forces}, Phys. Rev. Lett., 22 (1969), pp.~631--634.

\bibitem{LebRosSpe-88}
{\sc J.~L. Lebowitz, H.~A. Rose, and E.~R. Speer}, {\em Statistical mechanics
  of the nonlinear {S}chr\"odinger equation}, J. Statist. Phys., 50 (1988),
  pp.~657--687.

\bibitem{LeiMyr-77}
{\sc J.~M. {Leinaas} and J.~{Myrheim}}, {\em {On the theory of identical
  particles}}, Nuovo Cimento B Serie, 37 (1977), pp.~1--23.

\bibitem{Lerda-92}
{\sc A.~Lerda}, {\em {Anyons}}, Springer-Verlag, Berlin--Heidelberg, 1992.

\bibitem{LevyLeblond-69}
{\sc J.-M. L\'evy-Leblond}, {\em Nonsaturation of gravitational forces}, J.
  Math. Phys., 10 (1969), pp.~806--812.

\bibitem{Lewin-ICMP}
{\sc M.~Lewin}, {\em {Mean-Field limit of Bose systems: rigorous results}},
  Preprint (2015) arXiv:1510.04407.

\bibitem{Lewin-11}
\leavevmode\vrule height 2pt depth -1.6pt width 23pt, {\em Geometric methods
  for nonlinear many-body quantum systems}, J. Funct. Anal., 260 (2011),
  pp.~3535--3595.

\bibitem{Lewin_proc-12}
\leavevmode\vrule height 2pt depth -1.6pt width 23pt, {\em A nonlinear
  variational problem in relativistic quantum mechanics}, Proceeding of the 6th
  European Congress of Mathematics, Krakow (Poland),  (2012).

\bibitem{LewNamRou-ICMP}
{\sc M.~Lewin, P.~Nam, and N.~Rougerie}, {\em {Bose gases at positive
  temperature and non-linar Gibbs measures}}, Preprint (2016) arXiv:1602.05166.

\bibitem{LewNamRou-14}
\leavevmode\vrule height 2pt depth -1.6pt width 23pt, {\em Derivation of
  {H}artree's theory for generic mean-field {B}ose systems}, Adv. Math., 254
  (2014), pp.~570--621.

\bibitem{LewNamRou-15}
\leavevmode\vrule height 2pt depth -1.6pt width 23pt, {\em {A note on 2D
  focusing many-boson systems}}.
\newblock arXiv:1509.09045, 2015.

\bibitem{LewNamRou-14b}
\leavevmode\vrule height 2pt depth -1.6pt width 23pt, {\em Remarks on the
  quantum de {F}inetti theorem for bosonic systems}, Appl. Math. Res. Express
  (AMRX), 2015 (2015), pp.~48--63.

\bibitem{LewNamRou-14d}
\leavevmode\vrule height 2pt depth -1.6pt width 23pt, {\em Derivation of
  nonlinear {G}ibbs measures from many-body quantum mechanics}, Journal de
  l'Ecole Polytechnique, 2 (2016), pp.~553--606.

\bibitem{LewNamRou-14c}
\leavevmode\vrule height 2pt depth -1.6pt width 23pt, {\em The mean-field
  approximation and the non-linear {S}chr\"odinger functional for trapped
  {B}ose gases}, Trans. Amer. Math. Soc, 368 (2016), pp.~6131--6157.

\bibitem{LewNamSch-14}
{\sc M.~Lewin, P.~T. Nam, and B.~Schlein}, {\em Fluctuations around {H}artree
  states in the mean-field regime}, Amer. J. Math., 137 (2015), pp.~1613--1650.

\bibitem{LewNamSerSol-13}
{\sc M.~Lewin, P.~T. Nam, S.~Serfaty, and J.~P. Solovej}, {\em Bogoliubov
  spectrum of interacting {B}ose gases}, Comm. Pure Appl. Math., 68 (2015),
  pp.~413--471.

\bibitem{LewRou-13}
{\sc M.~Lewin and N.~Rougerie}, {\em Derivation of {P}ekar's {P}olarons from a
  {M}icroscopic {M}odel of {Q}uantum {C}rystals}, SIAM J. Math. Anal., 45
  (2013), pp.~1267--1301.

\bibitem{LewRou-13b}
\leavevmode\vrule height 2pt depth -1.6pt width 23pt, {\em On the binding of
  polarons in a mean-field quantum crystal}, ESAIM Control Optim. Calc. Var.,
  19 (2013), pp.~629--656.

\bibitem{LewSei-09}
{\sc M.~Lewin and R.~Seiringer}, {\em Strongly correlated phases in rapidly
  rotating {B}ose gases}, J. Stat. Phys., 137 (2009), pp.~1040--1062.

\bibitem{Lieb-73b}
{\sc E.~H. Lieb}, {\em The classical limit of quantum spin systems}, Comm.
  Math. Phys., 31 (1973), pp.~327--340.

\bibitem{Lieb-76}
\leavevmode\vrule height 2pt depth -1.6pt width 23pt, {\em The stability of
  matter}, Rev. Mod. Phys., 48 (1976), pp.~553--569.

\bibitem{Lieb-77}
\leavevmode\vrule height 2pt depth -1.6pt width 23pt, {\em {Existence and
  uniqueness of the minimizing solution of Choquard's nonlinear equation}},
  Studies in Applied Mathematics, 57 (1977), pp.~93--105.

\bibitem{Lieb-80}
\leavevmode\vrule height 2pt depth -1.6pt width 23pt, {\em {The Number of Bound
  States of One- Body Schr\"odinger Operators and the Weyl Problem}},
  Proceedings of the Amer. Math. Soc. Symposia in Pure Math., 36 (1980),
  pp.~241--252.

\bibitem{LieLeb-72}
{\sc E.~H. Lieb and J.~L. Lebowitz}, {\em The constitution of matter:
  {E}xistence of thermodynamics for systems composed of electrons and nuclei},
  Advances in Math., 9 (1972), pp.~316--398.

\bibitem{LieLos-01}
{\sc E.~H. Lieb and M.~Loss}, {\em Analysis}, vol.~14 of Graduate Studies in
  Mathematics, American Mathematical Society, Providence, RI, 2nd~ed., 2001.

\bibitem{LieNar-76}
{\sc E.~H. Lieb and H.~Narnhofer}, {\em The thermodynamic limit for jellium},
  J. Stat. Phys., 14 (1976), pp.~465--465.

\bibitem{LieRus-73a}
{\sc E.~H. Lieb and M.~B. Ruskai}, {\em A fundamental property of
  quantum-mechanical entropy}, Phys. Rev. Lett., 30 (1973), pp.~434--436.

\bibitem{LieRus-73b}
\leavevmode\vrule height 2pt depth -1.6pt width 23pt, {\em Proof of the strong
  subadditivity of quantum-mechanical entropy}, J. Math. Phys., 14 (1973),
  pp.~1938--1941.
\newblock With an appendix by B. Simon.

\bibitem{LieSei-02}
{\sc E.~H. Lieb and R.~Seiringer}, {\em {Proof of Bose-Einstein Condensation
  for Dilute Trapped Gases}}, Phys. Rev. Lett., 88 (2002), p.~170409.

\bibitem{LieSei-06}
\leavevmode\vrule height 2pt depth -1.6pt width 23pt, {\em Derivation of the
  {G}ross-{P}itaevskii equation for rotating {B}ose gases}, Commun. Math.
  Phys., 264 (2006), pp.~505--537.

\bibitem{LieSei-09}
\leavevmode\vrule height 2pt depth -1.6pt width 23pt, {\em The {S}tability of
  {M}atter in {Q}uantum {M}echanics}, Cambridge Univ. Press, 2010.

\bibitem{LieSeiSol-05}
{\sc E.~H. Lieb, R.~Seiringer, and J.~P. Solovej}, {\em Ground-state energy of
  the low-density {F}ermi gas}, Phys. Rev. A, 71 (2005), p.~053605.

\bibitem{LieSeiSolYng-05}
{\sc E.~H. Lieb, R.~Seiringer, J.~P. Solovej, and J.~Yngvason}, {\em The
  mathematics of the {B}ose gas and its condensation}, Oberwolfach {S}eminars,
  Birkh{\"a}user, 2005.

\bibitem{LieSeiYng-00}
{\sc E.~H. Lieb, R.~Seiringer, and J.~Yngvason}, {\em Bosons in a trap: A
  rigorous derivation of the {G}ross-{P}itaevskii energy functional}, Phys.
  Rev. A, 61 (2000), p.~043602.

\bibitem{LieSeiYng-01}
\leavevmode\vrule height 2pt depth -1.6pt width 23pt, {\em A rigorous
  derivation of the {Gross-Pitaevskii} energy functional for a two-dimensional
  {Bo}se gas}, Comm. Math. Phys., 224 (2001), pp.~17--31.

\bibitem{LieSeiYng-02b}
\leavevmode\vrule height 2pt depth -1.6pt width 23pt, {\em Superfluidity in
  dilute trapped bose gases}, Phys. Rev. B, 66 (2002), p.~134529.

\bibitem{LieSeiYng-05}
\leavevmode\vrule height 2pt depth -1.6pt width 23pt, {\em {Justification of
  $c$-Number Substitutions in Bosonic Hamiltonians}}, Phys. Rev. Lett., 94
  (2005), p.~080401.

\bibitem{LieSeiYng-09}
\leavevmode\vrule height 2pt depth -1.6pt width 23pt, {\em {Y}rast line of a
  rapidly rotating {B}ose gas: {G}ross-{P}itaevskii regime}, Phys. Rev. A, 79
  (2009), p.~063626.

\bibitem{LieSim-77b}
{\sc E.~H. Lieb and B.~Simon}, {\em The {T}homas-{F}ermi theory of atoms,
  molecules and solids}, Adv. Math., 23 (1977), pp.~22--116.

\bibitem{LieSol-15}
{\sc E.~H. Lieb and J.~P. Solovej}, {\em {Proof of the Wehrl-type entropy
  conjecture for symmmetric $SU(N)$ coherent states}}.
\newblock arXiv:1506.05263, 2015.

\bibitem{LieThi-75}
{\sc E.~H. Lieb and W.~E. Thirring}, {\em Bound on kinetic energy of fermions
  which proves stability of matter}, Phys. Rev. Lett., 35 (1975), pp.~687--689.

\bibitem{LieThi-76}
\leavevmode\vrule height 2pt depth -1.6pt width 23pt, {\em Inequalities for the
  moments of the eigenvalues of the {S}chr{\"o}dinger {H}amiltonian and their
  relation to {S}obolev inequalities}, Studies in Mathematical Physics,
  Princeton University Press, 1976, pp.~269--303.

\bibitem{LieThi-84}
\leavevmode\vrule height 2pt depth -1.6pt width 23pt, {\em Gravitational
  collapse in quantum mechanics with relativistic kinetic energy}, Ann.
  Physics, 155 (1984), pp.~494--512.

\bibitem{LieTho-97}
{\sc E.~H. Lieb and L.~E. Thomas}, {\em Exact ground state energy of the
  strong-coupling polaron}, Commun. Math. Phys., 183 (1997), pp.~511--519.

\bibitem{LieYau-87}
{\sc E.~H. Lieb and H.-T. Yau}, {\em The {C}handrasekhar theory of stellar
  collapse as the limit of quantum mechanics}, Commun. Math. Phys., 112 (1987),
  pp.~147--174.

\bibitem{LieYng-98}
{\sc E.~H. Lieb and J.~Yngvason}, {\em Ground state energy of the low density
  {B}ose gas}, Phys. Rev. Lett., 80 (1998), pp.~2504--2507.

\bibitem{LieYng-01}
\leavevmode\vrule height 2pt depth -1.6pt width 23pt, {\em The ground state
  energy of a dilute two-dimensional {B}ose gas}, J. Stat. Phys., 103 (2001),
  p.~509.

\bibitem{Spielman-09}
{\sc Y.-J. Lin, R.~L. Compton, K.~Jim\'{e}nez-Garc\'{\i}a, J.~Porto, and I.~B.
  Spielman}, {\em {Synthetic magnetic fields for ultracold neutral atoms}},
  Nature, 462 (2009), pp.~628--632.

\bibitem{Lions-82a}
{\sc P.-L. Lions}, {\em Principe de concentration-compacit{\'e} en calcul des
  variations}, C. R. Acad. Sci. Paris S{\'e}r. I Math., 294 (1982),
  pp.~261--264.

\bibitem{Lions-84}
\leavevmode\vrule height 2pt depth -1.6pt width 23pt, {\em The
  concentration-compactness principle in the calculus of variations. {T}he
  locally compact case, {P}art {I}}, Ann. Inst. H. Poincar{\'e} Anal. Non
  Lin{\'e}aire, 1 (1984), pp.~109--149.

\bibitem{Lions-84b}
\leavevmode\vrule height 2pt depth -1.6pt width 23pt, {\em The
  concentration-compactness principle in the calculus of variations. {T}he
  locally compact case, {P}art {II}}, Ann. Inst. H. Poincar{\'e} Anal. Non
  Lin{\'e}aire, 1 (1984), pp.~223--283.

\bibitem{Lions-CdF}
\leavevmode\vrule height 2pt depth -1.6pt width 23pt, {\em Mean-field games and
  applications}.
\newblock Lectures at the {Coll\`ege de France}, unpublished, Nov 2007.

\bibitem{LorHirBet-11}
{\sc J.~L{\"o}rinczi, F.~Hiroshima, and V.~Betz}, {\em {Feynman-Kac-Type
  Theorems and Gibbs Measures on Path Space: With Applications to Rigorous
  Quantum Field Theory}}, Gruyter - de Gruyter Studies in Mathematics, Walter
  de Gruyter GmbH \& Company KG, 2011.

\bibitem{Lundholm-15}
{\sc D.~Lundholm}, {\em Geometric extensions of many-particle {H}ardy
  inequalities}, J. Phys. A: Math. Theor., 48 (2015), p.~175203.

\bibitem{LunRou-15}
{\sc D.~Lundholm and N.~Rougerie}, {\em {The average field approximation for
  almost bosonic extended anyons}}, J. Stat. Phys., 161 (2015), pp.~1236--1267.

\bibitem{LunRou-16}
\leavevmode\vrule height 2pt depth -1.6pt width 23pt, {\em {Emergence of
  fractional statistics for tracer particles in a Laughlin liquid}}, Phys. Rev.
  Lett., 116 (2016), p.~170401.

\bibitem{LunSol-13a}
{\sc D.~Lundholm and J.~P. Solovej}, {\em {Hardy and Lieb-Thirring inequalities
  for anyons}}, Comm. Math. Phys., 322 (2013), pp.~883--908.

\bibitem{LunSol-13b}
\leavevmode\vrule height 2pt depth -1.6pt width 23pt, {\em Local exclusion
  principle for identical particles obeying intermediate and fractional
  statistics}, Phys. Rev. A, 88 (2013), p.~062106.

\bibitem{LunSol-14}
\leavevmode\vrule height 2pt depth -1.6pt width 23pt, {\em {Local exclusion and
  Lieb-Thirring inequalities for intermediate and fractional statistics}}, Ann.
  Henri Poincar\'e, 15 (2014), pp.~1061--1107.

\bibitem{MadChevWohDal-00}
{\sc K.~W. Madison, F.~Chevy, W.~Wohlleben, and J.~Dalibard}, {\em Vortex
  formation in a stirred {B}ose-{E}instein condensate}, Phys. Rev. Lett., 84
  (2000), pp.~806--809.

\bibitem{Maeda-10}
{\sc M.~Maeda}, {\em On the symmetry of the ground states of nonlinear
  {S}chr\"odinger equations with potential}, Adv. Nonlinear Stud., 10 (2010),
  pp.~895--925.

\bibitem{YacobiEtal-04}
{\sc J.~Martin, S.~Ilani, B.~Verdene, J.~Smet, V.~Umansky, D.~Mahalu, D.~Schuh,
  G.~Abstreiter, and A.~Yacoby}, {\em Localization of fractionally charged
  quasi-particles}, Science, 305 (2004), pp.~980--983.

\bibitem{Mashkevich-96}
{\sc S.~Mashkevich}, {\em {Finite-size anyons and perturbation theory}}, Phys.
  Rev. D, 54 (1996), pp.~6537--6543.

\bibitem{Mehta-04}
{\sc M.~Mehta}, {\em Random matrices. Third edition}, Elsevier/Academic Press,
  2004.

\bibitem{MesSpo-82}
{\sc J.~Messer and H.~Spohn}, {\em Statistical mechanics of the isothermal
  {L}ane-{E}mden equation}, J. Statist. Phys., 29 (1982), pp.~561--578.

\bibitem{Mischler-11}
{\sc S.~Mischler}, {\em {Estimation quantitative et uniforme en temps de la
  propagation du chaos et introduction aux limites de champ moyen pour des
  syst\`emes de particules}}.
\newblock Cours de l'Ecole doctorale {EDDIMO}, 2011.

\bibitem{MisMou-13}
{\sc S.~Mischler and C.~Mouhot}, {\em {Kac's Program in Kinetic Theory}},
  Inventiones mathematicae, 193 (2013), pp.~1--147.

\bibitem{MiySpo-07}
{\sc T.~Miyao and H.~Spohn}, {\em The bipolaron in the strong coupling limit},
  Annales Henri Poincar{\'e}, 8 (2007), pp.~1333--1370.

\bibitem{MorFed-07}
{\sc A.~Morris and D.~Feder}, {\em Gaussian potentials facilitate access to
  quantum {H}all states in rotating {B}ose gases}, Phys. Rev. Lett., 99 (2007),
  p.~240401.

\bibitem{Myrheim-99}
{\sc J.~Myrheim}, {\em Anyons}, in Topological aspects of low dimensional
  systems, A.~Comtet, T.~Jolic{\oe}ur, S.~Ouvry, and F.~David, eds., vol.~69 of
  Les Houches - Ecole d'Ete de Physique Theorique, 1999, pp.~265--413.

\bibitem{NamNap-16}
{\sc P.~Nam and M.~Napi\'orkowski}, {\em {A note on the validity of Bogoliubov
  correction to mean-field dynamics}}, Preprint (2016) arXiv:1604.05240.

\bibitem{NamNap-15}
\leavevmode\vrule height 2pt depth -1.6pt width 23pt, {\em {Bogoliubov
  correction to the mean-field dynamics of interacting bosons}}, Preprint
  (2015) arXiv:1509.04631.

\bibitem{NamRouSei-15}
{\sc P.~T. Nam, N.~Rougerie, and R.~Seiringer}, {\em {Ground states of large
  Bose systems: The Gross-Pitaevskii limit revisited}}, Analysis and PDEs, 9
  (2016), pp.~459--485.

\bibitem{NamSei-14}
{\sc P.-T. Nam and R.~Seiringer}, {\em Collective excitations of {B}ose gases
  in the mean-field regime}, Arch. Rat. Mech. Anal, 215 (2015), pp.~381--417.

\bibitem{Ning-09}
{\sc Y.~Ning, C.~Song, Z.~Guan, X.~Ma, X.~Chen, J.~Jia, and Q.~Xue}, {\em
  Observation of surface superconductivity and direct vortex imaging of a {Pb}
  thin island with a scanning tunneling microscope}, Europhys. Lett., 85
  (2009), p.~27004.

\bibitem{OhyPet-93}
{\sc M.~Ohya and D.~Petz}, {\em Quantum entropy and its use}, Texts and
  Monographs in Physics, Springer-Verlag, Berlin, 1993.

\bibitem{Onsager-39}
{\sc L.~Onsager}, {\em Electrostatic interaction of molecules}, J. Phys. Chem.,
  43 (1939), pp.~189--196.

\bibitem{Ouvry-94}
{\sc S.~Ouvry}, {\em {$\delta$-perturbative interactions in the Aharonov-Bohm
  and anyons models}}, Phys. Rev. D, 50 (1994), pp.~5296--5299.

\bibitem{Ouvry-07}
\leavevmode\vrule height 2pt depth -1.6pt width 23pt, {\em {Anyons and lowest
  Landau level anyons}}, S\'eminaire Poincar\'e, 11 (2007), pp.~77--107.

\bibitem{Pan-02}
{\sc X.-B. Pan}, {\em Surface superconductivity in applied magnetic fields
  above {$H_{c 2}$}}, Communications in Mathematical Physics, 228 (2002),
  p.~327.

\bibitem{PapBer-01}
{\sc T.~Papenbrock and G.~F. Bertsch}, {\em Rotational spectra of weakly
  interacting {B}ose-{E}instein condensates}, Phys. Rev. A, 63 (2001),
  p.~023616.

\bibitem{ParFedCirZol-01}
{\sc B.~Paredes, P.~Fedichev, J.~I. Cirac, and P.~Zoller}, {\em
  {$\frac{1}{2}$-Anyons in Small Atomic Bose-Einstein Condensates}}, Phys. Rev.
  Lett., 87 (2001), p.~010402.

\bibitem{PetSmi-01}
{\sc C.~Pethick and H.~Smith}, {\em Bose-Einstein Condensation of Dilute
  Gases}, Cambridge University Press, 2001.

\bibitem{PetSer-14}
{\sc M.~{Petrache} and S.~{Serfaty}}, {\em {Next Order Asymptotics and
  Renormalized Energy for Riesz Interactions}}, J. Inst. Math. Jussieu,
  (2014).

\bibitem{PetRagVer-89}
{\sc D.~Petz, G.~A. Raggio, and A.~Verbeure}, {\em Asymptotics of
  {V}aradhan-type and the {G}ibbs variational principle}, Comm. Math. Phys.,
  121 (1989), pp.~271--282.

\bibitem{Pickl-15}
{\sc P.~Pickl}, {\em Derivation of the time dependent {G}ross {P}itaevskii
  equation with external fields}, Rev. Math. Phys., 27 (2015), p.~1550003.

\bibitem{PitStr-03}
{\sc L.~Pitaevskii and S.~Stringari}, {\em Bose-Einstein Condensation}, Oxford
  Science Publications, Oxford, 2003.

\bibitem{RagWer-89}
{\sc G.~A. Raggio and R.~F. Werner}, {\em Quantum statistical mechanics of
  general mean-field systems}, Helv. Phys. Acta, 62 (1989), pp.~980--1003.

\bibitem{RamETALKet-01}
{\sc C.~Raman, J.~R. Abo-Shaeer, J.~M. Vogels, K.~Xu, and W.~Ketterle}, {\em
  Vortex nucleation in a stirred {Bose-Einstein} condensate}, Phys. Rev. Lett.,
  87 (2001), p.~210402.

\bibitem{ReeSim2}
{\sc M.~Reed and B.~Simon}, {\em Methods of {M}odern {M}athematical {P}hysics.
  {II}. {F}ourier analysis, self-adjointness}, Academic Press, New York, 1975.

\bibitem{ReeSim4}
\leavevmode\vrule height 2pt depth -1.6pt width 23pt, {\em Methods of {M}odern
  {M}athematical {P}hysics. {IV}. {A}nalysis of operators}, Academic Press, New
  York, 1978.

\bibitem{RegChaJolJai-06}
{\sc N.~{Regnault}, C.~C. {Chang}, T.~{Jolicoeur}, and J.~K. {Jain}}, {\em
  {Composite fermion theory of rapidly rotating two-dimensional bosons}},
  Journal of Physics B, 39 (2006), pp.~S89--S99.

\bibitem{RegJol-04}
{\sc N.~{Regnault} and T.~{Jolicoeur}}, {\em Quantum {H}all fractions for
  spinless bosons}, Phys. Rev. B, 69 (2004), p.~235309.

\bibitem{RegJol-07}
\leavevmode\vrule height 2pt depth -1.6pt width 23pt, {\em {Parafermionic
  states in rotating {B}ose-{E}instein condensates}}, Phys. Rev. B, 76 (2007),
  p.~235324.

\bibitem{Renner-07}
{\sc R.~Renner}, {\em Symmetry of large physical systems implies independence
  of subsystems}, Nature Physics, 3 (2007), pp.~645--649.

\bibitem{Ricaud-14}
{\sc J.~Ricaud}, {\em {On uniqueness and non-degeneracy of anisotropic
  polarons}}.
\newblock arXiv:1412.1230, 2014.

\bibitem{RonRizDal-11}
{\sc M.~Roncaglia, M.~Rizzi, and J.~Dalibard}, {\em From rotating atomic rings
  to quantum {H}all states}, www.nature.com, Scientific Reports, 1 (2011).

\bibitem{Rosenbljum-76}
{\sc G.~Rosenbljum}, {\em {Distribution of the discrete spectrum of singular
  differential operators}}, Soviet Math. (Iz. VUZ), 164 (1976), pp.~75--86.

\bibitem{RotSer-14}
{\sc S.~Rota~Nodari and S.~Serfaty}, {\em Renormalized energy equidistribution
  and local charge balance in 2d {C}oulomb system}, Int. Math. Res. Not., 11
  (2015), pp.~3035--3093.

\bibitem{Rougerie-these}
{\sc N.~Rougerie}, {\em {La th\'eorie de Gross-Pitaevskii pour un condensat de
  Bose-Einstein en rotation : vortex et transitions de phase}}.
\newblock tel-00547404, 2010.
\newblock Phd thesis.

\bibitem{Rougerie-11}
\leavevmode\vrule height 2pt depth -1.6pt width 23pt, {\em {Annular
  Bose-Einstein Condensates in the Lowest Landau Level}}, App. Math. Res.
  Express, 2011 (2011), pp.~95--121.

\bibitem{Rougerie-11b}
\leavevmode\vrule height 2pt depth -1.6pt width 23pt, {\em {The giant vortex
  state for a Bose-Einstein condensate in a rotating anharmonic trap : extreme
  rotation regimes}}, Journal de Math\'ematiques pures et appliqu\'ees, 95
  (2011), pp.~296--347.

\bibitem{Rougerie-12}
\leavevmode\vrule height 2pt depth -1.6pt width 23pt, {\em Vortex rings in fast
  rotating {Bose-Einstein} condensates}, Archive for Rational Mechanics and
  Analysis, 203 (2012), pp.~69 -- 135.

\bibitem{Rougerie-xedp13}
\leavevmode\vrule height 2pt depth -1.6pt width 23pt, {\em Sur la
  mod\'elisation de l'interaction entre polarons et cristaux quantiques},
  S\'eminaire Laurent Schwartz,  (2012-2013).

\bibitem{Rougerie-LMU}
\leavevmode\vrule height 2pt depth -1.6pt width 23pt, {\em {De Finetti
  theorems, mean-field limits and Bose-Einstein condensation}}.
\newblock arXiv:1506.05263, 2014.
\newblock LMU lecture notes.

\bibitem{Rougerie-cdf}
\leavevmode\vrule height 2pt depth -1.6pt width 23pt, {\em Th\'eor\`emes de de
  {F}inetti, limites de champ moyen et condensation de {B}ose-{E}instein}.
\newblock arXiv:1409.1182, 2014.
\newblock Lecture notes for a cours Peccot.

\bibitem{Rougerie-xedp15}
\leavevmode\vrule height 2pt depth -1.6pt width 23pt, {\em {From bosonic
  grand-canonical ensembles to nonlinear {G}ibbs measures}}, 2014-2015.
\newblock S{\'e}minaire Laurent Schwartz.

\bibitem{Rougerie-INSMI}
\leavevmode\vrule height 2pt depth -1.6pt width 23pt, {\em {Estimations
  d'incompressibilit\'e pour la phase de Laughlin}}.
\newblock Lettre de l'INSMI, 2015.

\bibitem{RouSer-14}
{\sc N.~Rougerie and S.~Serfaty}, {\em Higher-dimensional {C}oulomb gases and
  renormalized energy functionals}, Communications on Pure and Applied
  Mathematics, 69 (2016), p.~519.

\bibitem{RouSerYng-13a}
{\sc N.~Rougerie, S.~Serfaty, and J.~Yngvason}, {\em Quantum {H}all states of
  bosons in rotating anharmonic traps}, Phys. Rev. A, 87 (2013), p.~023618.

\bibitem{RouSerYng-13b}
\leavevmode\vrule height 2pt depth -1.6pt width 23pt, {\em Quantum {H}all
  phases and plasma analogy in rotating trapped {B}ose gases}, J. Stat. Phys.,
  154 (2014), pp.~2--50.

\bibitem{RouYng-14}
{\sc N.~Rougerie and J.~Yngvason}, {\em Incompressibility estimates for the
  {L}aughlin phase}, Comm. Math. Phys., 336 (2015), pp.~1109--1140.

\bibitem{RouYng-15}
\leavevmode\vrule height 2pt depth -1.6pt width 23pt, {\em Incompressibility
  estimates for the {L}aughlin phase, part {II}}, Comm. Math. Phys., 339
  (2015), pp.~263--277.

\bibitem{SafDevMar-01}
{\sc I.~Safi, P.~Devillard, and T.~Martin}, {\em Partition noise and statistics
  in the fractional quantum {H}all effect}, Phys. Rev. Lett., 86 (2001),
  pp.~4628--4631.

\bibitem{JamGen-63}
{\sc D.~Saint-James and P.~de~Gennes}, {\em Onset of superconductivity in
  decreasing fields}, Phys. Lett., 7 (1963), pp.~306--308.

\bibitem{SamGlaJinEti-97}
{\sc L.~Saminadayar, D.~C. Glattli, Y.~Jin, and B.~Etienne}, {\em Observation
  of the $e/3$ fractionally charged {L}aughlin quasiparticle}, Phys. Rev.
  Lett., 79 (1997), pp.~2526--2529.

\bibitem{SanSer-07}
{\sc E.~Sandier and S.~Serfaty}, {\em Vortices in the magnetic
  {G}inzburg-{L}andau model}, Progress in Nonlinear Differential Equations and
  their Applications, 70, Birkh\"auser Boston, Inc., Boston, MA, 2007.

\bibitem{SanSer-12}
\leavevmode\vrule height 2pt depth -1.6pt width 23pt, {\em From the
  {G}inzburg-{L}andau model to vortex lattice problems}, Commun. Math. Phys.,
  313 (2012), pp.~635--743.

\bibitem{SanSer-14a}
\leavevmode\vrule height 2pt depth -1.6pt width 23pt, {\em {1D} log gases and
  the renormalized energy: crystallization at vanishing temperature}, Probab.
  Theory Related Fields, 162 (2014), pp.~1--52.

\bibitem{SanSer-14}
\leavevmode\vrule height 2pt depth -1.6pt width 23pt, {\em {2D Coulomb Gases
  and the Renormalized Energy}}, Annals of Proba., 43 (2014), pp.~2026--2083.

\bibitem{Schlein-08}
{\sc B.~Schlein}, {\em Derivation of effective evolution equations from
  microscopic quantum dynamics}, arXiv eprints,  (2008).
\newblock Lecture Notes for a course at ETH Zurich.

\bibitem{Seiringer-02}
{\sc R.~Seiringer}, {\em Gross-{P}itaevskii theory of the rotating {B}ose gas},
  Commun. Math. Phys., 229 (2002), pp.~491--509.

\bibitem{Seiringer-03}
\leavevmode\vrule height 2pt depth -1.6pt width 23pt, {\em Ground state
  asymptotics of a dilute, rotating gas}, J. Phys. A, 36 (2003),
  pp.~9755--9778.

\bibitem{Seiringer-ICMP10}
\leavevmode\vrule height 2pt depth -1.6pt width 23pt, {\em {Hot topics in cold
  gases}}, in Proceedings of the XVIth International Congress on Mathematical
  Physics, P.~Exner, ed., World Scientific, 2010, pp.~231--245.

\bibitem{Seiringer-11}
\leavevmode\vrule height 2pt depth -1.6pt width 23pt, {\em The excitation
  spectrum for weakly interacting bosons}, Commun. Math. Phys., 306 (2011),
  pp.~565--578.

\bibitem{SeiYngZag-12}
{\sc R.~Seiringer, J.~Yngvason, and V.~A. Zagrebnov}, {\em {Disordered
  Bose-Einstein condensates with interaction in one dimension}}, J. Stat.
  Mech., 2012 (2012), p.~P11007.

\bibitem{Sen-91}
{\sc D.~Sen}, {\em {Quantum and statistical mechanics of anyons}}, Nuclear
  Phys. B, 630 (1991), pp.~397--408.

\bibitem{SenChi-92}
{\sc D.~Sen and R.~Chitra}, {\em {Anyons as perturbed bosons}}, Phys. Rev. B,
  45 (1992), pp.~881--894.

\bibitem{Serfaty-14}
{\sc S.~Serfaty}, {\em Ginzburg-{L}andau vortices, {C}oulomb gases, and
  renormalized energies}, J. Stat. Phys., 154 (2014), pp.~660--680.

\bibitem{Serfaty-15}
\leavevmode\vrule height 2pt depth -1.6pt width 23pt, {\em {Coulomb Gases and
  Ginzburg-Landau Vortices}}, Zurich Lectures in Advanced Mathematics, Euro.
  Math. Soc., 2015.

\bibitem{SheRad-04}
{\sc D.~E. Sheehy and L.~Radzihovsky}, {\em Vortex lattice inhomogeneity in
  spatially inhomogeneous superfluids}, Phys. Rev. A, 70 (2004), p.~051602.

\bibitem{SheRad-04b}
\leavevmode\vrule height 2pt depth -1.6pt width 23pt, {\em Vortices in
  spatially inhomogeneous superfluids}, Phys. Rev. A, 70 (2004), p.~063620.

\bibitem{Simon-74}
{\sc B.~Simon}, {\em The {$P(\phi )_{2}$} {E}uclidean (quantum) field theory},
  Princeton University Press, Princeton, N.J., 1974.
\newblock Princeton Series in Physics.

\bibitem{Simon-80}
\leavevmode\vrule height 2pt depth -1.6pt width 23pt, {\em The classical limit
  of quantum partition functions}, Comm. Math. Phys., 71 (1980), pp.~247--276.

\bibitem{Solovej-90}
{\sc J.~P. Solovej}, {\em Asymptotics for bosonic atoms}, Lett. Math. Phys., 20
  (1990), pp.~165--172.

\bibitem{Souriau-67}
{\sc J.~M. Souriau}, {\em {Quantification g\'eom\'etrique. Applications}}, Ann.
  Inst. Henri Poincar\'e, section A, 6 (1967), pp.~311--341.

\bibitem{Souriau-70}
\leavevmode\vrule height 2pt depth -1.6pt width 23pt, {\em Structure des
  syst\`emes dynamiques}, Maîtrises de mathématiques, Dunod, Paris, 1970.

\bibitem{Stormer-69}
{\sc E.~St{\o}rmer}, {\em Symmetric states of infinite tensor products of
  {$C^{\ast} $}-algebras}, J. Functional Analysis, 3 (1969), pp.~48--68.

\bibitem{StoTsuGos-99}
{\sc H.~St\"{o}rmer, D.~Tsui, and A.~Gossard}, {\em The fractional quantum
  {H}all effect}, Rev. Mod. Phys., 71 (1999), pp.~S298--S305.

\bibitem{Strongin-64}
{\sc M.~Strongin, A.~Paskin, D.~Schweitzer, O.~Kammerer, and P.~Craig}, {\em
  Observation of surface superconductivity and direct vortex imaging of a pb
  thin island with a scanning tunneling microscope}, Phys. Rev. Lett., 12
  (1964), pp.~442--444.

\bibitem{Summers-12}
{\sc S.~J. {Summers}}, {\em {A Perspective on Constructive Quantum Field
  Theory}}.
\newblock arXiv:1203.3991, 2012.

\bibitem{Suto-03}
{\sc A.~S{\"u}t{\H{o}}}, {\em Thermodynamic limit and proof of condensation for
  trapped bosons}, J. Statist. Phys., 112 (2003), pp.~375--396.

\bibitem{Tao-06}
{\sc T.~Tao}, {\em Nonlinear dispersive equations}, vol.~106 of CBMS Regional
  Conference Series in Mathematics, Published for the Conference Board of the
  Mathematical Sciences, Washington, DC, 2006.
\newblock Local and global analysis.

\bibitem{ThoTzv-10}
{\sc L.~Thomann and N.~Tzvetkov}, {\em Gibbs measure for the periodic
  derivative nonlinear {S}chr\"odinger equation}, Nonlinearity, 23 (2010),
  p.~2771.

\bibitem{Trugenberger-92b}
{\sc C.~Trugenberger}, {\em {Ground state and collective excitations of
  extended anyons}}, Phys. Lett. B, 288 (1992), pp.~121--128.

\bibitem{Trugenberger-92}
\leavevmode\vrule height 2pt depth -1.6pt width 23pt, {\em {The anyon fluid in
  the Bogoliubov approximation}}, Phys. Rev. D, 45 (1992), pp.~3807--3817.

\bibitem{Tzvetkov-08}
{\sc N.~Tzvetkov}, {\em Invariant measures for the defocusing nonlinear
  {S}chr\"odinger equation}, Ann. Inst. Fourier (Grenoble), 58 (2008),
  pp.~2543--2604.

\bibitem{VanWinter-64}
{\sc C.~{Van Winter}}, {\em Theory of finite systems of particles. {I}. {T}he
  {G}reen function}, Mat.-Fys. Skr. Danske Vid. Selsk., 2 (1964).

\bibitem{VelWig-73}
{\sc G.~Velo and A.~Wightman}, eds., {\em {Constructive quantum field theory:
  The 1973 Ettore Majorana international school of mathematical physics}},
  Lecture notes in physics, Springer-Verlag, 1973.

\bibitem{Verbeure-11}
{\sc A.~Verbeure}, {\em Many-bosons systems, half a century later}, Theoretical
  and Mathematical Physics, Springer, 2011.

\bibitem{Viefers-08}
{\sc S.~Viefers}, {\em Quantum {H}all physics in rotating {B}ose-{E}instein
  condensates}, J. Phys. C, 20 (2008), p.~123202.

\bibitem{VieHanRei-00}
{\sc S.~Viefers, T.~H. Hansson, and S.~M. Reimann}, {\em Bose condensates at
  high angular momenta}, Phys. Rev. A, 62 (2000), p.~053604.

\bibitem{Wehrl-78}
{\sc A.~Wehrl}, {\em General properties of entropy}, Rev. Modern Phys., 50
  (1978), pp.~221--260.

\bibitem{Weinstein-83}
{\sc M.~I. Weinstein}, {\em Nonlinear {S}chr{\"o}dinger equations and sharp
  interpolation estimates}, Comm. Math. Phys., 87 (1983), pp.~567--576.

\bibitem{Werner-92}
{\sc R.~F. Werner}, {\em Large deviations and mean-field quantum systems}, in
  Quantum probability \& related topics, QP-PQ, VII, World Sci. Publ., River
  Edge, NJ, 1992, pp.~349--381.

\bibitem{Westerberg-93}
{\sc E.~Westerberg}, {\em Mean-field approximation for anyons in a magnetic
  field}, Int. J. Mod. Phys. B, 7 (1993), pp.~2177--2199.

\bibitem{Wigner-34}
{\sc E.~P. Wigner}, {\em On the interaction of electrons in metals}, Phys.
  Rev., 46 (1934), pp.~1002--1011.

\bibitem{Wigner-55}
\leavevmode\vrule height 2pt depth -1.6pt width 23pt, {\em Characteristic
  vectors of bordered matrices with infinite dimensions}, Ann. of Math., 62
  (1955), pp.~pp. 548--564.

\bibitem{Wigner-67}
\leavevmode\vrule height 2pt depth -1.6pt width 23pt, {\em Random matrices in
  physics}, SIAM Review, 9 (1967), pp.~1--23.

\bibitem{Wilczek-82a}
{\sc F.~Wilczek}, {\em {Magnetic flux, angular momentum, and statistics}},
  Phys. Rev. Lett., 48 (1982), p.~1144.

\bibitem{Wilczek-82b}
\leavevmode\vrule height 2pt depth -1.6pt width 23pt, {\em {Quantum mechanics
  of fractional-spin particles}}, Phys. Rev. Lett., 49 (1982), pp.~957--959.

\bibitem{Wilczek-90}
\leavevmode\vrule height 2pt depth -1.6pt width 23pt, {\em {Fractional
  Statistics and Anyon Superconductivity}}, World Scientific, Singapore, 1990.

\bibitem{Wu-84}
{\sc Y.~Wu}, {\em {General theory for quantum statistics in two dimensions}},
  Phys. Rev. Lett., 52 (1984), p.~2103.

\bibitem{Wu-91}
\leavevmode\vrule height 2pt depth -1.6pt width 23pt, {\em {Braid groups,
  anyons and gauge invariance}}, Int. J. Mod. Phys. B, 5 (1991), p.~1649.

\bibitem{Yngvason-10}
{\sc J.~{Yngvason}}, {\em {The interacting Bose gas: A continuing challenge}},
  Phys. Particles Nuclei, 41 (2010), pp.~880--884.

\bibitem{ZhaSreGemJai-14}
{\sc Y.~Zhang, G.~J. Sreejith, N.~D. Gemelke, and J.~K. Jain}, {\em {Fractional
  angular momentum in cold atom systems}}, Phys. Rev. Lett., 113 (2014),
  p.~160404.

\bibitem{ZhaSreJai-15}
{\sc Y.~Zhang, G.~J. Sreejith, and J.~K. Jain}, {\em {Creating and manipulating
  non-Abelian anyons in cold atom systems using auxiliary bosons}}, Phys. Rev.
  B., 92 (2015), p.~075116.

\bibitem{Zhislin-71}
{\sc G.~M. Zhislin}, {\em On the finiteness of the discrete spectrum of the
  energy operator of negative atomic and molecular ions}, Teoret. Mat. Fiz., 21
  (1971), pp.~332--341.

\end{thebibliography}

\end{document}